\documentclass[12pt]{article}
\usepackage[margin=1in]{geometry}

\usepackage[utf8]{inputenc} 
\usepackage[T1]{fontenc}    
\usepackage{algorithmicx,amsmath,amssymb,amsthm,bbm,bm,booktabs,breqn,caption,enumerate,enumitem,epsf,graphicx,mathrsfs,mathtools,nicefrac,psfrag,rotating,setspace,subfigure,subfiles,txfonts,units} 
\usepackage[singlelinecheck=false]{caption}

\usepackage[backend=biber,style=numeric,citestyle=numeric-comp,sorting=none,giveninits=true,maxcitenames=1,uniquelist=false,maxbibnames=10,natbib=true]{biblatex}
\bibliography{references}

\makeatletter
\newcommand*\bigcdot{\mathpalette\bigcdot@{1.0}}
\newcommand*\bigcdot@[2]{\mathbin{\vcenter{\hbox{\scalebox{#2}{$\m@th#1*$}}}}}
\makeatother

\def\T{{ \mathrm{\scriptscriptstyle T} }}

\DeclareMathOperator{\Tr}{Tr}
\DeclareMathOperator{\cv}{\mathcal{L}}
\DeclareMathOperator{\Prob}{\mathbb{P}}

\DeclareMathOperator{\diag}{diag}
\newcommand{\Koracle}{K^{(\text{o})}}
\newcommand{\Loracle}{\bm{L}^{(\text{o})}}
\newcommand{\Coracle}{\bm{C}^{(\text{o})}}
\newcommand{\lamoracle}{\lambda^{(\text{o})}}
\newcommand{\lamhatoracle}{\hat{\lambda}^{(\text{o})}}
\newcommand{\Koraclehat}{\hat{K}}

\DeclareMathOperator*{\argmin}{arg\,min}
\DeclareMathOperator*{\argmax}{arg\,max}
\newcommand{\indep}{\rotatebox[origin=c]{90}{$\models$}}
\DeclareMathOperator{\E}{\mathbb{E}}
\DeclareMathOperator{\V}{\mathbb{V}}

\DeclareMathOperator{\tdist}{\stackrel{d}{\to}}
\DeclareMathOperator{\C}{Cov}

\DeclareMathOperator{\vecM}{vec}
\DeclareMathOperator{\asim}{\stackrel{\cdot}{\sim}}

\DeclareMathOperator{\edist}{\stackrel{\text{\scriptsize \textit{d}}}{=}}
\DeclareMathOperator{\im}{im}
\DeclarePairedDelimiter\abs{\lvert}{\rvert}%
\DeclarePairedDelimiter\norm{\lVert}{\rVert}%

\makeatletter
\let\oldabs\abs
\def\abs{\@ifstar{\oldabs}{\oldabs*}}
\let\oldnorm\norm
\def\norm{\@ifstar{\oldnorm}{\oldnorm*}}
\makeatother

\newtheorem{assumption}{\textit{Assumption}}

\newtheorem{remark}{\textit{Remark}}
\newtheorem{algorithm}{\textit{Algorithm}}
\newtheorem{proposition}{\textit{Proposition}}
\newtheorem{corollary}{\textit{Corollary}}
\newtheorem{theorem}{\textit{Theorem}}
\newtheorem{lemma}{\textit{Lemma}}

\title{Factor analysis in high dimensional biological data with dependent observations}

\author{
  Chris McKennan\\ 
  Department of Statistics\\
  University of Pittsburgh\\
  Pittsburgh, PA 15260 \\
  \texttt{chm195@pitt.edu}
}

\begin{document}
\maketitle

\begin{abstract} 
   Factor analysis is a critical component of high dimensional biological data analysis. However, modern biological data contain two key features that irrevocably corrupt existing methods. First, these data, which include longitudinal, multi-treatment and multi-tissue data, contain samples that break critical independence requirements necessary for the utilization of prevailing methods. Second, biological data contain factors with large, moderate and small signal strengths, and therefore violate the ubiquitous ``pervasive factor'' assumption essential to the performance of many methods. In this work, I develop a novel statistical framework to perform factor analysis and interpret its results in data with dependent observations and factors whose signal strengths span several orders of magnitude. I then prove that my methodology can be used to solve many important and previously unsolved problems that routinely arise when analyzing dependent biological data, including high dimensional covariance estimation, subspace recovery, latent factor interpretation and data denoising. Additionally, I show that my estimator for the number of factors overcomes both the notorious ``eigenvalue shadowing'' problem, as well as the biases due to the pervasive factor assumption that plague existing estimators. Simulated and real data demonstrate the superior performance of my methodology in practice.
\end{abstract} 


\section{Introduction}
\label{section:Introduction}
Factor analysis is an indispensable component of high throughput biological data analysis. However, existing methods rely on critical assumptions that are not satisfied in modern biological data.\par 
\indent Suppose $\bm{Y} \in \mathbb{R}^{p \times n}$ contains the gene expression or DNA methylation of $p$ genomic units measured in $n$ samples, where $10^4 \lesssim p \lesssim 10^6$ and $n \lesssim 10^2$ in typical genetic and epigenetic data. For latent factors $\bm{C} \in \mathbb{R}^{n \times K}$ and loadings $\bm{L} \in \mathbb{R}^{p \times K}$, I consider the following general factor model:
\begin{align}
\label{equation:ModelIntro}
    \bm{Y} = \bm{L}\bm{C}^{\T} + \bm{E}, \quad \E(\bm{E}_{g \bigcdot}) = \bm{0}, \quad \V(\bm{E}_{g \bigcdot}) = \bm{V}_g, \quad g \in \{1,\ldots,p\},
\end{align}
where $\bm{E}$ is a random matrix with $g$th row and $i$th column $\bm{E}_{g \bigcdot} \in \mathbb{R}^n$ and $\bm{E}_{\bigcdot i} \in \mathbb{R}^p$. As discussed in Section \ref{section:Theory}, the rows of $\bm{E}$ may be dependent, provided the eigenvalues of $\V(\bm{E}_{\bigcdot i})$ remain bounded. The goal is to estimate $K$, $\bm{C}$, $\bm{L}$ and $\bm{V}_g$ so as to accomplish the following common objectives in biological data analysis:
\begin{enumerate}[label=(\alph*)]
    \item Characterize and prioritize the sources of variation contained in $\bm{C}$ \citep{DieselParticles,FactorCorr,eLife,Knowles}.\label{item:goal:Prioritize}
    \item Understand the relationships between biological samples \citep{FactorSamples,TwinDNAmPlos} and genomic units \citep{SVD_genomics}.\label{item:goal:Relation}
    \item Denoise $\bm{Y}$ to empower inference in gene expression and DNA methylation quantitative trait loci (eQTL and meQTL) studies \citep{Knowles,Marcus}.\label{item:goal:Denoise}
\end{enumerate}
\indent Existing methods to perform factor analysis and their theoretical guarantees can, to a large extent, be partitioned into two groups based on their assumptions on $\bm{V}_g$ and the $K$ latent factors' signal strengths, where signal strengths are quantified as the $K$ non-zero eigenvalues $\lambda_1 \geq \cdots \geq \lambda_K > 0$ of $n\{p^{-1}\bm{L}(n^{-1}\bm{C}^{\T}\bm{C})\bm{L}^{\T}\}$. The first group relies on the standard assumption that the columns of $\bm{E}$ are independent and identically distributed \citep{BCV_svd,BaiLi,FanEigen,bcv,PA_Dobriban,ChrisANDDan,DPA,FanCorrMatrix}, implying $\bm{V}_g=v_g I_n$ for some $v_g > 0$ for all $g \in \{1,\ldots,p\}$. However, this critical assumption is violated by the cornucopia of biological data with dependent samples, which include longitudinal data \citep{Martino,SwedishTwins,LongRNAseq,DieselParticles,URECA}, multi-tissue data \citep{GTEX,MultiTissueMeth,MultiTissueDNAm}, multi-treatment data \citep{Knowles,Marcus}, as well as data from related individuals \citep{TwinDanish,KinshipBaboons,TwinDNAmPlos}. Not only does the independence assumption made in by aforementioned articles imply their theoretical guarantees are not applicable to these dependent data, I show their estimators for $K, \bm{C}$ and $\bm{L}$ are irrevocably corrupted by the dependence between the columns of $\bm{E}$ in practice.\par 
\indent The second group of methods allow for dependence between the columns of $\bm{E}$, but rely on the ``pervasive factor'' assumption in which $\lambda_K \asymp n$ \citep{BaiNg,AhnHorenstein,FanJRSSB,Reviewer2_ShrinkageFactors,Reviewer2_TimeOrdered,RandomJASA,Biometrika_K}, or assume $\lambda_K \to \infty$ as $n,p \to \infty$ \citep{Onatski_Corr}. Intuitively speaking, this implies a scree plot of the eigenvalues of $p^{-1}\bm{Y}^{\T}\bm{Y}$ should reveal an unambiguous gap between the $K$th and $[K+1]$th eigenvalues. While such an assumption makes it possible to place general assumptions on the dependence between the entries of $\bm{E}$, it is patently violated in nearly all biological data \citep{bcv,CATE,ChrisANDDan}. For example, $\lambda_1 \asymp n$ and $\lambda_K \lesssim 1$ in the gene expression data example presented in Section \ref{section:RealData}. This assumption is more than a mere technicality. In fact, these methods are so dependent on the assumption that $\lambda_K \to \infty$ that they consistently fail to recover moderate and weak factors \citep{bcv}, which I show biases estimators from and under-powers inference using downstream methods that rely on estimates for $\bm{C}$ and $\bm{L}$.\par 
\indent The purpose of this work is to facilitate objectives \ref{item:goal:Prioritize}, \ref{item:goal:Relation} and \ref{item:goal:Denoise} by providing a novel framework, efficient estimators and the requisite theory to perform factor analysis and interpret the results in dependent biological data with nearly arbitrary eigenvalues $\lambda_1,\ldots,\lambda_K$. First, I characterize the types of dependence typically observed in biological data in Section \ref{section:Model}, and for $K$ known, extend a recently proposed method to estimate $\lambda_1,\ldots,\lambda_K$, $\bm{C}$, $\bm{L}$ and $\bm{V}_g$ in Section \ref{section:CandL}. A critical component of my method is a novel eigenvalue bias correction for dependent data that ensures the estimates for $\lambda_1,\ldots,\lambda_K$ and the left and right singular vectors of $\bm{L}\bm{C}^{\T}$ are as efficient as those derived from data with independent samples. Accurate estimates for these quantities are crucial to analyzing biological data with dependent \citep{DieselParticles,FactorCorr,Knowles,Marcus} samples, and I prove in Section \ref{section:Theory} that my estimates for them enable objectives \ref{item:goal:Prioritize} and \ref{item:goal:Relation}. In addition to proving the estimate for $\bm{C}$ capacitates objective \ref{item:goal:Denoise}, I show the estimates for the left singular vectors and $\lambda_1,\ldots,\lambda_K$ can be leveraged to derive estimates for and asymptotic distributions of the eigenvectors and eigenvalues of $\V\left(\bm{Y}_{\bigcdot i}\right) \in \mathbb{R}^{p \times p}$ when $\bm{Y}_{\bigcdot 1},\ldots,\bm{Y}_{\bigcdot n}$ are dependent. As far as I am aware, the abovementioned methodology provides the first provably accurate set of  estimators that achieve objectives \ref{item:goal:Prioritize}, \ref{item:goal:Relation} and \ref{item:goal:Denoise} in dependent biological data.\par 
\indent Second, I extend the above methodology in Section \ref{section:EstK} when $K$ is unknown by framing the estimation of $K$ as a model selection problem, and introduce the Oracle rank, $\Koracle$, as that which minimizes a weighted generalization error. This has the effect of excluding components of $\bm{L}\bm{C}^{\T}$ whose signal strengths are below the noise level of the data, and are therefore too weak to justify the added estimation uncertainty that results from their inclusion in the model for $\bm{L}\bm{C}^{\T}$. As far as I am aware, my estimate for $\Koracle$ is the first consistent estimator for the number of latent factors in dependent data with $\lambda_{K}\lesssim 1$ and $\lambda_{1}\lesssim n$, and therefore circumvents the ``eigenvalue shadowing'' problem \citep{PA_Dobriban} and the biases that accompany the pervasive factor assumption. I lastly use simulated and real genetic data in Sections \ref{section:Simulation} and \ref{section:RealData} to illustrate the power of my framework and estimators in practice. The proofs of all theoretical statements are given in the Supplementary Material, and an \texttt{R} package implementing my method is available from https://github.com/chrismckennan/CorrConf.

\section{Notation and a model for the data}
\label{section:Model}
\subsection{Notation}
\label{subsection:Notation}
Let $n > 0$ be an integer. I let $\bm{1}_n \in \mathbb{R}^n$ be the vectors of all ones,  $I_n \in \mathbb{R}^{n \times n}$ be the identity matrix, $[n]=\lbrace 1,\ldots,n \rbrace$ and $\bm{x}_i$ be the $i$th element of $\bm{x} \in \mathbb{R}^n$. For $\bm{M} \in \mathbb{R}^{n \times m}$, I let $\bm{M}_{ij} \in\mathbb{R}$, $\bm{M}_{\bigcdot j} \in \mathbb{R}^{n}$ and $\bm{M}_{i \bigcdot} \in \mathbb{R}^{m}$ be the $(i,j)$th element, $j$th column and $i$th row of $\bm{M}$, respectively, and define $P_{M}$ and $P_{M}^{\perp}$ to be the orthogonal projection matrices onto $\im( \bm{M}) = \lbrace \bm{M}\bm{v} : \bm{v} \in \mathbb{R}^m \rbrace$ and $\ker(\bm{M}^{\T}) = \lbrace \bm{u} \in \mathbb{R}^{n} : \bm{M}^{\T} \bm{u}=\bm{0} \rbrace$. If $m=n$, I let $\abs*{\bm{M}}$ and $\abs*{\bm{M}}_+$ be the determinant and pseudo-determinant, respectively, and for $\bm{M}=\bm{M}^T$ and $s \in [n]$, let $\Lambda_s(\bm{M})$ be the $s$th largest eigenvalue of $\bm{M}$. For random vectors $\bm{X},\bm{Y}\in\mathbb{R}^{n}$, I let $\bm{X} \asim (\bm{\mu},\bm{V})$ if $\E(\bm{X})=\bm{\mu}$ and $\V(\bm{X})=\bm{V}$, and $\bm{X}\edist \bm{Y}$ if $\bm{X},\bm{Y}$ have the same distribution.

\subsection{The model for the data}
\label{subsection:model}
Let $\bm{Y} \in \mathbb{R}^{p \times n}$ be the observed data, where $\bm{Y}_{gi}$ is the observation at genomic unit $g \in [p]$ in sample $i \in [n]$. I assume that Model \eqref{equation:ModelIntro} holds for some non-random latent loadings $\bm{L} \in \mathbb{R}^{p \times K}$ and random latent factors $\bm{C} \in \mathbb{R}^{n \times K}$, where
\begin{align}
\label{equation:Vmodel}
    \bm{V}_g = \sum\limits_{j=1}^b v_{g,j}\bm{B}_j, \quad \bar{\bm{V}} = p^{-1}\sum\limits_{g=1}^p \bm{V}_g = \sum\limits_{j=1}^b \bar{v}_j \bm{B}_j, \quad g \in [p]
\end{align}
for some observed matrices $\bm{B}_1,\ldots,\bm{B}_b$ that parametrize the correlation across samples. This is a ubiquitous model for $\bm{V}_g$ in modern high throughput biological data, and can be used to model the correlation structure in multi-tissue data \citep{MultiTissueMeth,MultiTissueDNAm}, longitudinal data \citep{LongRNAseq,URECA}, multi-treatment or multi-condition data \citep{Knowles,Marcus}, data from related individuals \citep{TwinDanish,KinshipBaboons,TwinDNAmPlos}, or a combination of these data types \citep{Martino,SwedishTwins,DieselParticles}. If $\bm{E}_{\bigcdot 1},\ldots,\bm{E}_{\bigcdot n}$ are independent and identically distributed, then $b=1$ and $\bm{B}=I_n$. I assume the unknown variance multipliers $\bm{v}_g = (v_{g,1},\ldots,v_{g,b})^{\T}$ lie in the convex set $\Theta = \lbrace \bm{x} \in \mathbb{R}^b : \bm{A}_v\bm{x} \geq \bm{0} \rbrace$ for some known $\bm{A}_v \in \mathbb{R}^{q \times b}$. The matrix $\bm{A}_v$ will typically be $I_b$, but can take other values depending on the parametrization of $\bm{B}_1,\ldots,\bm{B}_b$.\par
\indent I assume throughout that $\bm{C}$ is independent of $\bm{E}$. Similar to previous work that assumes $b=1$, $\bm{B}_1=I_n$ and allows $\lambda_K \lesssim 1$, the assumptions I place on the dependence between the rows of $\bm{E}$ will depend on whether or not an estimate for $K$ is available \citep{FanEigen,bcv,DPA,PA_Dobriban}. To avoid confusing technicalities, I save the details for Section \ref{section:Theory}.\par
\indent The dependence between the entries of $\bm{C}_{\bigcdot r}$ may depend on $r$. For example, columns corresponding to technical variables like batch number may have independent entries, and others representing biological factors like cell composition may have dependent entries. Therefore, I only assume $\E(n^{-1}\bm{C}^{\T}\bm{C})$ exists and is full rank, and, unless otherwise stated, place no assumptions on the dependence between the elements of $\bm{C}$. Therefore, $\E(\bm{Y}) = \bm{L}\lbrace \E(\bm{C}) \rbrace^{\T}$ and
\begin{align*}
    \C(\bm{Y}_{gi},\bm{Y}_{hj}) &= \bm{\ell}_g^{\T}\C(\bm{C}_{i \bigcdot},\bm{C}_{j \bigcdot})\bm{\ell}_h + \C(\bm{E}_{gi},\bm{E}_{hj}), \quad g,h\in[p]; i,j\in[n].   
\end{align*}
Evidently, this dependence structure is far more general than that considered by previous authors, who typically only consider data where $\C(\bm{E}_{gi},\bm{E}_{hj})$ does not depend on $i$ or $j$ and $\C(\bm{C}_{i \bigcdot},\bm{C}_{j \bigcdot}) = \bm{\Psi}I(i=j)$ for some non-singular $\bm{\Psi} \in \mathbb{R}^{K \times K}$ \cite{FanJRSSB,FanEigen,PA_Dobriban,bcv,CATE,DPA,FanCorrMatrix}.\par 
\indent A more general model would be $\bm{Y}=\bm{\Gamma}\bm{Z}^{\T}+\bm{L}\bm{C}^{\T}+\bm{E}$, where $\bm{Z} \in \mathbb{R}^{n \times r}$ are observed nuisance covariates, like the intercept or treatment condition, that may not be of immediate interest. One can get back to Models \eqref{equation:ModelIntro} and \eqref{equation:Vmodel} by multiplying $\bm{Y}$ on the right by a matrix $\bm{Q}_Z \in \mathbb{R}^{n \times (n-r)}$ whose columns form an orthonormal basis for $\ker(\bm{Z}^{\T})$, where
\begin{align}
\label{equation:RorateOutZ}
    \bm{Y}\bm{Q}_Z = \bm{L}(\bm{Q}_Z^{\T}\bm{C})^{\T} + \tilde{\bm{E}}, \quad \tilde{\bm{E}}_{g \bigcdot}=\bm{Q}_Z^T \bm{E}_{g \bigcdot} \asim (\bm{0},\sum_{j=1}^b v_{g,j}\bm{Q}_Z^{\T}\bm{B}_j \bm{Q}_Z), \quad g \in [p].
\end{align}
I therefore work exclusively with Models \eqref{equation:ModelIntro} and \eqref{equation:Vmodel} and assume any nuisance covariates have already been rotated out.\par
\indent It is easy to see that conditional on $\bm{C}$ and provided $\dim\{\im(\bm{L})\}=\dim\{\im(\bm{C})\}=K$, $\bm{L}\bm{C}^{\T}$, and therefore $\im(\bm{L})$ and $\im(\bm{C})$, are identifiable in Model \eqref{equation:ModelIntro}. However, $\bm{L}$ and $\bm{C}$ are themselves not identifiable. To facilitate interpretation and make my methodology useful to biological practitioners, I use the IC3 identification conditions in \citet{BaiLi} and define $\Coracle$ and $\Loracle$ to be
\begin{align}
    \label{equation:LCoracleInitial}
    ( \Coracle,\Loracle ) \in  \{(\bar{\bm{C}},\bar{\bm{L}})\in\mathbb{R}^{n \times K}\times \mathbb{R}^{p \times K} : \bar{\bm{L}}\bar{\bm{C}}^{\T}=\bm{L}\bm{C}^{\T}, n^{-1}\bar{\bm{C}}^{\T}\bar{\bm{C}}=I_K, np^{-1}\bar{\bm{L}}^{\T}\bar{\bm{L}}=\diag(\lambda_1,\ldots,\lambda_K)\}.
\end{align}
Provided $\lambda_1,\ldots,\lambda_K$ are non-degenerate, $\Coracle_{\bigcdot r}$ and $\Loracle_{\bigcdot r}$ are identifiable up to sign parity and are proportional to the $r$th right and left singular vectors of $\bm{L}\bm{C}^{\T}$ for all $r \in [K]$. As defined, $\Coracle_{\bigcdot 1},\ldots,\Coracle_{\bigcdot K}$ and $\Loracle_{\bigcdot 1},\ldots,\Loracle_{\bigcdot K}$ are empirically uncorrelated factors and loadings, where $\Coracle_{\bigcdot r}$ has the natural and intelligible interpretation as being the factor with the $r$th largest effect on expression or methylation. This identification condition is ubiquitous in the biological literature, and has proven to be quite efficacious when analyzing data with independent \citep{SVD_genomics,eLife} and dependent \citep{DieselParticles,FactorCorr,Knowles,Marcus} samples.

\section{Estimation when $K$ is known}
\label{section:CandL}
\subsection{An algorithm to estimate $\Loracle$, $\lambda_1,\ldots,\lambda_K$ and $\Coracle$}
\label{subsection:FALCO}
Here I describe my method to estimate $\Loracle$, $\lambda_1,\ldots,\lambda_K$, $\Coracle$ and $\bm{V}_1,\ldots,\bm{V}_p$ assuming $K$ is known, which extends the method to recover $\im(\bm{C})=\im\{\Coracle\}$ and $\bm{V}_1,\ldots,\bm{V}_p$ proposed in \citet{CorrConf}. Unlike standard Principal Components Analysis (PCA) in data where $\bm{E}_{\bigcdot 1},\ldots,\bm{E}_{\bigcdot n}$ are independent and identically distributed, one must be careful to avoid including variation from $\bm{E}$ that is shared across samples in the estimate for $\im(\bm{C})$. Further, even if $\im(\bm{C})$ were known, estimating $\Loracle$, $\lambda_1,\ldots,\lambda_K$ and $\Coracle$ is challenging because they can no longer be estimated using the singular value decomposition of $\bm{Y}$.\par
\indent To elaborate on both of these points, let $\bm{S}=p^{-1}\bm{Y}^{\T}\bm{Y}$ and note that PCA's estimate for $\im(\bm{C})$, which is simply the span of the first $K$ eigenvectors of $\bm{S}$, can be expressed as
\begin{align}
\label{equation:LossPCA}
   P_{\hat{C}^{(PCA)}}= \argmax_{\substack{  \bm{H}\in \mathbb{R}^{n \times n},\, \bm{H}^{\T}=\bm{H}\\ \bm{H}^2=\bm{H},\,\Tr(\bm{H})=K  }} \Tr(\bm{S} \bm{H}),
\end{align}
where there is a one-to-one correspondence between the estimators $P_{\hat{C}^{(PCA)}}$ and $\im\lbrace \hat{\bm{C}}^{(PCA)} \rbrace$. Consider the simple case when $\bm{E}_{\bigcdot 1},\ldots,\bm{E}_{\bigcdot n}$ are independent and identically distributed. Then $b=1$, $\bm{B}_1=I_n$, $\E(p^{-1}\bm{E}^{\T}\bm{E}) = \bar{v}_1 I_n$ and $\bm{S}$, in expectation, can be expressed as
\begin{align*}
    \E(\bm{S} \mid \bm{C}) = \bm{C}(p^{-1}\bm{L}^{\T}\bm{L})\bm{C}^{\T} + \E(p^{-1}\bm{E}^{\T}\bm{E}) = n^{-1}\Coracle\diag(\lambda_1,\ldots,\lambda_K)\{\Coracle\}^{\T} + \bar{v}_1 I_n.
\end{align*}
Since $\Tr\{\E(p^{-1}\bm{E}^{\T}\bm{E})\bm{H}\}=\Tr\{(\bar{v}_1 I_n)\bm{H}\}=K\bar{v}_1$ does not depend on $\bm{H}$, this implies variation in $\bm{E}$ has little influence on the objective in \eqref{equation:LossPCA}, and therefore $\im\lbrace \hat{\bm{C}}^{(PCA)} \rbrace$. Further, since adding a multiple of the identity does not change a matrix's eigenvectors, the orthonormal columns of $n^{-1/2}\Coracle$ are the first $K$ eigenvectors of $\E(\bm{S} \mid \bm{C})$. This suggests $n^{-1/2}\Coracle$ can be accurately estimated as the first $K$ eigenvectors of $\bm{S}$, which form an ordered orthonormal basis for $\im\lbrace \hat{\bm{C}}^{(PCA)} \rbrace$. However, both of these lines of reasoning break down when $\bm{E}_{\bigcdot 1},\ldots,\bm{E}_{\bigcdot n}$ are dependent. In such cases, $\Tr\{\E(p^{-1}\bm{E}^{\T}\bm{E})\bm{H}\}$ will depend on $\bm{H}$ because $\bar{\bm{V}}=\E(p^{-1}\bm{E}^{\T}\bm{E})$ will no longer be a multiple of the identity. Therefore, the solution to \eqref{equation:LossPCA} will be driven by variation in $\bm{E}$, thereby corrupting PCA's estimate for $\im(\bm{C})$. Further, since the eigenvectors of $\E(\bm{S} \mid \bm{C})=n^{-1}\Coracle\diag(\lambda_1,\ldots,\lambda_K)\{\Coracle\}^{\T} + \bar{\bm{V}}$ are no longer $n^{-1/2}\Coracle$, the eigenvectors of $\bm{S}$ should not be used to estimate $\Coracle$.\par 
\indent Besides illuminating issues when applying standard factor analysis techniques in data with dependent samples, the above discussion also suggests that accounting for $\bar{\bm{V}}$ may circumvent these issues. Suppose that $\bar{\bm{V}} \succ \bm{0}$ were known, and define
\begin{align}
\label{equation:IcaseLoss}
    P_{\hat{C}} = \argmax_{\substack{  \bm{H}\in \mathbb{R}^{n \times n},\, \bm{H}^{\T}=\bm{H}\\ \bm{H}^2=\bm{H},\,\Tr(\bm{H})=K  }} \Tr\lbrace ( \bar{\bm{V}}^{-1}\bm{S}\bar{\bm{V}}^{-1} ) (\bm{H}\bar{\bm{V}}^{-1}\bm{H})^{\dagger}\rbrace.
\end{align}
Because $\Tr\{\bar{\bm{V}}^{-1}(\bm{H}\bar{\bm{V}}^{-1}\bm{H})^{\dagger}\}=K$ for all $\bm{H}$, the objective function in \eqref{equation:IcaseLoss} satisfies
\begin{align*}
    \E[ \Tr\lbrace ( \bar{\bm{V}}^{-1}\bm{S}\bar{\bm{V}}^{-1}) (\bm{H}\bar{\bm{V}}^{-1}\bm{H})^{\dagger}\rbrace \mid\bm{C} ]= \Tr\lbrace  \bar{\bm{V}}^{-1}\bm{C}(p^{-1}\bm{L}^{\T}\bm{L})\bm{C}^{\T}\bar{\bm{V}}^{-1} (\bm{H}\bar{\bm{V}}^{-1}\bm{H})^{\dagger}\rbrace + K.
\end{align*}
Since the only term involving $\bm{H}$ in the above expression takes its maximum when $\bm{H}=P_{C}$, this simple analysis argues that \eqref{equation:IcaseLoss} properly accounts for $\bar{\bm{V}}$ when estimating $\im(\bm{C})$. One can then use $\im(\hat{\bm{C}})$ to estimate $\lambda_1,\ldots,\lambda_K$, and subsequently choose an appropriate ordered basis for $\im(\hat{\bm{C}})$ to estimate $\Loracle$ and $\Coracle$. These steps are presented below in Algorithm \ref{algorithm:EstC}, which I call FALCO (\textbf{f}actor \textbf{a}na\textbf{l}ysis in \textbf{co}rrelated data), which estimates $\im(\bm{C})$, $\lambda_1,\ldots,\lambda_K$, $\Coracle$ and $\Loracle$, and uses the warm start technique detailed in \cite{CorrConf} to estimate $\bar{\bm{V}}$.

\begin{algorithm}[FALCO]
\label{algorithm:EstC}
Let $\bm{Y} \in \mathbb{R}^{p \times n}$ and $\bm{M} \in \mathbb{R}^{b \times b}$, where $\bm{M}_{rs}=n^{-1}\Tr(\bm{B}_r\bm{B}_s)$ for $r,s \in [b]$. Fix some $\alpha \in (0,1)$ and integer $K_{\max} \in (0,n \wedge p]$ and let $\bm{V}(\bm{\theta})=\sum_{j=1}^b \bm{\theta}_j \bm{B}_j$ for any $\bm{\theta}\in\mathbb{R}^b$.
\begin{enumerate}[label=(\alph*)]
\item Initialize $\hat{\bar{\bm{v}}}=(\hat{\bar{v}}_1,\ldots,\hat{\bar{v}}_b)^{\T}$ as $\hat{\bar{\bm{v}}} = \argmax_{\bm{\theta} \in \Theta}(- \log\lbrace\abs*{\bm{V}(\bm{\theta})} \rbrace - \Tr[ (p^{-1}\bm{Y}^{\T}\bm{Y}) \lbrace \bm{V}(\bm{\theta}) \rbrace^{-1} ])$. Set $k=1$ and $\hat{\bar{\bm{V}}} = \bm{V}(\hat{\bar{\bm{v}}})$.\label{item:EstC:0}
\item \begin{enumerate}[label=(\roman*)]
    \item Define $P_{\hat{C}}$, and therefore $\im(\hat{\bm{C}})$, to be
    \begin{align}
    \label{equation:Alg1:C}
        P_{\hat{C}} = \argmax_{\substack{  \bm{H}\in \mathbb{R}^{n \times n},\, \bm{H}^{\T}=\bm{H}\\ \bm{H}^2=\bm{H},\,\Tr(\bm{H})=k  }} \Tr[\{ \hat{\bar{\bm{V}}}^{-1}(p^{-1}\bm{Y}^{\T}\bm{Y})\hat{\bar{\bm{V}}}^{-1} \}(\bm{H}\hat{\bar{\bm{V}}}^{-1}\bm{H})^{\dagger}].
    \end{align}\label{item:EstC:C}
    \item Let $\hat{\bm{M}}\in\mathbb{R}^{b\times b}$ be $\hat{\bm{M}}_{rs}=n^{-1}\Tr( P_{\hat{C}}^{\perp}\bm{B}_r P_{\hat{C}}^{\perp}\bm{B}_s )$. If $\Lambda_{b}(\hat{\bm{M}} ) \leq \alpha \Lambda_{b}( \bm{M} )$, go to Step \ref{item:EstC:Ctilde}.\label{item:EstC:U}
    \item Set $\hat{\bar{\bm{v}}} = \argmax_{\bm{\theta} \in \Theta}( -\log( \abs*{P_{\hat{C}}^{\perp}\bm{V}(\bm{\theta})P_{\hat{C}}^{\perp}}_{+} ) - \Tr[ (p^{-1}\bm{Y}^{\T}\bm{Y}) \lbrace P_{\hat{C}}^{\perp}\bm{V}(\bm{\theta})P_{\hat{C}}^{\perp}\rbrace^{\dagger} ])$ and $\hat{\bar{\bm{V}}} = \bm{V}(\hat{\bar{\bm{v}}})$.\label{item:EstC:V}
    
    \item Repeat Steps \ref{item:EstC:C}, \ref{item:EstC:U} and \ref{item:EstC:V} two times, and stop on Step \ref{item:EstC:C} of the third iteration.\label{item:EstC:iter}
\end{enumerate}\label{item:EstC:ImageC}
\item Let $\tilde{\bm{C}} \in \mathbb{R}^{n \times k}$ be any matrix such that $\im(\tilde{\bm{C}})=\im(\hat{\bm{C}})$. Define $\tilde{\bm{L}} = \bm{Y}\hat{\bar{\bm{V}}}^{-1}\tilde{\bm{C}}( \tilde{\bm{C}}^{\T}\hat{\bar{\bm{V}}}^{-1}\tilde{\bm{C}} )^{-1}$.\label{item:EstC:Ctilde}
\item Define $\hat{\lambda}_r$ to be the $r$th largest eigenvalue of $n^{-1}\tilde{\bm{C}}\lbrace np^{-1}\tilde{\bm{L}}^{\T}\tilde{\bm{L}} - (n^{-1} \tilde{\bm{C}}^{\T}\hat{\bar{\bm{V}}}^{-1}\tilde{\bm{C}} )^{-1} \rbrace \tilde{\bm{C}}^{\T}$ for $r\in[K]$. If $k=K$, let $\hat{\lambda}_r$ be the estimate for $\lambda_r$ for all $r\in[K]$. \label{item:EstC:eigen}
\item Let $\tilde{\bm{U}} \in \mathbb{R}^{k \times k}$ be a unitary matrix that satisfies
\begin{align*}
    \tilde{\bm{U}}^{\T}(n^{-1} \tilde{\bm{C}}^{\T}\tilde{\bm{C}} )^{1/2} \lbrace np^{-1}\tilde{\bm{L}}^{\T}\tilde{\bm{L}} - (n^{-1} \tilde{\bm{C}}^{\T}\hat{\bar{\bm{V}}}^{-1}\tilde{\bm{C}} )^{-1} \rbrace ( n^{-1}\tilde{\bm{C}}^{\T}\tilde{\bm{C}} )^{1/2} \tilde{\bm{U}} = \diag(\hat{\lambda}_1,\ldots,\hat{\lambda}_k)
\end{align*}
and define $\hat{\bm{L}} = \tilde{\bm{L}}(n^{-1} \tilde{\bm{C}}^{\T}\tilde{\bm{C}} )^{1/2} \tilde{\bm{U}}$ and $\hat{\bm{C}} = \tilde{\bm{C}}(n^{-1} \tilde{\bm{C}}^{\T}\tilde{\bm{C}} )^{-1/2}\tilde{\bm{U}}$. If $k=K$, let $\hat{\bm{L}}$ and $\hat{\bm{C}}$ be the estimates for $\Loracle$ and $\Coracle$.\label{item:Estc:LandBasis}
\item If $k < K_{\max}$, update $k \leftarrow k + 1$ and return to Step \ref{item:EstC:ImageC}.
\end{enumerate}
\end{algorithm}

\begin{remark}
\label{remark:Alg1:Vg}
The estimator $\hat{\bar{\bm{v}}}$ in Step \ref{item:EstC:ImageC}\ref{item:EstC:V} is exactly the restricted maximum likelihood (REML) estimator for $\bm{\theta}$ under the model $\bm{Y}\sim MN_{p \times n}( \bm{L}\hat{\bm{C}}^{\T},I_p,\bm{V}(\bm{\theta}) )$. We can estimate $\bm{V}_g = \bm{V}(\bm{v}_g)$ for any $k$ with REML using the model $\bm{Y}_{g \bigcdot} \sim N(\hat{\bm{C}}\bm{L}_{g \bigcdot},\bm{V}(\bm{v}_g))$.
\end{remark}

\begin{remark}
\label{remark:Alg1:Ind}
If $b=1$, $\bm{B}_1=I_n$, then $\hat{\bm{L}}_{\bigcdot r}$ and $\hat{\bm{C}}_{\bigcdot r}$ are proportional to the $r$th left and right singular vectors of $\bm{Y}$, and $\hat{\lambda}_r$ is the bias-corrected estimator proposed in \citet{ChrisANDDan}.
\end{remark}

With the exception of \eqref{equation:Alg1:C} and Step \ref{item:EstC:U} of \ref{item:EstC:ImageC}, Steps \ref{item:EstC:0} and \ref{item:EstC:ImageC} of Algorithm \ref{algorithm:EstC} resemble the iterative method proposed in \citet{CorrConf} to simultaneously estimate $\im(\bm{C})$ and $\bar{\bm{V}}$, where the estimate for $\bar{\bm{V}}$ when $\dim\lbrace \im(\bm{C}) \rbrace$ is assumed to be $k-1$ is used as a starting point when $\dim\lbrace \im(\bm{C}) \rbrace = k$. This ``warm start'' technique helps ensure that variation attributable to $\bar{\bm{V}}$ is not mistakenly assigned to $\bm{C}$. Step \ref{item:EstC:ImageC}\ref{item:EstC:U} flags estimates for $P_{C}$ where the subsequent restricted log-likelihood function in \ref{item:EstC:V} may not identify $\bar{\bm{v}}$. This step is necessary when $k \asymp n \wedge p$, and allows us to circumvent the common, but problematic, restriction that the maximum possible value of $K$, $K_{\max}$, be at most finite when estimating $K$ in Section~\ref{section:EstK}. I set $\alpha=0.1$ in practice. The loss and estimator in \eqref{equation:Alg1:C} is unique to the above Algorithm, and helps generalize the problem of subspace estimation to data with correlated samples. And while \eqref{equation:Alg1:C} is ostensibly a challenging problem, I provide a simple and exact solution in Proposition \ref{proposition:ImageC} below.

\begin{proposition}
\label{proposition:ImageC}
If $\hat{\bar{\bm{V}}} \succ \bm{0}$, $P_{\hat{C}}$ in \eqref{equation:Alg1:C} is exactly $P_{\hat{\bar{V}}^{1/2}W}$, where the columns of $\bm{W} \in \mathbb{R}^{n \times k}$ are the first $k$ right singular vectors of $\bm{Y}\hat{\bar{\bm{V}}}^{-1/2}$.
\end{proposition}

\subsection{Intuition regarding the estimators in Algorithm \ref{algorithm:EstC}}
\label{subsection:FALCOIntuition}
\indent As far as I am aware, the estimates for the eigenvalues in Step \ref{item:EstC:eigen} and the estimates for $\Loracle$ and $\Coracle$ in Step \ref{item:Estc:LandBasis} are the first estimators for these quantities that account for the correlation between samples, and therefore warrant some discussion. Since many of the eigenvalues $\lambda_r$ will be moderate or small in biological data \cite{ChrisANDDan}, one must account for eigenvalue inflation. This is a well-studied phenomenon in data with independent samples \cite{FanEigen,ChrisANDDan}, and occurs because small errors in the estimates $\tilde{\bm{L}}_{1 \bigcdot},\ldots,\tilde{\bm{L}}_{p \bigcdot}$ accumulate and inflate the estimator $np^{-1}\tilde{\bm{L}}^{\T}\tilde{\bm{L}}=np^{-1}\sum_{g=1}^p \tilde{\bm{L}}_{g \bigcdot} \tilde{\bm{L}}_{g \bigcdot}^{\T}$. The term $(n^{-1} \tilde{\bm{C}}^{\T}\hat{\bar{\bm{V}}}^{-1}\tilde{\bm{C}} )^{-1}$ in Steps \ref{item:EstC:eigen} and \ref{item:Estc:LandBasis} corrects that bias, which as hinted in Remark \ref{remark:Alg1:Ind}, reduces to the usual bias correction used to deflate estimates for $\lambda_r$ when $b=1$ and $\bm{B}_1=I_n$ \cite{FanEigen,ChrisANDDan}.\par 
\indent Perhaps the most unnatural element of Algorithm \ref{algorithm:EstC} is Step \ref{item:Estc:LandBasis}. To justify this step, suppose $k=K$. Since $\tilde{\bm{L}}\tilde{\bm{C}}^{\T}$ only depends on $\im(\hat{\bm{C}})$ and not the choice of parametrization of $\tilde{\bm{C}}$, I require $\hat{\bm{L}}=\tilde{\bm{L}}\bm{R}$ and $\hat{\bm{C}} = \tilde{\bm{C}}\bm{R}^{-\T}$ for any non-singular $\bm{R}\in\mathbb{R}^{K\times K}$ to ensure $\hat{\bm{L}}\hat{\bm{C}}^{\T}=\tilde{\bm{L}}\tilde{\bm{C}}^{\T}$. Since $n^{-1}\{\Coracle\}^{\T}\Coracle=I_K$, I set $\bm{R}=(n^{-1}\tilde{\bm{C}}^{\T}\tilde{\bm{C}})^{1/2}\tilde{\bm{U}}$ for some unitary matrix $\tilde{\bm{U}} \in \mathbb{R}^{K \times K}$, which guarantees $n^{-1}\hat{\bm{C}}^{\T}\hat{\bm{C}}=I_K$. I choose $\tilde{\bm{U}}$ so that the inflation-corrected estimator for $np^{-1}\{\Loracle\}^{\T}\Loracle=\diag(\lambda_1,\ldots,\lambda_K)$,
\begin{align*}
    np^{-1}\hat{\bm{L}}^{\T}\hat{\bm{L}} - (n^{-1}\hat{\bm{C}}^{\T}\hat{\bar{\bm{V}}}^{-1}\hat{\bm{C}})^{-1} = \tilde{\bm{U}}^{\T}(n^{-1} \tilde{\bm{C}}^{\T}\tilde{\bm{C}} )^{1/2} \lbrace np^{-1}\tilde{\bm{L}}^{\T}\tilde{\bm{L}} - (n^{-1} \tilde{\bm{C}}^{\T}\hat{\bar{\bm{V}}}^{-1}\tilde{\bm{C}} )^{-1} \rbrace ( n^{-1}\tilde{\bm{C}}^{\T}\tilde{\bm{C}} )^{1/2} \tilde{\bm{U}},
\end{align*}
is exactly $\diag(\hat{\lambda}_1,\ldots,\hat{\lambda}_K)$, where $\hat{\lambda}_r$ is the inflation-corrected estimator of $\lambda_r$. Choosing such a $\tilde{\bm{U}}$ when $b=1,\bm{B}_1=I_n$ is trivial, since one can easily find a $\tilde{\bm{C}}$ such that $n^{-1}\tilde{\bm{C}}^{\T}\hat{\bar{\bm{V}}}^{-1}\tilde{\bm{C}} \propto n^{-1}\tilde{\bm{C}}^{\T}\tilde{\bm{C}} = I_K$ and $np^{-1}\tilde{\bm{L}}^{\T}\tilde{\bm{L}}$ is diagonal. This is certainly not the case in data with more complex correlation structures, since $n^{-1}\tilde{\bm{C}}^{\T}\tilde{\bm{C}}$ and $n^{-1}\tilde{\bm{C}}^{\T}\hat{\bar{\bm{V}}}^{-1}\tilde{\bm{C}}$ cannot both be multiples $I_K$ in general.\par 
\indent I lastly remark that $\hat{\bm{L}}$ will generally not have orthogonal columns. However, I show in Section \ref{subsection:AccuracyAlg1} that, quite remarkably, all of the estimators from Algorithm \ref{algorithm:EstC} for arbitrary $\bm{B}_1,\ldots,\bm{B}_b$ are at least as efficient as those derived from standard PCA when $b=1,\bm{B}_1=I_n$. And while it is not my primary goal, my proof techniques allow me to derive a central limit theorem for $\hat{\lambda}_r$ under far more general assumptions than those considered by other authors.

\section{Defining and estimating the Oracle rank, factors and loadings}
\label{section:EstK}
\subsection{Defining the Oracle rank}
\label{subsection:OracleK}
While Section~\ref{section:CandL} considers the case when $K$ is known, $K$ is typically unknown in real data. However, determining $K$, which is a notoriously challenging problem in data with independent samples, is particularly difficult in data with correlated samples. First, the true $K$ may not be the most appropriate choice for $K$, since the added benefit of estimating factors with negligibly small effects is offset by the cost of additional statistical uncertainty. Second, given that the goal is to analyze real biological data, any estimator must be amenable to data with both large and small eigenvalues $\lambda_r$. Lastly, the estimator must avoid mistaking latent structure due to the dependence between the columns of $\bm{E}$ as arising from $\bm{L}\bm{C}^{\T}$, which can lead to severe overestimates for $K$ \cite{DPA,CorrConf}.\par 
\indent To address these issues, I follow \citet{bcv} and treat the estimation of $K$ as a model selection problem. To define the optimal model, I set the Oracle rank, $\Koracle$, to be that which minimizes the following inverse-variance weighted generalization error:
\begin{align}
\label{equation:OracleK}
    \Koracle =  \argmin_{k \in \lbrace 0,1,\ldots,n \wedge p \rbrace}\lbrace \abs*{\hat{\bm{V}}^{(k)}}^{1/n} \E (\norm*{ [ \bm{L}\bm{C}^{\T} + \tilde{\bm{E}} - \hat{\bm{L}}^{(k)}\lbrace\hat{\bm{C}}^{(k)}\rbrace^{\T} ] \lbrace\hat{\bm{V}}^{(k)}\rbrace^{-1/2} }_F^2 \mid \bm{C},\bm{E} )\rbrace.
\end{align}
Here, $\hat{\bm{L}}^{(k)}\in\mathbb{R}^{p \times k}, \hat{\bm{C}}^{(k)}\in\mathbb{R}^{n \times k}$ and $\hat{\bm{V}}^{(k)}$ are the estimates $\hat{\bm{L}}$, $\hat{\bm{C}}$ and $\hat{\bar{\bm{V}}}$ defined in Steps \ref{item:EstC:ImageC}\ref{item:EstC:V} and \ref{item:Estc:LandBasis} at iteration $k$ of Algorithm \ref{algorithm:EstC} when $K_{\max} = n \wedge p$, $\tilde{\bm{E}}$ is independent of $(\bm{C},\bm{E})$ and $\tilde{\bm{E}} \edist \bm{E}$. The term $\abs*{\hat{\bm{V}}^{(k)}}^{1/n}$ is identical to re-scaling $\hat{\bm{V}}^{(k)}$ such that $\abs*{\hat{\bm{V}}^{(k)}}=1$ for all $k$, and makes \eqref{equation:OracleK} scale-invariant. If $b=1$ and $\bm{B}_1=I_n$, $\Koracle$ reduces to the Oracle rank defined in \citet{bcv}.\par
\indent Assuming for simplicity that $\abs*{\hat{\bm{V}}^{(k)}}=1$, one can rewrite the generalization error in \eqref{equation:OracleK} as
\begin{align}
\label{equation:OracleK:expanded}
     p\Tr[ \E(p^{-1}\tilde{\bm{E}}^{\T}\tilde{\bm{E}})\lbrace\hat{\bm{V}}^{(k)}\rbrace^{-1} ] + \norm*{[\bm{L}\bm{C}^{\T} - \hat{\bm{L}}^{(k)}\lbrace\hat{\bm{C}}^{(k)}\rbrace^{\T}] \lbrace\hat{\bm{V}}^{(k)}\rbrace^{-1/2} }_F^2.
\end{align}
The first term evaluates the accuracy of the estimate for $\bar{\bm{V}}=\E(p^{-1}\tilde{\bm{E}}^{\T}\tilde{\bm{E}})$, which, by Jensen's inequality, is minimized when $\hat{\bm{V}}^{(k)}$ is a scalar multiple of $\bar{\bm{V}}$. Therefore, weighting \eqref{equation:OracleK} by $\lbrace\hat{\bm{V}}^{(k)}\rbrace^{-1/2}$ ensures that $\hat{\bm{V}}^{(\Koracle)}$ captures the variation across the columns of $\bm{E}$. The second term in \eqref{equation:OracleK:expanded} measures the accuracy of $\hat{\bm{L}}^{(k)}\lbrace\hat{\bm{C}}^{(k)}\rbrace^{\T}$ as an estimator for $\bm{L}\bm{C}^{\T}$, where weighting by $\lbrace\hat{\bm{V}}^{(k)}\rbrace^{-1/2}$ prioritizes components of $\bm{L}\bm{C}^{\T}$ not already explained by the estimated model for $\bm{E}$. Note also that this term is not necessarily minimized at $k=K$. Instead, a factor is only included if its capacity to estimate $\bm{L}\bm{C}^{\T}$ outweighs its statistical uncertainty. I describe this precisely in Section \ref{subsection:Ktheory}. 

\subsection{Defining the Oracle factors and loadings}
\label{subsection:OracleFactors}
Given $\Koracle$, the Oracle then must choose the best rank-$\Koracle$ approximation to the rank-$K$ latent signal matrix $\bm{L}\bm{C}^{\T}$. For
\begin{align*}
    \mathcal{S}_{k}= \lbrace (\bar{\bm{C}}, \bar{\bm{L}}) \in \mathbb{R}^{n \times k} \times \mathbb{R}^{p \times k}: n^{-1}\bar{\bm{C}}^{\T}\bar{\bm{C}}=I_{k}, \text{$\bar{\bm{L}}^{\T}\bar{\bm{L}}$ is diagonal with non-increasing elements} \rbrace,
\end{align*}
we take inspiration from the generalized PCA loss considered in \citet{gPCA} and define the Oracle factors and loadings, $\Coracle$ and $\Loracle$, to be
\begin{align}
\label{equation:OracleCL}
( \Coracle,\Loracle ) = \argmin_{(\bar{\bm{C}}, \bar{\bm{L}}) \in \mathcal{S}_{\Koracle}} \norm*{ ( \bm{L}\bm{C}^{\T} - \bar{\bm{L}}\bar{\bm{C}}^{\T} )\bar{\bm{V}}^{-1/2} }_F^2.
\end{align}
If $\Koracle=K$, then $\Coracle$ and $\Loracle$ are exactly as defined in \eqref{equation:LCoracleInitial}. Otherwise, like \eqref{equation:OracleK}, weighting by $\bar{\bm{V}}^{-1/2}$ prioritizes variation in $\bm{L}\bm{C}^{\T}$ not captured by the true model for $\bm{E}$. Note also that $\Loracle\{ \Coracle \}^{\T}$ is the minimizer of \eqref{equation:OracleK:expanded} when $\hat{\bm{V}}^{(\Koracle)} = \bar{\bm{V}}$ is known and \eqref{equation:OracleK:expanded} is treated as a function of $\hat{\bm{L}}^{(\Koracle)}\lbrace\hat{\bm{C}}^{(\Koracle)}\rbrace^{\T}$. Therefore, taken together with the definition of $\Koracle$, $\Loracle\{ \Coracle \}^{\T}$ can be interpreted as the best approximation to $\bm{L}\bm{C}^{\T}$ whose components' signal strengths justify their estimation uncertainty. I show in Section~\ref{section:Theory} that by replacing $K$ with $\Koracle$ in Steps \ref{item:EstC:eigen} and \ref{item:Estc:LandBasis} of Algorithm \ref{algorithm:EstC}, Algorithm \ref{algorithm:EstC} recovers both $\Coracle$ and $\Loracle$.






\subsection{Estimating the Oracle rank}
\label{subsection:EstOracleK}
We extend the procedure developed in \cite{CorrConf} to estimate the Oracle rank in Algorithm \ref{algorithm:CBCV} below. Implicit in Algorithm \ref{algorithm:CBCV} is the assumption that the rows of $\bm{E}$ are independent, which is the standard assumption when estimating $K$ in biological data \cite{BCV_svd,bcv,DPA,PA_Dobriban,CorrConf}. Simulations in Section \ref{section:Simulation} show that Algorithm \ref{algorithm:CBCV} is robust to dependencies commonly observed in biological data.

\begin{algorithm}[CBCV+]
\label{algorithm:CBCV}
Let $K_{\max}=\lceil \eta (n \wedge p) \rceil$ for $\eta\in (0,1)$ and $\bm{Q} \in \mathbb{R}^{n \times n}$ be sampled uniformly from the set of all $n \times n$ unitary matrices. Partition the rows of $\bm{Y}$ uniformly at random into $F \geq 2$ folds. 
\begin{enumerate}[label=\textit{(\alph*)}]
\item For $f \in [F]$, arrange the rows of $\bm{Y}$ such that $\bm{Y} = \begin{bmatrix}
\bm{Y}_{(-f)}\\
\bm{Y}_{f}
\end{bmatrix} = \begin{bmatrix}
\bm{L}_{(-f)}\bm{C}^T\\
\bm{L}_{f}\bm{C}^T
\end{bmatrix} + \begin{bmatrix}
\bm{E}_{(-f)}\\
\bm{E}_{f}
\end{bmatrix}.$
Define $\bm{Y}_{(-f)} \in \mathbb{R}^{p_{(-f)} \times n}$ and $\bm{Y}_{f} \in \mathbb{R}^{p_{f} \times n}$ to be the training and test sets, respectively.\label{item:CBCV:arrange}
\item For all $k \in \lbrace 0,1,\ldots, K_{\max}\rbrace$, obtain $\hat{\bm{C}} \in \mathbb{R}^{n \times k}$ and $\hat{\bm{V}}_{(-f)}$ from $\bm{Y}_{(-f)}$ using Algorithm \ref{algorithm:EstC}.\label{item:CBCV:LogDet}
\item For each $k\in \lbrace 0,1,\ldots, K_{\max}\rbrace$, let $\bar{\bm{Y}}_{f} = \bm{Y}_{f} \hat{\bm{V}}_{(-f)}^{-1/2}\bm{Q}$ and $\hat{\bar{\bm{C}}} = \bm{Q}^{\T}\hat{\bm{V}}_{(-f)}^{-1/2}\hat{\bm{C}}$. Define the loss for this fold, dimension pair as the leave-one-out cross validation loss:
\begin{align}
\label{equation:LOOLoss}
\cv_{f}( k ) = \abs*{\hat{\bm{V}}_{(-f)}}^{1/n}\sum\limits_{i=1}^n \norm*{\bar{\bm{Y}}_{f_{\bigcdot i}} - \hat{\bm{L}}_{f,(-i)} \hat{\bar{\bm{C}}}_{i \bigcdot} }_{2}^2.
\end{align}  
Here, $\hat{\bm{L}}_{f,(-i)}$ is the ordinary least squares regression coefficient from the regression of $\bar{\bm{Y}}_{f,(-i)}$ onto $\hat{\bar{\bm{C}}}_{(-i)}$, where $\bar{\bm{Y}}_{f,(-i)}$ and $\hat{\bar{\bm{C}}}_{(-i)}^T$ are submatrices of $\bar{\bm{Y}}_{f}$ and $\hat{\bar{\bm{C}}}^T$ with the $i$th columns removed.\label{item:CBCV:Loss}
\item Repeat steps \ref{item:CBCV:arrange}--\ref{item:CBCV:Loss} for folds $f=1,\ldots,F$ and define $\Koraclehat = \mathop{\argmin}\limits_{k \in \lbrace 0,1,\ldots,K_{\max} \rbrace} \lbrace \sum_{f=1}^F \cv_{f}( k )\rbrace$.
\end{enumerate}\label{item:CBCV:ChooseK}
\end{algorithm}

Provided the rows of $\bm{E}$ are independent, step \ref{item:CBCV:arrange} partitions $\bm{Y}$ into independent training and test sets, which are used to determine $\hat{\bm{C}},\hat{\bm{V}}_{(-f)}$ and estimate the out-of-sample expected loss defined in \eqref{equation:OracleK}, respectively. Besides ensuring that \eqref{equation:LOOLoss} approximates the expected loss in \eqref{equation:OracleK}, re-scaling $\bm{Y}_f$ by $\hat{\bm{V}}_{(-f)}^{-1/2}$ in step \ref{item:CBCV:Loss} helps alleviate the deleterious effects of correlated data points in leave-one-out cross validation \cite{GeneralizedCV}. Further rotating the test data by $\bm{Q}$ uniformizes the leverage scores of both $\hat{\bm{V}}_{(-f)}^{-1/2}\bm{C}$ and $\hat{\bm{V}}_{(-f)}^{-1/2}\hat{\bm{C}} \in \mathbb{R}^{n \times k}$, which helps guarantee \eqref{equation:LOOLoss} is well behaved for $k \asymp n$. This latter point allows us to avoid the common requirement among existing estimators that $K_{\max}$ be at most finite \cite{BaiNg,Reviewer2_TimeOrdered,Reviewer2_ShrinkageFactors}. While subtle, this is quite important, as such estimators are typically sensitive to $K_{\max}$ \cite{AhnHorenstein,RandomJASA}.


\section{Theoretical guarantees}
\label{section:Theory}
\subsection{Assumptions}
\label{subsection:Assumptions}
In all assumptions and theoretical results, I assume that Models \eqref{equation:ModelIntro} and \eqref{equation:Vmodel} hold, where the symmetric matrices $\bm{B}_1,\ldots,\bm{B}_b \in\mathbb{R}^{n \times n}$ are observed and $b,K = O(1)$ as $n,p \to \infty$. I define $\bm{A}=p^{-1}\E(\bm{L}\bm{C}^{\T}\bm{C}\bm{L}^{\T})$, $\bar{\bm{V}}=p^{-1}\sum_{g=1}^p\bm{V}_g$ and $\delta^2=\abs*{\bar{\bm{V}}}^{1/n}$ throughout, where $\delta^2=\bar{v}_1$ if $b=1,\bm{B}_1=I_n$ and $\delta^2 \asymp 1$ in general under Assumption \ref{assumption:CandL}\ref{item:assumErrors:V}. Lastly, I let $c > 2$ be an arbitrarily large universal constant that does not depend on $n$ or $p$.

\begin{assumption}
\label{assumption:CandL}
Define $\bm{M} \in \mathbb{R}^{b \times b}$ to be $\bm{M}_{ij}=n^{-1}\Tr(\bm{B}_i\bm{B}_j)$. Then:
\begin{enumerate}[label=(\alph*)]
    \item $c^{-1}I_n \preceq \bm{V}_g$, $\abs{v_{g,j}}\leq c$, $c^{-1}I_b \preceq \bm{M}$, $\norm*{\bm{B}_j}_2\leq c$ for all $j \in [b]$ and $\E\lbrace \exp(\bm{t}^{\T}\bm{E}_{g \bigcdot}) \rbrace \leq \exp(c \norm*{\bm{t}}_2^2)$ for all $g \in [p]$ and $\bm{t} \in \mathbb{R}^n$.\label{item:assumErrors:V}
    \item $\bm{C} \in \mathbb{R}^{n \times K}$ is a random matrix that is independent of $\bm{E}$, where $\bm{\Psi}_n=\E\lbrace n^{-1}\bm{C}^{\T}(\delta^{-2}\bar{\bm{V}})^{-1}\bm{C}\rbrace$, $c^{-1}I_K\preceq \bm{\Psi}_n \preceq c I_K$ and $\Delta_{n,p}=\norm*{n^{-1}\bm{C}^{\T} (\delta^{-2}\bar{\bm{V}})^{-1} \bm{C} - \bm{\Psi}_n}_2 = o_P(1)$ as $n,p \to \infty$.\label{item:assumCandL:C}
    \item $\bm{L}$ is a non-random matrix. The $K$ non-zero eigenvalues of $np^{-1}\bm{L}\bm{\Psi}_n \bm{L}^{\T}$ satisfy $0<\gamma_K \leq \cdots \leq \gamma_1 \leq cn$, where for each $r \in [K]$, either $\limsup_{n,p \to \infty}\gamma_r < \infty$ or $\lim_{n,p \to \infty}\gamma_r = \infty$. Further, $\bm{L}_{g \bigcdot}^{\T}\bm{\Psi}_n \bm{L}_{g \bigcdot} \leq c$ for all $g \in [p]$.\label{item:assumCandL:Eigen}
\end{enumerate}
\end{assumption}

The assumptions on $\bm{B}_j$ and $\bm{M}$ imply no one direction dominates the variation in $\bm{E}_{g \bigcdot}$ and that $\bm{v}_g$ is identifiable, respectively, for all $g \in [p]$. With the exception of Section \ref{subsection:InferenceC}, I place no assumptions on $\bm{C}$ besides what are stated in \ref{item:assumCandL:C}, where it can be shown that $\Delta_{n,p}=o_P(1)$ under general assumptions \citep{CorrConf}. I place assumptions on $\gamma_1,\ldots,\gamma_K$ and not on the eigenvalues of $\bm{A}$ to facilitate statements regarding $\Koracle$ in Section~\ref{subsection:Ktheory}, but note that $\Lambda_r(\bm{A}) \asymp \gamma_r$ for all $r \in [K]$. Unlike previous work \cite{BaiLi,BiometrikaConfounding,CATE,DPA,CorrConf,Biometrika_K}, I only require $\gamma_1 \lesssim n$ and do not assume $\gamma_1/\gamma_K$ is bounded as $n,p \to \infty$. This allows one to analyze genetic data, where it is the norm rather than the exception for the data to contain both strong and weak factors \cite{CATE,ChrisANDDan}. This assumption is more than a mere technical condition, since as I show in Section \ref{section:Simulation}, many methods fail in practice when $\gamma_1/\gamma_K$ is too large. The assumption that $\bm{E}_{g \bigcdot}$ is sub-Gaussian in \ref{item:assumErrors:V} is standard among authors who assume the entries of $\bm{E}_{g \bigcdot}$ are independent and identically distributed and consider both strong and weak factors \cite{FanEigen,ChrisANDDan}. The assumed dependence between the rows of $\bm{E}$ and the relationship between $n$ and $p$ will depend on whether or not $\Koracle$ is known, which helps make my results as general as possible, and, as I show in Section~\ref{subsection:AccuracyAlg1}, allows me to extend existing results that assume $b=1,\bm{B}_1=I_n$ and $K$ is known. I lastly place an assumption on the estimates from Algorithm~\ref{algorithm:EstC}.

\begin{assumption}
\label{assumption:FALCO}
Let $\alpha$ be as defined in the initialization of Algorithm \ref{algorithm:EstC}. Then $\alpha \in [c^{-1},1-c^{-1}]$ and the estimators for $\hat{\bar{\bm{v}}}$ from Steps \ref{item:EstC:0} and \ref{item:EstC:ImageC}\ref{item:EstC:V} in Algorithm \ref{algorithm:EstC} are such that
\begin{align*}
    \hat{\bar{\bm{v}}} \in \Theta_* = \Theta \cap \lbrace \bm{x} \in \bm{R}^b : \norm*{\bm{x}}_2 \leq 2bc, \sum\limits_{j=1}^b \bm{x}_j \bm{B}_j - (2c)^{-1} I_n \succ \bm{0} \rbrace.
\end{align*}
\end{assumption}
This technical condition makes the parameter space for $\bar{\bm{v}}$ compact, and is analogous to Assumption D in \cite{BaiLi} and Assumption 2 in \cite{CATE}.

\subsection{Properties of $\Koracle$ and $\Koraclehat$}
\label{subsection:Ktheory}
I first demonstrate the properties of $\Koracle$, as well as its estimator from Algorithm \ref{algorithm:CBCV}, $\Koraclehat$, in Theorem \ref{theorem:Khat} below, where I let $\gamma_0 = \infty$ and $\gamma_{K+1}=0$ for the remainder of Section \ref{section:Theory}.

\begin{theorem}
\label{theorem:Khat}
Suppose Assumptions \ref{assumption:CandL} and \ref{assumption:FALCO} hold such that $\Delta_{n,p} = O_P(d_{n,p})$ for some non-random sequence $d_{n,p} \to 0$ as $n,p \to \infty$. Assume the following hold for $F,\eta$ defined in Algorithm \ref{algorithm:CBCV}:
\begin{enumerate}[label=(\roman*)]
    \item The rows of $\bm{E}$ are independent, $n/p \to 0$ as $n,p \to \infty$ and $F \in [2,c], \eta\in[c^{-1},1-c^{-1}]$.\label{item:Khat:Enp}
    \item There exists an $s \in \{0\}\cup [K]$ such that $\gamma_{s+1} + c^{-1}\leq \delta^2 < \gamma_{s}$.\label{item:Khat:Crand}
\end{enumerate}
Fix any $\epsilon > 0$ and let $a_{n,p}=\max( n^{1/2}p^{-1/2},n^{-1/2},d_{n,p} )$. Then there exists a constant $m_{\epsilon} > 0$ that depends on $\epsilon$, but not $n$ or $p$, such that if $\delta^2 + a_{n,p} m_{\epsilon}\leq \gamma_s$ for all $n,p$ suitably large, $\liminf_{n,p \to \infty}\Prob\lbrace\Koraclehat=\Koracle = s\rbrace \geq 1-\epsilon$.
\end{theorem}

\begin{remark}
\label{remark:delta2}
The sequence $a_{n,p} \to 0$ provides insight into how much larger $\gamma_r$ must be than the noise level $\delta^2 \asymp 1$ to ensure both the Oracle and Algorithm \ref{algorithm:CBCV} select the $r$th factor. If $\Delta_{n,p}= o_P(1)$ and $\delta^2+ c^{-1} \leq \gamma_s$, then $\lim_{n,p \to \infty}\Prob\lbrace\Koraclehat=\Koracle = s\rbrace = 1$.  
\end{remark}


Theorem \ref{theorem:Khat} shows Algorithm \ref{algorithm:CBCV} tends to select the same number of factors as the Oracle, where both only include the factor $r \in [K]$ if its signal strength $\gamma_r$ is greater than the noise level $\delta^2$. This is congruent with the goals of the Oracle estimator established in Section \ref{subsection:OracleK}, which is designed to only return factors whose signal strengths are large enough to outweigh their estimation uncertainty. This is contrary to parallel and analysis \citep{PA_Dobriban} and estimators proposed in \citet{DPA}, which, besides only being applicable when $b=1$ and $\bm{B}_1=I_n$, ignore a factor's estimation uncertainty when selecting $K$. This could be why the latter's estimates for $K$ were exceedingly large in their data application.\par
\indent The condition that $n/p \to 0$ in \ref{item:Khat:Enp} is appropriate in genetic and epigenetic data, where $n \lesssim 10^2$ and $10^4\lesssim p \lesssim 10^6$. Independence between the rows of $\bm{E}$ is a standard assumption among methods with $b=1$, $\bm{B}_1=I_n$ and $\lambda_K \lesssim 1$ \cite{bcv,BCV_svd,PA_Dobriban,DPA}, and more generally, when $\gamma_K = o(n)$ \cite{Onatski_Corr}.\par 
\indent Theorem~\ref{theorem:Khat} is, as far as I am aware, the first result to establish the consistency of an estimate for the number of latent factors in dependent data with $\gamma_1 \lesssim n$ and $\gamma_K \asymp 1$. This is more than a mere technical triumph. For example, several popular estimators, like parallel analysis \cite{PA_Dobriban}, suffer from the problem of eigenvalue shadowing, in which factors with large eigenvalues prohibit the recovery of factors with moderate or small eigenvalues. Other methods, which do allow correlation between the entries of $\bm{E}$ \citep{BaiNg,Reviewer2_TimeOrdered,Reviewer2_ShrinkageFactors,AhnHorenstein,RandomJASA}, are so dependent on the assumption that $\gamma_K \asymp n$ that they too consistently fail to recover factors with moderate to small eigenvalues.

\subsection{The accuracy of the estimators from Algorithm \ref{algorithm:EstC}}
\label{subsection:AccuracyAlg1}
Here I give theoretical results regarding the accuracy of the estimators from Algorithm \ref{algorithm:EstC} assuming $\Koracle$ is known, along with theory that facilitates interpreting the latent factors $\bm{C}$. For remainder of Section \ref{section:Theory}, I let $K_{\max}$ and the iteration number $k \in \{0\}\cup [K_{\max}]$ be as defined in Algorithm \ref{algorithm:EstC}, and let $\lamoracle_r = \Lambda_r[ p^{-1}\Loracle \lbrace \Coracle \rbrace^{\T}\Coracle \lbrace \Loracle \rbrace^{\T} ]$ for all $r \in [\Koracle]$. I first state an assumption that I will utilize for the remainder of Section \ref{section:Theory}.

\begin{assumption}
\label{assumption:DependenceE}
\begin{enumerate}[label=(\alph*)]
    \item $p \geq c^{-1}n$, $\Koracle \geq 1$ is known,  $\gamma_{\Koracle}, (\gamma_{\Koracle}/\gamma_{\Koracle+1}-1) \geq c^{-1}$, $\gamma_{\Koracle+1} \leq c$ and $n/(p\gamma_{\Koracle}) \to 0$ as $n,p \to \infty$.\label{item:DependenceE:Gamma}
    \item One of the following holds:
    \begin{enumerate}[label=(\roman*)]
        \item There exists a non-random $\bm{A} \in \mathbb{R}^{pn \times pn}$ with $\norm*{\bm{A}}_2 \leq c$ such that $\vecM(\bm{E})\edist\bm{A}\vecM(\tilde{\bm{E}})$, where the entries of $\tilde{\bm{E}}$ are independent with $\E\lbrace \exp(t\tilde{\bm{E}}_{gi}) \rbrace \leq \exp(c t^2)$ for all $t \in \mathbb{R}$, $g \in [p]$ and $i \in [n]$.\label{item:assumErrors:Corr:U}
        \item The rows of $\bm{E}$ can be partitioned into sets with at most $c$ elements, such that $\bm{E}_{g \bigcdot}$ and $\bm{E}_{h \bigcdot}$ are independent if rows $g$ and $h$ are in different sets.\label{item:assumErrors:Corr:Networks}
    \end{enumerate}\label{item:DependenceE:DependenceE}
\end{enumerate}
\end{assumption}

Assumption \ref{assumption:DependenceE} is more general than the assumptions used to prove Theorem \ref{theorem:Khat}, where the assumptions on $\gamma_{\Koracle}$ in \ref{item:DependenceE:Gamma} mirror those placed on $\gamma_s$ in \ref{item:Khat:Crand} of Theorem \ref{theorem:Khat}. The typical dependence assumption $\bm{E}=\bm{R}_1\tilde{\bm{E}}\bm{R}_2^{\T}$ for $\bm{R}_1 \in \mathbb{R}^{p \times p}$ and $\bm{R}_2 \in \mathbb{R}^{n \times n}$ corresponds to $\bm{A} = \bm{R}_2 \otimes \bm{R}_1$ \cite{Limma,SimpleCorr}. Condition \ref{item:DependenceE:DependenceE}\ref{item:assumErrors:Corr:U} is more general than that considered in \citet{FanEigen}, which besides assuming $b=1,\bm{B}_1=I_n$ and $K$ was known, required $\bm{Y}=\bm{U}\bm{D}\tilde{\bm{E}}$ for some unitary matrix $\bm{U} \in \mathbb{R}^{p \times p}$ and diagonal matrix $\bm{D} \in \mathbb{R}^{p \times p}$. Condition \ref{item:DependenceE:DependenceE}\ref{item:assumErrors:Corr:Networks} assumes genomic units can be partitioned into non-overlapping networks, and is common in DNA methylation data \citep{DNAMethyRegions}.\par 
\indent I first show that the bias-corrected estimates $\hat{\lambda}_r$, defined in Algorithm \ref{algorithm:EstC}, accurately estimate $\lamoracle_r$.

\begin{theorem}
\label{theorem:Lambda}
Suppose Assumptions \ref{assumption:CandL}, \ref{assumption:FALCO} and \ref{assumption:DependenceE} hold and $K_{\max} \geq \Koracle$. Then for $k = \Koracle$,
\begin{align}
\label{equation:LambdaRate}
    \hat{\lambda}_r/\lamoracle_r = 1 + O_P\lbrace (\gamma_r p)^{-1/2} + n/(\gamma_r p) + (\gamma_r n)^{-1} \rbrace, \quad r \in [\Koracle].
\end{align}
\end{theorem}

\begin{remark}
\label{remark:FanRate}
As far as I am aware, with the exception of the $(\gamma_rn)^{-1}$ term, the rate of convergence in Theorem \ref{theorem:Lambda} is as fast as the best known rate for PCA when $\bm{Y}_{\bigcdot 1},\ldots,\bm{Y}_{\bigcdot n}$ are independent and identically distributed \citep{ChrisANDDan}.
\end{remark}

\begin{remark}
\label{remark:NaiveLambda}
When $k=\Koracle$, the estimator $\hat{\lambda}_r^{(\text{naive})} = \Lambda_r(p^{-1}\tilde{\bm{C}}\tilde{\bm{L}}^{\T}\tilde{\bm{L}}\tilde{\bm{C}}^{\T})$ that ignores the bias term $(n^{-1}\tilde{\bm{C}}^{\T}\hat{\bar{\bm{V}}}^{-1}\tilde{\bm{C}})^{-1}$ in Step \ref{item:EstC:eigen} of Algorithm \ref{algorithm:EstC} is inflated and behaves as
\begin{align*}
    \hat{\lambda}_r^{(\text{naive})}/\lamoracle_r \geq 1 + \tilde{c}/\lamoracle_r + O_P\lbrace (\gamma_r p)^{-1/2} + n/(\gamma_r p) + (\gamma_r n)^{-1} \rbrace, \quad r \in [\Koracle]
\end{align*}
for some constant $\tilde{c} > 0$. If $b=1,\bm{B}_1=I_n$, the inequality becomes an equality with $\tilde{c}=\bar{v}_1$ \citep{ChrisANDDan}.
\end{remark}

Theorem \ref{theorem:Lambda} and Remark \ref{remark:NaiveLambda} show that my bias-corrected estimator for $\lamoracle_r$ corrects eigenvalue inflation. This is relevant whenever $p >> n$ and $\lamoracle_r$ is moderate or small, which is typically the case in genetic and epigenetic data. I next demonstrate the properties of $\hat{\bm{L}}$. 


\begin{theorem}
\label{theorem:L}
Suppose the assumptions of Theorem \ref{theorem:Lambda} hold, fix any $\epsilon > 0$, let $r \in [\Koracle]$ and let $F_r^{(\epsilon)}$ be the event $\{\lamoracle_{r-1}/\lamoracle_r,\lamoracle_{r}/\lamoracle_{r+1} \geq 1+\epsilon\}$. Then for $k=\Koracle$, $a \in \{-1,1\}$ and if $\{\log(p)\}^2/n \to 0$,
\begin{align}
\label{equation:Linfty}
    \norm*{\hat{\bm{L}}_{\bigcdot r} - a\Loracle_{\bigcdot r}}_{\infty} = O_P\{ \log(p)n^{-1/2} + n^{1/2}(\gamma_{\Koracle}p)^{-1/2} \} \text{ on $F_r^{(\epsilon)}$}.
\end{align}
Further, if $\Koracle=K$, $n^{3/2}/\{p\gamma_{\Koracle}\} \to 0$ and the technical conditions in Section \ref{subsection:supp:TechL} in the Supplement hold,
\begin{align}
\label{equation:Lgls}
    [\lbrace(\hat{\bm{C}}^{\T}\hat{\bm{V}}_g^{-1}\hat{\bm{C}})^{-1}\rbrace_{rr}]^{-1/2}\lbrace \hat{\bm{L}}_{g r}^{(GLS)} - a \Loracle_{g r} \rbrace \edist Z + o_P(1), \quad g \in [p]
\end{align}
as $n,p \to \infty$, where $\hat{\bm{V}}_g$ is the restricted maximum likelihood estimate for $\bm{V}_g$ described in Remark \ref{remark:Alg1:Vg}, $\hat{\bm{L}}_{g \bigcdot}^{(GLS)}$ is the corresponding generalized least squares estimate for $\Loracle_{g \bigcdot}$ using the design matrix $\hat{\bm{C}}$ and $Z \sim N(0,1)$.
\end{theorem}


\begin{remark}
\label{remark:LInnerProduct}
I show in Section \ref{subsection:supp:techF} of the Supplement that $\Prob\{F_r^{(\epsilon)}\} \to 1$ as $n,p \to \infty$ under standard eigengap assumptions. The conditions that $\{\log(p)\}^2/n \to 0$ and $n^{3/2}/\{p\gamma_{\Koracle}\} \to 0$ are standard in genetic and epigenetic data \cite{CATE,ChrisANDDan}.
\end{remark}

\begin{remark}
\label{remark:PCA_CLT}
I show in Section \ref{section:supp:PopCov} of the Supplement that Theorem \ref{theorem:Lambda} can be leveraged to derive a central limit theorem for the eigenvalues of $\V(\bm{Y}_{\bigcdot i}) \in \mathbb{R}^{p \times p}$ if $a_{n,p}=n^{3/2}/(p\gamma_r)\to 0$, and that \eqref{equation:Linfty} holds with $\Loracle_{\bigcdot r}$ replaced with a scalar multiple of the $r$th eigenvector of $\V(\bm{Y}_{\bigcdot i})$. This significantly extends the eigenvalue and eigenvector convergence results in \citet{FanEigen}, which required $a_{n,p}(p^{1/2}n^{-1/2})\to 0$ and $\bm{U}\bm{Y}$ have independent sub-Gaussian entries for some unitary matrix $\bm{U}\in\mathbb{R}^{p \times p}$. As far as I am aware, this is the first result proving the asymptotic normality of eigenvalue estimates in high dimensional data with dependent observations.
\end{remark}

Both \eqref{equation:Linfty} and \eqref{equation:Lgls} are quite useful in practice and facilitate objective \ref{item:goal:Relation} from Section~\ref{section:Introduction}. The former implies a standard principal component plot of $\hat{\bm{L}}_{\bigcdot r_1}$ versus $\hat{\bm{L}}_{\bigcdot r_2}$ mirrors the information contained in a plot of $\Loracle_{\bigcdot r_1}$ versus $\Loracle_{\bigcdot r_2}$, and the latter justifies inference on the components of $\Loracle$. This is quite important, as practitioners are often interested in determining the genomic units whose expression or methylation depends on $\Coracle$ \cite{DieselParticles}. I lastly demonstrate the accuracy of my estimator for $\Coracle$.


\begin{theorem}
\label{theorem:AngleC}
Suppose the assumptions of Theorem \ref{theorem:Lambda} hold and let $k=\Koracle$. Then
\begin{align}
\label{equation:Subspace}
    \norm*{ P_{\Coracle} - P_{\hat{C}} }_F^2 = O_P\lbrace (\gamma_{\Koracle} p)^{-1/2} + n/(\gamma_{\Koracle} p) + (\gamma_{\Koracle} n)^{-1} \rbrace, \quad \abs{\hat{\bar{\bm{v}}}_j - \bar{v}_j} = O_P(n^{-1})
\end{align}
for all $j \in [b]$. Further, if $r,\epsilon$ and $F_r^{(\epsilon)}$ are as defined in Theorem \ref{theorem:L},
\begin{align}
\label{equation:CinnerChat}
    \abs*{ \hat{\bm{C}}_{ \bigcdot r}^{\T}\Coracle_{\bigcdot r} }/( \norm*{\hat{\bm{C}}_{ \bigcdot r}}_2 \norm*{\Coracle_{\bigcdot r}}_2 ) = 1- O_P\lbrace (\gamma_r p)^{-1/2} + n/(\gamma_r p) + (\gamma_r n)^{-1} \rbrace \text{ on $F_r^{(\epsilon)}$}.
\end{align}
\end{theorem}


Theorem \ref{theorem:AngleC} shows that Algorithm \ref{algorithm:EstC} effectively recovers $\im\{ \Coracle \}$, and my novel bias-corrected estimator for $\Coracle$ is just as accurate as the standard principal components estimator when $b=1$ and $\bm{B}=I_n$ \cite{ChrisANDDan}. Like Theorem \ref{theorem:L}, this implies that a plot of $\hat{\bm{C}}_{ \bigcdot r_1}$ versus $\hat{\bm{C}}_{ \bigcdot r_2}$ mirrors the information contained in the plot of $\Coracle_{ \bigcdot r_1}$ versus $\Coracle_{ \bigcdot r_2}$.

\subsection{Data denoising}
\label{subsection:Denoising}
In this section, I provide the requisite theory to guarantee that one can perform accurate inference conditional on my estimate for $\bm{C}$, which is often referred to as denoising the data matrix $\bm{Y}$ \cite{DPA}. This is critical when inferring eQTLs and meQTLs, where accounting for $\bm{C}$ has been shown to reduce potential confounding and empower inference \citep{ConfoundingeQTL,ConfoundingmeQTL}. It also has application in DNA methylation twin studies, in which one goal is to recover $\bm{V}_1,\ldots,\bm{V}_p$ to determine the latent cell type-independent heritability of DNA methylation \citep{TwinDNAmPlos,SwedishTwins,TwinDanish}. Theorem~\ref{theorem:eQTL} below, as far as I am aware, is the first result showing that denoising is possible in data with correlated samples. 

\begin{theorem}
\label{theorem:eQTL}
Suppose Assumptions \ref{assumption:CandL}, \ref{assumption:FALCO} and \ref{assumption:DependenceE} hold with $K=\Koracle$ and $n^{3/2}/(p\gamma_K) \to 0$ as $n,p\to\infty$. Fix a $g \in [p]$ and suppose for some non-random vector $\bm{s}_g \in \mathbb{R}^d$, $\bm{E}_{g\bigcdot} = \bm{X}_g \bm{s}_g + \bm{R}_{g}$, where $\bm{X}_g$ and $\bm{R}_g$ satisfy the following:
\begin{enumerate}[label=(\roman*)]
    \item $\bm{X}_g$ and $\bm{R}_g$ are independent, mean $\bm{0}$ and independent of $\bm{C}$, where $\bm{X}_g$ is observed and independent of all but at most $c$ rows of $\bm{E}$. Further, $d=O(1)$ and
    $\norm*{n^{-1}\bm{X}_g^{\T} \bm{X}_g-\bm{\Sigma}_g}_2 = o_P(1)$ for some non-random $\bm{\Sigma}_g \succ \bm{0}$ as $n \to \infty$.
    \item $\E\lbrace \exp(\bm{t}^{\T}\bm{e}) \rbrace \leq \exp(\norm*{\bm{t}}_2^2 c)$ for $\bm{e}\in\{\bm{X}_g\bm{s}_g,\bm{R}_g\}$, $\V(\bm{X}_g\bm{s}_g) =\bm{V}(\bm{\tau}_g)$ and $\V(\bm{R}_g) =\bm{V}(\bm{\alpha}_g)$ for some $\bm{\tau}_g,\bm{\alpha}_g \in \mathbb{R}^b$, where $\bm{\alpha}_g \in \Theta_*$.
\end{enumerate}
Define
\begin{align}
    \label{equation:Denoise:REML}
    \hat{\bm{\alpha}}_g &= \argmax_{\bm{\theta} \in \Theta_*}[ -\log\lbrace \abs*{P_{(\hat{\bm{C}}, \bm{X}_g)}^{\perp}\bm{V}(\bm{\theta})P_{(\hat{\bm{C}}, \bm{X}_g)}^{\perp}}_{+} \rbrace - \bm{Y}_{g \bigcdot}^{\T}\lbrace P_{(\hat{\bm{C}}, \bm{X}_g)}^{\perp}\bm{V}(\bm{\theta})P_{(\hat{\bm{C}}, \bm{X}_g)}^{\perp}\rbrace^{\dagger}\bm{Y}_{g \bigcdot} ]\\
    \label{equation:Denoise:sg}
    \hat{\bm{s}}_g &= [ \bm{X}_g^{\T}\lbrace P_{\hat{C}}^{\perp}\bm{V}(\hat{\bm{\alpha}}_g)P_{\hat{C}}^{\perp} \rbrace^{\dagger}\bm{X}_g ]^{-1} \bm{X}_g^{\T}\lbrace P_{\hat{C}}^{\perp}\bm{V}(\hat{\bm{\alpha}}_g)P_{\hat{C}}^{\perp} \rbrace^{\dagger} \bm{Y}_{g \bigcdot}
\end{align}
to be the restricted maximum likelihood estimator for $\bm{\alpha}_g$ and denoised estimate for $\bm{s}_g$. Then for $\hat{\bm{s}}_g^{(known)}$ the generalized least squares estimate for $\bm{s}_g$ from the regression of $\bm{X}_g$ onto $\bm{E}_{g \bigcdot}$ assuming $\bm{\alpha}_g$ is known, $\norm*{\hat{\bm{\alpha}}_g - \bm{\alpha}_g}_2=o_P(1)$ and
\begin{align}
\label{equation:Denoise}
    n^{1/2}\norm*{ \hat{\bm{s}}_g - \hat{\bm{s}}_g^{(known)} }_2 = o_P(1)
\end{align}
as $n,p \to \infty$.
\end{theorem}

\begin{remark}
\label{remark:Denoise:Vg}
By replacing $(\hat{\bm{C}}, \bm{X}_g)$ in \eqref{equation:Denoise:REML} with $\hat{\bm{C}}$, the proof of Theorem~\ref{theorem:eQTL} shows that $\norm{ \bm{V}(\hat{\bm{\alpha}}_g) - \bm{V}_g }_2 = o_P(1)$. This is useful in DNA methylation twin studies, where the goal is often to estimate the latent factor-adjusted heritability of DNA methylation \citep{TwinDNAmPlos,SwedishTwins,TwinDanish}.
\end{remark}

\begin{remark}
\label{remark:Denose:Xg}
In eQTL and meQTL studies, $\bm{X}_g$ is a function of the genotypes of the samples. I provide examples of how $\bm{X}_g$ is constructed in practice in Sections~\ref{section:Simulation}~and~\ref{section:RealData}.
\end{remark}

Equation \eqref{equation:Denoise} shows that inference with the denoised estimate for $\bm{s}_g$ is asymptotically equivalent to that when $\bm{L}\bm{C}^{\T}=\bm{0}$, which is critically important in eQTL and meQTL studies. For example, I show in Section~\ref{section:RealData} that Algorithm \ref{algorithm:EstC} and the results of Theorem~\ref{theorem:eQTL} can be used to perform inference to identify eQTLs that is far more powerful than existing methods.

\subsection{Characterizing the variation in $\bm{C}$}
\label{subsection:InferenceC}
Biologists routinely regress estimated latent factors onto observed technical and biological covariates to identify and characterize the most important sources of variation in $\bm{Y}$. Such inference is used to perform quality control \citep{Jessie,Marcus}, empower eQTL an meQTL detection algorithms \citep{Panama} and make biological conclusions \citep{DieselParticles,FactorCorr}. Theorem~\ref{theorem:XandC} below provides the first model-based framework and set of statistical guarantees aimed at characterizing the variation in $\bm{C}$ in dependent data.

\begin{theorem}
\label{theorem:XandC}
Let $\bm{X} \in \mathbb{R}^{n}$ be a random vector such that $n^{-1}\bm{X}^{\T}\bm{X} = \sigma_x^2 + O_P(n^{-1/2})$ for $\sigma_x^2=\E(n^{-1}\bm{X}^{\T}\bm{X})$. Suppose Assumptions \ref{assumption:CandL}, \ref{assumption:FALCO} and \ref{assumption:DependenceE} hold, $\Koracle=K$, $\Lambda_r(\bm{A})/\Lambda_{r+1}(\bm{A})\geq 1+c^{-1}$ for all $r \in [K]$, $\E(n^{-1}\bm{C}^{\T}\bm{C})=I_K$, $np^{-1}\bm{L}^{\T}\bm{L}$ is diagonal with decreasing diagonal elements and the following assumptions on $\bm{C}$ hold:
\begin{enumerate}[label=(\roman*)]
    \item $\bm{C} = \bm{X}\bm{\omega}^{\T} + \bm{R}$, where $\bm{\omega} \in \mathbb{R}^K$ is non-random, $\bm{R}$ is independent of $\bm{X}$ and $\E(\bm{R})=\bm{0}$.
    \item For $j\in[b]$, let $\bm{\Psi}_j \in \mathbb{R}^{K \times K}$ be a non-random, symmetric matrix such that $\norm*{\bm{\Psi}_j}_2 \leq c$. Then $\V\{ \vecM(\bm{R}) \} = \sum_{j=1}^b \bm{\Psi}_j \otimes \bm{B}_j \succeq c^{-1}I_n$.\label{item:Cinference:Var}
\end{enumerate}
Let $r \in [K]$ and $\hat{\bm{\omega}}_r$ be the generalized least squares estimate for $\bm{\omega}_r$ assuming the incorrect model $\hat{\bm{C}}_{\bigcdot r} \asim (\bm{X}\bm{\omega}_r,\bm{V}(\bm{\theta}))$ for some $\bm{\theta} \in \mathbb{R}^b$, where $\bm{\theta}$ is estimated via restricted maximum likelihood (REML). If $n^{3/2}/(p\gamma_r)\to 0$ as $n,p \to \infty$ and the regularity conditions in Section \ref{subsection:supp:TechC} of the Supplement hold, the following are true:
\begin{enumerate}[label=(\alph*)]
    \item If $\bm{X}$ is dependent on at most $c$ rows $\bm{E}$ and the null hypothesis $\bm{\omega}=\bm{0}$ holds, then for $\hat{\bm{\theta}}$ the REML estimate for $\bm{\theta}$ and $Z\sim N(0,1)$, $[\bm{X}^{\T}\{ \bm{V}(\hat{\bm{\theta}}) \}^{-1}\bm{X}]^{1/2}\hat{\bm{\omega}}_r \edist Z + o_P(1)$ as $n,p \to \infty$.\label{item:XandC:H0}
    \item If $\bm{X}$ is independent of $\bm{E}$, then $\hat{\bm{\omega}}_r = a\bm{\omega}_r + O_P(n^{-1/2})$ for $a \in \{-1,1\}$.\label{item:XandC:HA}
\end{enumerate}
\end{theorem}

\begin{remark}
\label{remark:CInference:Koracle}
The assumption $\Koracle=K$ is for simplicity of presentation. I state an equivalent version of Theorem~\ref{theorem:XandC} when $\Koracle \neq K$ in Section \ref{subsection:supp:RestateXC} of the Supplement. The assumptions on $\E(n^{-1}\bm{C}^{\T}\bm{C})=I_K$ and $np^{-1}\bm{L}^{\T}\bm{L}$ are without loss of generality in \ref{item:XandC:H0}, and are used to identify $\bm{\omega}$ in \ref{item:XandC:HA}. Note that under these assumptions, $np^{-1}\bm{L}_{\bigcdot r}^{\T}\bm{L}_{\bigcdot s}=\Lambda_r(\bm{A})I(r=s)$ for all $r,s \in [K]$.
\end{remark}

\begin{remark}
\label{remark:CInference:Var}
The model for $\V\{\vecM(\bm{R})\}$ assumes $\bm{B}_1,\ldots,\bm{B}_b$ parametrize the variance of linear combinations of the columns of $\bm{C}$. This is natural, since $\bm{B}_1,\ldots,\bm{B}_b$ are constructed to parametrize the dependence between samples.
\end{remark}


Item~\ref{item:XandC:HA} shows Algorithm~\ref{algorithm:EstC}'s estimators can be used to estimate the linear dependence between $\bm{C}$ and $\bm{X}$, where the conditions on $\E(n^{-1}\bm{C}^{\T}\bm{C})$ and $np^{-1}\bm{L}^{\T}\bm{L}$ help identify the columns of $\bm{C}$ and order them from most important to least important. Item \ref{item:XandC:H0} has many applications, but is particularly useful in eQTL studies. There, practitioners often attempt to account for the genetic relatedness between individuals when estimating $\bm{C}$, and subsequently test for associations between genotype $\bm{X}$ and latent factors $\bm{C}$ \cite{Panama,Knowles}. Loci whose genotypes are correlated with $\bm{C}$ might be indicative of systematic trans-eQTLs, and modifying the genetic relatedness matrix to account for the genotypes of such SNPs has been shown to increase the power to detect eQTLs \cite{Panama}.

\section{A simulation study}
\label{section:Simulation}

\subsection{Simulation setup}
\label{subsection:SimSetup} 
\indent I simulated the eQTL-dependent expression of $p=15000$ genes across three treatment conditions in $n/3=60$ unrelated individuals to compare Algorithms \ref{algorithm:EstC} and \ref{algorithm:CBCV} with other factor analysis procedures. To mirror the complexity of real data, I set $K = 35$ and generated 100 gene expression datasets according to Model \eqref{equation:ModelIntro}, where $\bm{E}$ was simulated according to Theorem~\ref{theorem:eQTL}:
\begin{equation}
\label{equation:SimParams}
\begin{aligned}
    \bm{L}_{gk} &\sim (1-\pi_k)\delta_0 + \pi_k N_1(0,\tau_k^2), \quad g \in [p];k \in [K]\\
    \bm{C} &\sim MN_{n \times K}(\bm{0},I_n,I_K)\\
    \bm{E}_{g\bigcdot} &= \bm{X}_g s_g + \bm{R}_g,\quad \bm{X}_g=\bm{G}_g\otimes \bm{1}_3 , \quad \bm{R}_g \sim N_{n}(\bm{0},I_{n/3} \otimes \bm{M}_g), \quad g \in [p]\\
    s_g &\sim 0.8\delta_0 + 0.2N_1(0,0.4^2), \quad g \in [p],
\end{aligned}
\end{equation}
where $\delta_0$ is the point mass at 0. The vector $\bm{G}_g \in \{0,1,2\}^{n/3}$ contains the genotypes at a single nucleotide polymorphism (SNP) that acts as an eQTL for gene $g$ if $s_g \neq 0$. The condition-specific intercepts $\bm{Z}=\bm{1}_{n/3}\otimes I_3$ were treated as observed nuisance covariates, and $\bm{M}_g \in \mathbb{R}^{3 \times 3}$ is the covariance, conditional on $\bm{C}$ and $\bm{X}_g$, of the expression of gene $g$ across treatment conditions, where $\abs*{P_Z^{\perp}\{I_{n/3}\otimes (p^{-1}\sum_{g=1}^p \bm{M}_g)\}P_Z^{\perp} }_{+} = 1$. As described in \eqref{equation:RorateOutZ}, I redefined $\bm{Y},\bm{C},\bm{X}_g,\bm{R}_g$ and $\bm{E}$ to be $\bm{Y}\bm{Q}_Z,\bm{Q}_Z^{\T}\bm{C},\bm{Q}_Z^{\T}\bm{X}_g,\bm{Q}_Z^{\T}\bm{R}_g$, and $\bm{E}\bm{Q}_Z$, respectively, prior to estimation and inference.\par 
\indent I set $\tau_k \in [0.12,1]$ and $\pi_k \in (0,1]$ so as to simulate data with strong, moderate and weak factors, where $\Koracle = 30$ in all simulations (Figure~\ref{Fig:Mg}). I then used genome-wide SNP data and gene annotations from 15000 randomly selected genes from the data example in Section~\ref{section:RealData} to simulate $\bm{G}_g$. In brief, I pruned SNPs for linkage disequilibrium, mapped SNPs to each gene's \textit{cis} region, defined as $\pm 10^6$ base pairs around its transcription start site \citep{GTEX}, and randomly chose one SNP within each \textit{cis} region to act as a potential eQTL for the corresponding gene. Genotypes $\bm{G}_g$ had independent entries, were independent of $\bm{R}_g$ and were simulated assuming Hardy-Weinberg Equilibrium with minor allele frequencies as estimated in Section~\ref{section:RealData}, where $\bm{G}_{g}=\bm{G}_h$ if genes $g\neq h$ had the same potential eQTL and $\bm{G}_{g} \indep \bm{G}_h$ otherwise. This implied that, on average, the expressions of 25\% of all genes with eQTLs were correlated with the expression of at least one other gene. Further, since $\V(\bm{X}_g) \propto \bm{Q}_Z^{\T}\{I_{n/3}\otimes (\bm{1}_3\bm{1}_3^{\T})\}\bm{Q}_Z$, $\V(\bm{E}_{g \bigcdot}) = \sum_{j=i}^6 v_{g,j}\bm{Q}_Z^{\T}(I_{n/3}\otimes \bm{A}_j)\bm{Q}_Z = \sum_{j=i}^6  v_{g,j}\bm{B}_{j}$ for some variance multipliers $v_{g,j}$ for all $g \in [p]$, where $\{\bm{A}_1,\ldots,\bm{A}_6\} \subset \{0,1\}^{3 \times 3}$ is a basis for the space of $3 \times 3$ symmetric matrices. Therefore, $\bm{V}_g=\V(\bm{E}_{g \bigcdot})$ follows Model \eqref{equation:Vmodel} with $b=6$.\par
\indent To reflect the complex gene-specific correlation structures observed in practice, $\bm{M}_g$ was simulated such that each condition had a different marginal variance and, as shown in Figure~\ref{Fig:Mg}, the three pairs of conditions had different correlation coefficients. Given only the expression matrix $\bm{Y}$, the first goal was to estimate $\Coracle$ and $\lamoracle_k$, which facilitate the characterization and prioritization of latent sources of variation and is a critical step in multi-condition studies \cite{FactorCorr,DieselParticles,Knowles,Marcus}. The second goal was to leverage these estimates to identify eQTLs by performing inference on $s_g$. Section~\ref{section:supp:Simulation} of the Supplement contains additional simulation details.

\begin{figure}
    \centering
    \includegraphics[height=2.75in]{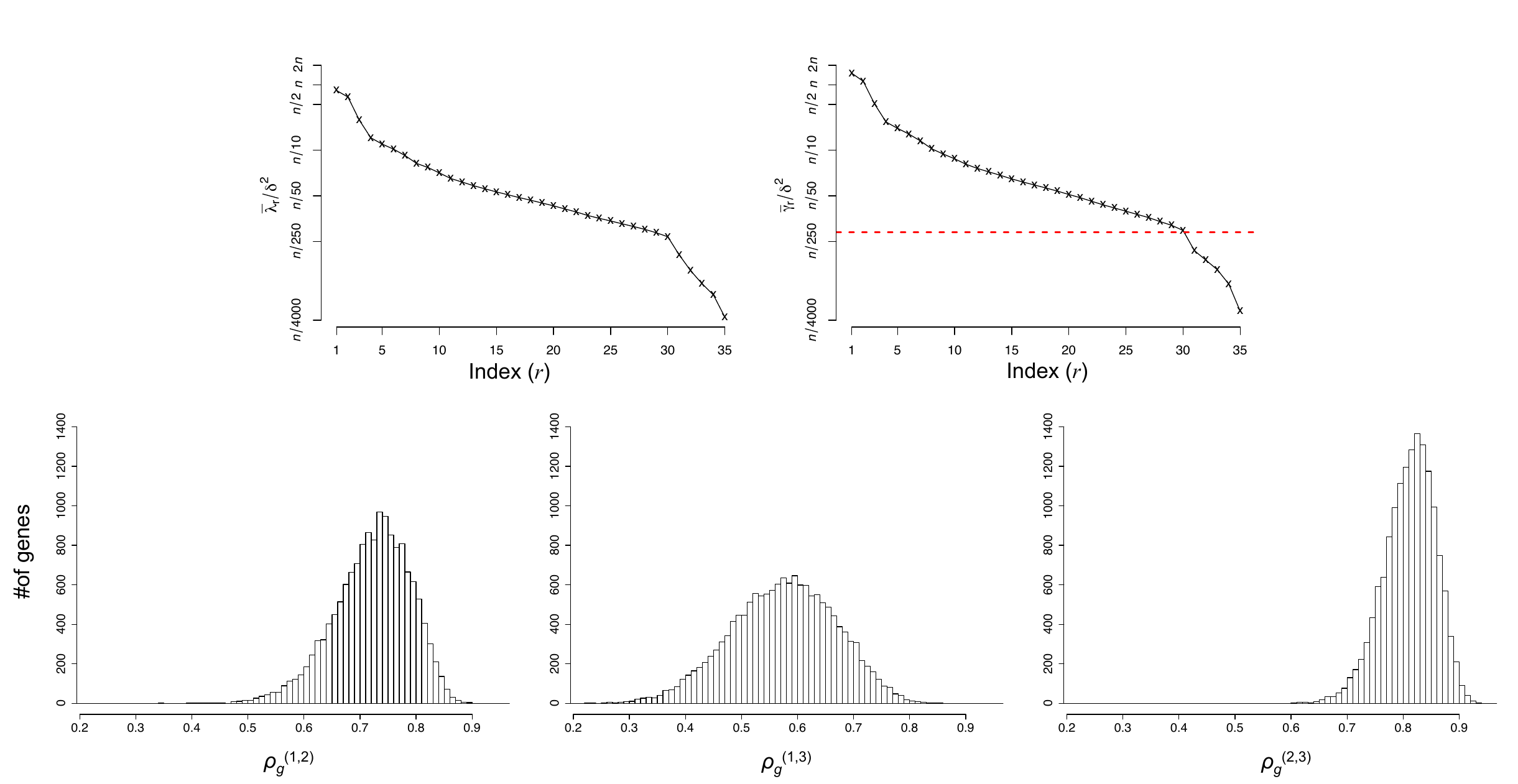}
    \caption{The average simulated $\lambda_r/\delta^2$ (top left) and $\gamma_r/\delta^2$ (top right) and the off-diagonal elements of the correlation matrices corresponding to $\bm{M}_1,\ldots,\bm{M}_p$ in one simulated dataset (bottom), where $\rho_g^{(i,j)} = \text{Corr}(\bm{R}_{g_i},\bm{R}_{g_j})$ for $i\neq j\in[3]$ and $g \in [p]$. The dashed red line is the line $y=1$.}
    \label{Fig:Mg}
\end{figure}

\subsection{Simulation results}
\label{subsection:SimResults}
\indent I first evaluated Algorithm \ref{algorithm:EstC}'s ability to recover $\lamoracle_1,\ldots,\lamoracle_{\Koracle}$ and $\Coracle$ assuming $\Koracle$ was known by comparing it to the most commonly used method to perform factor analysis in dependent biological data, PCA \citep{FactorCorr,DieselParticles,Knowles,Marcus}. The results are given in Figure~\ref{Fig:FALCO}, where the empirical factor and subspace correlations are $\vert \hat{\bm{A}}_{\bigcdot r}^{\T}\bm{A}_{\bigcdot r} \vert/(\Vert \hat{\bm{A}}_{\bigcdot r}\Vert_2 \Vert \bm{A}_{\bigcdot r}\Vert_2 )$ and $\mathop{\min}_{\bm{v}\in \im(\hat{\bm{A}})\setminus \{\bm{0}\}} \mathop{\max}_{\bm{u}\in \im(\bm{A})\setminus \{\bm{0}\}} \vert\bm{v}^{\T}\bm{u}\vert/(\Vert\bm{v}\Vert_2\Vert\bm{u}\Vert_2 )$, where $\bm{A}=\Coracle,\hat{\bm{A}}=\hat{\bm{C}}$ for FALCO and $\bm{A},\hat{\bm{A}}$ are the first $30$ right singular vectors of $\bm{Y},\bm{L}\bm{C}^{\T}$ for PCA. These demonstrate the fidelity of Algorithm \ref{algorithm:EstC}'s bias-corrected estimates for $\lamoracle_1,\ldots,\lamoracle_{\Koracle}$ and $\Coracle$ and clearly indicate that Algorithm \ref{algorithm:EstC} outperforms standard PCA. As discussed in Section \ref{subsection:FALCO}, PCA's poor performance can be attributed to the fact that the dependence between the columns of $\bm{E}$ precludes it from recovering factors with moderate to small eigenvalues.\par 

\begin{figure}[t!]
    \centering
    \includegraphics[height=1.8in]{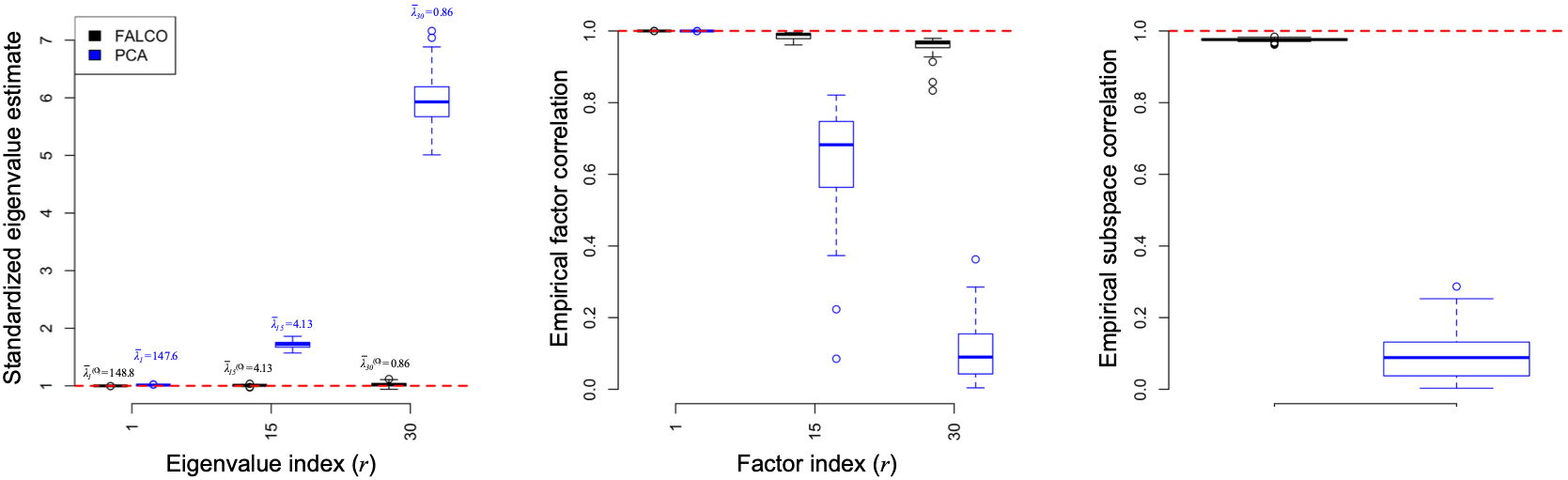}
    \caption{A comparison of Algorithm~\ref{algorithm:EstC} (FALCO) and PCA assuming $\Koracle=30$ was known, where $\bar{\lambda}_r^{\text{(O)}}$ and $\bar{\lambda}_r$ are the average simulated $\protect\lamoracle_r$ and $\lambda_r$. The standardized eigenvalue estimates for FALCO and PCA are $\hat{\lambda}_r/\protect\lamoracle_r$ and $\Lambda_r(p^{-1}\bm{Y}^{\T}\bm{Y})/\lambda_r$. Eigenvalues for points in the middle figure satisfy $\protect\lamoracle_{r-1}/\protect\lamoracle_r,\protect\lamoracle_{r}/\protect\lamoracle_{r+1}\geq 1.1$ or $\lambda_{r-1}/\lambda_r,\lambda_{r}/\lambda_{r+1}\geq 1.1$ for $\protect\lamoracle_0=\lambda_0=\infty$ and $\protect\lamoracle_{31}=0$.}
    \label{Fig:FALCO}
\end{figure}

\indent Next, I assessed my method's capacity to denoise $\bm{Y}$ and discover eQTLs by evaluating its power to identify genes $g$ with $s_g \neq 0$ when $K$ and $\Koracle$ were unknown. I compared my method to that routinely used to denoise data in dependent biological data, namely using one of the methods proposed in \citet{BaiNg} (BN), \citet{AhnHorenstein} (AH), \citet{Onatski_Corr} (ED), \citet{bcv} (BCV) or \citet{PA_Dobriban} (PA) to estimate $K$, and subsequently estimating $\bm{C}$ with PCA. Results were nearly identical when I replaced PCA with methods that attempt to account for heterogeneity across genes, like maximum quasi-likelihood \citep{BaiLi} or the algorithm proposed in \citet{bcv}. To make computation tractable and to be consistent with current practice, I estimated $\V(\bm{R}_g)$ via restricted maximum likelihood with each method's estimate for $\bm{C}$, $\hat{\bm{C}}$, by assuming $\V(\bm{R}_g)=\sigma_g^2\sum_{j=1}^6 \phi_j \bm{B}_j$, $\bm{Y}_{g \bigcdot} \sim N_n(\hat{\bm{C}}\bm{L}_{g \bigcdot},\V(\bm{R}_g))$ and $\bm{Y}_{g \bigcdot} \indep \bm{Y}_{h \bigcdot}$ for $g \neq h \in [p]$. I then estimated $s_g$ with each method via generalized least squares using the design matrix $[\bm{X}_g\, \hat{\bm{C}}]$, computed \textit{P} values with the normal approximation and used the Benjamini-Hochberg procedure \citep{BH} to control the false discovery rate.\par
\indent Figure~\ref{Fig:eQTLSimulations} contains the results. The fact that Algorithm~\ref{algorithm:CBCV} consistently estimates $\Koracle$ suggests Algorithm~\ref{algorithm:CBCV} is robust to dependencies across genomic units commonly observed in genetic and epigenetic data. The gain in power using my proposed denoised estimate for $s_g$ illustrates the importance of accounting for dependencies between $\bm{E}_{\bigcdot 1},\ldots,\bm{E}_{\bigcdot n}$ when estimating $\bm{C}$. A brief discussion of each competing method is given below.

\begin{itemize}[leftmargin=*]
    \item \textit{BN, AH, ED}: The theoretical arguments used in \citet{BaiNg,AhnHorenstein,Onatski_Corr} to prove the consistency of their estimates for $K$ and subsequent fidelity of PCA's estimate for $\bm{C}$ allow for general dependence between the entries of $\bm{E}$. However, they consistently underestimate $K$ because their theoretical arguments and estimators rely on the assumption that $\lambda_K \asymp n$ \citep{BaiNg,AhnHorenstein} or $\bm{V}_1=\cdots=\bm{V}_p$ and $\lambda_K \to \infty$ \citep{Onatski_Corr}. $\text{BN}_{IC}$ and $\text{BN}_{PC}$ in Figure~\ref{Fig:eQTLSimulations} refer to the $IC$ and $PC$ estimators defined \citet{BaiNg}.
    \item \textit{BCV}: This allows $\lambda_1\asymp n$ and $\lambda_K \lesssim 1$, but requires the entries of $\bm{E}$ be independent. When applied to the full data matrix $\bm{Y}$, denoted as $\text{BCV}_{\text{full}}$ in Figure~\ref{Fig:eQTLSimulations}, it severely overestimates $K$ because it attributes dependencies between $\bm{E}_{\bigcdot 1},\ldots,\bm{E}_{\bigcdot n}$ as arising from $\bm{L}\bm{C}^{\T}$. To circumvent this problem, I adopted a common strategy and let $\text{BCV}_{\text{ind}}$ be the estimator that applies BCV to each of the three conditions separately \citep{GTEX,StephensMultiTissue}, where an accurate estimate for $K$ would now be $\approx 3\times 35=105$. However, this effectively reduces the sample size by 67\%, which causes $\text{BCV}_{\text{ind}}$ to underestimate $K$.
    \item \textit{PA (Parallel Analysis)}: This method and BCV rely on similar assumptions, except it requires $\lambda_1 = o(n)$ and suffers from the ``eigenvalue shadowing'' problem in which factors with large eigenvalues preclude it from recovering those with moderate to small eigenvalues \citep{PA_Dobriban}. This explains why $\text{PA}_{\text{full}}$'s and $\text{PA}_{\text{ind}}$'s, the analogues of $\text{BCV}_{\text{full}}$ and $\text{BCV}_{\text{ind}}$, estimates for $K$ in Figure~\ref{Fig:eQTLSimulations} are smaller than $\text{BCV}_{\text{full}}$'s and $\text{BCV}_{\text{ind}}$'s.
\end{itemize}



\begin{figure}[t!]
\centering
\includegraphics[height=2.25in]{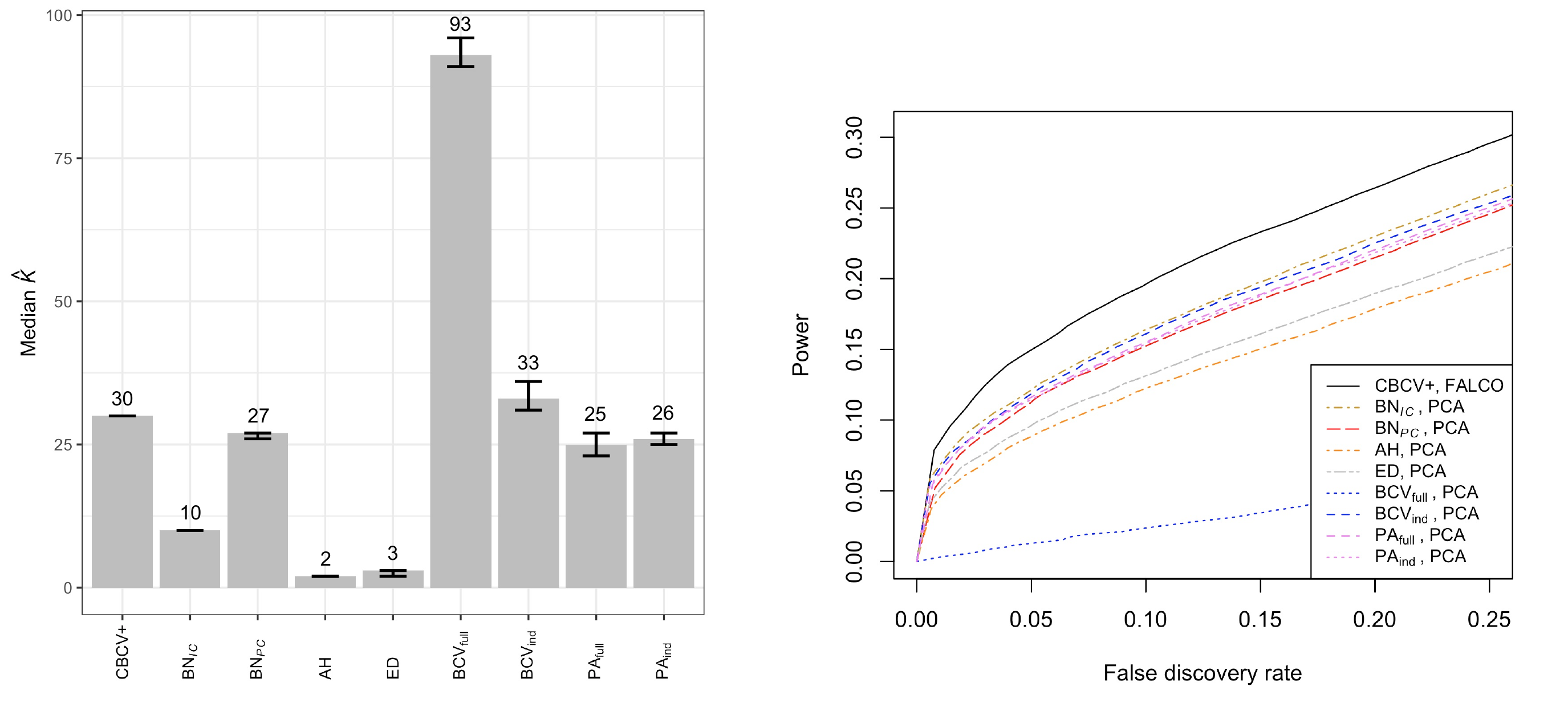}
\caption{Estimates for $K$ (left) and each method's power to identify non-zero $s_g$ (right), where error bars give the first and third quartile and the numbers above each bar denote each method's median estimate. ``CBCV+'' refers to Algorithm~\ref{algorithm:CBCV}, and my method, ``CBCV+, FALCO'', uses Algorithm~\ref{algorithm:CBCV} to estimate $\Koracle=30$ and  Algorithm~\ref{algorithm:EstC} to estimate $\Coracle$. At a 5\% false discovery rate, my method identified 25\% more eQTLs than the next most powerful method.} \label{Fig:eQTLSimulations}
\end{figure}

\section{Data application}
\label{section:RealData}
\indent I analyzed data from \citet{Knowles} to illustrate the power of Algorithm's~\ref{algorithm:EstC} and \ref{algorithm:CBCV} when applied to modern genetic data with dependent samples and large, moderate and small eigenvalues $\lambda_1,\ldots,\lambda_K$. As shown in Figure~\ref{Fig:Knowles:Experiment}, \citeauthor{Knowles} measured the expression of $p=12317$ genes in cardiomyocytes procured from 45 individuals, where each individual's cardiomyocytes were treated \textit{in vitro} with five dosages of the chemotherapeutic agent doxorubicin ($n = 45 \times 5$). The genotypes at $\approx 3\times 10^6$ SNPs were also collected for each individual. This non-trivial experimental design, coupled with the fact that, as shown in Figure~\ref{Fig:Knowles:Eig}, $\lambda_1,\ldots,\lambda_K$ appear to span several orders of magnitude, suggests existing methods are not equipped to perform factor analysis on these data.\par 
\indent One of the goals of this experiment was to identify eGenes, defined as genes whose expression under these conditions was regulated by at least one eQTL in the gene's \textit{cis} region. To do so, Theorem~\ref{theorem:eQTL} and the simulations in Section~\ref{section:Simulation} suggest one can empower such inference by estimating $\bm{C}$ and denoising the expression matrix. I therefore modeled $\bm{Y}$ as
\begin{align}
\label{equation:KnowleseQTL}
    \bm{Y}_{gi} = \bm{\Gamma}_{g \bigcdot}^{\T}\bm{Z}_{i \bigcdot} + \bm{L}_{g \bigcdot}^{\T}\bm{C}_{i \bigcdot}+\bm{E}_{gi}, \quad \bm{E}_{gi} = x_{g,m(i)} s_{g,d(i)} + \bm{R}_{gi}, \quad g\in[p]; i\in[n],
\end{align}
where $\bm{Z} \in \mathbb{R}^{n \times 5}$ contains the dose level-specific intercepts, $x_{g,m} \in \{0,1,2\}$ is individual $m$'s genotype at gene $g$'s potential eQTL, $s_{g,d} \in \mathbb{R}$ is the potential eQTL's effect for dose level $d \in [5]$ and $m(i)$ and $d(i)$ are the individual and dose level for sample $i$. Since $\bm{\Gamma}$ was of little interest in \citet{Knowles}, $\bm{Z}$ was treated as a nuisance covariate. While individuals were sampled from a founder population, I found no relationship between $\bm{Y}$ and the known kinship matrix. Therefore, I assumed $\bm{E}_{gi} \edist \bm{E}_{gi'}$ for $d(i)=d(i')$ and $\bm{E}_{gi} \indep \bm{E}_{gi'}$ for $m(i) \neq m(i')$, meaning the covariance of $\bm{e}_{g,m'} = (\bm{E}_{gi})_{\{i\in[n]:m(i) = m'\}} \in \mathbb{R}^5$ completely described $\bm{V}_g=\V(\bm{E}_{g \bigcdot})$. Initial data exploration then revealed that a suitable model for $\V(\bm{e}_{g,m'})$ was $\V(\bm{e}_{g,m'}) = \alpha_g^2 \bm{1}_5 \bm{1}_5^{\T} + \sum_{d=1}^5 \phi_{g,d}^2 \bm{a}_d \bm{a}_d^{\T}$ for all $g \in [p]$ and $m'\in[45]$, where $\bm{a}_d \in \{0,1\}^5$ is 1 in the $d$th coordinate and 0 everywhere else, meaning $\bm{V}_g$ followed \eqref{equation:Vmodel} with $b=6$.\par 
\indent I first used my method and each competing method described in Section~\ref{section:Simulation} to estimate $\bm{C}$ and, to investigate the latent variation explained by each method, the resulting mean marginal variance $\bar{\sigma}_d^2=p^{-1}\sum_{g=1}^p (\alpha_g^2 + \phi_{g,d}^2)$ for each dose level $d \in [5]$. Figure~\ref{Fig:Knowles:Var} contains the results, where the methods AH, $\text{BCV}_{\text{full}}$ and $\text{PA}_{\text{full}}$ estimated $K$ to be 2, 89 and 21, and were excluded because they were outperformed by ED, $\text{BCV}_{\text{ind}}$ and $\text{PA}_{\text{ind}}$, respectively, in all comparisons. First, while CBCV+ is nominally a stochastic algorithm, there was no variation in its estimate. This is contrary to BCV, whose stochasticity gives rise to a highly variable estimator. Second, and perhaps most interestingly, is that my method's estimates for $\bar{\sigma}_d^2$ are the only estimates that are strictly increasing in administered doxorubicin dose. While not explored in \citet{Knowles}, this is consistent with the observation that doxorubicin disrupts cardiomyocyte homeostasis in an individual- and dose-specific manner \citep{Doxtox}.\par 
\indent I next evaluated each method's ability to denoise $\bm{Y}$ and identify eGenes. I first pruned SNPs for linkage disequilibrium and mapped SNPs with minor allele frequencies $\geq 5\%$ to each gene's \textit{cis} region. I used \eqref{equation:KnowleseQTL} to model expression for each gene-SNP pair, where like Section~\ref{subsection:SimResults}, I estimated $\V(\bm{R}_{g \bigcdot})$ using restricted maximum likelihood with each method's estimate for $\bm{C}$, $\hat{\bm{C}}$, assuming $\V(\bm{R}_{g \bigcdot}) = \sigma_g^2(\bar{\alpha}^2\bm{B}_{\alpha} + \sum_{d=1}^5 \bar{\phi}_d^2 \bm{B}_{d})$, $\bm{Y}_{g\bigcdot}\sim N_n(\bm{Z}\bm{\Gamma}_{g \bigcdot}+\hat{\bm{C}}\bm{L}_{g\bigcdot},\V(\bm{R}_{g \bigcdot}))$ and $\bm{Y}_{g\bigcdot} \indep \bm{Y}_{h\bigcdot}$ for $g \neq h \in [p]$. I computed \textit{P} values for the null hypotheses $H_0: s_{g,1}=\cdots=s_{g,5}=0$ using the normal approximation and used TreeQTL \citep{TreeQTL} to identify eGenes at a 5\% false discovery rate. Lastly, I examined the overlap between each method's reported eGenes and those identified in heart tissues in GTEx, a comprehensive public resource containing tissue-specific eQTLs \citep{GTEX}, to assess the veracity of each method's findings. The results are presented in Figure~\ref{Fig:Knowles:eQTL}.\par
\indent While the fraction of eGenes identified by each method that overlap with GTEx-identified eGenes is relatively consistent across methods, my method identifies 48\% more eGenes than the next most powerful method. Further, over 70\% of the eGenes identified by the next three most powerful methods were also identified using my method. Like the simulation results from Section~\ref{subsection:SimResults}, this suggests my method's denoised estimates are far more powerful than those from existing methods, and highlights the importance of recovering factors with ostensibly moderate or weak signal strengths. While my and \citeauthor{Knowles}'s results are not directly comparable because the latter ignored the heterogeneity in dose-specific variances, it is worth noting that I identify over 20\% more eGenes than \citeauthor{Knowles}, who chose $K$ to maximize the number of detected eGenes.

\begin{figure}[t!]
\centering
\subfigure[][]{
\label{Fig:Knowles:Experiment}
\includegraphics[height=1.5in]{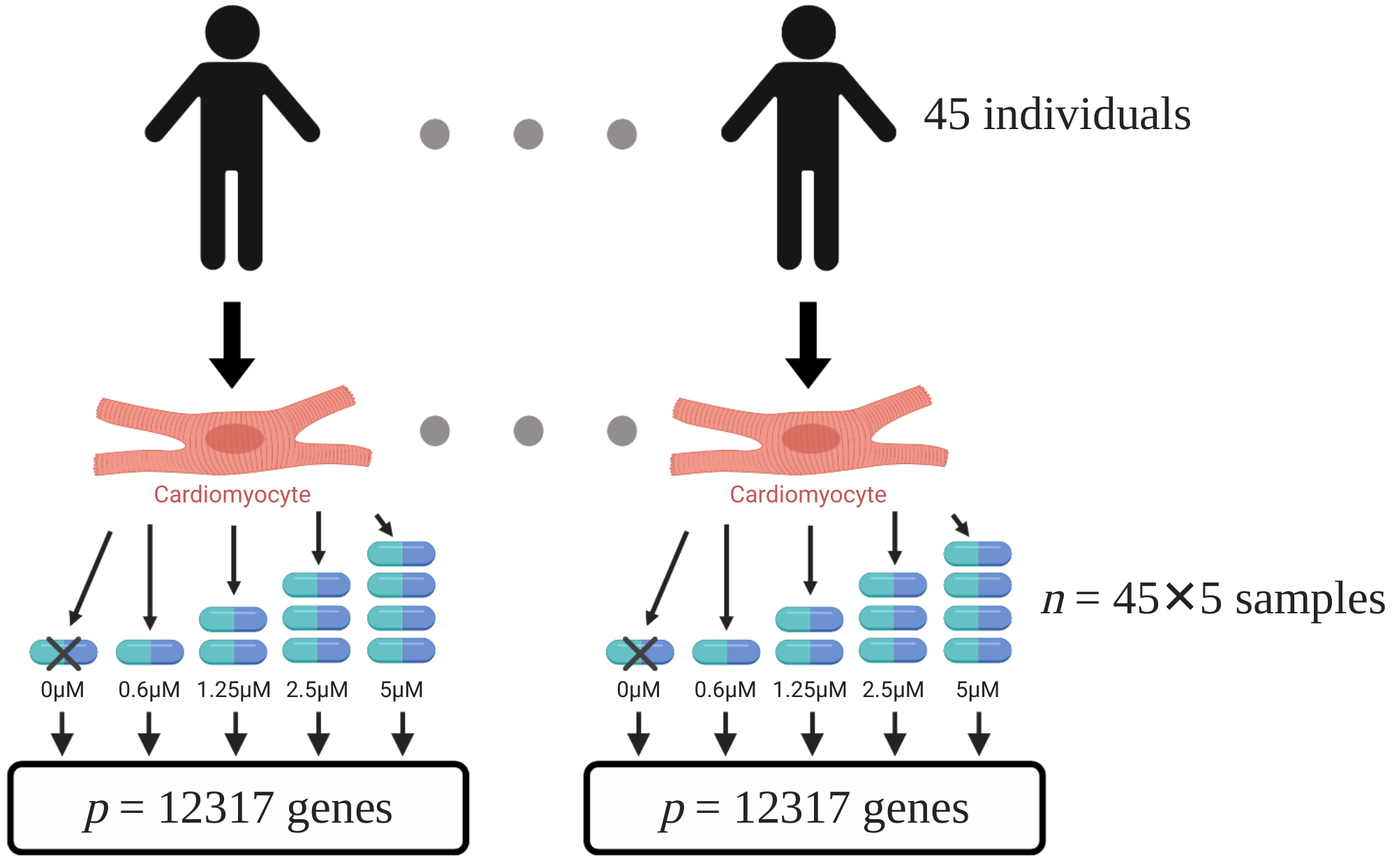}}
\hspace{8mm}
\subfigure[][]{
\label{Fig:Knowles:Eig}
\includegraphics[height=1.5in]{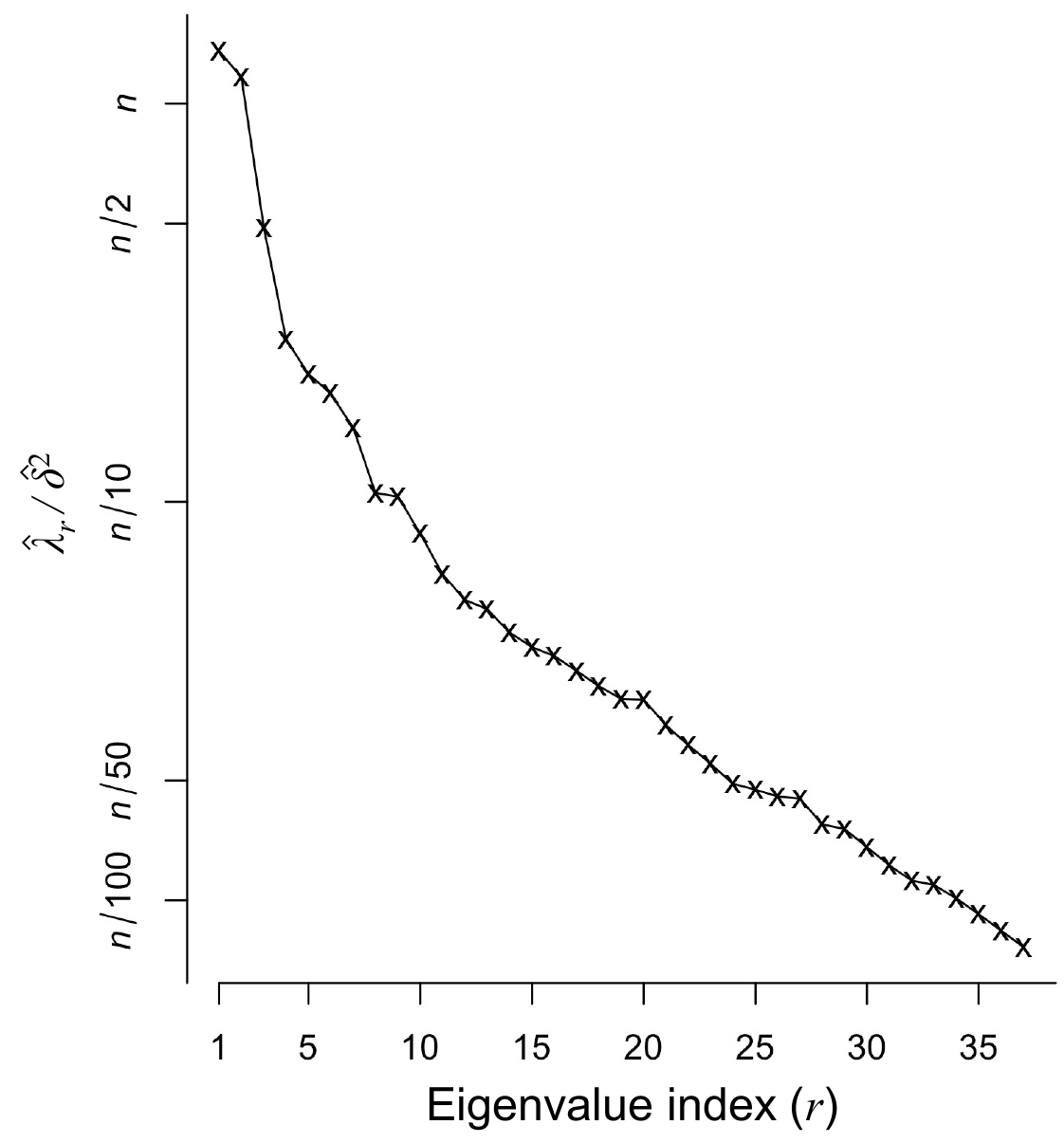}}\\
\subfigure[][]{
\label{Fig:Knowles:Var}
\includegraphics[height=1.5in]{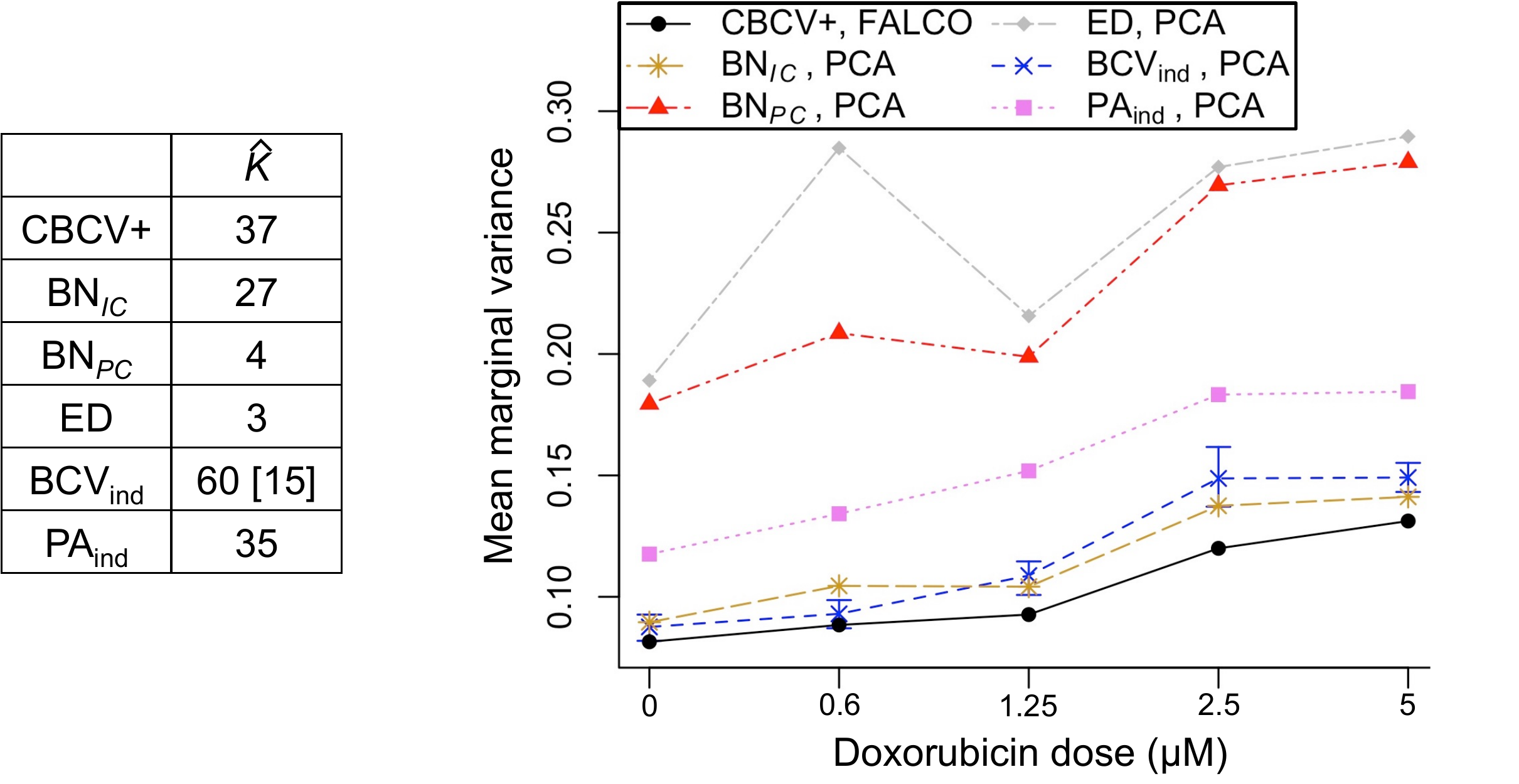}}
\subfigure[][]{
\label{Fig:Knowles:eQTL}
\includegraphics[height=1.5in]{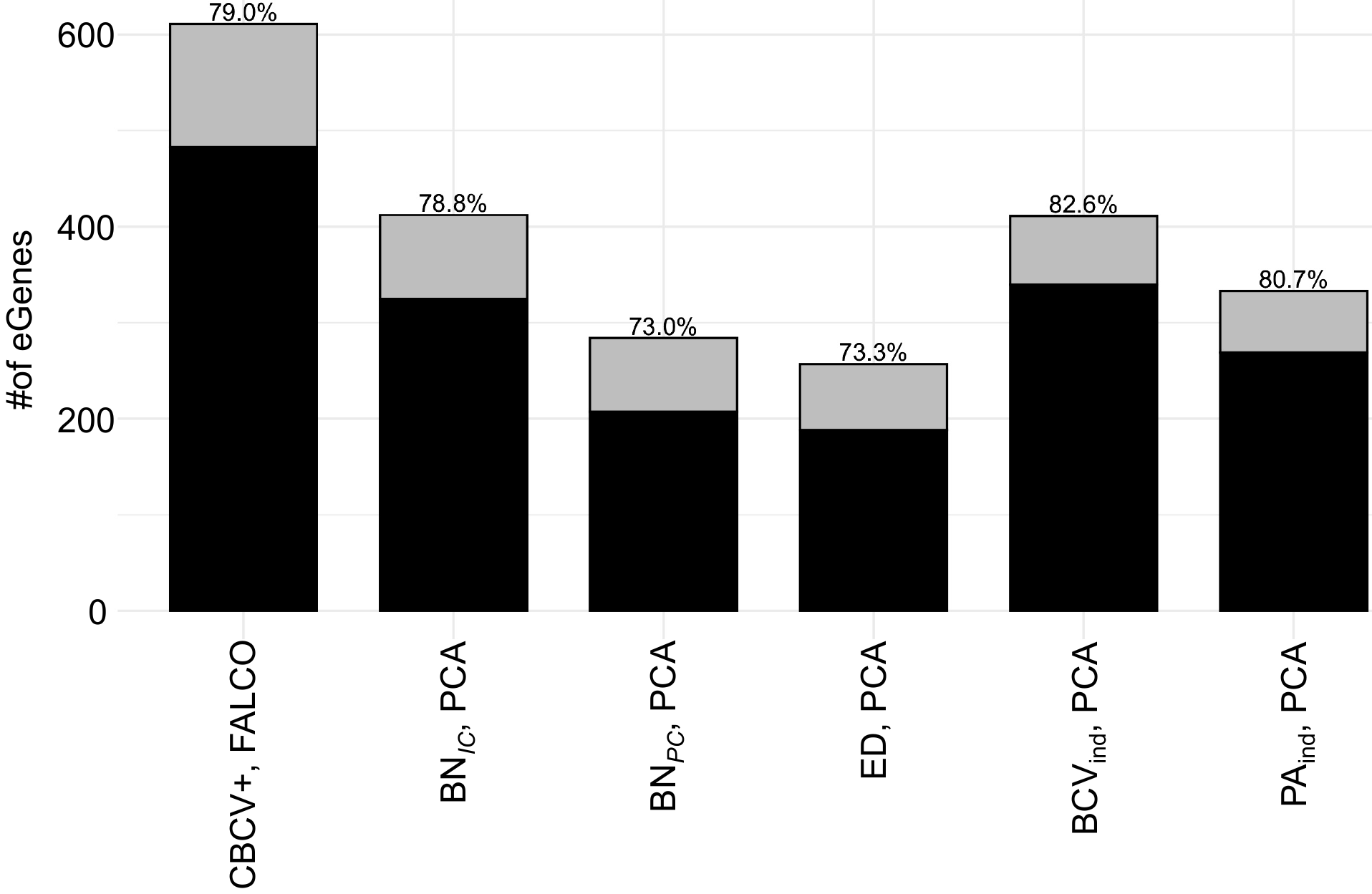}}
\caption{(a): Experimental design from \citet{Knowles}. (b): Algorithm~\ref{algorithm:EstC}-derived estimates for $\lamoracle_r$. (c): Estimates for $K$ or $\Koracle$ and the resulting estimated mean marginal variances. Numbers in square brackets and the error bars give the interquartile range and first or third quartiles, respectively, for variable stochastic algorithms. (d): Number of significant eGenes identified using each method that are also (black) and are not (grey) eGenes in heart tissues in GTEx, where the number above each bar gives the fraction that are also eGenes in heart tissues in GTEx; enrichment \textit{P} values for each method were $1.5\times 10^{-24}, 1.3\times 10^{-16},3.8\times 10^{-8},1.2\times 10^{-7},3.3\times 10^{-20},3.2\times 10^{-15}$, respectively. $\text{BCV}_{\text{ind}}$ was applied with $\hat{K}=60$, its median estimate.} 
\end{figure}

\section{Discussion}
\label{section:Discussion}
In this work, I developed a novel framework and new, provably accurate methodology to perform factor analysis, interpret its results and utilize its estimates in modern high throughput biological data with non-trivial dependence structures and factors whose signal strengths span several orders of magnitude. I also showed that my estimate for $K$ circumvents the ill-reputed ``eigenvalue shadowing'' problem, as well the biases that accompany the ``pervasive factor'' assumption. I lastly used simulated and real genetic data to illustrate the power of my methodology in application.\par 
\indent My results and those from the existing literature suggest there is a trade-off between two critical assumptions in high dimensional factor analysis: either allow the columns of $\bm{E}$ to have unknown dependence structure but require $\lambda_K \to \infty$, or allow $\lambda_K \lesssim 1$ but assume the practitioner knows $\bm{B}_1,\ldots,\bm{B}_b$. While it can be argued how relevant such prior knowledge is in other disciplines, biological practitioners have intimate knowledge of the experimental design and data collection process, and therefore will likely know $\bm{B}_1,\ldots,\bm{B}_b$. Given the results from Section~\ref{section:RealData}, this suggests biological statisticians should worry less about developing methodology that satisfies the aesthetically pleasing assumption that the entries of $\bm{E}$ have arbitrary dependence, and focus more on methodology that can accommodate data with strong, moderate and weak factors.

\section*{Acknowledgements}
I thank Carole Ober for providing the genetic data from Section~\ref{section:RealData}, which motivated this research. I also thank Dan Nicolae for his useful comments and suggestions that have substantially improved this manuscript. This work is supported in part by NIH grant R01 HL129735.







\newpage
\printbibliography 



\setcounter{equation}{0}
\setcounter{theorem}{0}
\setcounter{assumption}{0}
\setcounter{lemma}{0}
\setcounter{corollary}{0}
\setcounter{table}{0}
\setcounter{section}{0}
\setcounter{remark}{0}
\renewcommand{\theequation}{S\arabic{equation}}
\renewcommand{\thesection}{S\arabic{section}}
\renewcommand{\thetheorem}{S\arabic{theorem}}
\renewcommand{\theassumption}{S\arabic{assumption}}
\renewcommand{\thelemma}{S\arabic{lemma}}
\renewcommand{\thecorollary}{S\arabic{corollary}}
\renewcommand{\theproposition}{S\arabic{proposition}}
\renewcommand{\theremark}{S\arabic{remark}}

\newpage

\begin{center}
    {\Large Supplementary material for ``Factor analysis in high dimensional biological data with dependent observations''}
\end{center}
\vspace{4mm}

\section{Additional simulation details}
\label{section:supp:Simulation}
Here I provide the values for $\pi_k$ and $\tau_k^2$, defined in \eqref{equation:SimParams}, used to simulate $\bm{L}$.

\begin{minipage}{1\linewidth}
\small
\begin{tabular}{c|c|c|c|c|c|c|c|c|c|c|c|c}
$k$ & 1 & 2 & 3 & 4 & 5 & 6 & 7 & 8 & 9 & 10 & 11 & 12\\\hline
$\pi_k$ & 1 & 1 & 1 & 1 & 1 & 1 & 1 & 1 & 1 & 1 & 1 & 1\\\hline
$\tau_k$ & 0.89 & 0.83 & 0.53 & 0.39 & 0.35 & 0.33 & 0.29 & 0.24 & 0.24 & 0.22 & 0.19 & 0.18
\end{tabular}
\end{minipage}

\vspace{4mm}

\begin{minipage}{1\linewidth}
\small
\begin{tabular}{c|c|c|c|c|c|c|c|c|c|c|c|c}
$k$ & 13 & 14 & 15 & 16 & 17 & 18 & 19 & 20 & 21 & 22 & 23 & 24\\\hline
$\pi_k$ & 1 & 1 & 1 & 1 & 1 & 1 & 1 & 1 & 1 & 0.95 & 0.78 & 0.72\\\hline
$\tau_k$ & 0.18 & 0.17 & 0.16 & 0.15 & 0.14 & 0.14 & 0.13 & 0.13 & 0.12 & 0.12 & 0.12 & 0.12
\end{tabular}
\end{minipage}

\vspace{4mm}

\begin{minipage}{1\linewidth}
\small
\begin{tabular}{c|c|c|c|c|c|c|c|c|c|c|c}
$k$ & 25 & 26 & 27 & 28 & 29 & 30 & 31 & 32 & 33 & 34 & 35\\\hline
$\pi_k$ & 0.69 & 0.62 & 0.59 & 0.52 & 0.45 & 0.44 & 0.19 & 0.15 & 0.11 & 0.07 & 0.03\\\hline
$\tau_k$ & 0.12 & 0.12 & 0.12 & 0.12 & 0.12 & 0.12 & 0.12 & 0.12 & 0.12 & 0.12 & 0.12
\end{tabular}
\end{minipage}

\section{Notation used for the remainder the Supplementary Material}
\label{section:supp:Notation}
In addition to the notation used throughout the main text, I use the following notation throughout the remainder of the supplement. For any matrix $\bm{M} \in \mathbin{R}^{n \times m}$, define $\bm{Q}_M \in \mathbb{R}^{m \times \dim\left\lbrace \ker\left(\bm{M}^{\T}\right) \right\rbrace}$ to be a matrix whose columns form an orthonormal basis for $\ker\left(\bm{M}^{\T}\right)$. For any $\bm{x} \in \mathbb{R}^b$, define $\bm{V}\left(\bm{x}\right) = \sum_{j=1}^b \bm{x}_j\bm{B}_j$. Unless otherwise stated, for any sequence $\bm{X}_n \in \mathbb{R}^{r \times s}$, $n \geq 1$, I use the notation $\bm{X}_n = O_P\left(a_n\right)$ and $\bm{X}_n=o_P\left(a_n\right)$ if $\norm{\bm{X}_n}_2/a_n = O_P(1)$ and $\norm{\bm{X}_n}_2/a_n = o_P(1)$ as $n \to \infty$, respectively, where $\norm{\bm{X}_n}_2$ is the usual operator norm. We treat vectors $\bm{X}_n \in \mathbb{R}^r$ as matrices with one column.

\section{Technical conditions for theory presented in the main text}
\label{section:supp:TechCondition}
\subsection{Theorem \ref{theorem:L}}
\label{subsection:supp:TechL}
Recall \eqref{equation:Lgls} in Theorem \ref{theorem:L} required additional technical assumptions. Assumption \ref{assumption:supp:techL} below lists said conditions.
\begin{assumption}
\label{assumption:supp:techL}
Let $c > 1$ be a constant not dependent on $n$ or $p$, $r \in [\Koracle]$ be as defined in the statement of Theorem \ref{theorem:L} and $g \in [p]$.
\begin{enumerate}[label=(\alph*)]
    \item $\bm{E}_{g \bigcdot}$ is dependent on at most $c$ rows of $\bm{E}$, and is independent of all others.
    \item $\Prob\{F_r^{(\epsilon)}\} \to 1$ as $n,p \to \infty$.\label{item:supp:L:F}
    \item Let $\hat{\bm{v}}_g$ be the restricted maximum likelihood estimate (REML) for $\bm{v}_g$ defined in Remark \ref{remark:Alg1:Vg}, where $\hat{\bm{V}}_g = \bm{V}(\hat{\bm{v}}_g)$. Then the optimization to determine $\hat{\bm{v}}_g$ is restricted to the parameter space $\Theta_*$.
    \item For $\bm{a}_r \in \mathbb{R}^K$ the $r$th standard basis vector, the quantity
    \begin{align*}
        [\{ (\bm{C}^{\T}\bm{V}_g^{-1}\bm{C})^{-1} \}_{rr}]^{-1/2}\bm{a}_{r}^{\T}\left(\bm{C}^{\T}\bm{V}_g^{-1}\bm{C}\right)^{-1}\bm{C}^{\T}\bm{V}_g^{-1} \bm{E}_{g \bigcdot}
    \end{align*}
    is asymptotically $N(0,1)$.\label{item:supp:L:Normal}
\end{enumerate}
\end{assumption}
Asymptotic normality in \ref{item:supp:L:Normal} is satisfied in the following general scenario:
\begin{enumerate}[label=(\arabic*)]
    \item $\bm{E}_{g \bigcdot} \edist \bm{A}_g \tilde{\bm{e}}_{g}$, where $\bm{A}_g\bm{A}_g^{\T}=\bm{V}_g$ and $\tilde{\bm{e}}_{g} \in \mathbb{R}^n$ has independent entries with uniformly bounded sub-Gaussian norm.
    \item The entries of $\bm{A}_g^{-1/2}\bm{C}$ have uniformly bounded fourth moments.
\end{enumerate}
If \ref{item:supp:L:Normal} does not hold but all other conditions do hold, \eqref{equation:Lgls} can be replaced with
\begin{align*}
    [\lbrace(\hat{\bm{C}}^{\T}\hat{\bm{V}}_g^{-1}\hat{\bm{C}})^{-1}\rbrace_{rr}]^{-1/2}\lbrace \hat{\bm{L}}_{g r}^{(GLS)} - a \Loracle_{g r} \rbrace \edist W + o_P(1),
\end{align*}
where $W\asim (0,1)$. I give sufficient conditions to guarantee that \ref{item:supp:L:F} holds in Section \ref{subsection:supp:techF} below.

\subsection{Theorem \ref{theorem:XandC}}
\label{subsection:supp:TechC}
The conditions referenced in the statement of Theorem \ref{theorem:XandC} are given below. 
\begin{assumption}
\label{assumption:supp:techC}
Let $r \in [K],\bm{X},\bm{R}$ be as defined in Theorem \ref{theorem:XandC} and $\tilde{c} > 1$ be a constant not dependent on $n$ or $p$.
\begin{enumerate}[label=(\alph*)]
    \item $\hat{\bm{\theta}}$ is restricted to the convex set $\{ \bm{\theta} \in \mathbb{R}^b : (2c)^{-1}I_n \preceq \bm{V}(\bm{\theta}) \preceq 2bc^2 I_n \}$, where $c$ is as defined in the first paragraph of Section~\ref{subsection:Assumptions}.\label{item:supp:C:theta}
    \item $\norm{ n^{-1}\bm{C}^{\T}\bm{C} - \E(n^{-1}\bm{C}^{\T}\bm{C}) }_2 = O_P(n^{-1/2})$.\label{item:supp:C:EC}
    \item Let $\bm{a}_r \in \mathbb{R}^{K}$ be the $r$th standard basis vector. Then for $\tilde{\bm{\theta}} \in \mathbb{R}^b$ such that $\tilde{\bm{\theta}}_j = \bm{a}_r^{\T} \bm{\Psi}_j \bm{a}_r$ for all $j \in [b]$, $[\bm{X}^{\T}\{\bm{V}(\tilde{\bm{\theta}})\}^{-1}\bm{X}]^{-1/2}\bm{X}^{\T}\{\bm{V}(\tilde{\bm{\theta}})\}^{-1}\bm{R}_{\bigcdot r}$ is asymptotically $N(0,1)$.\label{item:supp:C:Normal}
\end{enumerate}
\end{assumption}
It is straightforward to find general conditions when item \ref{item:supp:C:EC} is satisfied (see Remark \ref{remark:supp:CrandomCond}, for example). Asymptotic normality holds under the following conditions:
\begin{enumerate}[label=(\arabic*)]
    \item $\vecM(\bm{R}) \edist \bm{D}\bm{\Xi}$, where $\bm{D}$ is a non-random matrix that satisfies $\bm{D}\bm{D}^{\T} = \sum_{j=1}^b \bm{\Psi}_j \otimes \bm{B}_j$ and $\bm{\Xi} \in \mathbb{R}^{nK}$ is a random matrix with independent entries such that $\E(\bm{\Xi})=\bm{0}$, $\E(\bm{\Xi}_i^2)=1$ and $\E(\bm{\Xi}_i^4)<\tilde{c}$ for all $i \in [nK]$.
    \item Let $\bm{D}=(\bm{D}_{rs})_{r,s \in [K]}$, where $\bm{D}_{rs} \in \mathbb{R}^{n \times n}$. Then the entries of $\bm{D}_{rs}\bm{X}$ have uniformly bounded fourth moments. Sufficient conditions for this to hold are:
    \begin{enumerate}[label=(\alph*)]
        \item $\bm{X}$ has mean $\bm{0}$ and is sub-exponential with uniformly sub-exponential norm (see \citet{Vershynin} for a definition of sub-exponential random vectors). This follows from the fact that the columns of $\bm{D}$ have uniformly bounded 2-norm.
        \item $\bm{X} \edist \bm{H}\tilde{\bm{X}}$, where $\bm{H} \in \mathbb{R}^{n \times n}$ is a non-random matrix with $\norm{\bm{H}}_2 \leq c$ and $\tilde{\bm{X}}$ is mean $\bm{0}$, has independent entries and $\E(\tilde{\bm{X}}_i^4) < c$ for all $i \in [n]$.
    \end{enumerate}
\end{enumerate}
If item \ref{item:supp:C:Normal} in Assumption \ref{assumption:supp:techC} does not hold, $Z$ in \ref{item:XandC:H0} of Theorem \ref{theorem:XandC} can be replaced with $W$ for $W \asim (0,1)$.

\subsection{Conditions that guarantee $\Prob\{F_r^{(\epsilon)}\} \to 1$}
\label{subsection:supp:techF}
Here I give the conditions necessary to ensure $\Prob\{F_r^{(\epsilon)}\} \to 1$ as $n,p \to \infty$, where $F_r^{(\epsilon)}$ was defined in Theorem \ref{theorem:L}. I study this by considering two scenarios: $\Koracle=K$ and $\Koracle < K$.

\begin{proposition}
\label{proposition:supp:PF1}
Suppose Assumptions \ref{assumption:CandL} and \ref{assumption:DependenceE} hold with $K=\Koracle$. Then for $\tau_r=\Lambda_{r}\{\E(p^{-1}\bm{L}\bm{C}^{\T}\bm{C}\bm{L}^{\T})\}$ and $\tau_{K+1}=0 < \tau_K \leq \cdots \leq \tau_1 < \tau_0=\infty$, $\lim_{n,p \to \infty}\Prob\{F_r^{(\epsilon)}\} = 1$ if $\tau_{r-1}/\tau_r,\tau_{r}/\tau_{r+1} \geq 1+\beta$ and $\norm{ \E(n^{-1}\bm{C}^{\T}\bm{C}) - n^{-1}\bm{C}^{\T}\bm{C} }_2 = o_P(1)$ as $n,p \to \infty$, where $\beta > 0$ is a constant that does not depend on $n$ or $p$.
\end{proposition}

\begin{proof}
This follows directly from Lemma \ref{lemma:supp:EigBound}.
\end{proof}

I next state and prove an analogous Proposition when $\Koracle<K$.

\begin{proposition}
\label{proposition:supp:PF2}
Suppose Assumptions \ref{assumption:CandL} and \ref{assumption:DependenceE} hold with $\Koracle < K$, and without loss of generality, assume $\E(n^{-1}\bm{C}^{\T}\bm{C})=I_K$ and $\bm{D} = np^{-1}\bm{L}^{\T}\bm{L} = \diag\left(\tau_1,\ldots,\tau_K\right)$, where $\tau_1 \geq \cdots \geq \tau_K > 0$ are defined in Proposition \ref{proposition:supp:PF1}. Define the non-random unitary matrix $\bm{W} \in \mathbb{R}^{K \times K}$ to be such that
\begin{align*}
    \bm{W}^{\T}\bm{D}^{1/2}\E\{ n^{-1}\bm{C}^{\T}\left(\delta^{-2} \bar{\bm{V}}\right)^{-1}\bm{C}\}\bm{D}^{1/2}\bm{W} = \diag\left( \gamma_1,\ldots,\gamma_K \right),
\end{align*}
and for $\bm{W}^{(\Koracle)} \in \mathbb{R}^{K \times \Koracle}$ the first $\Koracle$ columns of $\bm{W}$, let
\begin{align*}
    d_r = \Lambda_r[ \bm{D}^{1/2}\bm{W}^{(\Koracle)}\{ \bm{W}^{(\Koracle)} \}^{\T}\bm{D}^{1/2} ] = \Lambda_r[ \{ \bm{W}^{(\Koracle)} \}^{\T} \bm{D} \bm{W}^{(\Koracle)}  ], \quad r \in [\Koracle].
\end{align*}
Assume the following hold for some constant $c>1$ that does not depend on $n$ or $p$:
\begin{enumerate}[label=(\roman*)]
    \item $\bm{C}$ satisfies $\norm{ n^{-1}\bm{C}^{\T}\bm{\Delta}\bm{C} - \E(n^{-1}\bm{C}^{\T}\bm{\Delta}\bm{C}) }_2 = O_P(n^{-1/2})$ for any symmetric, positive definite $\bm{\Delta} \in \mathbb{R}^{n \times n}$ such that $\norm{\bm{\Delta}}_2 \leq c$.\label{item:supp:PropF2:C}
    \item $\gamma_k/\gamma_{k+1}\geq 1+c^{-1}$ for all $k \in [K]$ and $d_r/d_{r+1} \geq 1+c^{-1}$ for all $r \in [\Koracle]$, where $\gamma_{K+1}=d_{\Koracle+1}=0$.\label{item:supp:PropF2:EigenGap}
\end{enumerate}
Then $\lim_{n,p\to\infty}\Prob\{F_r^{\epsilon}\} = 1$ for all $r \in [\Koracle]$.
\end{proposition}

\begin{remark}
\label{remark:supp:PF2:d}
I show in Lemma \ref{lemma:supp:PrelimCinference} that under these assumptions $\lamoracle_r = d_r\{1+O_P(n^{-1/2})\}$ and $d_r \in [\tau_r-\tilde{c}\tau_{\Koracle+1},\tau_r+\tilde{c}\tau_{\Koracle+1}]$ for some constant $\tilde{c}>0$ that does not depend on $n$ or $p$. The latter implies that $d_r \approx \tau_r$ if $\tau_{\Koracle+1}$ is small (i.e. $\tau_{\Koracle+1}/\tau_r<<1$).
\end{remark}

\begin{remark}
Remark \ref{remark:supp:CrandomCond} in Section \ref{section:supp:lemmas} discusses general scenarios when \ref{item:supp:PropF2:C} holds.
\end{remark}

\begin{proof}
This follows directly from Lemma \ref{lemma:supp:PrelimCinference} stated in Section \ref{section:supp:lemmas}. 
\end{proof}

\subsection{Extending Theorem \ref{theorem:XandC} when $\Koracle < K$}
\label{subsection:supp:RestateXC}
Theorem \ref{theorem:XandC} can be extended to accommodate the case when $\Koracle < K$. A restatement of the Theorem to accommodate this scenario is given below.

\begin{theorem}[Restatement of Theorem \ref{theorem:XandC} when $\Koracle < K$]
\label{theorem:supp:RestateXC}
Let $\bm{X}$ be as defined in the statement of Theorem \ref{theorem:XandC}. Suppose Assumptions \ref{assumption:CandL}, \ref{assumption:FALCO} and \ref{assumption:DependenceE} hold, $\Koracle < K$ and let $\bm{W}^{(\Koracle)},d_1,\ldots,d_{\Koracle}$ be as defined in Proposition \ref{proposition:supp:PF2}. Assume the following conditions hold for some constant $c>0$ that does not depend on $n$ or $p$:
\begin{enumerate}[label=(\roman*)]
     \item $\bm{C}$ and $\bm{L}$ satisfy the identifiability conditions from the statement of Proposition \ref{proposition:supp:PF2}.
     \item $\bm{C} = \bm{X}\bm{\omega}^{\T}+\bm{R}$, where $\bm{X}$ is independent of $\bm{R}$, $\bm{\omega} \in \mathbb{R}^K$ is a constant and $\E(\bm{R})=\bm{0}$. For $j\in[b]$, let $\bm{\Psi}_j \in \mathbb{R}^{K \times K}$ be a non-random, symmetric matrix such that $\norm*{\bm{\Psi}_j}_2 \leq c$. Then $\V\{ \vecM(\bm{R}) \} = \sum_{j=1}^b \bm{\Psi}_j \otimes \bm{B}_j \succeq c^{-1}I_n$. Further, $\norm{\E(\bm{X})}_2 \leq cn^{1/2}$ and $\bm{\Xi}$ satisfies one of the following for all $\bm{\Xi}\in\{\bm{X}-\E(\bm{X}),\bm{R}\}$:
     \begin{enumerate}[label=(\arabic*)]
         \item $\vecM(\bm{\Xi})=\bm{G}\bm{\Delta}$, where $\bm{G}$ is a non-random square matrix that satisfies $\norm{\bm{G}}_2\leq c$, and $\bm{\Delta}$ is mean $\bm{0}$ and $\E(\bm{\Delta}_i^4) < c$ for every entry $i$ of $\bm{\Delta}$.
         \item $\E[\exp\{\vecM(\bm{\Xi})^{\T}\bm{t}\}] \leq \exp(c \norm{\bm{t}}_2^2)$ for all $\bm{t}$.
     \end{enumerate}
    \item Condition \ref{item:supp:C:theta} from Assumption \ref{assumption:supp:techC} holds.
    \item $\gamma_1,\ldots,\gamma_K$ and $d_1,\ldots,d_{\Koracle}$ satisfy Condition \ref{item:supp:PropF2:EigenGap} from the Statement of Proposition \ref{proposition:supp:PF2}.\label{item:supp:RestatC:d}
\end{enumerate}
Let $\bm{U} \in \mathbb{R}^{K \times \Koracle}$ be a non-random matrix with orthonormal columns whose columns are the eigenvectors of $\bm{D}^{1/2}\bm{W}^{(\Koracle)}\{\bm{W}^{(\Koracle)}\}^{\T}\bm{D}^{1/2}$. Then for some constant $c>0$ that does not depend on $n$ or $p$, the following hold for all $r \in [\Koracle]$:
\begin{enumerate}[label=(\alph*)]
    \item Suppose $\bm{X}$ is dependent on at most $c$ rows of $\bm{E}$ and $\bm{\omega}=\bm{0}$. Then if $n^{3/2}/(p\gamma_r) \to 0$ and for $\bm{\theta} \in \mathbb{R}^b$ such that $\bm{\theta}_j=\bm{U}_{\bigcdot r}^{\T}\bm{\Psi}_j\bm{U}_{\bigcdot r}$,
    \begin{align*}
        &n^{1/2}\abs{ [ \bm{X}^{\T}\{\bm{V}(\hat{\bm{\theta}})\}^{-1}\bm{X}]^{-1}\bm{X}^{\T}\{\bm{V}(\hat{\bm{\theta}})\}^{-1}\hat{\bm{C}}_{\bigcdot r} - [\bm{X}^{\T}\{\bm{V}(\bm{\theta})\}^{-1}\bm{X}]^{-1}\bm{X}^{\T}\{\bm{V}(\bm{\theta})\}^{-1}\Coracle_{\bigcdot r}  } = o_P(1)\\
        & [ \bm{X}^{\T}\{\bm{V}(\hat{\bm{\theta}})\}^{-1}\bm{X}]^{-1/2}\bm{X}^{\T}\{\bm{V}(\hat{\bm{\theta}})\}^{-1}\hat{\bm{C}}_{\bigcdot r} \edist Z + o_P(1),
    \end{align*}
    where $Z \asim (0,1)$.\label{result:supp:RestatC:null}
    \item Suppose $\bm{X}$ is independent of $\bm{E}$. Then for $a \in \{-1,1\}$ and if $n^{3/2}/(p\gamma_r) \to 0$,
    \begin{align*}
        [ \bm{X}^{\T}\{\bm{V}(\hat{\bm{\theta}})\}^{-1}\bm{X}]^{-1}\bm{X}^{\T}\{\bm{V}(\hat{\bm{\theta}})\}^{-1}\hat{\bm{C}}_{\bigcdot r} = a\bm{\omega}^{\T}\bm{U}_{\bigcdot r} + O_P(n^{-1/2}).
    \end{align*} 
\end{enumerate}

\end{theorem}
The proof of this theorem and Theorem \ref{theorem:XandC} are given in Section \ref{section:supp:InferenceC}.

\begin{remark}
\label{remark:supp:RetateXC:K}
If $\Koracle=K$, $\bm{U}=I_K$ and we can replace Condition \ref{item:supp:RestatC:d} with $\tau_{r-1}/\tau_{r},\tau_r/\tau_{r+1} \geq 1+c^{-1}$ for all $r \in [K]$, where $\tau_r$ is as defined in Proposition \ref{proposition:supp:PF1}. When $\Koracle < K$, I show in Lemma \ref{lemma:supp:PrelimCinference} that
\begin{align*}
    \bm{U}_{tr} = \begin{cases}
    O\left( \gamma_{\Koracle+1}/\gamma_t \right) & \text{if $t < r$}\\
    O\left\{ \gamma_{(\Koracle+1)\vee t}/\gamma_r \right\} & \text{if $t > r$}
    \end{cases}, \quad t \in [K]\setminus \{r\}
\end{align*}
\end{remark}

\begin{remark}
Asymptotic normality in result \ref{result:supp:RestatC:null} holds if $[\bm{X}^{\T}\{\bm{V}(\bm{\theta})\}^{-1}\bm{X}]^{-1/2}\bm{X}^{\T}\{\bm{V}(\bm{\theta})\}^{-1}\bm{C}\bm{U}_{\bigcdot r}$ is asymptotically normal.
\end{remark}

\section{Eigenvalues and eigenvectors of the population covariance matrix}
\label{section:supp:PopCov}
Here I derive properties of the eigenvalues and eigenvectors of the population covariance matrices $\V(\bm{Y}_{\bigcdot i})$. To do so, we need the following assumption:

\begin{assumption}
\label{assumption:supp:Covariance}
Let $c > 1$ be a constant. Then $K \geq 1$ is known, $\gamma_K \geq c^{-1}$ and the following hold:
\begin{enumerate}[label=(\alph*)]
    \item $\E(\bm{C}) = \bm{0}$,  $\bm{C}_{1 \bigcdot},\ldots,\bm{C}_{n \bigcdot}$ are identically distributed, $\V(\bm{C}_{1\bigcdot})=I_K$ and $\Lambda_r(\bm{L}\bm{L}^{\T})/\Lambda_{r+1}(\bm{L}\bm{L}^{\T}) \geq 1+c^{-1}$ for all $r \in [K]$.\label{item:supp:Cov:C}
    \item $n^{1/2}\vecM( n^{-1}\sum_{i=1}^n \bm{C}_{i\bigcdot}\bm{C}_{i\bigcdot}^{\T} - I_K ) \edist \bm{W} + o_P(1)$ as $n \to \infty$, where $\bm{W} \sim N(\bm{0},\bm{G})$ for some non-singular $\bm{G} \in \mathbb{R}^{K^2 \times K^2}$.\label{item:supp:Cov:limit}
\end{enumerate}
\end{assumption}
The assumption that the rows of $\bm{C}$ will likely hold when samples $i$ are identically distributed. Some examples include samples collected on related individuals (e.g. twin studies, samples related through a kinship matrix, etc.), data with repeated measurements and multi-tissue data collected from similar tissues, among others. Theorem~\ref{theorem:supp:CovEigs} gives the asymptotic properties of the eigenvalues and eigenvectors of $\V(\bm{Y}_{\bigcdot i})$.

\begin{theorem}
\label{theorem:supp:CovEigs}
Suppose Assumptions \ref{assumption:CandL}, \ref{assumption:FALCO}, \ref{assumption:DependenceE} and \ref{assumption:supp:Covariance} hold, and $\Koracle=K$. Then for $\eta_r^{(i)} = np^{-1}\Lambda_r\lbrace \V(\bm{Y}_{\bigcdot i}) \rbrace$ and $\bm{u}_r^{(i)}$ the $r$th eigenvector of $\V(\bm{Y}_{\bigcdot i})$,
\begin{align}
\label{equation:supp:PopEig}
    &n^{1/2}(\lamhatoracle_r/\eta_r^{(i)} - [1 + O_P\lbrace n/(\gamma_r p) \rbrace]) \edist z_r + o_P(1), \quad r \in [K]\\
\label{equation:supp:PopL}
    &\max_{i \in [n]}\norm*{ a\hat{\bm{L}}_{\bigcdot r} - \{p/n \eta_r^{(i)}\}^{1/2}\bm{u}_r^{(i)} }_{\infty} = O_P\{ \log(p)n^{-1/2} + n^{1/2}(\gamma_{\Koracle} p)^{-1/2} \}, \quad r \in [K]
\end{align}
as $n,p \to \infty$, where $(z_1,\ldots,z_K)^{\T} \sim N(\bm{0},\tilde{\bm{G}})$ does not depend on $i$,  $\tilde{\bm{G}}_{st}=\bm{G}_{(s-1)K+s,(t-1)K+t}$ for all $s,t \in [K]$, $a \in \{-1,1\}$ and the error term $o_P(1)$ is uniform across $i=1,\ldots,n$.
\end{theorem}

\begin{proof}
Assumptions \ref{assumption:CandL} and \ref{assumption:DependenceE} imply $\norm{ \V(\bm{E}_{\bigcdot i}) }_2 \leq c_1$ for some constant $c_1 > 0$ that does not depend on $n$ or $p$. Weyl's Theorem therefore implies for $\bm{A} = \E(p^{-1}\bm{L}\bm{C}^{\T}\bm{C}\bm{L}^{\T}) = np^{-1}\bm{L}\bm{L}^{\T}$, $\abs{\Lambda_r(\bm{A}) - \eta_r^{(i)}} \leq n/p c_1$, and by the eigengap assumption on $\bm{A}$ in Assumption \ref{assumption:supp:Covariance}, $\norm{ \tilde{\bm{u}}_r^{(i)} - \bm{u}_r^{(i)} }_2 \leq c_2 n/(p\gamma_r)$ by Lemma \ref{lemma:supp:EigVector} for $\tilde{\bm{u}}_r^{(i)}$ the $r$th eigenvector of $\bm{A}$ and $c_2 >0$ a constant that does not depend on $n$ or $p$. The rest of the proof follows from Theorems \ref{theorem:Lambda} and \ref{theorem:L}, as well as the co-factor expansion argument utilized in Appendix A of \citet{FanEigen}.
\end{proof}

Theorem \ref{theorem:supp:CovEigs} shows that Algorithm \ref{algorithm:EstC} recovers both the eigenvalues and eigenvectors of the gene-by-gene covariance matrix, where like \eqref{equation:Linfty}, \eqref{equation:supp:PopL} implies principal components plots of $\hat{\bm{L}}$ mirror the information contained in the population eigenvectors. The result in \eqref{equation:supp:PopEig} shows $\lamhatoracle_r$ is consistent and asymptotically normal if $a_n=n^{3/2}/(\gamma_r p) \to 0$, which is the first result proving the existence of consistent and asymptotically normal estimators for population eigenvalues of nearly arbitrary size in data with correlated samples. Theorem \ref{theorem:supp:CovEigs} also significantly extends the results of \citet{FanEigen}, which in order to show the asymptotic normality of sample eigenvalues, required (1) $\bm{U}\bm{Y} \in \mathbb{R}^{p \times n}$ have independent entries for some unitary matrix $\bm{U}\in\mathbb{R}^{p \times p}$ and (2) $(p^{1/2}n^{-1/2})a_n \to 0$.

\section{Proof of Proposition \ref{proposition:ImageC}}
\label{section:supp:Prop}
Here I prove Proposition \ref{proposition:ImageC}.
\begin{proof}[Proof of Proposition \ref{proposition:ImageC}]
Let $\bm{S}=p^{-1}\bm{Y}^{\T}\bm{Y}$. We can re-write the objective function in \eqref{equation:Alg1:C} to be
\begin{align*}
    \Tr\{ (\bm{U}^{\T}\bar{\bm{V}}^{-1}\bm{S}\bar{\bm{V}}^{-1}\bm{U})(\bm{U}\bar{\bm{V}}^{-1} \bm{U})^{-1} \} = \Tr(\tilde{\bm{U}}^{\T}\bar{\bm{V}}^{-1/2}\bm{S}\bar{\bm{V}}^{-1/2}\tilde{\bm{U}})
\end{align*}
where $\bm{U} \in \mathbb{R}^{n \times k}$ is any matrix such that $\im(\bm{U})=\im(\bm{H})$ and $\tilde{\bm{U}}=\bar{\bm{V}}^{-1/2}\bm{U}\left(\bm{U}^{\T}\bar{\bm{V}}^{-1}\bm{U}\right)^{-1/2}$. Since the latter has orthonormal columns, the objective achieves its maximum when the columns of $\tilde{\bm{U}}$ are the first $k$ eigenvectors of $\bar{\bm{V}}^{-1/2}\bm{S}\bar{\bm{V}}^{-1/2}$, which completes the proof.
\end{proof}

\section{Estimating the eigenvalues and eigenvectors of $\bm{C}\left(p^{-1}\bm{L}^{\T}\bm{L}\right)\bm{C}^{\T}$}
\label{section:supp:Eigen}

\subsection{Preliminaries}
\label{subsection:supp:EigenPreliminaries}
Without loss of generality, we may assume $np^{-1}\bm{L}^{\T}\bm{L}=\diag\left(\lambda_1,\ldots,\lambda_K\right)$ and $n^{-1}\bm{C}^{\T}\bm{C}=I_K$. We utilize similar techniques to those developed in \citet{ChrisANDDan}. For any estimate $\hat{\bm{V}} = \sum_{j=1}^b \left(\bar{v}_j + \epsilon_j\right)\bm{B}_j$ of $\bar{\bm{V}}=\E\left(p^{-1}\bm{E}^{\T}\bm{E}\right)$, define
\begin{subequations}
\label{equation:supp:Params}
\begin{align}
\label{equation:supp:epsilon}
\epsilon_V &= \norm{\bar{\bm{V}} - \hat{\bm{V}}}_2\\
\label{equation:supp:Ctilde}
\tilde{\bm{C}} &= \hat{\bm{V}}^{-1/2}\bm{C}\left(\bm{C}^T \hat{\bm{V}}^{-1}\bm{C}\right)^{-1/2}\bm{U}\\
\label{equation:supp:Ltilde}
\tilde{\bm{L}} &= p^{-1/2}\bm{L}\left(\bm{C}^T \hat{\bm{V}}^{-1}\bm{C}\right)^{1/2}\bm{U}\\
\label{equation:supp:QCtilde}
\bm{Q} &= \bm{Q}_{\hat{V}^{-1/2} C}
\end{align}
\end{subequations}
where $\bm{U} \in \mathbb{R}^{K \times K}$ is a rotation matrix such that
\begin{align}
\label{equation:supp:tau}
    \tilde{\bm{L}}^T \tilde{\bm{L}} = \diag\left(\tau_1,\ldots,\tau_K\right), \quad 0=\tau_{K+1} < \tau_K \leq \cdots \leq \tau_1 < \tau_0 = \infty.
\end{align}
By Lemma \ref{lemma:supp:EigBound} in Section \ref{section:supp:lemmas} and Assumptions \ref{assumption:CandL} and \ref{assumption:FALCO}, this implies that $c^{-1}\lambda_k \leq\tau_k \leq c\lambda_k$ for some constant $c > 1$ that does not depend on $n$ or $p$. Therefore, for any constant $c_2 > 1$, there exists a $c_1 > 1$ large enough that does not depend on $n$ or $p$ such that $\lambda_r/\lambda_{r+1} > c_1$ implies $\tau_{r}/\tau_{r+1} > c_2$ regardless of the choice of $\hat{\bm{V}}$. We use this to inductively define the indices $k_1,\ldots,k_J \in [K]$ in which an eigengap occurs. First, for some arbitrary $c_2 > 1$ and suitably large $c_1 > 1$, define
\begin{align*}
    k_1 = \min\left( \left\lbrace r \in [K]: \lambda_{r}/\lambda_{r+1} > c_1 \right\rbrace \right).
\end{align*}
If $k_1 = K$, we are done. Otherwise, define $k_j$ inductively as
\begin{align*}
    k_j =  \min\left( \left\lbrace r \in \left\lbrace k_{j-1}+1,\ldots,K \right\rbrace : \lambda_{r}/\lambda_{r+1} > c_1 \right\rbrace \right).
\end{align*}
We let $k_0=0$ and $J \in [K]$ be such that $k_J = K$. We refer to these indices $k_0,k_1,\ldots,k_J$ throughout the supplement.\par
\indent Following \citet{ChrisANDDan}, we define
\begin{align}
\label{equation:supp:E1E2}
    \bm{E}_1 &= \bm{E}\hat{\bm{V}}^{-1/2}\tilde{\bm{C}} \in \mathbb{R}^{p \times K}, \quad \bm{E}_2 = \bm{E}\hat{\bm{V}}^{-1/2}\bm{Q} \in \mathbb{R}^{p \times (n-K)}\\
    \bm{S} &= \begin{pmatrix}
    \tilde{\bm{C}}^T\\
    \bm{Q}^T
    \end{pmatrix}\left(p^{-1}\hat{\bm{V}}^{-1/2}\bm{Y}^T \bm{Y}\hat{\bm{V}}^{-1/2}\right)\begin{pmatrix}
    \tilde{\bm{C}}&\bm{Q}
    \end{pmatrix}\nonumber\\
    \label{equation:supp:S}
    & = \begin{pmatrix}
    \left(\tilde{\bm{L}} + p^{-1/2}\bm{E}_1 \right)^T \left(\tilde{\bm{L}} + p^{-1/2}\bm{E}_1 \right) & \left(\tilde{\bm{L}} + p^{-1/2}\bm{E}_1 \right)^T \left(p^{-1/2}\bm{E}_2 \right)\\
    \left(p^{-1/2}\bm{E}_2 \right)^T \left(\tilde{\bm{L}} + p^{-1/2}\bm{E}_1 \right) & p^{-1}\bm{E}_2^T \bm{E}_2
    \end{pmatrix}.
\end{align}
If $\begin{pmatrix}\hat{\bm{v}}^{\T} \, \hat{\bm{z}}^{\T} \end{pmatrix}^{\T} \in \mathbb{R}^{n \times K}$ for $\hat{\bm{v}} \in \mathbb{R}^{K \times K}, \hat{\bm{z}} \in \mathbb{R}^{n \times (n-K)}$ are the eigenvectors of $\bm{S}$, then
\begin{align*}
    \tilde{\bm{C}}\hat{\bm{v}} + \bm{Q}\hat{\bm{z}} \in \mathbb{R}^{n \times K}
\end{align*}
are the first $K$ right singular vectors of $\bm{Y}$. Our first goal is to understand $\hat{\bm{v}}$ and $\hat{\bm{z}}$.

\subsection{The top-left $K \times K$ block of $\bm{S}$}
\label{subsection:supp:UpperLeft}
We first develop theory to understand the behavior of the upper left block of $\bm{S}$, defined as
\begin{align*}
    \left(\tilde{\bm{L}} + p^{-1/2}\bm{E}_1 \right)^T \left(\tilde{\bm{L}} + p^{-1/2}\bm{E}_1 \right) \in \mathbb{R}^{K \times K}.
\end{align*}

\begin{lemma}
\label{lemma:supp:UpperBlock}
Suppose Assumptions \ref{assumption:CandL}, \ref{assumption:FALCO} and \ref{assumption:DependenceE} hold, and define 
\begin{align*}
    &\tilde{\bm{N}} = \tilde{\bm{L}} + p^{-1/2}\bm{E}_1, \quad \phi_1 = p^{-1/2}\left(1 + n^{1/2}\epsilon_V\right), \quad \phi_2 = \epsilon_V.
\end{align*}
Then
\begin{align}
\label{equation:supp:mus}
    \mu_s = \Lambda_s\left( \tilde{\bm{N}}^{\T}\tilde{\bm{N}} \right) = \tau_s + 1 + O_P\left( \phi_1 \lambda_s^{1/2} + \phi_2 \right), \quad s \in [K].
\end{align}
Additionally, let $\begin{pmatrix} 
\bm{V}_{11} & \bm{V}_{12} & \cdots & \bm{V}_{1J}\\
\bm{V}_{21} & \bm{V}_{22} & \cdots & \bm{V}_{2J}\\
\vdots & \vdots & \ddots & \vdots\\
\bm{V}_{J1} & \bm{V}_{J2} & \cdots & \bm{V}_{JJ}
\end{pmatrix} \in \mathbb{R}^{K \times K}$ be the right singular values of $\tilde{\bm{N}}$, where $\bm{V}_{rs} \in \mathbb{R}^{\left(k_r - k_{r-1}\right) \times \left(k_s - k_{s-1}\right)}$ for $r,s \in [J]$ and $k_r,k_s$ defined in Section \ref{subsection:supp:EigenPreliminaries}. Then
\begin{subequations}
\label{equation:supp:UpperEigVectors}
\begin{align}
\label{equation:supp:UpperEigVectors:rr}
    &\norm{ I_{\left(k_{r}-k_{r-1}\right)} - \bm{V}_{rr}^{\T}\bm{V}_{rr} }_2, \, \norm{ I_{\left(k_{r}-k_{r-1}\right)} - \bm{V}_{rr}\bm{V}_{rr}^{\T} }_2 = O_P\left\lbrace\left(\phi_1 \lambda_{k_r}^{-1/2} + \phi_2 \lambda_{k_r}^{-1}\right)^2\right\rbrace, \quad r \in [J]\\
\label{equation:supp:UpperEigVectors:rs}
    &\norm{\bm{V}_{rs}}_2 = O_P\left( \phi_1\lambda_{k_{\min(r,s)}}^{-1/2} + \phi_2\lambda_{k_{\min(r,s)}}^{-1} \right), \quad r\neq s \in [J].
\end{align}
\end{subequations}
\end{lemma}

\begin{proof}
First,
\begin{align*}
    \tilde{\bm{N}}^{\T}\tilde{\bm{N}} = \diag\left(\tau_1,\ldots,\tau_K\right) + \underbrace{p^{-1/2}\tilde{\bm{L}}^{\T}\bm{E}_1 + \left(p^{-1/2}\tilde{\bm{L}}^{\T}\bm{E}_1\right)^{\T}}_{\bm{A}^{(1)}} + \underbrace{p^{-1}\bm{E}_1^{\T}\bm{E}_1}_{\bm{A}^{(2)}}.
\end{align*}
We derive the properties of $\bm{A}^{(1)}$ and $\bm{A}^{(2)}$ below.
\begin{enumerate}[label=(\arabic*)]
    \item Define $\bar{\bm{L}} = n^{1/2}p^{-1/2} \bm{L}\left(np^{-1}\bm{L}^T \bm{L}\right)^{-1/2}$, $\hat{\bm{M}} = \left(n^{-1}\bm{C}^T\hat{\bm{V}}^{-1/2}\bm{C}\right)^{1/2}\bm{U}$ and $\bm{\Delta} = \hat{\bm{V}}^{-1} - \bm{V}^{-1}$. By definition,
    \begin{align*}
        \left(np^{-1}\bm{L}^T \bm{L}\right)^{1/2}\hat{\bm{M}}^{-1} = \hat{\bm{W}}\diag\left(\tau_1^{1/2},\ldots,\tau_K^{1/2}\right), \quad \tilde{\bm{L}} = \bar{\bm{L}}\hat{\bm{W}}\diag\left(\tau_1^{1/2},\ldots,\tau_K^{1/2}\right)
    \end{align*}
    where $\hat{\bm{W}} \in \mathbb{R}^{K \times K}$ is a random unitary matrix. Then
    \begin{align*}
        p^{-1/2}\tilde{\bm{L}}^T \bm{E}_1 =& p^{-1/2}\diag\left(\tau_1^{1/2},\ldots,\tau_K^{1/2}\right)\hat{\bm{W}}^T \bar{\bm{L}}^T \bm{E}\bm{V}^{-1}\left(n^{-1/2}\bm{C}\right)\left(n^{-1}\bm{C}^T\hat{\bm{V}}^{-1/2}\bm{C}\right)^{-1/2}\bm{U}\\
        &+ p^{-1/2}\diag\left(\tau_1^{1/2},\ldots,\tau_K^{1/2}\right)\hat{\bm{W}}^T \bar{\bm{L}}^T \bm{E}\bm{\Delta}\left(n^{-1/2}\bm{C}\right)\left(n^{-1}\bm{C}^T\hat{\bm{V}}^{-1/2}\bm{C}\right)^{-1/2}\bm{U}.
    \end{align*}
    For some large constant $c>0$ that does not depend on $n$ or $p$, we have
    \begin{align*}
        \norm{\bm{\Delta}\left(n^{-1/2}\bm{C}\right)\left(n^{-1}\bm{C}^T\hat{\bm{V}}^{-1/2}\bm{C}\right)^{-1/2}\bm{U}}_2 \leq c\epsilon_V
    \end{align*}
    and
    \begin{align*}
        \norm{\hat{\bm{W}}^T \bar{\bm{L}}^T \bm{E}\bm{V}^{-1}\left(n^{-1/2}\bm{C}\right)\left(n^{-1}\bm{C}^T\hat{\bm{V}}^{-1/2}\bm{C}\right)^{-1/2}\bm{U}}_2 \leq c\norm{\bar{\bm{L}}^T \bm{E}\bm{V}^{-1}\left(n^{-1/2}\bm{C}\right)}_2 = O_P\left(1\right),
    \end{align*}
    where the last equality follows by Lemma \ref{lemma:supp:Preliminaries}. Therefore,
    \begin{align*}
        \bm{A}^{(1)}_{rs} = O_P\left\lbrace \phi_1\left(\lambda_r^{1/2} + \lambda_s^{1/2}\right) \right\rbrace, \quad r,s \in [K].
    \end{align*}
    \item Define $\hat{\bm{M}}_C = \left(n^{-1}\bm{C}^T\hat{\bm{V}}^{-1/2}\bm{C}\right)^{-1/2}\bm{U}$. We see that
    \begin{align*}
        p^{-1}\bm{E}_1^T\bm{E}_1 =& \hat{\bm{M}}_C^T\left(n^{-1/2}\bm{C}\right)^T \hat{\bm{V}}^{-1}\left(p^{-1}\bm{E}^T\bm{E}\right)\hat{\bm{V}}^{-1}\left(n^{-1/2}\bm{C}\right) \hat{\bm{M}}_C\\
        =& \hat{\bm{M}}_C^T\left(n^{-1/2}\bm{C}\right)^T \bm{V}^{-1}\left(p^{-1}\bm{E}^T\bm{E}\right)\bm{V}^{-1}\left(n^{-1/2}\bm{C}\right) \hat{\bm{M}}_C\\
        &+ \hat{\bm{M}}_C^T\left(n^{-1/2}\bm{C}\right)^T \bm{\Delta}\left(p^{-1}\bm{E}^T\bm{E}\right)\bm{V}^{-1}\left(n^{-1/2}\bm{C}\right) \hat{\bm{M}}_C\\
        &+ \hat{\bm{M}}_C^T\left(n^{-1/2}\bm{C}\right)^T \bm{V}^{-1}\left(p^{-1}\bm{E}^T\bm{E}\right)\bm{\Delta}\left(n^{-1/2}\bm{C}\right) \hat{\bm{M}}_C\\
        &+ \hat{\bm{M}}_C^T\left(n^{-1/2}\bm{C}\right)^T \bm{\Delta}\left(p^{-1}\bm{E}^T\bm{E}\right)\bm{\Delta}\left(n^{-1/2}\bm{C}\right) \hat{\bm{M}}_C.
    \end{align*}
    By Lemma \ref{lemma:supp:Preliminaries}, we then get that
    \begin{align*}
        \norm{I_K - \bm{A}^{(2)}}_2 = O_P\left( p^{-1/2} + \phi_2 \right).
    \end{align*}
\end{enumerate}
Therefore, for $\bm{M}=\tilde{\bm{N}}^{\T}\tilde{\bm{N}}$
\begin{align*}
    \bm{M}_{rs} = \begin{cases}
    \tau_r + 1 + O_P\left( \phi_1\lambda_r^{1/2} + \phi_2 \right) & \text{if $r=s$}\\
    O_P\left( \phi_1\lambda_r^{1/2} + \phi_2 \right) & \text{if $r \neq s$}
    \end{cases}, \quad r,s \in [K]
\end{align*}
Let $\tilde{\bm{M}}_j \in \mathbb{R}^{\left(k_{j}-k_{j-1}\right) \times \left(k_{j}-k_{j-1}\right)}$ be a diagonal matrix containing the eigenvalues of the $k_j$th diagonal block of $\bm{M}$. By Lemma \ref{lemma:supp:EigApprox}, 
\begin{align*}
    \tilde{\bm{M}}_{j_{ss}} = \tau_{k_{j-1}+s} + \delta^2 + O_P\left(\phi_1\lambda_{k_{j-1}+s}^{1/2} + \phi_2\right).
\end{align*}
Then for $\bm{v}_{rs} \in \mathbb{R}^{\left( k_r-k_{r-1}\right)\times \left( k_s-k_{s-1}\right)}$, we can write $\bm{M}$ as
\begin{align*}
    \bm{M} =& \begin{pmatrix} \bm{v}_{11}\\\vdots\\\bm{v}_{J 1} \end{pmatrix}\tilde{\bm{M}}_1 \begin{pmatrix} \bm{v}_{11}^T&\cdots&\bm{v}_{J 1}^T \end{pmatrix} + \begin{pmatrix} \bm{0}_{k_1 \times \left(k_2-k_1\right)}\\\bm{v}_{2 2}\\\vdots\\\bm{v}_{J 2} \end{pmatrix}\tilde{\bm{M}}_2 \begin{pmatrix} \bm{0}_{k_1 \times \left(k_2-k_1\right)}^T&\bm{v}_{2 2}^T&\cdots&\bm{v}_{J 2}^T \end{pmatrix} + \cdots\\
    &+ \begin{pmatrix} \bm{0}_{k_1 \times \left(k_J-k_{J-1}\right)}\\\vdots\\\bm{0}_{\left(k_{J-1}-k_{J-2}\right) \times \left(k_J-k_{J-1}\right)}\\\bm{v}_{J J} \end{pmatrix}\tilde{\bm{M}}_J \begin{pmatrix} \bm{0}_{k_1 \times \left(k_J-k_{J-1}\right)}^T&\cdots&\bm{0}_{\left(k_{J-1}-k_{J-2}\right) \times \left(k_J-k_{J-1}\right)}^T&\bm{v}_{J J}^T \end{pmatrix} + \bm{\Delta},
\end{align*}
where $\bm{v}_{jj}$ is a unitary matrix for all $j=1,\ldots,J$ and $\norm{\bm{v}_{rs}}_2 = O_P\left(\phi_1\lambda_{k_s}^{-1/2} + \phi_2\lambda_{k_s}^{-1}\right)$ for $r \neq s$. The matrix $\bm{\Delta} = \begin{pmatrix}\bm{\Delta}_{11} & \bm{\Delta}_{21}^T & \cdots & \bm{\Delta}_{J1}^T\\ \bm{\Delta}_{21} & \bm{\Delta}_{22} & \cdots & \bm{\Delta}_{J2}^T\\
\vdots & \vdots & \ddots & \vdots\\
\bm{\Delta}_{J1} & \bm{\Delta}_{J2} & \cdots &\bm{\Delta}_{JJ}
\end{pmatrix}$, where $\bm{\Delta}_{rs} \in \mathbb{R}^{\left(k_{r}-k_{r-1}\right) \times \left(k_{s}-k_{s-1}\right)}$, is such that
\begin{align*}
    \norm{\bm{\Delta}_{rs}}_2 = O_P\left(K^2 \phi_1^2\right) + O_P\left\lbrace \phi_2^2 \sum\limits_{j=1}^{s-1}\lambda_j^{-1}\right\rbrace 1\left(s \geq 2\right), \quad s \leq r.
\end{align*}
Let $\bm{V}_j \in \mathbb{R}^{K \times \left(k_j-k_{j-1}\right)}$ be the eigenvectors corresponding to eigenvalues $\mu_{k_{j-1}+1},\ldots,\mu_{k_{j}}$. By Lemma \ref{lemma:supp:EigApprox} and Corollary \ref{corollary:supp:NormSpaceVhat},
\begin{align*}
&\mu_{s} = \tau_{s} + \delta^2 + O_P\left(\phi_1\lambda_{k_1}^{1/2} + \phi_2\right), \quad s \in \left[k_1\right]\\
&\norm{\begin{pmatrix} I_{k_1\times k_1} & \bm{0}_{k_1\times\left(K-k_1\right)}\\
\bm{0}_{\left(K-k_1\right)\times k_1} & \bm{0}_{\left(K-k_1\right)\times\left(K-k_1\right)}\end{pmatrix} -P_{V_1}}_F = O_P\left( \phi_1^2\lambda_1^{-1/2} + \phi_2\lambda_1^{-1} \right).
\end{align*}
To understand $\bm{V}_j$ for $j > 1$, define $\tilde{\bm{V}}_j = \begin{pmatrix} \bm{0}_{k_1 \times \left(k_{j}-k_{j-1}\right)}\\ \vdots\\ \bm{0}_{\left(k_{j-1}-k_{j-2}\right) \times \left(k_{j}-k_{j-1}\right)}\\ \bm{v}_{jj} \\ \vdots\\ \bm{v}_{Jj} \end{pmatrix}$. Then
\begin{align*}
    &P_{\left(\tilde{V}_1 \cdots \tilde{V}_{j-1}\right)}^{\perp}\tilde{\bm{V}}_j = \tilde{\bm{V}}_j - \left(\tilde{\bm{V}}_1 \cdots \tilde{\bm{V}}_{j-1}\right) \begin{pmatrix} \tilde{\bm{V}}_1^T \tilde{\bm{V}}_1 & \cdots &  \tilde{\bm{V}}_1^T \tilde{\bm{V}}_{j-1}\\
    \vdots & \ddots & \vdots\\
    \tilde{\bm{V}}_{j-1}^T \tilde{\bm{V}}_1 & \cdots &  \tilde{\bm{V}}_{j-1}^T \tilde{\bm{V}}_{j-1}\end{pmatrix}^{-1}\begin{pmatrix} \tilde{\bm{V}}_1^T \tilde{\bm{V}}_j\\ \vdots \\ \tilde{\bm{V}}_{j-1}^T \tilde{\bm{V}}_j \end{pmatrix} = \tilde{\bm{V}}_j - \bm{\Delta}_j\\
    &\norm{\bm{\Delta}_j}_2 = O_P\left( \phi_1\lambda_{k_{j-1}}^{-1/2} + \phi_2\lambda_{k_{j-1}}^{-1} \right).
\end{align*}
Let $\bm{R}_j$ be a symmetric matrix such that $\left(\tilde{\bm{V}}_j - \bm{\Delta}_j\right)\bm{R}_j$ has orthogonal columns. Therefore,
\begin{align*}
    \lambda_{k_j}^{-1}\bm{M}\left(\tilde{\bm{V}}_j - \bm{\Delta}_j\right)\bm{R}_j = \lambda_{k_j}^{-1}\left(\tilde{\bm{V}}_j - \bm{\Delta}_j\right)\bm{R}_j\tilde{\bm{M}}_j + O_P\left( \phi_1\lambda_{k_{j}}^{-1/2} + \phi_2\lambda_{k_{j}}^{-1} \right),
\end{align*}
which by Lemma \ref{lemma:supp:EigApprox} and Corollary \ref{corollary:supp:NormSpaceVhat},
\begin{align}
\label{equation:supp:mus_lemma}
    &\mu_{s} = \tau_s + \delta^2 + O_P\left(\phi_1\lambda_{k_j}^{1/2} + \phi_2\lambda_{k_j}\right), \quad s\in\left\lbrace k_{j-1}+1,\ldots,k_j  \right\rbrace\\
    &\norm{ \bm{0}_{k_1 \times k_1} \oplus \cdots \oplus \bm{0}_{\left(k_{j-1}-k_{j-2}\right) \times \left(k_{j-1}-k_{j-2}\right)} \oplus I_{\left(k_j-k_{j-1}\right)\times \left(k_j-k_{j-1}\right)} \oplus \bm{0}_{\left(k_{j+1}-k_{j}\right) \times \left(k_{j+1}-k_{j}\right)} \oplus \cdots \oplus \bm{0}_{\left(k_{J}-k_{J-1}\right) \times \left(k_{J}-k_{J-1}\right)} - P_{V_j} }_F\nonumber\\
\label{equation:supp:Vs_lemma}
    & = O_P\left( \phi_1\lambda_{k_{j}}^{-1/2} + \phi_2\lambda_{k_{j}}^{-1} \right).
\end{align}
Equation \eqref{equation:supp:mus_lemma} proves \eqref{equation:supp:mus}.\par
\indent To prove the remainder of the lemma, let $\bm{V}_j = \begin{pmatrix} \bm{V}_{1j} \\ \vdots \\ \bm{V}_{Jj}\end{pmatrix}$. Then by \eqref{equation:supp:Vs_lemma} and for $r > j$,
\begin{align*}
    \norm{I_{\left(k_j-k_{j-1}\right) \times \left(k_j-k_{j-1}\right)} - \bm{V}_{jj}^T \bm{V}_{jj}}_2 = O_P\left(\phi_1^2 \lambda_{k_j}^{-1} + \phi_2^2 \lambda_{k_j}^{-2}\right), \quad \norm{\bm{V}_{rj}}_2 = O_P\left(\phi_1\lambda_{k_j}^{-1/2} + \phi_2 \lambda_{k_j}^{-1}\right).
\end{align*}
Lastly, for any $s < j$, we have
\begin{align*}
    \bm{0} = \bm{V}_{ss}^T\left(\bm{V}_s^T \bm{V}_j\right) = \bm{V}_{ss}^T\bm{V}_{ss}\bm{V}_{sj} + O_P\left(\phi_1\lambda_{k_s}^{-1/2} + \phi_2 \lambda_{k_s}^{-1}\right) = \bm{V}_{sj} + O_P\left(\phi_1\lambda_{k_s}^{-1/2} + \phi_2 \lambda_{k_s}^{-1}\right).
\end{align*}
Therefore,
\begin{align*}
    \norm{\bm{V}_{sj}}_2 = O_P\left(\phi_1\lambda_{k_s}^{-1/2} + \phi_2 \lambda_{k_s}^{-1}\right), \quad s < j.
\end{align*}
This proves \eqref{equation:supp:UpperEigVectors} and completes the proof.
\end{proof}

\begin{remark}
\label{remark:supp:J}
Note that \eqref{equation:supp:UpperEigVectors} holds if we define $k_j$ in terms of $\tau_1,\ldots,\tau_K$, defined in \eqref{equation:supp:tau}, as follows: Let $k_0 = 0$ and define $k_j$ inductively as
\begin{align*}
    k_j = \min\left( \left\lbrace r \in \left\lbrace k_{j-1}+1,\ldots,K \right\rbrace: \tau_r/\tau_{r+1} \geq 1+\epsilon \right\rbrace \right), \quad j \in [J],
\end{align*}
where $k_J = K$ and $\epsilon > 0$ is an arbitrarily small constant.
\end{remark}

\subsection{Understanding $\hat{\bm{v}}$ and $\hat{\bm{z}}$ given an estimate for $\bm{V}$}
\label{subsection:supp:vandzhat}
We use the results of Lemma \ref{lemma:supp:UpperBlock} to study the properties of $\hat{\bm{v}}= \left(\hat{\bm{v}}_1 \cdots \hat{\bm{v}}_K\right) \in \mathbb{R}^{K \times K}$ and $\hat{\bm{z}} = \left(\hat{\bm{Z}}_1 \cdots \hat{\bm{v}}_K\right) \in \mathbb{R}^{(n-K) \times K}$, which were defined in Section \ref{subsection:supp:EigenPreliminaries}.

\begin{lemma}
\label{lemma:supp:vandzhat}
Suppose Assumptions \ref{assumption:CandL}, \ref{assumption:FALCO} and \ref{assumption:DependenceE} hold, and let $\hat{\bm{v}},\hat{\bm{z}}$ be as defined above, $k_0,k_1,\ldots,k_J$ be as defined in Section \ref{subsection:supp:EigenPreliminaries}, $\phi_1,\phi_2, \tilde{\bm{N}},\mu_1,\ldots,\mu_K$ be as defined in Lemma \ref{lemma:supp:UpperBlock} and $\bm{S}$ be as defined in \eqref{equation:supp:S}. Define $\hat{\mu}_s$ to be the $s$th eigenvalue of $\bm{S}$, and for $\bm{V}_{rs}$, $r,s \in [J]$, defined in the statement of Lemma \ref{lemma:supp:UpperBlock}, let $\bm{V}_j = \left(\bm{V}_{1j}^{\T}\cdots \bm{V}_{Jj}^{\T}\right)^{\T}$ for all $j \in [J]$. Set $\hat{\bm{V}}_j = \left(\hat{\bm{v}}_{k_{j-1}+1} \cdots \hat{\bm{v}}_{k_{j}}\right) = \left(\hat{\bm{V}}_{1j}^{\T}\cdots \hat{\bm{V}}_{Jj}^{\T}\right)^{\T}, \hat{\bm{Z}}_j = \left(\hat{\bm{z}}_{k_{j-1}+1} \cdots \hat{\bm{z}}_{k_{j}}\right)$ for all $j \in [J]$, where $\hat{\bm{V}}_{rj} \in \mathbb{R}^{(k_{r}-k_{r-1}) \times (k_{j}-k_{j-1})}$ for all $r \in [J]$, and let $f:[K] \to [J]$ be such that $s \in \left\lbrace k_{f(s)-1}+1,\ldots,k_{f(s)} \right\rbrace$. Then if $\phi_2/\lambda_{k_{t}} = o_P(1)$ for some $t \in [J]$, the following hold as $n,p \to \infty$:
\begin{subequations}
\label{equation:supp:vandzResults}
\begin{align}
    \label{equation:supp:muhat}
    &\hat{\mu}_s = \mu_s + O_P\left(np^{-1} + \phi_2^2 \lambda_s^{-1}\right), \quad s \in \left[k_{t}\right]\\
    \label{equation:supp:Vhatrr}
    &\norm{\hat{\bm{V}}_{rr}\hat{\bm{V}}_{rr}^{\T}-I_{(k_r-k_{r-1})}}_F = O_P\left(np^{-1}\lambda_{k_r}^{-1}+\phi_2^2\lambda_{k_r}^{-2} \right), \quad r \in [t]\\
    \label{equation:supp:Vhatrj}
    &\norm{\hat{\bm{V}}_{rj}}_F = \begin{cases}
    O_P\left(\frac{n}{p\lambda_{k_r}^{1/2}\lambda_{k_j}^{1/2}} + \phi_1 \lambda_{k_r}^{-1/2} + \phi_2\lambda_{k_r}^{-1} \right) & \text{if $r < j$ and $j \in [t]$}\\
    O_P\left(\frac{n}{p\lambda_{k_j}} + \phi_1 \lambda_{k_j}^{-1/2} + \phi_2\lambda_{k_j}^{-1} \right) & \text{if $r > j$ and $j \in [t]$}
    \end{cases}\\
    &\hat{\bm{z}}_s = \left(\hat{\mu}_s-1\right)^{-1}p^{-1/2}\bm{E}_2^{\T}\tilde{\bm{N}}\left(\bm{V}_{f(s)} \cdots \bm{V}_{J}\right)\left(\bm{V}_{f(s)} \cdots \bm{V}_{J}\right)^{\T}\hat{\bm{v}}_s\nonumber\\
    &+ \left(\hat{\mu}_s-1\right)^{-1}p^{-1/2}\bm{R}_s\bm{E}_2^{\T}\tilde{\bm{N}}\left(\bm{V}_{f(s)} \cdots \bm{V}_{J}\right)\left(\bm{V}_{f(s)} \cdots \bm{V}_{J}\right)^{\T}\hat{\bm{v}}_s\nonumber\\
    \label{equation:supp:zhatLemma}
    &+ O_P\left[ \lambda_s^{-1}\left\lbrace \left(\frac{n}{\lambda_s p}\right)^{1/2} + \phi_2 \lambda_{s}^{-2} \right\rbrace\left(np^{-1} + \phi_2^{2}\lambda_s^{-1}\right) \right], \quad s \in \left[k_{t}\right]\\
    \label{equation:supp:zhatnorm}
    &\norm{\hat{\bm{z}}_s}_2 = O_P\left( n^{1/2}p^{-1/2}\lambda_s^{-1/2} + \phi_2 \lambda_s^{-1} \right), \quad s \in \left[k_{t}\right],
\end{align}
\end{subequations}
where
\begin{align*}
    \bm{R}_s &= \left[ I_{n-K} + \left(\hat{\mu}_s - 1\right)^{-1}\left\lbrace I_{n-K}-p^{-1}\bm{E}_2^T \bm{E}_2 + O_P\left(np^{-1}+\phi_2^2\lambda_{k_{f(s)-1}}^{-1}\right)\right\rbrace I\left\lbrace f(s) > 1 \right\rbrace \right]^{-1} - I_{n-K}\\
    &= O_P\left\lbrace \lambda_s^{-1}\left(n^{1/2}p^{-1/2} + \phi_2\right) \right\rbrace, \quad s \in \left[k_t\right].
\end{align*}
\end{lemma}

\begin{proof}
Let $\bm{M}$ be as defined in Lemma \ref{lemma:supp:UpperBlock}. We first attempt to understand the components of $\bm{S}$. First, for some constant $c > 0$ not dependent on $n$ or $p$,
\begin{align*}
    \norm{p^{-1/2}\tilde{\bm{L}}_{\bigcdot k}^{\T}\bm{E}_2}_2 \leq c \norm{p^{-1/2}\tilde{\bm{L}}_{\bigcdot k}^{\T}\bm{E}}_2 = O_P\left(n^{1/2}p^{-1/2}\lambda_k^{1/2}\right), \quad k \in [K],
\end{align*}
where the equality follows by the proof of Lemma \ref{lemma:supp:UpperBlock}. Next, for $\bm{\Delta} = \hat{\bm{V}}^{-1} - \bm{V}^{-1}$,
\begin{align*}
    p^{-1}\bm{E}_2^T \bm{E}_{1} =& p^{-1}\left(\bm{Q}_{C}^T \hat{\bm{V}} \bm{Q}_C\right)^{-1/2}\bm{Q}_C^T\bm{E}^T \bm{E}\hat{\bm{V}}^{-1}\bm{C}\left(\bm{C}^T \hat{\bm{V}}^{-1}\bm{C}\right)^{-1/2}\bm{U}\\
    =& p^{-1}\left(\bm{Q}_{C}^T \hat{\bm{V}} \bm{Q}_C\right)^{-1/2}\bm{Q}_C^T\bm{E}^T \bm{E}\bm{V}^{-1}\bm{C}\left(\bm{C}^T \hat{\bm{V}}^{-1}\bm{C}\right)^{-1/2}\bm{U}\\
    & + p^{-1}\left(\bm{Q}_{C}^T \hat{\bm{V}} \bm{Q}_C\right)^{-1/2}\bm{Q}_C^T\bm{E}^T \bm{E}\bm{\Delta}\bm{C}\left(\bm{C}^T \hat{\bm{V}}^{-1}\bm{C}\right)^{-1/2}\bm{U},
\end{align*}
meaning
\begin{align*}
    \norm{p^{-1}\bm{E}_2^T \bm{E}_{1}}_2 = O_P\left(n^{1/2}p^{-1} + \phi_2\right)
\end{align*}
by Lemma \ref{lemma:supp:Preliminaries}. Lastly, for some constant $c > 0$ that does not depend on $n$ or $p$,
\begin{align*}
    \norm{p^{-1}\bm{E}_2^{\T}\bm{E}_2 - I_{n-K}}_2 \leq c\left( \norm{p^{-1}\bm{E}^{\T}\bm{E} - \bm{V}}_2 + \norm{ \bm{V}-\hat{\bm{V}} }_2 \right) = O_P\left(n^{1/2}p^{-1/2} + \phi_2\right).
\end{align*}
Weyl's Theorem then implies $\abs{\hat{\mu}_k/\mu_k - 1} = o_P(1)$ for all $k \in \left[k_1\right]$, meaning by the definition of $\hat{\bm{v}}_k$ and $\hat{\bm{z}}_k$,
\begin{align*}
\hat{\mu}_k\hat{\bm{v}}_{k} &= \bm{M}\hat{\bm{v}}_k + p^{-1}\tilde{\bm{N}}^T\bm{E}_2 \left(\hat{\mu}_k - p^{-1}\bm{E}_2^T  \bm{E}_2\right)^{-1}\bm{E}_2^T \tilde{\bm{N}} \hat{\bm{v}}_k, \quad k \in \left[k_1\right]\\
    \hat{\bm{z}}_k &= p^{-1/2}\left(\hat{\mu}_k - p^{-1}\bm{E}_2^T  \bm{E}_2\right)^{-1}\bm{E}_2^T\tilde{\bm{N}} \hat{\bm{v}}_k, \quad k \in \left[k_1\right]\\
    & = p^{-1/2}\left(\hat{\mu}_k - 1\right)^{-1}\bm{E}_2^T  \tilde{\bm{N}}\hat{\bm{v}}_k + p^{-1/2}\left(\hat{\mu}_k - 1\right)^{-1}\bm{R}_k\bm{E}_2^T  \tilde{\bm{N}}\hat{\bm{v}}_k, \quad k \in \left[k_1\right],
\end{align*}
where
\begin{align*}
    \bm{R}_k = \left\lbrace I_{n-K} + \left(\hat{\mu}_k - 1\right)^{-1}\left(I_{n-K}-p^{-1}\bm{E}_2^T \bm{E}_2\right) \right\rbrace^{-1} - I_{n-K} = O_P\left\lbrace \lambda_k^{-1}\left(n^{1/2}p^{-1/2} + \epsilon_V\right) \right\rbrace, \quad k \in \left[k_1\right]
\end{align*}
and
\begin{align*}
    \norm{\hat{\bm{z}}_k}_2 = O_P\left(n^{1/2}p^{-1/2}\lambda_k^{-1/2} + \phi_2 \lambda_k^{-1}\right), \quad k \in \left[k_1\right].
\end{align*}
By Weyl's theorem,
\begin{align*}
    \hat{\mu}_k = \mu_k + O_P\left( np^{-1} + \phi_2^2\lambda_k^{-1} \right) \quad , k \in \left[k_1\right].
\end{align*}
Define $\hat{\bm{V}}_1 = \left(\hat{\bm{v}}_1 \cdots \hat{\bm{v}}_{k_1}\right)$ and $\hat{\bm{Z}}_1 = \left(\hat{\bm{z}}_1 \cdots \hat{\bm{z}}_{k_1}\right)$. By definition, $\hat{\bm{V}}_1^T \hat{\bm{V}}_1 = I_{k_1\times k_1} - \hat{\bm{Z}}_1^{\T}\hat{\bm{Z}}_1$. Define
\begin{align*}
    \bm{M}^{(k)} = \bm{M} + p^{-1}\tilde{\bm{N}}^T \bm{E}_2 \left(\hat{\mu}_k - p^{-1}\bm{E}_2^T  \bm{E}_2\right)^{-1}\bm{E}_2^T \tilde{\bm{N}}, \quad k \in [k_1].
\end{align*}
Then
\begin{align*}
    \hat{\bm{v}}_k\hat{\mu}_k = \bm{M}^{(k)}\hat{\bm{v}}_k = \bm{M}\hat{\bm{v}}_k + \bm{\epsilon}_k, \quad k \in [k_1],
\end{align*}
meaning
\begin{align*}
    \bm{M}\hat{\bm{V}}_1\left(\hat{\bm{V}}_1^T \hat{\bm{V}}_1\right)^{-1/2} &= \hat{\bm{V}}_1\left(\hat{\bm{V}}_1^T \hat{\bm{V}}_1\right)^{-1/2}\diag\left(\hat{\mu}_1,\ldots,\hat{\mu}_{k_1}\right)\left(\hat{\bm{V}}_1^T \hat{\bm{V}}_1\right)^{-1/2} - \left(\bm{\epsilon}_1 \cdots \bm{\epsilon}_{k_1}\right)\left(\hat{\bm{V}}_1^T \hat{\bm{V}}_1\right)^{-1/2}\\
    +& \underbrace{\hat{\bm{V}}_1\left\lbrace I_{k_1}-\left(\hat{\bm{V}}_1^T \hat{\bm{V}}_1\right)^{-1/2} \right\rbrace\diag\left(\hat{\mu}_1,\ldots,\hat{\mu}_{k_1}\right)\left(\hat{\bm{V}}_1^T \hat{\bm{V}}_1\right)^{-1/2}}_{=O_P\left( np^{-1} + \phi_2^2 \lambda_{k_1}^{-1} \right)}.
\end{align*}
By Lemma \ref{lemma:supp:EigVector} and Corollary \ref{corollary:supp:NormSpaceVhat}, this shows that
\begin{align*}
    \norm{\hat{\bm{V}}_1 \hat{\bm{V}}_1^T - \bm{V}_1 \bm{V}_1^T}_F = O_P\left( np^{-1}\lambda_{k_1}^{-1} + \phi_2^2 \lambda_{k_1}^{-2} \right)
\end{align*}
and that
\begin{align*}
    \norm{P_{V_1}^{\perp}\hat{\bm{V}}_1 \hat{\bm{V}}_1^T}_F = O_P\left( np^{-1}\lambda_{k_1}^{-1} + \phi_2^2 \lambda_{k_1}^{-2} \right).
\end{align*}
Therefore, if we express the eigenvectors of $\bm{M}$ as $\left(\bm{V}_1 \bm{V}_1^{\perp}\right)$,
\begin{align*}
    &\bm{V}_1\diag\left(\mu_1,\ldots,\mu_{k_1}\right)\bm{V}_1^T + O_P\left(np^{-1} + \phi_2^2 \lambda_{k_1}^{-1}\right) = \bm{V}_1\diag\left(\mu_1,\ldots,\mu_{k_1}\right)\bm{V}_1^T\hat{\bm{V}}_1 \hat{\bm{V}}_1^T\\
    & + \bm{V}_1^{
    \perp}\diag\left(\mu_{k_1+1},\ldots,\mu_{K}\right)\left(\bm{V}_1^{\perp}\right)^T\hat{\bm{V}}_1 \hat{\bm{V}}_1^T = \bm{M}\hat{\bm{V}}_1 \hat{\bm{V}}_1^T=\hat{\bm{V}}_1 \diag\left(\hat{\mu}_1,\ldots,\hat{\mu}_{k_1}\right)\hat{\bm{V}}_1^T + O_P\left(np^{-1} + \phi_2^2 \lambda_{k_1}^{-1}\right),
\end{align*}
meaning
\begin{align*}
    \hat{\bm{V}}_1 \diag\left(\hat{\mu}_1,\ldots,\hat{\mu}_{k_1}\right)\hat{\bm{V}}_1^T = \bm{V}_1\diag\left(\mu_1,\ldots,\mu_{k_1}\right)\bm{V}_1^T + O_P\left(np^{-1} + \phi_2^2 \lambda_{k_1}^{-1}\right).
\end{align*}
Define $\bm{S}^{(1)} = \bm{S} - \begin{pmatrix}\hat{\bm{V}}_1\\ \hat{\bm{Z}}_1 \end{pmatrix}\diag\left(\hat{\mu}_1,\ldots,\hat{\mu}_{k_1}\right) \begin{pmatrix}\hat{\bm{V}}_1\\ \hat{\bm{Z}}_1 \end{pmatrix}^T = \begin{pmatrix} \bm{A}^{(1)} & \left(\bm{B}^{(1)}\right)^T \\ \bm{B}^{(1)} & \bm{D}^{(1)} \end{pmatrix}$, where
\begin{align*}
    \bm{A}^{(1)} =& \bm{M} - \hat{\bm{V}}_1 \diag\left(\hat{\mu}_1,\ldots,\hat{\mu}_{k_1}\right)\hat{\bm{V}}_1^T = \sum\limits_{k=k_1+1}^{K}\mu_k \bm{v}_k \bm{v}_k^T + O_P\left(np^{-1} + \phi_2^2 \lambda_{k_1}^{-1}\right)\\
    \bm{B}^{(1)} &= p^{-1/2}\bm{E}_2^T \tilde{\bm{N}} - p^{-1/2}\sum\limits_{k=1}^{k_1} \frac{\hat{\mu}_k}{\hat{\mu}_k - 1}\bm{E}_2^T \tilde{\bm{N}}\hat{\bm{v}}_k\hat{\bm{v}}_k^T - p^{-1/2}\sum\limits_{k=1}^{k_1}\frac{\hat{\mu}_k}{\hat{\mu}_k - 1} \bm{R}_k\bm{E}_2^T \tilde{\bm{N}}\hat{\bm{v}}_k\hat{\bm{v}}_k^T\\
    =& p^{-1/2}\bm{E}_2^T \tilde{\bm{N}}\left(\bm{V}_2\cdots\bm{V}_J\right) \left(\bm{V}_2\cdots\bm{V}_J\right)^T - \underbrace{p^{-1/2}\sum\limits_{k=1}^{k_1}\frac{1}{\hat{\mu}_k - 1}\bm{E}_2^T\tilde{\bm{N}}\hat{\bm{v}}_k \hat{\bm{v}}_k^T}_{=O_P\left\lbrace \left(\frac{n}{\lambda_{k_1}p}\right)^{1/2} + \phi_2\lambda_{k_1}^{-1}\right\rbrace}\\
    &- \underbrace{p^{-1/2}\sum\limits_{k=1}^{k_1}\frac{\hat{\mu}_k}{\hat{\mu}_k - 1} \bm{R}_k\bm{E}_2^T \tilde{\bm{N}}\hat{\bm{v}}_k\hat{\bm{v}}_k^T}_{=O_P\left[ \left\lbrace\left(\frac{n}{\lambda_{k_1}p}\right)^{1/2} + \phi_2\lambda_{k_1}^{-1}\right\rbrace\left(n^{1/2}p^{-1/2}+\phi_2\right) \right]} + O_P\left[ \left\lbrace\left(\frac{n}{\lambda_{k_1}p}\right)^{1/2} + \phi_2\lambda_{k_1}^{-2}\right\rbrace\left( np^{-1} + \phi_2^2\lambda_{k_1}^{-1} \right) \right]\\
    \bm{D}^{(1)} =& p^{-1}\bm{E}_2^T \bm{E}_2 - \underbrace{\sum\limits_{k=1}^{k_1}\hat{\mu}_k \hat{\bm{z}}_k \hat{\bm{z}}_k^T}_{=O_P\left(np^{-1} + \phi_2^2 \lambda_{k_1}^{-1}\right)}.
\end{align*}
Lastly, for $\tilde{\bm{\mu}}_j = \diag\left(\tau_{k_{j-1}+1},\ldots,\tau_{k_{j}}\right)$ and $\bar{\bm{L}}$ defined in the proof of Lemma \ref{lemma:supp:UpperBlock},
\begin{align}
\label{equation:supp:E2NV}
    \norm{p^{-1/2}\bm{E}_2^T \tilde{\bm{N}}\bm{V}_j}_2 \leq \norm{p^{-1/2}\bm{E}_2^T\bar{\bm{L}}}_2 \norm{\begin{pmatrix} \tilde{\bm{\mu}}_1^{1/2}\bm{V}_{1j}\\
    \tilde{\bm{\mu}}_2^{1/2}\bm{V}_{2j}\\ \vdots \\ \tilde{\bm{\mu}}_J^{1/2}\bm{V}_{Jj}\end{pmatrix}}_2 + \norm{p^{-1}\bm{E}_2^T\bm{E}_1}_2.
\end{align}
Since $\norm{\tilde{\bm{\mu}}_s^{1/2}\bm{V}_{sj}}_2 = O_P\left(\phi_1 + \phi_2\lambda_{k_1}^{-1/2}\right)$ for all $s \neq j$, $\norm{\bm{B}^{(1)}}_2 = O_P\left(\lambda_{k_2}^{1/2}n^{1/2}p^{-1/2} + \phi_2\right)$.\par
\indent The $k_1+1,\ldots,k_2$ eigenvalues and eigenvectors of $\bm{S}$ can be obtained using $\bm{A}^{(1)}, \bm{B}^{(1)}$ and $\bm{D}^{(1)}$. By the exact techniques used to analyze the first set of eigenvalues ($1,\ldots,k_1$), we get that
\begin{align*}
    \hat{\mu}_k\hat{\bm{v}}_k &= \bm{A}^{(1)}\hat{\bm{v}}_k + \left\lbrace \bm{B}^{(1)} \right\rbrace^T \left\lbrace\hat{\mu}_k - \bm{D}^{(1)}\right\rbrace^{-1}\bm{B}^{(1)} \hat{\bm{v}}_k, \quad k\in\left\lbrace k_1+1,\ldots,k_2 \right\rbrace\\
    \hat{\bm{z}}_k &= \left\lbrace\hat{\mu}_k - \bm{D}^{(1)}\right\rbrace^{-1} \bm{B}^{(1)} \hat{\bm{v}}_k, \quad k\in\left\lbrace k_1+1,\ldots,k_2 \right\rbrace.
\end{align*}
For $\hat{\bm{V}}_2 = \left(\hat{\bm{v}}_{k_1+1} \cdots \hat{\bm{v}}_{k_2}\right)$, these same techniques can also be used to show the following:
\begin{align}
    &\norm{\hat{\bm{z}}_k}_2 = O_P\left(n^{1/2}p^{-1/2}\lambda_k^{-1/2} + \phi_2\lambda_k^{-1}\right), \quad k\in\left\lbrace k_1+1,\ldots,k_2 \right\rbrace \nonumber\\
    &\hat{\mu}_k = \mu_k + O_P\left( np^{-1} + \phi_2\lambda_k^{-1} \right), \quad k\in\left\lbrace k_1+1,\ldots,k_2 \right\rbrace \nonumber\\
    \label{eqution:supp:Vhattmp}
    &\norm{ \hat{\bm{V}}_2\hat{\bm{V}}_2^{\T} - \bm{V}_2\bm{V}_2^{\T} }_F, \, \norm{P_{V_2}^{\perp}\hat{\bm{V}}_2\hat{\bm{V}}_2^{\T}}_F = O_P\left( np^{-1}\lambda_{k_2}^{-1}+\phi_2^2\lambda_{k_2}^{-2} \right)\\
    &\norm{ \hat{\bm{V}}_2\diag\left(\hat{\mu}_{k_1+1},\ldots,\hat{\mu}_{k_2}\right)\hat{\bm{V}}_2^{\T} - \bm{V}_2\diag\left(\mu_{k_1+1},\ldots,\mu_{k_2}\right)\bm{V}_2^{\T} }_2 = O_P\left( np^{-1}+\phi_2^2\lambda_{k_2}^{-1} \right)\nonumber.
\end{align}
Next, we see that for $k \in \left\lbrace k_1+1,\ldots,k_2 \right\rbrace$ and because $\hat{\bm{v}}_s^{\T}\hat{\bm{v}}_k = O_P\left(np^{-1}\lambda_{k_1}^{-1/2}\lambda_{k_2}^{-1/2} + \phi_2^2 \lambda_{k_1}^{-1}\lambda_{k_2}^{-1}\right)$ for $s \in [k_1]$,
\begin{align*}
    \bm{B}^{(1)}\hat{\bm{v}}_k =& p^{-1/2}\bm{E}_2^{\T}\tilde{\bm{N}}\left( \bm{V}_2 \cdots \bm{V}_J \right)\left( \bm{V}_2 \cdots \bm{V}_J \right)^{\T}\hat{\bm{v}}_k - p^{-1/2}p^{-1/2}\sum\limits_{s=1}^{k_1}\frac{1}{\hat{\mu}_s - 1}\bm{E}_2^T\tilde{\bm{N}}\hat{\bm{v}}_s \hat{\bm{v}}_s^T \hat{\bm{v}}_k\\
    &- p^{-1/2}\sum\limits_{s=1}^{k_1}\frac{\hat{\mu}_s}{\hat{\mu}_s - 1} \bm{R}_s\bm{E}_2^T \tilde{\bm{N}}\hat{\bm{v}}_s\hat{\bm{v}}_s^T\hat{\bm{v}}_k + O_P\left[ \left\lbrace\left(\frac{n}{\lambda_{k_1}p}\right)^{1/2} + \phi_2\lambda_{k_1}^{-2}\right\rbrace\left( np^{-1} + \phi_2^2\lambda_{k_1}^{-1} \right) \right]\\
    =& p^{-1/2}\bm{E}_2^{\T}\tilde{\bm{N}}\left( \bm{V}_2 \cdots \bm{V}_J \right)\left( \bm{V}_2 \cdots \bm{V}_J \right)^{\T}\hat{\bm{v}}_k + O_P\left[ \left\lbrace\left(\frac{n}{\lambda_{k_2}p}\right)^{1/2} + \phi_2\lambda_{k_2}^{-2}\right\rbrace\left( np^{-1} + \phi_2^2\lambda_{k_2}^{-1} \right) \right].
\end{align*}
We then have for $k \in \left\lbrace k_1+1,\ldots,k_2 \right\rbrace$,
\begin{align*}
    \hat{\bm{z}}_k =& \left\lbrace \hat{\mu}_k - \bm{D}^{(1)}\right\rbrace^{-1} \bm{B}^{(1)} \hat{\bm{v}}_k = \left(\hat{\mu}_k-1\right)^{-1}\bm{B}^{(1)}\hat{\bm{v}}_k + \left(\hat{\mu}_k-1\right)^{-1}\bm{R}_k\bm{B}^{(1)}\hat{\bm{v}}_k\\
    =& \left(\hat{\mu}_k-1\right)^{-1} p^{-1/2}\bm{E}_2^{\T}\tilde{\bm{N}}\left( \bm{V}_2 \cdots \bm{V}_J \right)\left( \bm{V}_2 \cdots \bm{V}_J \right)^{\T}\hat{\bm{v}}_k + \left(\hat{\mu}_k-1\right)^{-1}p^{-1/2}\bm{R}_k\bm{E}_2^{\T}\tilde{\bm{N}}\left( \bm{V}_2 \cdots \bm{V}_J \right)\left( \bm{V}_2 \cdots \bm{V}_J \right)^{\T}\hat{\bm{v}}_k\\
    &+O_P\left[\lambda_{k_2}^{-1} \left\lbrace\left(\frac{n}{\lambda_{k_2}p}\right)^{1/2} + \phi_2\lambda_{k_2}^{-2}\right\rbrace\left( np^{-1} + \phi_2^2\lambda_{k_2}^{-1} \right) \right]\\
    \bm{R}_k =& \left[ I_{n-K} + \left(\hat{\mu}_k - 1\right)^{-1}\left\lbrace I_{n-K}-p^{-1}\bm{E}_2^T \bm{E}_2 + O_P\left(np^{-1}+\phi_2^2\lambda_{k_1}^{-1}\right)\right\rbrace \right]^{-1} - I_{n-K} = O_P\left\lbrace \lambda_k^{-1}\left(n^{1/2}p^{-1/2} + \phi_2\right) \right\rbrace.
\end{align*}
Define $\hat{\bm{Z}}_2 = \left(\hat{\bm{z}}_{k_1+1}\cdots \hat{\bm{z}}_{k_2}\right)$ and let
\begin{align*}
    \bm{S}^{(2)} = \bm{S} - \begin{pmatrix}\hat{\bm{V}}_1\\ \hat{\bm{Z}}_1 \end{pmatrix}\diag\left(\hat{\mu}_1,\ldots,\hat{\mu}_{k_1}\right) \begin{pmatrix}\hat{\bm{V}}_1\\ \hat{\bm{Z}}_1 \end{pmatrix}^T - \begin{pmatrix}\hat{\bm{V}}_2\\ \hat{\bm{Z}}_2 \end{pmatrix}\diag\left(\hat{\mu}_{k_1+1},\ldots,\hat{\mu}_{k_2}\right) \begin{pmatrix}\hat{\bm{V}}_2\\ \hat{\bm{Z}}_2 \end{pmatrix}^T = \begin{pmatrix} \bm{A}^{(2)} & \left(\bm{B}^{(2)}\right)^T \\ \bm{B}^{(2)} & \bm{D}^{(2)} \end{pmatrix}.
\end{align*}
Then
\begin{align*}
    \bm{A}^{(2)} =& \bm{A}^{(1)} - \hat{\bm{V}}_1 \diag\left(\hat{\mu}_1,\ldots,\hat{\mu}_{k_1}\right) \hat{\bm{V}}_1^{\T} = \sum\limits_{k=k_2+1}^K \mu_k \bm{v}_k \bm{v}_k^{\T} + O_P\left(np^{-1} + \phi_2^2\lambda_{k_2}^{-1}\right)\\
    \bm{B}^{(2)} =& \bm{B}^{(1)} - \sum\limits_{k=k_1+1}^{k_2}\hat{\mu}_k\hat{\bm{z}}_k \hat{\bm{v}}_k^{\T} = p^{-1/2}\bm{E}_2^{\T}\tilde{\bm{N}}\left(\bm{V}_3 \cdots \bm{V}_J\right)\left(\bm{V}_3 \cdots \bm{V}_J\right)^{\T} - \underbrace{p^{-1/2}\sum\limits_{k=1}^{k_1}\frac{1}{\hat{\mu}_k - 1}\bm{E}_2^T\tilde{\bm{N}}\hat{\bm{v}}_k \hat{\bm{v}}_k^T}_{=O_P\left\lbrace \left(\frac{n}{\lambda_{k_1}p}\right)^{1/2} + \phi_2\lambda_{k_1}^{-1}\right\rbrace}\\
    &- \underbrace{p^{-1/2}\sum\limits_{k=k_1+1}^{k_2}\frac{1}{\hat{\mu}_k - 1}\bm{E}_2^T\tilde{\bm{N}}\left(\bm{V}_2 \cdots \bm{V}_J\right)\left(\bm{V}_2 \cdots \bm{V}_J\right)^{\T}\hat{\bm{v}}_k \hat{\bm{v}}_k^T}_{=O_P\left\lbrace \left(\frac{n}{\lambda_{k_2}p}\right)^{1/2} + \phi_2\lambda_{k_2}^{-1}\right\rbrace}\\
    &-\underbrace{p^{-1/2}\sum\limits_{k=1}^{k_1}\frac{\hat{\mu}_k}{\hat{\mu}_k - 1} \bm{R}_k\bm{E}_2^T \tilde{\bm{N}}\hat{\bm{v}}_k\hat{\bm{v}}_k^T}_{=O_P\left[ \left\lbrace\left(\frac{n}{\lambda_{k_1}p}\right)^{1/2} + \phi_2\lambda_{k_1}^{-1}\right\rbrace\left(n^{1/2}p^{-1/2}+\phi_2\right) \right]} - \underbrace{p^{-1/2}\sum\limits_{k=k_1+1}^{k_2}\frac{\hat{\mu}_k}{\hat{\mu}_k - 1} \bm{R}_k\bm{E}_2^T \tilde{\bm{N}}\left(\bm{V}_2 \cdots \bm{V}_J\right)\left(\bm{V}_2 \cdots \bm{V}_J\right)^{\T}\hat{\bm{v}}_k\hat{\bm{v}}_k^T}_{=O_P\left[ \left\lbrace\left(\frac{n}{\lambda_{k_2}p}\right)^{1/2} + \phi_2\lambda_{k_2}^{-1}\right\rbrace\left(n^{1/2}p^{-1/2}+\phi_2\right) \right]}\\
    &+O_P\left[ \left\lbrace\left(\frac{n}{\lambda_{k_2}p}\right)^{1/2} + \phi_2\lambda_{k_2}^{-2}\right\rbrace\left( np^{-1} + \phi_2^2\lambda_{k_2}^{-1} \right) \right]\\
    \bm{D}^{(2)}=& p^{-1}\bm{E}_2^{\T}\bm{E}_2 - \underbrace{\sum\limits_{k=1}^{k_2}\hat{\mu}_k \hat{\bm{z}}_k\hat{\bm{z}}_k^{\T}}_{=O_P\left(np^{-1}+\phi_2^2\lambda_{k_2}^{-1}\right)}.
\end{align*}
This can be carried out to understand the eigenstructure of the remaining groups $j=3,\ldots,t$, which proves \eqref{equation:supp:muhat}, \eqref{equation:supp:zhatLemma} and \eqref{equation:supp:zhatnorm}. For the remaining equalities, we first see that \eqref{equation:supp:Vhatrr} follows from \eqref{equation:supp:Vs_lemma} and \eqref{eqution:supp:Vhattmp}. Next, for $r > j$ and by \eqref{eqution:supp:Vhattmp},
\begin{align*}
    O_P\left(np^{-1}\lambda_{k_j}^{-1}+\phi_2^2\lambda_{k_j}^{-2}\right) = \hat{\bm{V}}_{jj}^{\T}\left(\hat{\bm{V}}_{jj} \hat{\bm{V}}_{rj}^{\T} - \bm{V}_{jj}\bm{V}_{rj}^{\T}\right) = \hat{\bm{V}}_{rj}^{\T} - \hat{\bm{V}}_{jj}^{\T}\bm{V}_{jj}\bm{V}_{rj}^{\T} +  O_P\left(np^{-1}\lambda_{k_j}^{-1}+\phi_2^2\lambda_{k_j}^{-2}\right).
\end{align*}
This implies $\norm{\hat{\bm{V}}_{rj}}_F$ follows \eqref{equation:supp:Vhatrj} for $r > j$ by \eqref{equation:supp:Vs_lemma}. The remaining part of \eqref{equation:supp:Vhatrj} follows from the fact that $\hat{\bm{V}}_r^{\T}\hat{\bm{V}}_j = -\hat{\bm{Z}}_r^{\T}\hat{\bm{Z}}_j$ for $r \neq j$.
\end{proof}

\begin{corollary}
\label{corollary:supp:RemoveBlocks}
Suppose the assumptions of Lemma \ref{lemma:supp:vandzhat} hold. Then
\begin{align*}
    \bm{S} - \sum\limits_{j=1}^t \left(\hat{\bm{V}}_j^{\T} \hat{\bm{Z}}_j^{\T}\right)^{\T} \diag\left(\hat{\mu}_{k_{j-1}+1},\ldots,\hat{\mu}_{k_j}\right) \left(\hat{\bm{V}}_j^{\T} \hat{\bm{Z}}_j^{\T}\right) = \begin{pmatrix}
    \bm{A}^{(t)} & \left(\bm{B}^{(t)}\right)^{\T}\\
    \bm{B}^{(t)} & p^{-1}\bm{E}_2^{\T}\bm{E}_2 - \bm{F}^{(t)}
    \end{pmatrix},
\end{align*}
where $\dim\left[ \im\left\lbrace\bm{F}^{(t)}\right\rbrace \right],\dim\left[ \im\left\lbrace\bm{A}^{(t)}\right\rbrace \right],\dim\left[ \im\left\lbrace\bm{B}^{(t)}\right\rbrace \right]\leq K$ and
\begin{align*}
    \bm{A}^{(t)} &= \begin{cases}
    \sum\limits_{j=t+1}^J \bm{V}_j \diag\left(\mu_{k_{j-1}+1},\ldots,\mu_{k_j}\right)\bm{V}_j^{\T} + O_P\left( np^{-1} + \phi_2^2 \lambda_{k_t}^{-1} \right) & \text{if $1 \leq t < J$}\\
    O_P\left( np^{-1} + \phi_2^2 \lambda_{K}^{-1} \right) & \text{if $t=J$}
    \end{cases}\\
    \norm{\bm{B}^{(t)}}_2 &= \begin{cases}
    O_P\left\lbrace \left(\frac{n \lambda_{k_{t+1}}}{p}\right)^{1/2} + \phi_2 \right\rbrace & \text{if $1 \leq t < J$}\\
    O_P\left\lbrace \left(\frac{n}{\lambda_K p}\right)^{1/2} + \phi_2 \lambda_K^{-1} \right\rbrace & \text{if $t=J$}
    \end{cases}\\
    \norm{\bm{F}^{(t)}}_2 &= O_P\left\lbrace np^{-1} + \phi_2^2 \lambda_{k_{\min(t,J)}} \right\rbrace
\end{align*}
as $n,p \to \infty$.
\end{corollary}

\begin{proof}
This is a direct consequence of the proof of Lemma \ref{lemma:supp:vandzhat}.
\end{proof}

\begin{remark}
\label{remark:supp:RemoveBlocks}
Corollary \ref{corollary:supp:RemoveBlocks} is analogous to Corollary S3 in \citet{CorrConf}. In fact, Lemma \ref{lemma:supp:V} (see below) can be proved using the proof of Lemma S7 in \citet{CorrConf}, where we simply replace the results of Corollary S3 in \citet{CorrConf} with Corollary \ref{corollary:supp:RemoveBlocks} above.
\end{remark}

\begin{remark}
\label{remark:supp:J2}
The conclusions of Lemma \ref{lemma:supp:vandzhat} and Corollary \ref{corollary:supp:RemoveBlocks} still hold if we re-define $k_0,k_1,\ldots,k_J$ to be those given in Remark \ref{remark:supp:J} above.
\end{remark}

\begin{corollary}
\label{corollary:supp:ImCtilde}
Suppose the assumptions of Lemma \ref{lemma:supp:vandzhat} hold. Define $\hat{\tilde{\bm{C}}}^{(t)} \in \mathbb{R}^{n \times k_t}$ to be the first $k_t$ eigenvectors of $\hat{\bm{V}}^{-1/2}\left(p^{-1}\bm{Y}^{\T}\bm{Y}\right)\hat{\bm{V}}^{-1/2}$ and let $\tilde{\bm{C}}^{(t)} = \left(\tilde{\bm{C}}_{\bigcdot 1} \cdots \tilde{\bm{C}}_{\bigcdot k_t}\right)$. Then
\begin{align*}
    \norm{ P_{\hat{\tilde{C}}^{(t)}} - P_{\tilde{C}^{(t)}} }_F^2 = O_P\left( np^{-1}\lambda_{k_t}^{-1} + \phi_2^2\lambda_{k_t}^{-2} \right)
\end{align*}
\end{corollary}

\begin{proof}
By definition,
\begin{align*}
    \hat{\tilde{\bm{C}}}^{(t)} = \tilde{\bm{C}}\left(\hat{\bm{V}}_1 \cdots \hat{\bm{V}}_t\right) + \bm{Q}\left(\hat{\bm{Z}}_1 \cdots \hat{\bm{Z}}_t\right),
\end{align*}
meaning
\begin{align*}
    \left\lbrace \tilde{\bm{C}}^{(t)} \right\rbrace^{\T}\hat{\tilde{\bm{C}}}^{(t)} = \begin{pmatrix}
    \hat{\bm{V}}_{11} & \cdots & \hat{\bm{V}}_{1t}\\
    \vdots & \ddots & \vdots\\
    \hat{\bm{V}}_{t1} & \cdots & \hat{\bm{V}}_{tt} 
    \end{pmatrix}.
\end{align*}
The result then follows by \eqref{equation:supp:Vhatrr} and \eqref{equation:supp:Vhatrj}.
\end{proof}

\begin{corollary}
\label{corollary:supp:Ctz}
Suppose the assumptions of Lemma \ref{lemma:supp:vandzhat} hold. Then for any non-random $\bm{M} \in \mathbb{R}^{n \times n}$ such that $\norm{\bm{M}}_2 \leq c$ for some constant $c > 0$ that does not depend on $n$ or $p$,
\begin{align*}
    \norm{\tilde{\bm{C}}^{\T}\hat{\bm{V}}^{1/2} \bm{M} \bm{Q}_{\tilde{C}}\hat{\bm{z}}_s}_2 = O_P\left[\phi_1 \lambda_s^{-1/2} + \frac{n}{\lambda_s p} + \phi_2\lambda_s^{-1} \right], \quad s \in \left[k_t\right].
\end{align*}
\end{corollary}


\begin{proof}
By \eqref{equation:supp:E2NV} and Lemma \ref{lemma:supp:Preliminaries}, we first see that for any $s \in \left[k_t\right]$ and non-random unit vector $\bm{u} \in \text{ker}\left(\bm{C}^{\T}\right)$,
\begin{align*}
    \lambda_s^{-1} p^{-1/2}\norm{\bm{u}^{\T}\bm{E}^{\T}\tilde{\bm{N}} \left(\bm{V}_{f(s)} \cdots \bm{V}_{J}\right)\left(\bm{V}_{f(s)} \cdots \bm{V}_{J}\right)^{\T}}_2 =& \lambda_s^{-1/2} p^{-1/2} O_P\left(\norm{ \bm{u}^{\T}\bm{E}^{\T}\bar{\bm{L}} }_2 \right)\\
    & + \lambda_s^{-1} O_P\left( \norm{p^{-1}\bm{u}^{\T}\bm{E}^{\T}\bm{E}\bm{V}^{-1}\left(n^{-1/2}\bm{C}\right)}_2 \right)\\
    & + O_P\left(\lambda_s^{-1}\phi_2\right)\\
    =& O_P\left( \phi_1\lambda_s^{-1/2} + \lambda_s^{-1}\phi_2 \right)
\end{align*}
as $n,p \to \infty$. Therefore, for all $s \in \left[k_t\right]$,
\begin{align*}
    &\norm{\tilde{\bm{C}}^{\T}\hat{\bm{V}}^{1/2} \bm{M} \bm{Q}_{\tilde{C}}\hat{\bm{z}}_s}_2 = O\left\lbrace \left(n^{-1/2}\bm{C}\right)^{\T} \bm{M} \hat{\bm{V}}\bm{Q}\left(\bm{Q}^{\T}\hat{\bm{V}}\bm{Q}\right)^{-1/2} \hat{\bm{z}}_s \right\rbrace\\
    =&\lambda_s^{-1}p^{-1/2} O_P\left\lbrace \norm{ \left(n^{-1/2}\bm{C}\right)^{\T}\bm{M}\hat{\bm{V}}^{1/2}\bm{Q}_C \left(\bm{Q}_C^{\T}\hat{\bm{V}}\bm{Q}_C\right)^{-1}\bm{Q}_C^{\T}\bm{E}^{\T}\tilde{\bm{N}} \left(\bm{V}_{f(s)} \cdots \bm{V}_{J}\right)\left(\bm{V}_{f(s)} \cdots \bm{V}_{J}\right)^{\T} }_2 \right\rbrace\\
    &+O_P\left(\lambda_s^{-1}n/p + \lambda_s^{-1}\phi_2^{-1}\right)\\
    =& \lambda_s^{-1}p^{-1/2} O_P\left\lbrace \norm{ \left(n^{-1/2}\bm{C}\right)^{\T}\bm{M}\bm{V}^{1/2}\bm{Q}_C \left(\bm{Q}_C^{\T}\bm{V}\bm{Q}_C\right)^{-1}\bm{Q}_C^{\T}\bm{E}^{\T}\tilde{\bm{N}} \left(\bm{V}_{f(s)} \cdots \bm{V}_{J}\right)\left(\bm{V}_{f(s)} \cdots \bm{V}_{J}\right)^{\T} }_2 \right\rbrace\\
    &+ O_P\left[ \frac{n}{\lambda_s p} + \phi_2\left\lbrace \lambda_s^{-1} + n^{1/2}\left(\lambda_s p\right)^{-1/2}\right\rbrace \right]\\
    =& O_P\left( \phi_1\lambda_s^{-1/2} +\frac{n}{\lambda_s p} + \phi_2\lambda_s^{-1} \right).
\end{align*}
\end{proof}

\begin{corollary}
\label{corollary:supp:CtC}
Suppose the assumptions of Corollary \ref{corollary:supp:Ctz} hold, and let $\hat{\tilde{\bm{C}}} \in \mathbb{R}^{n \times \Koracle}$ be the first $\Koracle$ right singular vectors of $\bm{Y}\hat{\bar{\bm{V}}}^{-1/2}$, where $\phi_2=\norm{\hat{\bar{\bm{V}}} - \bar{\bm{V}}}_2 = O_P(1/n)$. Then for $\tilde{\bm{C}}$ defined in \eqref{equation:supp:Params},
\begin{align*}
    &\norm{ \hat{\tilde{\bm{C}}}^{\T}\hat{\bar{\bm{V}}}^{1/2}\bm{M}\hat{\bar{\bm{V}}}^{1/2}\hat{\tilde{\bm{C}}} - \hat{\bm{v}}_1^{\T}\bm{A}^{\T}\tilde{\bm{C}}^{\T}\hat{\bar{\bm{V}}}^{1/2}\bm{M}\hat{\bar{\bm{V}}}^{1/2}\tilde{\bm{C}}\bm{A}\hat{\bm{v}}_1 }_2 = O_P\left( \phi_1\gamma_{\Koracle}^{-1/2}+\phi_3\gamma_{\Koracle}^{-1} \right)\\
    &\norm{ \hat{\tilde{\bm{C}}}^{\T}\hat{\bar{\bm{V}}}^{1/2}\bm{M}\hat{\bar{\bm{V}}}^{1/2}\tilde{\bm{C}}\bm{A} - \hat{\bm{v}}_1^{\T}\bm{A}^{\T}\tilde{\bm{C}}^{\T}\hat{\bar{\bm{V}}}^{1/2}\bm{M}\hat{\bar{\bm{V}}}^{1/2}\tilde{\bm{C}}\bm{A} }_2 = O_P\left( \phi_1\gamma_{\Koracle}^{-1/2}+\phi_3\gamma_{\Koracle}^{-1} \right),
\end{align*}
where $\bm{A} = \left(I_{\Koracle}\, \bm{0}\right)^{\T} \in \mathbb{R}^{K \times \Koracle}$, $\phi_1=p^{-1/2}$, $\phi_3=n/p+n^{-1}$ and $\hat{\bm{v}}_1$ is the upper $\Koracle \times \Koracle$ block of $\hat{\bm{v}}$.
\end{corollary}

\begin{proof}
Let $\hat{\bm{Z}} = \left(\hat{\bm{z}}_{\bigcdot 1} \cdots \hat{\bm{z}}_{\bigcdot \Koracle}\right)$ and $\hat{\bm{v}} = \begin{pmatrix}
\hat{\bm{v}}_1 & \hat{\bm{v}}_3\\
\hat{\bm{v}}_2 & \hat{\bm{v}}_4
\end{pmatrix} \in \mathbb{R}^{K \times K}$, where $\hat{\bm{v}}_1 \in \mathbb{R}^{K \times \Koracle}$ and $\hat{\bm{v}}_2 \in \mathbb{R}^{(K-\Koracle) \times \Koracle}$. By Lemma \ref{lemma:supp:vandzhat}, $\norm{\hat{\bm{v}}_2}_2 = O_P\left( \phi_1\gamma_{\Koracle}^{-1/2}+\phi_3\gamma_{\Koracle}^{-1} \right)$. Therefore,
\begin{align*}
    \hat{\tilde{\bm{C}}} = \tilde{\bm{C}}\bm{A}\hat{\bm{v}}_1 + \bm{Q}_{\tilde{C}}\hat{\bm{Z}} + O_P\left( \phi_1\gamma_{\Koracle}^{-1/2}+\phi_3\gamma_{\Koracle}^{-1} \right).
\end{align*}
Since $\norm{\hat{\bm{Z}}}_2^2 = O_P(\phi_3\gamma_{\Koracle}^{-1})$, both results follow after applying Corollary \ref{corollary:supp:Ctz}.
\end{proof}

\begin{corollary}
\label{corollary:supp:Cte}
Suppose the assumptions of Corollary \ref{corollary:supp:CtC} hold, and let $\bm{X} \in \mathbb{R}^{n}$ be a random vector that is independent of $\bm{C}$ but dependent on at most finitely many rows of $\bm{E}$. Then for $\hat{\tilde{\bm{C}}}$ defined in the statement of Corollary \ref{corollary:supp:CtC},
\begin{align*}
    \norm{ \bm{X}^{\T}\hat{\bar{\bm{V}}}^{1/2}\hat{\tilde{\bm{C}}}_{\bigcdot r} - \bm{X}^{\T} \hat{\bar{\bm{V}}}^{1/2}\tilde{\bm{C}}\hat{\bm{v}}_{\bigcdot r} }_2 = \left\lbrace \norm{\E(\bm{X})}_2+ \norm{\bm{X}-\E(\bm{X})}_2 \right\rbrace O_P\left( \phi_1\gamma_{r}^{-1/2}+\phi_3\gamma_{r}^{-1} \right), \quad r \in [\Koracle].
\end{align*}
\end{corollary}

\begin{proof}
First,
\begin{align*}
    \bm{X}^{\T} \hat{\bar{\bm{V}}}^{1/2}\hat{\tilde{\bm{C}}}_{\bigcdot r} = \bm{X}^{\T} \hat{\bar{\bm{V}}}^{1/2}\tilde{\bm{C}}\hat{\bm{v}}_{\bigcdot r} + \bm{X}^{\T} \hat{\bar{\bm{V}}}^{1/2}\bm{Q}_{\tilde{C}}\hat{\bm{z}}_{\bigcdot r}.
\end{align*}
Let $\bm{\mu}_x = \E(\bm{X})$ and $r \in [\Koracle]$. By Corollary \ref{corollary:supp:Ctz},
\begin{align*}
    \bm{\mu}_x^{\T} \hat{\bar{\bm{V}}}^{1/2}\bm{Q}_{\tilde{C}}\hat{\bm{z}}_{\bigcdot r} = \norm{ \bm{\mu}_x }_2 O_P\left( \phi_1\gamma_{r}^{-1/2}+\phi_3\gamma_{r}^{-1} \right).
\end{align*}
Therefore, to complete the proof, it suffices to assume $\E(\bm{X})=\bm{0}$ and $\norm{\bm{X}}_2 = 1$. Let $\delta_r=\phi_1\gamma_{r}^{-1/2}+\phi_3\gamma_{r}^{-1}$ and $\mathcal{S} \in [p]$ be such that $\bm{X}$ is independent of $\bm{E}_{g \bigcdot}$ for all $g \in \mathcal{S}^c$. By the proof of Corollary \ref{corollary:supp:Ctz},
\begin{align*}
    \norm{ \bm{X}^{\T} \hat{\bar{\bm{V}}}^{1/2}\bm{Q}_{\tilde{C}}\hat{\bm{z}}_{\bigcdot r} }_2 \leq & \gamma_r^{-1}p^{-1/2}\norm{ \bm{X}^{\T}\bar{\bm{V}}^{1/2}\bm{Q}_C\left(\bm{Q}_C^{\T}\bar{\bm{V}}\bm{Q}_C\right)^{-1}\bm{Q}_C^{\T}\bm{E}^{\T}\tilde{\bm{N}} \left(\bm{V}_{f(r)} \cdots \bm{V}_{J}\right)\left(\bm{V}_{f(r)} \cdots \bm{V}_{J}\right)^{\T}\tilde{\bm{v}}_{\bigcdot r}}_2\\
    &+ O_P(\delta_r)\\
    \leq & \gamma_r^{-1}p^{-1/2}\norm{ \tilde{\bm{X}}^{\T}\bm{E}^{\T}_{\mathcal{S}\bigcdot}\tilde{\bm{L}}_{\mathcal{S}\bigcdot}}_2 + \gamma_r^{-1}\norm{ \tilde{\bm{X}}^{\T}\left(p^{-1}\bm{E}^{\T}_{\mathcal{S}\bigcdot}\bm{E}_{1_{\mathcal{S}\bigcdot}}\right) }_2\\
    &+ \gamma_r^{-1}p^{-1/2}\norm{ \tilde{\bm{X}}^{\T}\bm{E}^{\T}_{\mathcal{S}^c\bigcdot}\tilde{\bm{L}}_{\mathcal{S}^c\bigcdot}\tilde{\bm{V}}}_2 + \gamma_r^{-1}\norm{ \tilde{\bm{X}}^{\T}\left(p^{-1}\bm{E}^{\T}_{\mathcal{S}^c\bigcdot}\bm{E}_{1_{\mathcal{S}^c\bigcdot}}\right) }_2 + O_P(\delta_r)
\end{align*}
where $\tilde{\bm{X}} = \bm{Q}_C\left(\bm{Q}_C^{\T}\bar{\bm{V}}\bm{Q}_C\right)^{-1}\bm{Q}_C^{\T}\bar{\bm{V}}^{1/2}\bm{X}$ (note $\tilde{\bm{X}} \in \ker(\bm{C}^{\T})$), $\tilde{\bm{V}} = \left(\bm{V}_{f(r)} \cdots \bm{V}_{J}\right)$ and $\bm{A}_{\mathcal{S}\bigcdot}$, $\bm{A}_{\mathcal{S}^c\bigcdot}$ are the sub-matrices of $\bm{A} \in \mathbb{R}^{p \times m}$ restricted to the rows $g \in \mathcal{S}$ and $g \in \mathcal{S}^c$, respectively. Since $\abs{\mathcal{S}}$ is at most finite, $\norm{ \tilde{\bm{L}}_{\mathcal{S}\bigcdot} }_2 \leq n^{1/2}p^{-1/2}c$ for some constant $c > 0$, meaning
\begin{align*}
    \gamma_r^{-1}p^{-1/2}\norm{ \tilde{\bm{X}}^{\T}\bm{E}^{\T}_{\mathcal{S}\bigcdot}\tilde{\bm{L}}_{\mathcal{S}\bigcdot}}_2 = O_P(\gamma_r^{-1}p^{-1/2}n^{1/2}) \norm{ \tilde{\bm{L}}_{\mathcal{S}\bigcdot} }_2 = O_P\{n/(\gamma_rp)\}.
\end{align*}
Similarly,
\begin{align*}
    \gamma_r^{-1}\norm{ \tilde{\bm{X}}^{\T}\left(p^{-1}\bm{E}^{\T}_{\mathcal{S}\bigcdot}\bm{E}_{1_{\mathcal{S}\bigcdot}}\right) }_2 = O_P\{n/(\gamma_rp)\}.
\end{align*}
Let $\bm{\Gamma}=\tilde{\bm{L}}^{\T}\tilde{\bm{L}}$ and
\begin{align*}
  n^{1/2}p^{-1/2}\bm{L}\left(np^{-1}\bm{L}^{\T}\bm{L}\right)^{-1/2} = \bar{\bm{L}} = \tilde{\bm{L}}\hat{\bm{W}}\bm{\Gamma}^{-1/2}
\end{align*}
for some unitary matrix $\hat{\bm{W}} \in \mathbb{R}^{K \times K}$. Then by Lemma \ref{lemma:supp:UpperBlock}, $\norm{ \bm{\Gamma}^{1/2}\tilde{\bm{V}} }_2 = O_P\left(\gamma_r^{1/2}\right)$, meaning
\begin{align*}
    \gamma_r^{-1}p^{-1/2}\norm{ \tilde{\bm{X}}^{\T}\bm{E}^{\T}_{\mathcal{S}^c\bigcdot}\tilde{\bm{L}}_{\mathcal{S}^c\bigcdot}\tilde{\bm{V}}}_2 =  \gamma_r^{-1/2}p^{-1/2}\norm{ \tilde{\bm{X}}^{\T}\bm{E}^{\T}_{\mathcal{S}^c\bigcdot}\bar{\bm{L}}_{\mathcal{S}^c\bigcdot}}_2 O_P\left(1\right) = O_P(\delta_r).
\end{align*}
Lastly,
\begin{align*}
    \gamma_r^{-1}\norm{ \tilde{\bm{X}}^{\T}\left(p^{-1}\bm{E}^{\T}_{\mathcal{S}^c\bigcdot}\bm{E}_{1_{\mathcal{S}^c\bigcdot}}\right) }_2 = O_P(\delta_r)
\end{align*}
by the proof of Corollary \ref{corollary:supp:Ctz} (since $\abs{\mathcal{S}}$ is at most finite). This completes the proof.
\end{proof}

\begin{corollary}
\label{corollary:supp:REML}
Suppose the assumption of Corollary \ref{corollary:supp:CtC} hold with $K = \Koracle$ and fix a $g \in [p]$. Assume that $\bm{E}_{g \bigcdot} = \bm{X}\bm{s} + \bm{R}_g$, where $\bm{X} \in \mathbb{R}^{n \times d}$ is an observed, mean 0 random variable that is independent of $\bm{R}_g$ and $\bm{C}$, dependent on at most finitely many other rows of $\bm{E}$ and $\norm{n^{-1}\bm{X}^{\T}\bm{X}-\bm{\Sigma}_x}_2 = o_P(1)$ for some non-random $\bm{\Sigma}_x \succ \bm{0}$. Assume that $\V(\bm{R}_g) = \bm{V}(\bm{\alpha}_g)$ for some $\bm{\alpha}_g \in \Theta_*$ and $\bm{X},\bm{R}_g$ have uniformly bounded sub-Gaussian norm. Then if $\norm{ \hat{\bar{\bm{V}}}-\bar{\bm{V}} }_2 = O_P(1/n)$, the estimator
\begin{align*}
    \hat{\bm{\alpha}}_g &= \argmax_{\bm{\theta} \in \Theta_*} f\left( \bm{\theta} \right)\\
    f\left( \bm{\theta} \right) &= -n^{-1}\log\left\lbrace \abs{ P_{M}^{\perp}\bm{V}(\bm{\theta})P_{M}^{\perp} }_+ \right\rbrace - n^{-1}(P_{M}^{\perp}\bm{Y}_{g \bigcdot})^{\T}\left\lbrace P_{M}^{\perp} \bm{V}(\bm{\theta}) P_{M}^{\perp}\right\rbrace^{\dagger}(P_{M}^{\perp}\bm{Y}_{g \bigcdot})
\end{align*}
satisfies $\norm{ \hat{\bm{\alpha}}_g - \bm{\alpha}_g }_2 = O_P\left\lbrace n^{-1/2} + n^{1/2}(p\gamma_{K})^{-1/2}\right\rbrace$, where $\bm{M} = \left( (\hat{\bar{\bm{V}}}^{1/2}\hat{\tilde{\bm{C}}})\, \bm{X} \right)$.
\end{corollary}

\begin{proof}
First,
\begin{align*}
    n^{-1}(P_{M}^{\perp}\bm{Y}_{g \bigcdot})^{\T}\left[ P_{M}^{\perp} \bm{V}(\bm{\theta}) P_{M}^{\perp}\right]^{\dagger}(P_{M}^{\perp}\bm{Y}_{g \bigcdot}) = n^{-1}\Tr\left[ (P_{M}^{\perp}\bm{Y}_{g \bigcdot})(P_{M}^{\perp}\bm{Y}_{g \bigcdot})^{\T} \left\lbrace P_{M}^{\perp} \bm{V}(\bm{\theta}) P_{M}^{\perp}\right\rbrace^{\dagger} \right].
\end{align*}
Since $\norm{ \bm{V}(\bm{\theta}) }_2$ and $\norm{ \left\lbrace \bm{V}(\bm{\theta}) \right\rbrace^{-1} }_2$ is uniformly bounded for all $\bm{\theta} \in \Theta_*$, $\norm{\bm{B}_j}_2 \leq c$ for some constant $c > 0$ and $\bm{R}_g$ is sub-Gaussian, we need only show that
\begin{align*}
    \norm{ n^{-1}(P_{M}^{\perp}\bm{Y}_{g \bigcdot})(P_{M}^{\perp}\bm{Y}_{g \bigcdot})^{\T} - n^{-1}\bm{R}_g\bm{R}_g^{\T} }_2 = O_P\left\lbrace n^{-1/2} + n^{1/2}(p\gamma_{K})^{-1/2} \right\rbrace
\end{align*}
to complete the proof. Let $\hat{\bm{C}} = \hat{\bar{\bm{V}}}^{1/2}\hat{\tilde{\bm{C}}}$. Then for $\tilde{\bm{\ell}}_g = p^{1/2}\tilde{\bm{L}}_{g \bigcdot}$ (where $\norm{\tilde{\bm{\ell}}_g}_2 \leq c$ for some constant $c>0$),
\begin{align*}
    \bm{Y}_{g \bigcdot} =& n^{1/2}\hat{\bar{\bm{V}}}^{1/2}\tilde{\bm{C}}\tilde{\bm{\ell}}_g + \bm{X}\bm{s} + \bm{R}_g\\
    n^{-1/2}P_{\hat{C}}^{\perp}\bm{Y}_{g \bigcdot} =& P_{\hat{C}}^{\perp}\hat{\bar{\bm{V}}}^{1/2}\tilde{\bm{C}}\tilde{\bm{\ell}}_g + n^{-1/2}P_{\hat{C}}^{\perp}\bm{X}\bm{s} + n^{-1/2}P_{\hat{C}}^{\perp}\bm{R}_g.
\end{align*}
Corollary \ref{corollary:supp:CtC} shows that $\norm{P_{\hat{C}}^{\perp}\hat{\bar{\bm{V}}}^{1/2}\tilde{\bm{C}}\tilde{\bm{\ell}}_g}_2 = O_P\left\lbrace n^{1/2}(p\gamma_K)^{-1/2} \right\rbrace$. Since $P_{M}^{\perp} = P_{P_{\hat{C}}^{\perp}X}^{\perp}P_{\hat{C}}^{\perp}$, we only have to understand how
\begin{align*}
    n^{-1/2}P_{M}^{\perp}\bm{R}_g = n^{-1/2}P_{P_{\hat{C}}^{\perp}X}^{\perp}P_{\hat{C}}^{\perp}\bm{R}_g = n^{-1/2}\bm{R}_g - n^{-1/2}P_{\hat{C}}^{\perp}\bm{X}\left( \bm{X}^{\T}P_{\hat{C}}^{\perp}\bm{X} \right)^{-1}\bm{X}^{\T}P_{\hat{C}}^{\perp} \bm{R}_g
\end{align*}
behaves. First,
\begin{align*}
    n^{-1}\norm{\bm{X}^{\T}\bm{X}-\bm{X}^{\T}P_{\hat{C}}^{\perp}\bm{X}}_2= n^{-1}\norm{ \bm{X}^{\T}\hat{\bm{C}}\left(\hat{\bm{C}}^{\T}\hat{\bm{C}}\right)^{-1}\hat{\bm{C}}^{\T}\bm{X} }_2 \geq c n^{-1}\norm{\bm{X}^{\T}\hat{\bm{C}}}_2^2.
\end{align*}
for some constant $c > 0$. By Corollary \ref{corollary:supp:Cte} and because $\bm{X}$ has uniformly sub-Gaussian norm,
\begin{align*}
    n^{-1/2}\norm{\bm{X}^{\T}\hat{\bm{C}}}_2 = O_P\left( n^{-1/2}+\phi_1\gamma_{K}^{-1/2}+\phi_3\gamma_K^{-1} \right).
\end{align*}
Second, by the same technique as we used above,
\begin{align*}
    n^{-1}\bm{X}^{\T}P_{\hat{C}}^{\perp} \bm{R}_g = n^{-1}\bm{X}^{\T}\bm{R}_g - n^{-1}\bm{X}^{\T} \hat{\bm{C}}\left(\hat{\bm{C}}^{\T}\hat{\bm{C}}\right)^{-1}\hat{\bm{C}}^{\T} \bm{R}_g = O_P\left\lbrace n^{-1/2}+ \left(  n^{-1/2}+\phi_1\gamma_{K}^{-1/2}+\phi_3\gamma_K^{-1} \right)^2 \right\rbrace.
\end{align*}
Putting all this together shows that
\begin{align*}
    n^{-1/2}\norm{ P_{M}^{\perp}\bm{Y}_{g \bigcdot} - \bm{R}_g }_2 = O_P\left\lbrace n^{-1/2} + n^{1/2}(p\gamma_{K})^{-1/2} \right\rbrace
\end{align*}
and completes the proof.
\end{proof}

\subsection{Estimating $\bar{\bm{V}}$}
\label{subsection:supp:EstV}
Here we derive the asymptotic properties of the estimates for $\bm{V}$ from step \ref{item:EstC:ImageC}\ref{item:EstC:V} in Algorithm \ref{algorithm:EstC}.
\begin{lemma}
\label{lemma:supp:V}
Suppose Assumptions \ref{assumption:CandL}, \ref{assumption:FALCO} and \ref{assumption:DependenceE} hold, and assume for some $t \in [J]$, the current estimate for $\bar{\bm{v}}$, $\hat{\bar{\bm{v}}}^{(0)}$, satisfies $\lambda_{k_t}^{-1}\norm{ \bm{V}\left\lbrace \hat{\bar{\bm{v}}}^{(0)} \right\rbrace - \bm{V} }_2 = o_P(1)$ as $n,p \to \infty$. Then for $P_{\hat{C}}$ defined in step \ref{item:EstC:ImageC}\ref{item:EstC:C} of Algorithm \ref{algorithm:EstC} with $k=k_t$ and $\hat{\bar{\bm{V}}} = \bm{V}\left\lbrace \hat{\bar{\bm{v}}}^{(0)} \right\rbrace$, the resulting estimator $\hat{\bar{\bm{v}}}$ in step \ref{item:EstC:ImageC}\ref{item:EstC:V} of Algorithm \ref{algorithm:EstC} satisfies
\begin{align}
\label{equation:supp:vbarhatasym}
    \norm{ \hat{\bar{\bm{v}}} - \bar{\bm{v}} }_2 = O_P\left\lbrace \max\left(\lambda_{k_t+1},1\right)n^{-1} \right\rbrace
\end{align}
as $n,p \to \infty$.
\end{lemma}

\begin{proof}
We first remark that by Lemma \ref{lemma:supp:StochEqui}, the quasi log-likelihood
\begin{align*}
    f(\bm{\theta}) = -n^{-1}\log\left\lbrace \abs{\bm{V}\left(\bm{\theta}\right)} \right\rbrace - \left(np\right)^{-1}\Tr\left[ \bm{E}^{\T}\bm{E}\left\lbrace \bm{V}\left(\bm{\theta}\right) \right\rbrace^{-1} \right]
\end{align*}
is stochastically equicontinuous for $\bm{\theta} \in \Theta_*$, where $\abs{ f\left(\bar{\bm{v}}\right) - \E\left\lbrace f\left(\bar{\bm{v}}\right) \right\rbrace } = o_P(1)$ as $n,p \to \infty$. As discussed in Remark \ref{remark:supp:RemoveBlocks}, the remainder of proof is exactly the same as the proof of Lemma S7 in \citet{CorrConf}, except we replace Corollary S3 in \citet{CorrConf} with Corollary \ref{corollary:supp:RemoveBlocks} stated above. The remaining details have been omitted. 
\end{proof}

\begin{corollary}
\label{corollary:supp:VRest}
Suppose the Assumptions of Lemma \ref{lemma:supp:V} hold, and let $\tilde{k} = O(1)$ as $n,p \to \infty$. Then if $P_{\hat{C}}$ defined in step \ref{item:EstC:ImageC}\ref{item:EstC:C} of Algorithm \ref{algorithm:EstC} is defined using $k=\tilde{k} \geq k_t$, then the resulting estimator $\hat{\bar{\bm{v}}}$ in step \ref{item:EstC:ImageC}\ref{item:EstC:V} of Algorithm \ref{algorithm:EstC} satisfies \eqref{equation:supp:vbarhatasym} as $n,p \to \infty$.
\end{corollary}

\begin{proof}
The proof follows exactly from the reasoning presented in the proof of Lemma S7 in \citet[page 40 of the Supplementary Material]{CorrConf}. The details are omitted.
\end{proof}

\begin{corollary}
\label{corollary:supp:VAlgorithm}
Let $c > K$ be a large constant not dependent on $n$ or $p$ and let $\tilde{k} \in [K]$ be such that $\limsup_{n,p \to \infty} \gamma_{\tilde{k}+1} < \infty$, where $\gamma_{K+1}=0$. Suppose Assumptions \ref{assumption:CandL}, \ref{assumption:FALCO} and \ref{assumption:DependenceE} hold. Then the estimate for $\hat{\bar{\bm{v}}}$ upon completion of step \ref{item:EstC:ImageC} of Algorithm \ref{algorithm:EstC} satisfies $\norm{\hat{\bar{\bm{v}}} - \bar{\bm{v}}}_2 = O_P\left(n^{-1}\right)$ for all $k \in \left[ \tilde{k},\min(c,K_{\max})\right]$. 
\end{corollary}

\begin{proof}
Let $s = \max\left\lbrace k \in [K]: \limsup_{n,p \to \infty} \lambda_{k} = \infty \right\rbrace$, where $s = 0$ if $\limsup_{n,p \to \infty} \lambda_{1} < \infty$. If $s = 0$, then the proof of Lemma \ref{lemma:supp:V} shows that $\norm{\hat{\bar{\bm{v}}} - \bar{\bm{v}}}_2 = O_P\left(n^{-1}\right)$ upon completion of step \ref{item:EstC:0} of Algorithm \ref{algorithm:EstC}. If $s > 1$, then by assumption, $\limsup_{n,p \to \infty} \lambda_{s+1} < \infty$, where $\lambda_{K+1} = 0$. Since $\phi_2 = O_P(1)$, $\phi_2/\lambda_{s} = o_P(1)$ as $n,p \to \infty$. The result then follows by Corollary \ref{corollary:supp:VRest}.
\end{proof}

\begin{remark}
This proves that $\norm{ \hat{\bar{\bm{v}}} - \bar{\bm{v}} }_2 = O_P(n^{-1})$ in \eqref{equation:Subspace} of Theorem \ref{theorem:AngleC}.
\end{remark}

\subsection{Estimating $\lamoracle_1,\ldots,\lamoracle_{\Koracle}$}
\label{subsection:supp:ActualEigs}
Let $\bm{M}$ be any deterministic matrix such that $n^{-1}\bm{C}^T \bm{M}\bm{C}$ is full rank with bounded minimum eigenvalue. In the main text $\bm{M} = I_n$, but $\bm{M}$ can be anything in general (i.e. maybe we only want to estimate the eigenvalues for a subset of samples). By Lemma~\ref{lemma:supp:OracleSignal} and \eqref{equation:Subspace} from Theorem \ref{theorem:AngleC} (whose proof is invariant to the choice of parametrization of $\bm{C}$), it suffices to re-define $\Loracle,\Coracle$ to be
\begin{align*}
    (\Coracle,\Loracle) \in \argmin_{(\bar{\bm{C}},\bar{\bm{L}}) \in \mathcal{S}_{\Koracle}} \norm{ (\bm{L}\bm{C}^{\T} - \bar{\bm{L}}\bar{\bm{C}}^{\T})\hat{\bm{V}}^{-1/2} }_F^2.
\end{align*}
Define
\begin{align*}
    \bm{A} = \begin{cases}
    I_{\Koracle} \oplus 0_{(K-\Koracle)\times (K-\Koracle)} & \text{if $\Koracle < K$}\\
    I_K & \text{if $K=\Koracle$}
    \end{cases}.
\end{align*}
Then $\lamoracle_1,\ldots,\lamoracle_{\Koracle}$ are exactly the eigenvalues of $\bm{A}\hat{\bm{\Gamma}}\bm{A}\hat{\bm{F}}\bm{A}$, where
\begin{subequations}
\label{equation:supp:GammaF}
\begin{align}
    &\bm{\Gamma} = \diag\left(\tau_1,\ldots,\tau_K\right), \quad \bm{F} = \tilde{\bm{C}}^{\T}\hat{\bm{V}}^{1/2}\bm{M}\hat{\bm{V}}^{1/2}\tilde{\bm{C}}\\
    &\hat{\bm{\Gamma}} = \diag\left(\hat{\mu}_1,\ldots,\hat{\mu}_K\right)-I_K, \quad \hat{\bm{F}} = \hat{\tilde{\bm{C}}}^{\T}\hat{\bm{V}}^{1/2}\bm{M}\hat{\bm{V}}^{1/2}\hat{\tilde{\bm{C}}}
\end{align}
\end{subequations}
The proof of the accuracy of these estimates is given below, which we use to prove Theorem \ref{theorem:Lambda}.
\begin{lemma}
\label{lemma:supp:TrueEigs}
Suppose Assumptions \ref{assumption:CandL}, \ref{assumption:FALCO} and \ref{assumption:DependenceE} hold. Then the estimate $\hat{\lambda}_s$ defined in step \ref{item:EstC:eigen} of Algorithm \ref{algorithm:EstC} satisfies
\begin{align*}
    \lamhatoracle_s/\lamoracle_s = 1 + O_P\left( p^{-1/2}\lambda_s^{-1/2} + \frac{n}{p\lambda_s} + \epsilon_V \lambda_s^{-1} \right), \quad s \in \left[ \Koracle \right].
\end{align*}
\end{lemma}

\begin{proof}
By the assumption of an eigengap between $\gamma_{\Koracle}$ and $\gamma_{\Koracle+1}$ in Assumption \ref{assumption:DependenceE}, it suffices to assume $\Koracle = k_{j_*}$ for some $j_* \in [J]$. Suppose
\begin{subequations}
\label{equation:supp:F1Gamma1}
\begin{align}
   &\bm{F} = \begin{pmatrix} \bm{F}_1 & \bm{F}_2\\\bm{F}_2^{\T} & \bm{F}_3 \end{pmatrix}, \quad \hat{\bm{F}} = \begin{pmatrix} \hat{\bm{F}}_1 & \hat{\bm{F}}_2\\\hat{\bm{F}}_2^{\T} & \hat{\bm{F}}_3 \end{pmatrix}\\
   &\bm{\Gamma}_1 = \diag\left(\tau_1,\ldots,\tau_{\Koracle}\right), \quad \hat{\bm{\Gamma}}_1 = \diag\left(\hat{\mu}_1,\ldots,\hat{\mu}_{\Koracle}\right)- I_{\Koracle}\\
   &\hat{\bm{v}} = \begin{pmatrix}
   \hat{\bm{v}}_1 & \hat{\bm{v}}_3\\
   \hat{\bm{v}}_2 & \hat{\bm{v}}_4
   \end{pmatrix}, \quad \hat{\bm{Z}} = \left(\hat{\bm{z}}_{\bigcdot 1}\cdots \hat{\bm{z}}_{\bigcdot \Koracle}\right)
\end{align}
\end{subequations}
where $\bm{F}_1,\hat{\bm{F}}_1,\hat{\bm{v}}_1 \in \mathbb{R}^{\Koracle \times \Koracle}$ and $\hat{\bm{v}}_2 \in \mathbb{R}^{(K-\Koracle) \times \Koracle}$. We abuse notation when defining $\hat{\bm{v}}_1,\hat{\bm{v}}_2$ here. These are not the same as the vectors $\hat{\bm{v}}_s$ defined in the proof of Lemma \ref{lemma:supp:vandzhat}. Our goal is to estimate the eigenvalues of $\bm{\Gamma}_1^{1/2}\bm{F}_1 \bm{\Gamma}_1^{1/2}$. First,
\begin{align*}
    \hat{\bm{\Gamma}}_1^{1/2}\hat{\bm{F}}_1\hat{\bm{\Gamma}}_1^{1/2} =& \hat{\bm{\Gamma}}_1^{1/2}\left(\hat{\bm{v}}_1^{\T}, \hat{\bm{v}}_2^{\T}\right)\bm{F} \left(\hat{\bm{v}}_1^{\T}, \hat{\bm{v}}_2^{\T}\right)^{\T}\hat{\bm{\Gamma}}^{1/2} + \hat{\bm{\Gamma}}_1^{1/2}\left(\hat{\bm{v}}_1^{\T}, \hat{\bm{v}}_2^{\T}\right) \tilde{\bm{C}}^T \hat{\bm{V}}^{1/2} \bm{M}\hat{\bm{V}}^{1/2} \tilde{\bm{Q}}\hat{\bm{Z}}\hat{\bm{\Gamma}}_1^{1/2}\\
    & + \left\lbrace \hat{\bm{\Gamma}}_1^{1/2}\left(\hat{\bm{v}}_1^{\T}, \hat{\bm{v}}_2^{\T}\right) \tilde{\bm{C}}^T \hat{\bm{V}}^{1/2} \bm{M}\hat{\bm{V}}^{1/2} \tilde{\bm{Q}}\hat{\bm{Z}}\hat{\bm{\Gamma}}_1^{1/2}\right\rbrace^T + O_P\left(np^{-1}+\phi_2^2\lambda_{\Koracle}^{-1}\right)\\
    =& \bm{H}^T\left(\hat{\bm{v}}_1 \hat{\bm{\Gamma}}_1 \hat{\bm{v}}_1^{\T}\right)^{1/2} \hat{\tilde{\bm{F}}}_1 \left(\hat{\bm{v}}_1 \hat{\bm{\Gamma}}_1 \hat{\bm{v}}_1^{\T}\right)^{1/2}\bm{H} +O_P\left(np^{-1}+\phi_2^2\lambda_{\Koracle}^{-1}\right)\\
    \hat{\tilde{\bm{F}}}_1 =&  \left(I_{\Koracle},\hat{\bm{v}}_1^{-\T}\hat{\bm{v}}_2^{\T}\right) \bm{F}\left(I_{\Koracle},\hat{\bm{v}}_1^{-\T}\hat{\bm{v}}_2^{\T}\right)^{\T} + \left(I_{\Koracle},\hat{\bm{v}}_1^{-\T}\hat{\bm{v}}_2^{\T}\right)\tilde{\bm{C}}^T\hat{\bm{V}}^{1/2} \bm{M}\hat{\bm{V}}^{1/2}\tilde{\bm{Q}}\hat{\bm{Z}}\hat{\bm{v}}_1^{-1}\\
    &+\left\lbrace \left(I_{\Koracle},\hat{\bm{v}}_1^{-\T}\hat{\bm{v}}_2^{\T}\right)\tilde{\bm{C}}^T\hat{\bm{V}}^{1/2} \bm{M}\hat{\bm{V}}^{1/2}\tilde{\bm{Q}}\hat{\bm{Z}}\hat{\bm{v}}_1^{-1} \right\rbrace^T
\end{align*}
where $\bm{H} \in \mathbb{R}^{K \times K}$ is a unitary matrix such that $\bm{H}\hat{\bm{\Gamma}}_1^{1/2}\hat{\bm{v}}_1^T = \left(\hat{\bm{v}}_1\hat{\bm{\Gamma}}_1\hat{\bm{v}}_1\right)^{1/2}$. By Lemma \ref{lemma:supp:vandzhat}, we can write $\hat{\bm{v}}_1$ as
\begin{align*}
    \hat{\bm{v}}_1 = \underbrace{\diag\left(\bm{U}_1,\ldots,\bm{U}_{j_*}\right)}_{=\bm{U}} + \bm{\epsilon}
\end{align*}
where $\bm{U}_j \in \mathbb{R}^{\left(k_j - k_{j-1}\right) \times \left(k_j - k_{j-1}\right)}$ is a unitary matrix and
\begin{align*}
    \bm{\epsilon} = \begin{pmatrix}
    \underbrace{O_P\left(\phi_1\lambda_1^{-1/2}+\phi_2\lambda_1^{-1} + np^{-1}\lambda_1^{-1}\right)}_{{\Koracle} \times 1} \cdots \underbrace{O_P\left(\phi_1\lambda_{\Koracle}^{-1/2}+\phi_2\lambda_{\Koracle}^{-1} + np^{-1}\lambda_{\Koracle}^{-1}\right)}_{{\Koracle} \times 1}
    \end{pmatrix}.
\end{align*}
Next,
\begin{align*}
    \hat{\bm{v}}_1^{-1} = \bm{U}^T + \bm{U}^{T}\left\lbrace\sum\limits_{t=0}^{\infty}\left(\bm{\epsilon}\bm{U}\right)^{t}\right\rbrace \left(\bm{\epsilon}\bm{U}\right)
\end{align*}
where
\begin{align*}
    \bm{\epsilon}\bm{U} = \left(\bm{\epsilon}_1\bm{U}_1 \cdots \bm{\epsilon}_{j*}\bm{U}_{j*}\right) = \left( \underbrace{O_P\left(\phi_1\lambda_{k_1}^{-1/2}+\phi_2\lambda_{k_1}^{-1} + np^{-1}\lambda_{k_1}^{-1}\right)}_{\Koracle \times k_1} \cdots \underbrace{O_P\left(\phi_1\lambda_{k_{j_*}}^{-1/2}+\phi_2\lambda_{k_{j_*}}^{-1} + np^{-1}\lambda_{k_{j_*}}^{-1}\right)}_{\Koracle \times \left(k_{j_*}-k_{{j_*}-1}\right)} \right).
\end{align*}
Therefore,
\begin{align*}
    \hat{\bm{v}}_1^{-1} = \bm{U}^T + \begin{pmatrix}
    \underbrace{O_P\left(\phi_1\lambda_1^{-1/2}+\phi_2\lambda_1^{-1} + np^{-1}\lambda_1^{-1}\right)}_{\Koracle \times 1} \cdots \underbrace{O_P\left(\phi_1\lambda_{\Koracle}^{-1/2}+\phi_2\lambda_{\Koracle}^{-1} + np^{-1}\lambda_{\Koracle}^{-1}\right)}_{\Koracle \times 1}
    \end{pmatrix}.
\end{align*}
By the proof of Lemma \ref{lemma:supp:vandzhat},
\begin{align}
    \label{equation:supp:vAGammaAv}
    \hat{\bm{v}}_1 \hat{\bm{\Gamma}}_1  \hat{\bm{v}}_1^{\T} = \sum\limits_{j=1}^{j_*}\bar{\bm{V}}_j \tilde{\bm{M}}_j\bar{\bm{V}}_j^{\T} + O_P\left(\phi_1\lambda_{\Koracle}^{-1/2} + np^{-1} + \phi_2\lambda_{\Koracle}^{-1}\right)
\end{align}
where $\tilde{\bm{M}}_j = \diag\left(\mu_{k_{j-1}+1}-1,\ldots,\mu_{k_j}-1\right)$ and $\bar{\bm{V}}_j = \left(\bm{V}_{1j}^{\T} \cdots \bm{V}_{j_*j}^{\T}\right)^{\T}$ for $\bm{V}_{rs}$, $r,s \in [j_*]$, defined in the statement of Lemma \ref{lemma:supp:UpperBlock}. By \eqref{equation:supp:UpperEigVectors:rr} and \eqref{equation:supp:UpperEigVectors:rs},
\begin{align*}
    \bm{G}_j = \left(\bar{\bm{V}}_j^{\T}\bar{\bm{V}}_j\right)^{-1/2} = I_{(k_j-k_{j-1})} - O_P\left( \phi_1^2 \lambda_{k_j}^{-1} + \phi_2^2 \lambda_{k_j}^{-2} \right).
\end{align*}
Therefore,
\begin{align*}
    \hat{\bm{v}}_1 \hat{\bm{\Gamma}}_1  \hat{\bm{v}}_1^{\T} = \sum\limits_{j=1}^{j_*}\tilde{\bm{V}}_j\bm{G}_j \tilde{\bm{M}}_j\bm{G}_j\tilde{\bm{V}}_j^{\T} + O_P\left(\phi_1\lambda_{\Koracle}^{-1/2} + np^{-1} + \phi_2\lambda_{\Koracle}^{-1}\right) = \tilde{\bm{V}}\diag\left(\tilde{\mu}_1,\ldots,\tilde{\mu}_{\Koracle}\right)\tilde{\bm{V}}^{\T}
\end{align*}
where for $\phi_3 = \phi_2 + n/p$ and by Lemmas \ref{lemma:supp:EigApprox} and \ref{lemma:supp:EigVector},
\begin{align*}
    &\tilde{\mu}_s = \tau_{s}\left\lbrace 1+ O_P\left(\phi_1\lambda_s^{-1/2} + \phi_3\lambda_s^{-1}\right) \right\rbrace, \quad s \in [\Koracle]\\
    &\norm{ \tilde{\bm{V}}_{\bigcdot (k_{j-1}+1):k_j} - \tilde{\bm{V}}_j\bm{G}_j\bm{W}_j }_F = O_P\left\lbrace \lambda_{k_j}^{-1}\left( \phi_1\lambda_{\Koracle}^{-1/2} + np^{-1} + \phi_2\lambda_{\Koracle}^{-1}\right) \right\rbrace, \quad j \in [j_*],
\end{align*}
where $\bm{W}_j \in \mathbb{R}^{(k_j-k_{j-1}) \times (k_j-k_{j-1})}$ is a unitary matrix. Next, let $\bm{A}_{ij},\tilde{\bm{V}}_{ij} \in \mathbb{R}^{\left(k_i - k_{i-1}\right) \times \left(k_j - k_{j-1}\right)}$ be the sub-matrices of $\left(\hat{\bm{v}}_1 \hat{\bm{\Gamma}}_1 \hat{\bm{v}}_1^T\right)^{1/2}$ and $\tilde{\bm{V}}$, respectively, containing the $\left[k_{i-1}+1\right]$th through $k_{i}$th rows and $\left[k_{j-1}+1\right]$th through $k_{j}$th columns. Then
\begin{align*}
    \bm{A}_{rs} =
    \sum\limits_{j=1}^{j_*} \tilde{\bm{V}}_{rj}\diag\left( \tilde{\mu}_{k_{j-1}+1}^{1/2},\ldots,\tilde{\mu}_{k_j}^{1/2}\right)\tilde{\bm{V}}_{sj}^T = O_P\left(\phi_1 + \phi_3 \lambda_{k_{\min(r,s)}}^{-1/2}\right), \quad r \neq s \in [j_*].
\end{align*}
First,
\begin{align}
\label{equation:supp:Vtilderonehalf}
    \norm{\left(\tilde{\bm{V}}_{rr} \tilde{\bm{V}}_{rr}^{\T}\right)^{-1/2} - I_{(k_r - k_{r-1})}}_F = O_P\left( \phi_1^2\lambda_{k_r}^{-1} + \phi_3^2 \lambda_{k_r}^{-1} \right), \quad r \in [j_*].
\end{align}
Therefore,
\begin{align*}
    \bm{A}_{rr} &= \tilde{\bm{V}}_{rr} \diag\left( \tilde{\mu}_{k_{r-1}+1}^{1/2},\ldots,\tilde{\mu}_{k_r}^{1/2}\right) \tilde{\bm{V}}_{rr}^T + \sum\limits_{j \neq r}^{j_*}\tilde{\bm{V}}_{rj} \diag\left( \tilde{\mu}_{k_{j-1}+1}^{1/2},\ldots,\tilde{\mu}_{k_j}^{1/2}\right) \tilde{\bm{V}}_{rj}^T\\
    &= \left\lbrace \tilde{\bm{V}}_{rr}\diag\left( \tilde{\mu}_{k_{r-1}+1},\ldots,\tilde{\mu}_{k_r}\right)\tilde{\bm{V}}_{rr}^T\right\rbrace^{1/2} + O_P\left( \phi_1 + \phi_3\lambda_{k_r}^{-1/2} \right)\\
    & = \left\lbrace \bm{V}_{rr}\diag\left( \mu_{k_{r-1}+1}-1,\ldots,\mu_{k_r}-1\right)\bm{V}_{rr}^T\right\rbrace^{1/2} + O_P\left( \phi_1 + \phi_3\lambda_{k_r}^{-1/2} \right)\\
    & = \diag\left(\tau_{k_{r-1}+1}^{1/2},\ldots,\tau_{k_r}^{1/2}\right) + O_P\left(\phi_1 + \phi_3\lambda_{k_r}^{-1/2}\right), \quad r \in [j_*],
\end{align*}
where the second equality follows by \eqref{equation:supp:Vtilderonehalf} and \eqref{equation:supp:vAGammaAv}, the third equality follows from \eqref{equation:supp:vAGammaAv} and the last equality follows from Lemma \ref{lemma:sqrtDeriv} and the fact that
\begin{align*}
    \norm{\bm{V}_{rr}\diag\left( \mu_{k_{r-1}+1}-1,\ldots,\mu_{k_r}-1\right)\bm{V}_{rr}^T - \diag\left(\tau_{k_{r-1}+1},\ldots,\tau_{k_r}\right)}_2 = O_P\left( \phi_1\lambda_{k_r}^{1/2} + \phi_2 \right)
\end{align*}
by the proof of Lemma \ref{lemma:supp:UpperBlock}. Define
\begin{align*}
    \bm{R} &= \left(\hat{\bm{v}}_1\hat{\bm{\Gamma}}_1 \hat{\bm{v}}_1^T\right)^{1/2} - \bm{\Gamma}_1^{1/2}.
\end{align*}
First,
\begin{align*}
    \left(I_{\Koracle},\hat{\bm{v}}_1^{-\T}\hat{\bm{v}}_2^{\T}\right) \bm{F}\left(I_{\Koracle},\hat{\bm{v}}_1^{-\T}\hat{\bm{v}}_2^{\T}\right)^{\T} = \bm{F}_1 + \bm{F}_2\hat{\bm{v}}_2\hat{\bm{v}}_1^{-1} + \left(\bm{F}_2\hat{\bm{v}}_2\hat{\bm{v}}_1^{-1}\right)^{\T} + \hat{\bm{v}}_1^{-\T}\hat{\bm{v}}_2^{\T}\bm{F}_3\hat{\bm{v}}_2\hat{\bm{v}}_1^{-1}.
\end{align*}
By \eqref{equation:supp:Vhatrj}, the expansion of $\hat{\bm{v}}_1^{-1}$ above and Corollary \ref{corollary:supp:Ctz},
\begin{align*}
    &\hat{\bm{v}}_1^{-\T}\hat{\bm{v}}_2^{\T}\bm{F}_3\hat{\bm{v}}_2\hat{\bm{v}}_1^{-1},\, \bm{F}_2\hat{\bm{v}}_2\hat{\bm{v}}_1^{-1} = \left( \underbrace{O_P\left( \phi_1\lambda_1^{-1/2} + \phi_3 \lambda_1^{-1} \right)}_{\Koracle \times 1} \cdots \underbrace{O_P\left( \phi_1\lambda_{\Koracle}^{-1/2} + \phi_3 \lambda_{\Koracle}^{-1} \right)}_{\Koracle \times 1} \right)\\
    & \left(I_{\Koracle},\hat{\bm{v}}_1^{-\T}\hat{\bm{v}}_2^{\T}\right)\tilde{\bm{C}}^{\T}\hat{\bm{V}}^{1/2}\bm{M}\hat{\bm{V}}^{1/2} \bm{Q}\hat{\bm{Z}}\hat{\bm{v}}_1^{-1} = \left( \underbrace{O_P\left( \phi_1\lambda_1^{-1/2} + \phi_3 \lambda_1^{-1} \right)}_{\Koracle \times 1} \cdots \underbrace{O_P\left( \phi_1\lambda_{\Koracle}^{-1/2} + \phi_3 \lambda_{\Koracle}^{-1} \right)}_{\Koracle \times 1} \right).
\end{align*}
This shows that
\begin{align*}
    &\hat{\tilde{\bm{F}}}_1 = \bm{F}_{1} + \bm{\Delta}\\
    &\bm{\Delta}_{rs} = O_P\left\lbrace \phi_1\left(\lambda_r^{-1/2}+\lambda_s^{-1/2}\right) + \phi_3\left(\lambda_r^{-1}+\lambda_s^{-1}\right) \right\rbrace, 
    \quad r,s \in [\Koracle].
\end{align*}
Next,
\begin{align*}
    \left\lbrace\bm{R}\left(\bm{F}_1+\bm{\Delta}\right)\bm{\Gamma}_1^{1/2} + \bm{\Gamma}_1^{1/2}\left(\bm{F}_1+\bm{\Delta}\right)\bm{R} + \bm{\Gamma}_1^{1/2}\bm{\Delta}\bm{\Gamma}_1^{1/2}\right\rbrace_{rs} =& O_P\left\lbrace \phi_1\left(\lambda_r^{1/2}+\lambda_s^{1/2}\right)\right.\\
    &+ \left.  \phi_3\left(\lambda_r^{1/2}\lambda_s^{-1/2}+\lambda_s^{1/2}\lambda_r^{-1/2}\right) \right\rbrace, \quad r,s \in [\Koracle].
\end{align*}
This shows that
\begin{align*}
    \left\lbrace \left(\hat{\bm{v}}_1 \hat{\bm{\Gamma}}_1 \hat{\bm{v}}_1^{\T}\right)^{1/2}\hat{\tilde{\bm{F}}}_1 \left(\hat{\bm{v}}_1 \hat{\bm{\Gamma}}_1 \hat{\bm{v}}_1^{\T}\right)^{1/2} \right\rbrace_{rs} =& \left( \bm{\Gamma}_1^{1/2}\bm{F}_1\bm{\Gamma}_1^{1/2} \right)_{rs}\\
    &+ O_P\left\lbrace \phi_1\left(\lambda_r^{1/2}+\lambda_s^{1/2}\right) + \phi_3\left(\lambda_r^{1/2}\lambda_s^{-1/2}+\lambda_s^{1/2}\lambda_r^{-1/2}\right) \right\rbrace, \quad r,s \in [\Koracle].
\end{align*}
Let $\bm{U} \in \mathbb{R}^{\Koracle \times \Koracle}$ be the eigenvectors of $\bm{\Gamma}_1^{1/2}\bm{F}_1\bm{\Gamma}_1^{1/2}$. By Lemma \ref{lemma:supp:U},
\begin{align*}
    \bm{U}_{rs} = O\left( \lambda_{r \vee s}^{1/2}\lambda_{r \wedge s}^{-1/2} \right), \quad r,s \in [\Koracle],
\end{align*}
meaning for any matrix $\bm{\Delta} \in \mathbb{R}^{\Koracle \times \Koracle}$ that satisfies
\begin{align*}
    \bm{\Delta}_{rs} = O_P\left\lbrace \phi_1\left(\lambda_r^{1/2}+\lambda_s^{1/2}\right) + \phi_3\left(\lambda_r^{1/2}\lambda_s^{-1/2}+\lambda_s^{1/2}\lambda_r^{-1/2}\right) \right\rbrace, \quad r,s \in [\Koracle],
\end{align*}
\begin{align*}
    \bm{U}_{\bigcdot r}^{\T} \bm{\Delta} \bm{U}_{\bigcdot s} = O_P\left\lbrace \phi_1\left(\lambda_r^{1/2}+\lambda_s^{1/2}\right) + \phi_3\left(\lambda_r^{1/2}\lambda_s^{-1/2}+\lambda_s^{1/2}\lambda_r^{-1/2}\right) \right\rbrace, \quad r,s \in [\Koracle].
\end{align*}
Putting this all together give us
\begin{align*}
    \hat{\bm{G}}_{rs}=&\bm{U}_{\bigcdot r}^{\T}\left(\hat{\bm{v}}_1 \hat{\bm{\Gamma}}_1 \hat{\bm{v}}_1^{\T}\right)^{1/2}\hat{\tilde{\bm{F}}}_1 \left(\hat{\bm{v}}_1 \hat{\bm{\Gamma}}_1 \hat{\bm{v}}_1^{\T}\right)^{1/2}\bm{U}_{\bigcdot s} = \lamoracle_{r}I\left(r=s\right)\\
    &+O_P\left\lbrace \phi_1\left(\lambda_r^{1/2}+\lambda_s^{1/2}\right) + \phi_3\left(\lambda_r^{1/2}\lambda_s^{-1/2}+\lambda_s^{1/2}\lambda_r^{-1/2}\right) \right\rbrace, \quad r,s \in [\Koracle].
\end{align*}
This can then be written as
\begin{align*}
    \hat{\bm{G}} = \tilde{\eta}_1 \bm{w}_1 \bm{w}_1^T + \tilde{\eta}_2 \bm{w}_2 \bm{w}_2^T + \cdots \tilde{\eta}_{\Koracle} \bm{w}_{\Koracle} \bm{w}_{\Koracle}^T + O_P\left\lbrace \phi_1^2 + \left(\phi_2 + \frac{n}{p}\right)^2\lambda_{\Koracle}^{-1} \right\rbrace
\end{align*}
where
\begin{align*}
    \tilde{\eta}_s = \lamoracle_s\left[ 1 + O_P\left\lbrace \phi_1\lambda_s^{-1/2} + \phi_3\lambda_s^{-1} \right\rbrace \right], \quad s \in [\Koracle]
\end{align*}
and
\begin{align*}
    \bm{w}_1 = \begin{pmatrix}
    1\\
    O_P\left\lbrace \phi_1\lambda_1^{-1/2} + \phi_3\left(\lambda_1\lambda_2\right)^{-1/2} \right\rbrace\\
    \vdots\\
    O_P\left\lbrace \phi_1\lambda_1^{-1/2} + \phi_3\left(\lambda_1\lambda_{\Koracle}\right)^{-1/2} \right\rbrace
    \end{pmatrix}, \, \bm{w}_2 = \begin{pmatrix}
    0\\
    1\\
    \vdots\\
    O_P\left\lbrace \phi_1\lambda_2^{-1/2} + \phi_3\left(\lambda_2\lambda_{\Koracle}\right)^{-1/2} \right\rbrace
    \end{pmatrix}, \cdots, \bm{w}_{\Koracle} = \begin{pmatrix}
    0\\
    0\\
    \vdots\\
    1
    \end{pmatrix}.
\end{align*}
To estimate the first eigenvalue, we see that
\begin{align*}
    \norm{\bm{w}_1}_2 = 1 + O_P\left\lbrace \phi_1^2\lambda_1^{-1} + \phi_3^2\left(\lambda_1\lambda_{\Koracle}\right)^{-1} \right\rbrace
\end{align*}
and
\begin{align*}
    \bm{w}_1^T \bm{w}_k = O_P\left\lbrace \phi_1\lambda_1^{-1/2} + \phi_3\left(\lambda_1\lambda_k\right)^{-1/2} \right\rbrace, \quad k \in \left\lbrace 2,\ldots,\Koracle \right\rbrace.
\end{align*}
Therefore,
\begin{align*}
    \left(\tilde{\eta}_1 \norm{\bm{w}_1}_2^2\right)^{-1}\bm{w}_1^T\hat{\bm{G}}\bm{w}_1 =& 1 + O_P\left( \frac{\lambda_2}{\lambda_1}\frac{\phi_1^2}{\lambda_1} + \frac{\phi_3^2}{\lambda_1^2} \right)\\
    \left(\tilde{\eta}_1 \norm{\bm{w}_1}_2^2\right)^{-1} \hat{\bm{G}}\bm{w}_1 =& \norm{\bm{w}_1}_2^{-1}\bm{w}_1 + \frac{\tilde{\eta}_2\left(\bm{w}_1^T \bm{w}_2\right)}{\tilde{\eta}_1 \norm{\bm{w}_1}_2^2}\bm{w}_2 + \cdots + \frac{\tilde{\eta}_{\Koracle}\left(\bm{w}_1^T \bm{w}_{\Koracle}\right)}{\tilde{\eta}_1 \norm{\bm{w}_1}_2^2}\bm{w}_{\Koracle}\\
    =& \norm{\bm{w}_1}_2^{-1}\bm{w}_1 + O_P\left( \phi_1\lambda_1^{-1/2} + \phi_3\lambda_1^{-1}\right).
\end{align*}
Therefore,
\begin{align*}
    \lamhatoracle_1 = \lamoracle_1\left\lbrace 1 + O_P\left( \phi_1\lambda_1^{-1/2} + \phi_3\lambda_1^{-1} \right) \right\rbrace.
\end{align*}
For the remaining eigenvalues, we use a similar technique to that used in the proof of Lemma \ref{lemma:supp:UpperBlock}. I will only determine $\lamoracle_2$. The remaining eigenvalues can be derived by a trivial extension. First,
\begin{align*}
    P_{w_1}^{\perp}\bm{w}_2 = \bm{w}_2 - \norm{\bm{w}_1}_2^{-2}\left(\bm{w}_1^T\bm{w}_2\right)\bm{w}_1 = \bm{w}_2 - \begin{pmatrix}
    O_P\left\lbrace \phi_1\lambda_1^{-1/2} + \phi_3\left(\lambda_1\lambda_2\right)^{-1/2} \right\rbrace\\
    O_P\left\lbrace \phi_1^2\lambda_1^{-1} + \phi_3^2\left(\lambda_1\lambda_2\right)^{-1} \right\rbrace\\
    \vdots\\
    O_P\left\lbrace \phi_1^2\lambda_1^{-1} + \frac{\phi_3^2}{\lambda_1 \left(\lambda_1\lambda_{\Koracle}\right)^{1/2}} + \frac{\phi_1\phi_3}{\lambda_1\lambda_{\Koracle}^{1/2}} \right\rbrace
    \end{pmatrix} = \bm{w}_2 - \bm{\Delta}_2.
\end{align*}
Therefore,
\begin{align*}
    \tilde{\eta}_2\bm{w}_2\bm{w}_2^T = \tilde{\eta}_2\left(P_{w_1}^{\perp}\bm{w}_2\right)\left(P_{w_1}^{\perp}\bm{w}_2\right)^T + \tilde{\eta}_2\left(P_{w_1}^{\perp}\bm{w}_2\right) \bm{\Delta}_2^T + \tilde{\eta}_2\bm{\Delta}_2\left(P_{w_1}^{\perp}\bm{w}_2\right)^T + \tilde{\eta}_2 \bm{\Delta}_2\bm{\Delta}_2^T
\end{align*}
where
\begin{align*}
    \norm{P_{w_1}^{\perp}\bm{w}_2}_2 = 1 + O_P\left\lbrace \phi_1^2\lambda_2^{-1} + \phi_3^2\left(\lambda_2\lambda_{\Koracle}\right)^{-1} \right\rbrace
\end{align*}
and for $k > 2$,
\begin{align*}
    \left(P_{w_1}^{\perp}\bm{w}_2\right)^T \bm{w}_k = \bm{w}_2^T \bm{w}_k - \bm{\Delta}_2^T \bm{w}_k = O_P\left\lbrace \phi_1\lambda_2^{-1/2} + \phi_3\left(\lambda_2\lambda_k\right)^{-1/2} \right\rbrace.
\end{align*}
A similar technique to that used above shows that
\begin{align*}
    \lamhatoracle_2 = \lamoracle_2\left\lbrace 1 + O_P\left( \phi_1\lambda_2^{-1/2} + \phi_3\lambda_2^{-1} \right) \right\rbrace.
\end{align*}
\end{proof}

\begin{remark}
\label{remark:supp:EigVectorsEta}
It is easy to see that if $(1+\epsilon)\lamoracle_{s+1}\leq \eta_{s} \leq (1-\epsilon)\lamoracle_{s-1}$ for some constant $\epsilon \in (0,1)$, then Corollary \ref{corollary:supp:NormSpaceVhat} shows that if $\tilde{\bm{w}}_s$ is the $s$th eigenvector of $\hat{\bm{G}}$, then $\norm{\tilde{\bm{w}}_s - \bm{w}_s}_2 = O_P\left(\phi_1 \lambda_s^{-1/2} + \phi_3\lambda_s^{-1}\right)$.
\end{remark}

\subsection{Estimating $\Coracle$ and $\Loracle$}
\label{subsection:supp:Eigenvectors}
We use the above work to prove Theorem \ref{theorem:AngleC}. We note that Corollary \ref{corollary:supp:VAlgorithm} shows that $\abs{\hat{\bar{\bm{v}}}_j - \bar{v}_j} = O_P(n^{-1})$ in \eqref{equation:Subspace} of Theorem \ref{theorem:AngleC}.

\begin{proof}[Proof of the rest of \eqref{equation:Subspace} in Theorem \ref{theorem:AngleC}]
Define $\delta = p^{-1/2}\lambda_{\Koracle}^{-1/2} + n/\{p\lambda_{\Koracle}\} + \{n\lambda_{\Koracle}\}^{-1}$. By Corollary \ref{corollary:supp:VAlgorithm}, the estimate $\hat{\bar{\bm{V}}}$ in Step \ref{item:EstC:ImageC}\ref{item:EstC:V} when $k = \Koracle$ satisfies $\phi_2 = \norm{ \hat{\bar{\bm{V}}} - \bar{\bm{V}} }_2 = O_P\left(n^{-1}\right)$. Let $\hat{\tilde{\bm{C}}} \in \mathbb{R}^{n \times \Koracle}$ be the first $\Koracle$ eigenvectors of $\hat{\bar{\bm{V}}}^{-1/2}\left(p^{-1}\bm{Y}^{\T}\bm{Y}\right) \hat{\bar{\bm{V}}}^{-1/2}$. By Lemma \ref{lemma:supp:OracleSignal},
\begin{align*}
    n^{-1/2}\Coracle &= \hat{\bar{\bm{V}}}^{1/2}\left( \tilde{\bm{C}}_{\bigcdot 1} \cdots \tilde{\bm{C}}_{\bigcdot \Koracle} \right)\bm{R} + O_P[\{\lambda_{\Koracle}n\}^{-1}],
\end{align*}
where $\tilde{\bm{C}}$ is defined in \eqref{equation:supp:Ctilde} and $\bm{R} \in \mathbb{R}^{\Koracle \times \Koracle}$ is an invertible matrix that satisfies $\norm{\bm{R}}_2 = O(1)$ as $n,p \to \infty$. Further,
\begin{align*}
    n^{-1/2}\hat{\bm{C}} = \hat{\bar{\bm{V}}}^{1/2}\hat{\tilde{\bm{C}}} = \hat{\bar{\bm{V}}}^{1/2}\tilde{\bm{C}} \left(\hat{\bm{v}}_{\bigcdot 1} \cdots \hat{\bm{v}}_{\bigcdot \Koracle}\right)\hat{\bm{R}} + \hat{\bar{\bm{V}}}^{1/2}\bm{Q}_{\tilde{C}}\left(\hat{\bm{z}}_{\bigcdot 1} \cdots \hat{\bm{z}}_{\bigcdot \Koracle}\right)\hat{\bm{R}},
\end{align*}
where $\hat{\bm{R}} \in \mathbb{R}^{\Koracle \times \Koracle}$ is an invertible matrix that satisfies $\norm{\hat{\bm{R}}}_2 = O(1)$ as $n,p \to \infty$. \eqref{equation:supp:Vhatrj} in Lemma \ref{lemma:supp:vandzhat} then shows that
\begin{align*}
    n^{-1/2}\hat{\bm{C}} = n^{-1/2}\Coracle \hat{\tilde{\bm{R}}} + \hat{\bar{\bm{V}}}^{1/2}\bm{Q}_{\tilde{C}}\left(\hat{\bm{z}}_{\bigcdot 1} \cdots \hat{\bm{z}}_{\bigcdot \Koracle}\right)\hat{\bm{R}} + O_P(\delta),
\end{align*}
where $\hat{\tilde{\bm{R}}} \in \mathbb{R}^{\Koracle \times \Koracle}$ is an invertible matrix (with probability tending to 1 as $n,p \to \infty$) and $\norm{\hat{\tilde{\bm{R}}}}_2 = O_P(1)$ as $n,p \to \infty$. For notational convenience, I re-define $\bm{C} \leftarrow n^{-1/2}\Coracle$ and $\hat{\bm{C}} \leftarrow n^{-1/2}\hat{\bm{C}}$ for the remainder of the proof.\par 
\indent First,
\begin{align*}
    2^{-1}\norm{ P_{\hat{C}} - P_{C} }_F^2 = \Koracle - \Tr\left\lbrace \left(\bm{C}^{\T}\bm{C}\right)^{-1} \bm{C}^{\T}\hat{\bm{C}} \left(\hat{\bm{C}}^{\T}\hat{\bm{C}}\right)^{-1}\hat{\bm{C}}^{\T}\bm{C} \right\rbrace,
\end{align*}
where by Corollary \ref{corollary:supp:Ctz},
\begin{align*}
    \bm{C}^{\T}\hat{\bm{C}} =& \bm{C}^{\T} \bm{C}\hat{\tilde{\bm{R}}} + O_P(\delta)\\
    \hat{\bm{C}}^{\T}\hat{\bm{C}} =& \hat{\tilde{\bm{R}}}^{\T}\bm{C}^{\T} \bm{C}\hat{\tilde{\bm{R}}} + O_P(\delta).
\end{align*}
This completes the proof.
\end{proof}

\begin{proof}[Proof of \eqref{equation:Linfty} in Theorem \ref{theorem:L} and \eqref{equation:CinnerChat} in Theorem \ref{theorem:AngleC}]
By Lemma \ref{lemma:supp:OracleSignal} and because $\norm{\hat{\bar{\bm{V}}}-\bar{\bm{V}}}_2 = O_P\left(n^{-1}\right)$ when $k=\Koracle$, it suffices to re-define $\Coracle$ and $\Loracle$ to be 
\begin{align*}
\lbrace \Coracle,\Loracle \rbrace = \argmin_{(\bar{\bm{C}}, \bar{\bm{L}}) \in \mathcal{S}_{\Koracle}} \norm{ ( \bm{L}\bm{C}^{\T} - \bar{\bm{L}}\bar{\bm{C}}^{\T} )\hat{\bar{\bm{V}}}^{-1/2} }_F^2.
\end{align*}
This implies for $\tilde{\bm{L}},\tilde{\bm{C}}$ defined in \eqref{equation:supp:Params} and $\bm{F}_1,\bm{\Gamma}_1$ defined in \eqref{equation:supp:GammaF} and \eqref{equation:supp:F1Gamma1} (with $\bm{M}=I_n$), there exists a unitary matrix $\bm{W} \in \mathbb{R}^{\Koracle \times \Koracle}$ such that
\begin{subequations}
\label{equation:supp:W}
\begin{align}
    &n^{-1/2}\Coracle = \hat{\bar{\bm{V}}}^{1/2}\tilde{\bm{C}}\bm{A}\bm{F}_1^{-1/2}\bm{W}\\
    &n^{1/2}p^{-1/2}\Loracle = \tilde{\bm{L}}\bm{A}\bm{F}_1^{1/2}\bm{W}\\
    &\bm{W}^{\T}\bm{F}_1^{1/2}\bm{\Gamma}_1 \bm{F}_1^{1/2}\bm{W} = \diag\left\lbrace \lamoracle_1,\ldots,\lamoracle_{\Koracle} \right\rbrace,
\end{align}
\end{subequations}
where $\bm{A} = \left(I_{\Koracle}\,\bm{0}\right)^{\T} \in \mathbb{R}^{K \times \Koracle}$. Additionally, for $\hat{\bm{F}}_1,\hat{\bm{\Gamma}}_1$ defined in \eqref{equation:supp:GammaF} and \eqref{equation:supp:F1Gamma1} (with $\bm{M}=I_n$),
\begin{subequations}
\label{equation:supp:What}
\begin{align}
    &n^{-1/2}\hat{\bm{C}} = \hat{\bar{\bm{V}}}^{1/2}\hat{\tilde{\bm{C}}} \hat{\bm{F}}_1^{-1/2}\hat{\bm{W}} = \hat{\bar{\bm{V}}}^{1/2}\tilde{\bm{C}}\begin{pmatrix}
    \hat{\bm{v}}_1\\
    \hat{\bm{v}}_2
    \end{pmatrix} \hat{\bm{F}}_1^{-1/2}\hat{\bm{W}} + \hat{\bar{\bm{V}}}^{1/2}\bm{Q}_{\tilde{C}}\hat{\bm{Z}}\hat{\bm{F}}_1^{-1/2}\hat{\bm{W}}\\
    &n^{1/2}p^{-1/2}\hat{\bm{L}} = \hat{\tilde{\bm{L}}} \hat{\bm{F}}_1^{1/2}\hat{\bm{W}}\\
    &\hat{\bm{W}}^{\T}\hat{\bm{F}}_1^{1/2}\hat{\bm{\Gamma}}_1 \hat{\bm{F}}_1^{1/2}\hat{\bm{W}} = \diag\left\lbrace \lamhatoracle_1,\ldots,\lamhatoracle_{\Koracle} \right\rbrace,
\end{align}
\end{subequations}
where $\hat{\tilde{\bm{C}}} \in \mathbb{R}^{n \times \Koracle}$ are the first $\Koracle$ right singular vectors of $\bm{Y}\hat{\bar{\bm{V}}}^{1/2}$, $\hat{\bm{v}}_1 \in \mathbb{R}^{\Koracle \times \Koracle},\hat{\bm{v}}_2 \in \mathbb{R}^{(K-\Koracle) \times \Koracle}$ and $\hat{\bm{Z}} \in \mathbb{R}^{(n-K) \times \Koracle}$ are defined in \eqref{equation:supp:F1Gamma1} and 
\begin{align}
\label{equation:supp:TildeLHat}
    \hat{\tilde{\bm{L}}} = p^{-1/2}\bm{Y}\hat{\bar{\bm{V}}}^{1/2}\hat{\tilde{\bm{C}}} = \tilde{\bm{L}}\begin{pmatrix}
    \hat{\bm{v}}_1\\
    \hat{\bm{v}}_2
    \end{pmatrix} + p^{-1/2}\bm{E}_1\begin{pmatrix}
    \hat{\bm{v}}_1\\
    \hat{\bm{v}}_2
    \end{pmatrix} + p^{-1/2}\bm{E}_2\hat{\bm{Z}}.
\end{align}
In the above equation, $\bm{E}_1 \in \mathbb{R}^{p \times K},\bm{E}_2 \in \mathbb{R}^{p \times (n-K)}$ are defined in \eqref{equation:supp:E1E2}. Let $\bm{D} = \diag\left\lbrace \lamoracle_1,\ldots,\lamoracle_{\Koracle} \right\rbrace$ and $\hat{\bm{D}} = \diag\left\lbrace \lamhatoracle_1,\ldots,\lamhatoracle_{\Koracle} \right\rbrace$. Then for $\bm{U},\hat{\bm{U}} \in \mathbb{R}^{\Koracle \times \Koracle}$ the eigenvectors of $\bm{\Gamma}_1^{1/2}\bm{F}_1\bm{\Gamma}_1^{1/2}$ and $\hat{\bm{\Gamma}}_1^{1/2}\hat{\bm{F}}_1\hat{\bm{\Gamma}}_1^{1/2}$, respectively,
\begin{align*}
    &\bm{W}^T \bm{F}_1^{1/2}\bm{\Gamma}_1\bm{F}_1^{1/2}\bm{W} = \bm{D} = \bm{U}^T \bm{\Gamma}_1^{1/2}\bm{F}_1\bm{\Gamma}_1^{1/2}\bm{U} \Rightarrow \bm{F}_1^{-1/2}\bm{W} = \bm{F}_1^{-1}\bm{\Gamma}_1^{-1/2}\bm{U}\bm{D}^{1/2}\\
    &\hat{\bm{W}}^T \hat{\bm{F}}_1^{1/2}\hat{\bm{\Gamma}}_1\hat{\bm{F}}_1^{1/2}\hat{\bm{W}} = \hat{\bm{D}} = \hat{\bm{U}}^T \hat{\bm{\Gamma}}_1^{1/2}\hat{\bm{F}}_1\hat{\bm{\Gamma}}_1^{1/2}\hat{\bm{U}} \Rightarrow \hat{\bm{F}}_1^{-1/2}\hat{\bm{W}} = \hat{\bm{F}}_1^{-1}\hat{\bm{\Gamma}}_1^{-1/2}\hat{\bm{U}}\hat{\bm{D}}^{1/2}
\end{align*}
where
\begin{align*}
    \hat{\bm{F}}_1^{-1/2}\hat{\bm{W}} = \hat{\bm{F}}_1^{-1}\hat{\bm{\Gamma}}_1^{-1/2}\hat{\bm{U}}\hat{\bm{D}}^{1/2} = \hat{\bm{\Gamma}}_1^{1/2} \underbrace{\left( \hat{\bm{\Gamma}}_1^{1/2}\hat{\bm{F}}_1\hat{\bm{\Gamma}}_1^{1/2} \right)^{-1}}_{=\hat{\bm{U}}\hat{\bm{D}}^{-1}\hat{\bm{U}}^T}\hat{\bm{U}}\hat{\bm{D}}^{1/2} = \hat{\bm{\Gamma}}_1^{1/2} \hat{\bm{U}}\hat{\bm{D}}^{-1/2}.
\end{align*}
This implies for $\bm{F}_2$ defined in \eqref{equation:supp:F1Gamma1},
\begin{align*}
    n^{-1}\{ \Coracle \}^{\T}\hat{\bm{C}} =& \bm{W}^T \bm{F}_1^{-1/2}\left( \bm{F}_1\hat{\bm{v}}_1 + \bm{F}_2\hat{\bm{v}}_2 + \tilde{\bm{C}}^T \hat{\bm{V}} \bm{Q}_{\tilde{C}}\hat{\bm{Z}} \right)\hat{\bm{F}}_1^{-1/2}\hat{\bm{W}}\\
     =& \bm{D}^{1/2}\bm{U}^{\T}\bm{\Gamma}_1^{-1/2}\left( \hat{\bm{v}}_1 + \bm{F}_1^{-1}\bm{F}_2\hat{\bm{v}}_2 + \bm{F}_1^{-1}\tilde{\bm{C}}^T \hat{\bm{V}} \bm{Q}_{\tilde{C}}\hat{\bm{Z}} \right)\hat{\bm{F}}_1^{-1/2}\hat{\bm{W}}\\
     =& \bm{D}^{1/2}\bm{U}^{\T}\bm{\Gamma}_1^{-1/2} \left( \hat{\bm{v}}_1 \bm{\Gamma}_1 \hat{\bm{v}}_1^{\T}\right)^{1/2}\bm{H}\hat{\bm{U}}\hat{\bm{D}}^{-1/2} + \bm{F}_1^{-1}\left( \bm{F}_2\hat{\bm{v}}_2 + \tilde{\bm{C}}^T \hat{\bm{V}} \bm{Q}_{\tilde{C}}\hat{\bm{Z}} \right)\hat{\bm{\Gamma}}_1^{1/2} \hat{\bm{U}}\hat{\bm{D}}^{-1/2},
\end{align*}
where $\bm{H} \in \mathbb{R}^{\Koracle \times \Koracle}$ is the same unitary matrix defined in the proof of Lemma \ref{lemma:supp:TrueEigs} and satisfies $\hat{\bm{v}}_1\hat{\bm{\Gamma}}_1^{1/2} = \left( \hat{\bm{v}}_1 \bm{\Gamma}_1 \hat{\bm{v}}_1^{\T}\right)^{1/2}\bm{H}$. By the proof of Lemma \ref{lemma:supp:TrueEigs}, $\hat{\bm{U}} = \bm{H}^{\T}\bm{U}\tilde{\bm{W}}$ for the unitary matrix $\tilde{\bm{W}} = \left(\tilde{\bm{w}}_1 \cdots \tilde{\bm{w}}_{\Koracle}\right) \in \mathbb{R}^{\Koracle \times \Koracle}$, where $\tilde{\bm{w}}_s$ is defined in Remark \ref{remark:supp:EigVectorsEta}. Next, \eqref{equation:supp:Vhatrj} in Lemma \ref{lemma:supp:vandzhat} and Corollary \ref{corollary:supp:Ctz} imply
\begin{align*}
    \bm{F}_1^{-1}\left( \bm{F}_2\hat{\bm{v}}_2 + \tilde{\bm{C}}^T \hat{\bm{V}} \bm{Q}_{\tilde{C}}\hat{\bm{Z}} \right) =& \left( \underbrace{O_P\left\lbrace p^{-1/2}\gamma_1^{-1/2}+n/(\gamma_1p) + (\gamma_1n)^{-1} \right\rbrace}_{\Koracle \times 1} \right.\\
    &\left.\cdots \underbrace{O_P\left\lbrace p^{-1/2}\gamma_{\Koracle}^{-1/2}+n/(\gamma_{\Koracle}p) + (\gamma_{\Koracle}n)^{-1} \right\rbrace}_{\Koracle \times 1} \right).
\end{align*}
Further, $\hat{\bm{U}}_{rs} = O_P\left( \gamma_{r\vee s}^{1/2}\gamma_{r\wedge s}^{-1/2} \right)$ by Lemma \ref{lemma:supp:U}. Therefore, for any $r \in [\Koracle]$,
\begin{align*}
    n^{-1}\{ \Coracle_{\bigcdot r} \}^{\T}\hat{\bm{C}}_{\bigcdot r} = \left\lbrace \frac{\lamoracle_r}{\lamhatoracle_r} \right\rbrace^{1/2} \bm{U}_{\bigcdot r}^{\T} \bm{\Gamma}_1^{-1/2}\left( \hat{\bm{v}}_1 \bm{\Gamma}_1 \hat{\bm{v}}_1^{\T}\right)^{1/2} \bm{U} \tilde{\bm{w}}_r + O_P\left\lbrace p^{-1/2}\gamma_{r}^{-1/2}+n/(\gamma_{r}p) + (\gamma_{r}n)^{-1} \right\rbrace.
\end{align*}
Therefore, we need only understand how $\bm{U}_{\bigcdot r}^{\T} \bm{\Gamma}_1^{-1/2}\left( \hat{\bm{v}}_1 \bm{\Gamma}_1 \hat{\bm{v}}_1^{\T}\right)^{1/2} \bm{U} \tilde{\bm{w}}_r$ behaves. As we did in the proof of Lemma \ref{lemma:supp:TrueEigs}, let $\bm{R} = \left( \hat{\bm{v}}_1 \bm{\Gamma}_1 \hat{\bm{v}}_1^{\T}\right)^{1/2} - \bm{\Gamma}^{1/2}$. We showed in the proof of Lemma \ref{lemma:supp:TrueEigs} that $\bm{R}_{rs}= O_P\left( \phi_1 + \phi_3\gamma_{r \wedge s}^{-1/2} \right)$ for $\phi_1=p^{-1/2}$ and $\phi_3=n/p + n^{-1}$. Therefore,
\begin{align*}
    &\bm{U}_{\bigcdot r}^{\T} \bm{\Gamma}_1^{-1/2}\left( \hat{\bm{v}}_1 \bm{\Gamma}_1 \hat{\bm{v}}_1^{\T}\right)^{1/2} \bm{U} \tilde{\bm{w}}_r = \tilde{\bm{w}}_{r_r} + \bm{U}_{\bigcdot r}^{\T} \bm{\Gamma}_1^{-1/2}\bm{R} \bm{U} \tilde{\bm{w}}_r\\
    &\bm{U}_{\bigcdot r}^{\T} \bm{\Gamma}_1^{-1/2}\bm{R} \bm{U}_{\bigcdot s} = O_P\left(\phi_1\gamma_r^{-1/2} + \phi_3\gamma_r^{-1/2}\gamma_s^{-1/2}\right), \quad s \in [\Koracle].
\end{align*}
The proof of Lemma \ref{lemma:supp:TrueEigs} and Remark \ref{remark:supp:EigVectorsEta} show that on the event $F_r^{(\epsilon)}$, defined in the statement of Theorem \ref{theorem:L},
\begin{align*}
    \abs{\tilde{\bm{w}}_{r_r}} &= 1 - O_P\left(\phi_1\gamma_r^{-1/2} + \phi_3\gamma_r^{-1}\right)\\
    \tilde{\bm{w}}_{r_s} &= O_P\left( \phi_1\gamma_r^{-1/2} + \phi_3\left\lbrace\gamma_r\gamma_s\right\rbrace^{-1/2} \right), \quad r < s \in [\Koracle].
\end{align*}
This completes the proof of \eqref{equation:CinnerChat}.\par 
\indent It remains to prove \eqref{equation:Linfty}. Using the expression of $\hat{\tilde{\bm{L}}}$ in \eqref{equation:supp:TildeLHat},
\begin{align*}
    \hat{\bm{L}} = p^{1/2}n^{-1/2}\tilde{\bm{L}}_1\hat{\bm{v}}_1\hat{\bm{\Gamma}}_1^{-1/2}\hat{\bm{U}}\hat{\bm{D}}^{1/2} + p^{1/2}n^{-1/2}\tilde{\bm{L}}_2\hat{\bm{v}}_2 \hat{\bm{\Gamma}}_1^{-1/2}\hat{\bm{U}}\hat{\bm{D}}^{1/2} + n^{-1/2}\bm{E}_1\begin{pmatrix}
    \hat{\bm{v}}_1\\
    \hat{\bm{v}}_2
    \end{pmatrix} + n^{-1/2}\bm{E}_2 \hat{\bm{Z}},
\end{align*}
where $\tilde{\bm{L}} = (\tilde{\bm{L}}_1 \, \tilde{\bm{L}}_2)$. First, for any $g \in [p]$ and some constant $c > 0$ and random sequence $a_{n,p}=O_P(1/n)$ that does not depend on $g$,
\begin{align*}
    \norm{ \bm{E}_{1_{g \bigcdot}} }_2 \leq c\norm{ n^{-1/2} \bm{C}^{\T} \hat{\bar{\bm{V}}}^{-1/2} \bm{E}_{g \bigcdot} }_2 \leq c\norm{ n^{-1/2} \bm{C}^{\T} \bar{\bm{V}}^{-1/2}\bm{E}_{g \bigcdot} }_2 + a_{n,p}\norm{ \bm{E}_{g \bigcdot} }_2
\end{align*}
Since $n^{-1/2} \bm{C}^{\T} \bar{\bm{V}}^{-1/2}\bm{E}_{g \bigcdot}$ is sub-Gaussian with uniformly bounded sub-Gaussian norm,
\begin{align*}
    \sup_{g \in [p]}\norm{ n^{-1/2} \bm{C}^{\T} \bar{\bm{V}}^{-1/2}\bm{E}_{g \bigcdot} }_2 = O_P\left\lbrace \log(p) \right\rbrace.
\end{align*}
Further, for some constant $c > 0$ and all $t \geq 0$,
\begin{align*}
    &\sup_{g \in [p]}\norm{ \bm{E}_{g \bigcdot} }_2^2 \leq cn + \sup_{g \in [p]}\abs{ \bm{E}_{g \bigcdot}^{\T}\bm{E}_{g \bigcdot} - \E\left( \bm{E}_{g \bigcdot}^{\T}\bm{E}_{g \bigcdot} \right) }\\
    &\Prob\left\lbrace \abs{ \bm{E}_{g \bigcdot}^{\T}\bm{E}_{g \bigcdot} - \E\left( \bm{E}_{g \bigcdot}^{\T}\bm{E}_{g \bigcdot} \right) } \geq tn^{1/2} \right\rbrace \leq 2\exp\left\lbrace -\min\left(ct^2,ct\right) \right\rbrace,
\end{align*}
where the second equality follows from Lemma \ref{lemma:supp:Preliminaries}. Therefore, if $\log(p)/n^{1/2} \to 0$,
\begin{align*}
    \sup_{g \in [p]}\norm{ \bm{E}_{g \bigcdot} }_2 \leq c^{1/2} n^{1/2}\left[ 1+O_P\left\lbrace \log(p)n^{-1/2} \right\rbrace \right].
\end{align*}
This shows that
\begin{align*}
    \norm{ n^{-1/2}\bm{E}_1\begin{pmatrix}
    \hat{\bm{v}}_1\\
    \hat{\bm{v}}_2
    \end{pmatrix}_{\bigcdot r} + n^{-1/2}\bm{E}_2 \hat{\bm{Z}}_{\bigcdot r} }_{\infty} = O_P\left\lbrace \log(p)n^{-1/2}+\left( \frac{n}{\gamma_{r}p} \right)^{1/2} + n^{-1} \right\rbrace.
\end{align*}
Next, since $\norm{\hat{\bm{\Gamma}}_1^{-1/2}\hat{\bm{U}}\hat{\bm{D}}^{1/2}}_2 = O_P(1)$ and by Lemma \ref{lemma:supp:vandzhat},
\begin{align*}
    \hat{\bm{v}}_2\norm{\hat{\bm{\Gamma}}_1^{-1/2}\hat{\bm{U}}\hat{\bm{D}}^{1/2}}_2 = O_P\left( \phi_1\gamma_{\Koracle}^{-1/2} + \phi_3\gamma_{\Koracle}^{-1} \right).
\end{align*}
Therefore, it suffices to assume $\hat{\bm{L}} = p^{1/2}n^{-1/2}\tilde{\bm{L}}_1\hat{\bm{v}}_1\hat{\bm{\Gamma}}_1^{-1/2}\hat{\bm{U}}\hat{\bm{D}}^{1/2}$. We then see that
\begin{align*}
    \Loracle_{\bigcdot r}-\hat{\bm{L}}_{\bigcdot r} = p^{1/2}n^{-1/2}\tilde{\bm{L}}_1\left[ \{\lamoracle_r\}^{1/2}\bm{\Gamma}_1^{-1/2}\bm{U}_{\bigcdot r} - \{\lamhatoracle_r\}^{1/2}\hat{\bm{v}}_1\hat{\bm{\Gamma}}_1^{-1/2}\underbrace{\bm{H}^{\T}\bm{U}\tilde{\bm{w}}_r}_{=\hat{\bm{U}}_{\bigcdot r}} \right].
\end{align*}
Since $\norm{\gamma_r^{1/2}\hat{\bm{v}}_1\hat{\bm{\Gamma}}_1^{-1/2}\hat{\bm{U}}_{\bigcdot r}}_2 = O_P(1)$,
\begin{align*}
    \{\lamoracle_r\}^{1/2}\bm{\Gamma}_1^{-1/2}\bm{U}_{\bigcdot r} - \{\lamhatoracle_r\}^{1/2}\hat{\bm{v}}_1\hat{\bm{\Gamma}}_1^{-1/2}\bm{H}^{\T}\bm{U}\tilde{\bm{w}}_r =& \{\lamoracle_r\}^{1/2}\left( \bm{\Gamma}_1^{-1/2}\bm{U}_{\bigcdot r} - \hat{\bm{v}}_1\hat{\bm{\Gamma}}_1^{-1/2}\bm{H}^{\T}\bm{U}\tilde{\bm{w}}_r \right)\\
    &+ O_P\left( \phi_1\gamma_{r}^{-1/2} + \phi_3 \gamma_r^{-1} \right).
\end{align*}
Next, since $\bm{H}^{\T} = \hat{\bm{\Gamma}}_1^{1/2} \hat{\bm{v}}_1^{\T}\left( \hat{\bm{v}}_1 \hat{\bm{\Gamma}}_1 \hat{\bm{v}}_1^{\T} \right)^{-1/2}$ and for $\bm{a}_r \in \mathbb{R}^{\Koracle}$ the $r$th standard basis vector,
\begin{align*}
    \hat{\bm{v}}_1\hat{\bm{\Gamma}}_1^{-1/2}\bm{H}^{\T}\bm{U}\tilde{\bm{w}}_r = \hat{\bm{v}}_1 \hat{\bm{v}}_1^{\T} \left( \hat{\bm{v}}_1 \hat{\bm{\Gamma}}_1 \hat{\bm{v}}_1^{\T} \right)^{-1/2} \bm{U}_{\bigcdot r} + \hat{\bm{v}}_1 \hat{\bm{v}}_1^{\T} \left( \hat{\bm{v}}_1 \hat{\bm{\Gamma}}_1 \hat{\bm{v}}_1^{\T} \right)^{-1/2} \bm{U}\left(\tilde{\bm{w}}_r - \bm{a}_r\right).
\end{align*}
By the proof of Lemma \ref{lemma:supp:TrueEigs} and Remark \ref{remark:supp:EigVectorsEta}, $\norm{ \tilde{\bm{w}}_r - \bm{a}_r }_2 = O_P\left( \phi_1\gamma_r^{-1/2} + \phi_3\gamma_{r}^{-1/2}\gamma_{\Koracle}^{-1/2} \right)$. Therefore,
\begin{align*}
    \{\lamoracle_r\}^{1/2}\norm{ \hat{\bm{v}}_1 \hat{\bm{v}}_1^{\T} \left( \hat{\bm{v}}_1 \hat{\bm{\Gamma}}_1 \hat{\bm{v}}_1^{\T} \right)^{-1/2} \bm{U}\left(\tilde{\bm{w}}_r - \bm{a}_r\right) }_2 = O_P\left( \phi_1\gamma_{\Koracle}^{-1/2} + \phi_3\gamma_{\Koracle}^{-1} \right).
\end{align*}
Next, for $\bm{R}$ defined above,
\begin{align*}
    &\left( \hat{\bm{v}}_1 \hat{\bm{\Gamma}}_1 \hat{\bm{v}}_1^{\T} \right)^{-1/2} = \bm{\Gamma}_1^{-1/2}\left\lbrace I_{\Koracle} - \sum\limits_{t=0}^{\infty}(-1)^t \left(\bm{R}\bm{\Gamma}_1^{-1/2}\right)^t (\bm{R}\bm{\Gamma}_1^{-1/2}) \right\rbrace\\
    &\{\lamoracle_r\}^{1/2}\gamma_{\Koracle}^{-1/2}\bm{R}\bm{\Gamma}_1^{-1/2}\bm{U}_{\bigcdot r} = O_P\left(\phi_1\gamma_{\Koracle}^{-1/2} + \phi_3\gamma_{\Koracle}^{-1}\right).
\end{align*}
Lastly, by Lemma \ref{lemma:supp:vandzhat},
\begin{align*}
    \norm{ \{\lamoracle_r\}^{1/2}\left(I_{\Koracle} - \hat{\bm{v}}_1\hat{\bm{v}}_1^{\T}\right)\bm{\Gamma}_1^{-1/2}\bm{U}_{\bigcdot r} }_2 = \norm{ I_{\Koracle} - \hat{\bm{v}}_1\hat{\bm{v}}_1^{\T} }_2O_P(1) = O_P\left(\phi_1\gamma_{\Koracle}^{-1/2} + \phi_3\gamma_{\Koracle}^{-1}\right).
\end{align*}
Putting this all together gives us
\begin{align*}
    \norm{ \{\lamoracle_r\}^{1/2}\bm{\Gamma}_1^{-1/2}\bm{U}_{\bigcdot r} - \{\lamhatoracle_r\}^{1/2}\hat{\bm{v}}_1\hat{\bm{\Gamma}}_1^{-1/2}\bm{H}^{\T}\bm{U}\tilde{\bm{w}}_r }_2 = O_P\left( \phi_1\gamma_{\Koracle}^{-1/2} + \phi_3 \gamma_{\Koracle}^{-1} \right).
\end{align*}
This shows that on the event $F_r^{(\epsilon)}$,
\begin{align*}
    \norm{ \Loracle_{\bigcdot r} - a\hat{\bm{L}}_{\bigcdot r} }_{\infty} = O_P\left\lbrace \log(p)n^{-1/2} + \left( \frac{n}{\gamma_{\Koracle} p} \right)^{1/2} + n^{-1} \right\rbrace
\end{align*}
for some $a \in \{-1,1\}$, which completes the proof.
\end{proof}


We next prove \eqref{equation:Lgls} in Theorem \ref{theorem:L}.

\begin{proof}[Proof of \eqref{equation:Lgls} in Theorem \ref{theorem:L}]
Let $\hat{\bm{V}}_g = \bm{V}(\hat{\bm{v}}_g)$, where $\hat{\bm{v}}_g$ is the restricted maximum likelihood estimate for $\bm{v}_g$ using the design matrix $\hat{\bm{C}} \in \mathbb{R}^{n \times K}$ (defined in Step \ref{item:Estc:LandBasis} of Algorithm \ref{algorithm:EstC} when $k=K$). By Corollary \ref{corollary:supp:REML}, $\norm{ \hat{\bm{V}}_g - \bm{V}_g }_2 = O_P\left\lbrace n^{-1/2} + n^{1/2}(p\gamma_K)^{-1/2} \right\rbrace$. Let $\hat{\tilde{\bm{C}}} \in \mathbb{R}^{n \times K}$ be the first $K$ right singular vectors of $\bm{Y}\hat{\bar{\bm{V}}}^{-1/2}$. The generalised least squares estimate for $\Loracle_{g\bigcdot}$ is then
\begin{align*}
    \hat{\bm{L}}_{g\bigcdot}^{(GLS)} = \left(\hat{\bm{C}}^{\T}\hat{\bm{V}}_g^{-1}\hat{\bm{C}}\right)^{-1}\hat{\bm{C}}^{\T}\hat{\bm{V}}_g^{-1}\bm{Y}_{g \bigcdot} =& \hat{\bm{W}}^{\T}\hat{\bm{F}}_1^{1/2}\left(\hat{\tilde{\bm{C}}}^{\T}\hat{\bar{\bm{V}}}^{1/2}\hat{\bm{V}}_g^{-1}\hat{\bar{\bm{V}}}^{1/2}\hat{\tilde{\bm{C}}}\right)^{-1}\hat{\tilde{\bm{C}}}^{\T}\hat{\bar{\bm{V}}}^{1/2}\hat{\bm{V}}_g^{-1}\hat{\bar{\bm{V}}}^{1/2}\tilde{\bm{C}}\tilde{\bm{\ell}}_g\\
    &+ n^{-1/2}\hat{\bm{W}}^{\T}\hat{\bm{F}}_1^{1/2}\left(\hat{\tilde{\bm{C}}}^{\T}\hat{\bar{\bm{V}}}^{1/2}\hat{\bm{V}}_g^{-1}\hat{\bar{\bm{V}}}^{1/2}\hat{\tilde{\bm{C}}}\right)^{-1}\hat{\tilde{\bm{C}}}^{\T}\hat{\bar{\bm{V}}}^{1/2}\hat{\bm{V}}_g^{-1} \bm{E}_{g \bigcdot}\\
    =& \bm{M}_g + \bm{R}_g,
\end{align*}
where $\hat{\bm{F}}_1$ is defined in \eqref{equation:supp:F1Gamma1}, $\hat{\bm{W}}$ is as defined in \eqref{equation:supp:What}, $\tilde{\bm{C}}$ is given in \eqref{equation:supp:Ctilde} and $\tilde{\bm{\ell}}_g = p^{1/2}n^{-1/2}\tilde{\bm{L}}_{g \bigcdot}$ for $\tilde{\bm{L}}_{g \bigcdot}$ defined in \eqref{equation:supp:Ltilde}. Note that $\norm{\tilde{\bm{\ell}}_g}_2 \leq c$ for some constant $c > 0$ and $\norm{ \hat{\bar{\bm{V}}} - \bar{\bm{V}} }_2= O_P(1/n)$. Let $\epsilon_j$, $j \in [b]$, be such that $\hat{\bm{V}}_g - \bm{V}_g = \sum_{j=1}^b \epsilon_j \bm{B}_j$. Then $\epsilon_j^2 = o_P(n^{-1/2})$ by Corollary \ref{corollary:supp:REML}, meaning
\begin{align*}
    \hat{\tilde{\bm{C}}}^{\T}\hat{\bar{\bm{V}}}^{1/2}\hat{\bm{V}}_g^{-1}\hat{\bar{\bm{V}}}^{1/2}\tilde{\bm{C}} =& \hat{\tilde{\bm{C}}}^{\T}\hat{\bar{\bm{V}}}^{1/2}\bm{V}_g^{-1}\hat{\bar{\bm{V}}}^{1/2}\tilde{\bm{C}} - \sum\limits_{j=1}^b \epsilon_j\left(\hat{\tilde{\bm{C}}}^{\T}\hat{\bar{\bm{V}}}^{1/2}\bm{V}_g^{-1} \bm{B}_j \bm{V}_g^{-1}\hat{\bar{\bm{V}}}^{1/2}\tilde{\bm{C}}\right) + o_P(n^{-1/2})\\
    =& \hat{\bm{v}}^{\T} \tilde{\bm{C}}^{\T}\hat{\bar{\bm{V}}}^{1/2}\bm{V}_g^{-1}\hat{\bar{\bm{V}}}^{1/2}\tilde{\bm{C}} - \sum\limits_{j=1}^b \epsilon_j\left(\hat{\bm{v}}^{\T} \tilde{\bm{C}}^{\T}\hat{\bar{\bm{V}}}^{1/2}\bm{V}_g^{-1}\hat{\bar{\bm{V}}}^{1/2}\tilde{\bm{C}}\bm{V}_g^{-1} \bm{B}_j \bm{V}_g^{-1}\hat{\bar{\bm{V}}}^{1/2}\tilde{\bm{C}}\right)\\
    &+ o_P(n^{-1/2})\\
    =&\hat{\bm{v}}^{\T} \tilde{\bm{C}}^{\T}\hat{\bar{\bm{V}}}^{1/2} \left(\bm{V}_g^{-1}- \bm{V}_g^{-1}\sum\limits_{j=1}^b \epsilon_j\bm{B}_j \bm{V}_g^{-1} \right)\hat{\bar{\bm{V}}}^{1/2}\tilde{\bm{C}} + o_P(n^{-1/2})\\
    =& \hat{\bm{v}}^{\T} \tilde{\bm{C}}^{\T}\hat{\bar{\bm{V}}}^{1/2}\hat{\bm{V}}_g^{-1}\hat{\bar{\bm{V}}}^{1/2}\tilde{\bm{C}} + o_P(n^{-1/2}),
\end{align*}
where the second equality follows from Corollary \ref{corollary:supp:CtC}. A similar technique shows that
\begin{align*}
    \hat{\tilde{\bm{C}}}^{\T}\hat{\bar{\bm{V}}}^{1/2}\hat{\bm{V}}_g^{-1}\hat{\bar{\bm{V}}}^{1/2}\hat{\tilde{\bm{C}}} = \hat{\bm{v}}^{\T}\tilde{\bm{C}}^{\T}\hat{\bar{\bm{V}}}^{1/2}\hat{\bm{V}}_g^{-1}\hat{\bar{\bm{V}}}^{1/2}\tilde{\bm{C}}\hat{\bm{v}} + o_P(n^{-1/2}).
\end{align*}
Therefore,
\begin{align*}
    \bm{M}_g = \hat{\bm{W}}^{\T}\hat{\bm{F}}_1^{1/2}\hat{\bm{v}}^{-1}\tilde{\bm{\ell}}_g + o_P(n^{-1/2}).
\end{align*}
Let $\bm{a}_r \in \mathbb{R}^K$ be the $r$th standard basis vector. By Corollary \ref{corollary:supp:Cte} and the fact that $\norm{n^{-1/2}\bm{E}_{g \bigcdot}}_2 = O_P(1)$,
\begin{align*}
    \bm{R}_{g_r} =& n^{-1/2}\bm{a}_r^{\T}\hat{\bm{W}}^{\T}\hat{\bm{F}}_1^{1/2}\hat{\bm{v}}^{-1}\left(\tilde{\bm{C}}^{\T}\hat{\bar{\bm{V}}}^{1/2}\bm{V}_g^{-1}\hat{\bar{\bm{V}}}^{1/2}\tilde{\bm{C}}\right)^{-1}\tilde{\bm{C}}^{\T}\hat{\bm{V}}^{1/2}\bm{V}_g^{-1}\bm{E}_{g \bigcdot} + o_P(n^{-1/2})\\
    =& O_P( \norm{ \hat{\bm{v}}^{-\T} \hat{\bm{F}}_1^{1/2}\hat{\bm{W}}\bm{a}_r - \bm{F}_1^{1/2}\bm{W}\bm{a}_r }_2 ) + \bm{a}_r^{\T}[ \{\Coracle\}^{\T}\bm{V}_g^{-1}\Coracle ]^{-1}\{\Coracle\}^{\T} \bm{V}_g^{-1}\bm{E}_{g \bigcdot} + o_P(n^{-1/2}).
\end{align*}
To complete the proof, we need only show that
\begin{align*}
    \norm{ \hat{\bm{v}}^{-\T} \hat{\bm{F}}_1^{1/2}\hat{\bm{W}}\bm{a}_r - \bm{F}_1^{1/2}\bm{W}\bm{a}_r }_2 = o_P(n^{-1/2})
\end{align*}
on the event $F_r^{(\epsilon)}$, where $\bm{F}_1$ is defined in \eqref{equation:supp:F1Gamma1} and $\bm{W}$ is as defined in \eqref{equation:supp:W}. We note that if $\bm{U}\in\mathbb{R}^{K \times K}$ contains the eigenvectors of $\bm{\Gamma}_1^{1/2}\bm{F}_1\bm{\Gamma}_1^{1/2}$ and $\tilde{\bm{w}}_r$ is as defined in Remark \ref{remark:supp:EigVectorsEta},
\begin{align*}
   \hat{\bm{v}}^{-\T} \hat{\bm{F}}_1^{1/2}\hat{\bm{W}}\bm{a}_r - \bm{F}_1^{1/2}\bm{W}\bm{a}_r =  \{ \lamhatoracle_r \}^{1/2} \left( \hat{\bm{v}} \hat{\bm{\Gamma}}\hat{\bm{v}}^{\T} \right)^{-1/2}\bm{U}\tilde{\bm{w}}_r - \{ \lamoracle_r \}^{1/2}\bm{\Gamma}^{-1/2}\bm{U}_{\bigcdot r}.
\end{align*}
However, we showed this was $o_P(n^{-1/2})$ in the proof of \eqref{equation:Linfty} in Theorem \ref{theorem:L} and \eqref{equation:CinnerChat} in Theorem \ref{theorem:AngleC} above.
\end{proof}

\section{Inference on $\Coracle$}
\label{section:supp:InferenceC}
In this section, we prove Theorem \ref{theorem:XandC}. We first state and prove two useful lemmas regarding the REML estimate for the covariance of linear combinations of the columns of $\bm{C}$. The first lemma shows that the correlation between $\bm{X}$ and $\hat{\bm{C}}_{\bigcdot r}$ mirrors that of $\bm{X}$ and $\Coracle_{\bigcdot r}$.

\begin{lemma}
\label{lemma:supp:XtC}
Let $\bm{X} \in \mathbb{R}^n$ be as defined in the statement of Theorem \ref{theorem:XandC}, suppose Assumptions \ref{assumption:CandL}, \ref{assumption:FALCO} and \ref{assumption:DependenceE} hold. Let $r \in [\Koracle]$ and let the event $F_r^{(\epsilon)}$ be as defined in the statement of \ref{theorem:L} for some small $\epsilon > 0$. Then if $\bm{X}$ is dependent on at most $c$ rows of $\bm{E}$ for some constant $c\geq 0$,
\begin{align}
    \abs{ \frac{\bm{X}^{\T} \Coracle_{\bigcdot r}}{\norm{\bm{X}}_2 \norm{\Coracle_{\bigcdot r}}_2} - \frac{\bm{X}^{\T} \hat{\bm{C}}_{\bigcdot r}}{ \norm{\bm{X}}_2\norm{\hat{\bm{C}}_{\bigcdot r}}_2 } } = O_P\left\lbrace (p\gamma_r)^{-1/2} + (n/p+n^{-1})\gamma_r^{-1} \right\rbrace.
\end{align}
\end{lemma}

\begin{proof}[Proof of Theorem \ref{theorem:XandC}]
Since we are studying the empirical correlation, it suffices to assume $\norm{\bm{X}}_2=1$ and to redefine $\hat{\bm{C}}$ to be $\hat{\bm{C}} \leftarrow n^{-1/2}\hat{\bm{C}}$. Then using notation defined in the proof of \eqref{equation:CinnerChat} and \eqref{equation:Linfty} in Section \ref{subsection:supp:Eigenvectors},
\begin{align*}
\bm{X}^{\T}\hat{\bm{C}}_{\bigcdot r} =& \{\lamhatoracle_r\}^{-1/2}\bm{X}^{\T}\hat{\bar{\bm{V}}}^{1/2}\tilde{\bm{C}}\begin{pmatrix}\hat{\bm{v}}_1 \\ \hat{\bm{v}}_2 \end{pmatrix} \hat{\bm{\Gamma}}_1^{1/2}\hat{\bm{U}}_{\bigcdot r} + \{\lamhatoracle_r\}^{-1/2}\bm{X}^{\T}\hat{\bar{\bm{V}}}^{1/2} \bm{Q}_{\tilde{C}}\hat{\bm{Z}}\hat{\bm{\Gamma}}_1^{1/2}\hat{\bm{U}}_{\bigcdot r}\\
\bm{X}^{\T}\Coracle_{\bigcdot r} =& \{\lamoracle_r\}^{-1/2}\bm{X}^{\T}\hat{\bar{\bm{V}}}^{1/2}\tilde{\bm{C}}_1\bm{\Gamma}_1^{1/2}\bm{U}_{r \bigcdot},
\end{align*}
where $\tilde{\bm{C}}_1 = \left(\tilde{\bm{C}}_{\bigcdot 1} \cdots \tilde{\bm{C}}_{\bigcdot \Koracle}\right)$ and $\tilde{\bm{C}}_2 = \left(\tilde{\bm{C}}_{\bigcdot (\Koracle+1)} \cdots \tilde{\bm{C}}_{\bigcdot K}\right)$. Next,
\begin{align*}
    \{\lamhatoracle_r\}^{-1/2}\bm{X}^{\T}\hat{\bar{\bm{V}}}^{1/2}\tilde{\bm{C}}\begin{pmatrix}\hat{\bm{v}}_1 \\ \hat{\bm{v}}_2 \end{pmatrix} \hat{\bm{\Gamma}}_1^{1/2}\hat{\bm{U}}_{\bigcdot r} =& \{\lamhatoracle_r\}^{-1/2}\bm{X}^{\T}\hat{\bar{\bm{V}}}^{1/2}\tilde{\bm{C}}_1\left( \hat{\bm{v}}_1 \hat{\bm{\Gamma}}_1\hat{\bm{v}}_1^{\T} \right)^{1/2}\bm{U}\tilde{\bm{w}}_r\\
    & + \{\lamhatoracle_r\}^{-1/2}\bm{X}^{\T}\hat{\bar{\bm{V}}}^{1/2}\tilde{\bm{C}}_2 \hat{\bm{v}}_2 \hat{\bm{\Gamma}}_1^{1/2}\hat{\bm{U}}_{\bigcdot r}.
\end{align*}
By Lemmas \ref{lemma:supp:vandzhat} and \ref{lemma:supp:U},
\begin{align*}
    \{\lamhatoracle_r\}^{-1/2}\norm{ \hat{\bm{v}}_2 \hat{\bm{\Gamma}}_1^{1/2}\hat{\bm{U}}_{\bigcdot r} }_2 = O_P\left\lbrace \left(p\gamma_r \right)^{-1/2} + \left(n/p+n^{-1}\right)\gamma_r^{-1} \right\rbrace.
\end{align*}
By the proof of Corollary \ref{corollary:supp:Cte},
\begin{align*}
    \bm{X}^{\T}\hat{\bar{\bm{V}}}^{1/2} \bm{Q}_{\tilde{C}}\hat{\bm{Z}}_{\bigcdot t} = O_P\left\lbrace \left(p\gamma_t \right)^{-1/2}+ \left(n/p+n^{-1}\right)\gamma_t^{-1} \right\rbrace, \quad t \in [\Koracle],
\end{align*}
meaning $\{\lamhatoracle_r\}^{-1/2}\bm{X}^{\T}\hat{\bar{\bm{V}}}^{1/2} \bm{Q}_{\tilde{C}}\hat{\bm{Z}}\hat{\bm{\Gamma}}_1^{1/2}\hat{\bm{U}}_{\bigcdot r} = O_P\left\lbrace \left(p\gamma_t \right)^{-1/2}+ \left(n/p+n^{-1}\right)\gamma_t^{-1} \right\rbrace$. Therefore, we only need to show that
\begin{align*}
    \norm{ \left\lbrace \frac{\lamoracle_r}{\lamhatoracle_r} \right\rbrace^{1/2}\{\lamoracle_r\}^{-1/2}\left( \hat{\bm{v}}_1 \hat{\bm{\Gamma}}_1\hat{\bm{v}}_1^{\T} \right)^{1/2}\bm{U}\tilde{\bm{w}}_r - \{\lamoracle_r\}^{-1/2}\bm{\Gamma}_1^{1/2}\bm{U}_{r \bigcdot} }_2 = O_P\left\lbrace \left(p\gamma_r \right)^{-1/2}+ \left(n/p+n^{-1}\right)\gamma_r^{-1} \right\rbrace
\end{align*}
to complete the proof. Let $\delta_r = \left(p\gamma_r \right)^{-1/2}+ \left(n/p+n^{-1}\right)\gamma_r^{-1}$. Since $\norm{ \{\lamoracle_r\}^{-1/2}\left( \hat{\bm{v}}_1 \hat{\bm{\Gamma}}_1\hat{\bm{v}}_1^{\T} \right)^{1/2}\bm{U}\tilde{\bm{w}}_r }_2=O_P(1)$ and $\left\lbrace \frac{\lamoracle_r}{\lamhatoracle_r} \right\rbrace^{1/2} = O_P(\delta_r)$, this amounts to showing 
\begin{align*}
    \{\lamoracle_r\}^{-1/2}\norm{\left( \hat{\bm{v}}_1 \hat{\bm{\Gamma}}_1\hat{\bm{v}}_1^{\T} \right)^{1/2}\bm{U}\tilde{\bm{w}}_r- \bm{\Gamma}_1^{1/2}\bm{U}_{r \bigcdot}}_2=O_P(\delta_r).
\end{align*}
Define $\bm{R} = \left( \hat{\bm{v}}_1 \hat{\bm{\Gamma}}_1\hat{\bm{v}}_1^{\T} \right)^{1/2} - \bm{\Gamma}_1^{1/2}$. By the proof of Lemma \ref{lemma:supp:TrueEigs}, $\bm{R}_{st} = O_P\left( \phi_1+\phi_3\gamma_{s \wedge t}^{-1/2}\right)$, where $\phi_1=p^{-1/2}$ and $\phi_3=n/p+1/n$, meaning $(\bm{R}\bm{U})_{st} = O_P\left( \phi_1+\phi_3\gamma_{s \wedge t}^{-1/2} \right)$ by Lemma \ref{lemma:supp:U}. Since $\norm{\tilde{\bm{w}}_r-\bm{a}_r}_2 = O_P\left(\phi_1\gamma_r^{-1/2}+\phi_3\gamma_r^{-1/2}\gamma_{\Koracle}^{-1/2}\right)$ for $\bm{a}_r\in\mathbb{R}^{\Koracle}$ the $r$th standard basis vector,
\begin{align*}
    \{\lamoracle_r\}^{-1/2}\norm{ \bm{R}\bm{U}\tilde{\bm{w}}_r }_2 = O_P(\delta_r).
\end{align*}
By the proof of Lemma \ref{lemma:supp:TrueEigs} and the fact that $\tilde{\bm{w}}_1,\ldots,\tilde{\bm{w}}_{\Koracle}$ are orthogonal with unit norm, $\tilde{\bm{w}}_{r_s} = O_P\left( \phi_1\gamma_{r \wedge s}^{-1/2} + \phi_3\gamma_r^{-1/2}\gamma_s^{-1/2} \right)$ for all $s \neq r \in [\Koracle]$. Therefore,
\begin{align*}
    \norm{\{\lamoracle_r\}^{-1/2}\bm{\Gamma}_1^{1/2}\bm{U}\left(\bm{a}_r - \tilde{\bm{w}}_r\right)}_2 = O_P\left(\delta_r\right),
\end{align*}
which completes the proof.
\end{proof}

We next prove REML estimates for the variance of $\bm{C}$ are consistent.

\begin{lemma}
\label{lemma:supp:REMLC}
Suppose Assumptions \ref{assumption:CandL}, \ref{assumption:FALCO} and \ref{assumption:DependenceE} hold and $\Koracle \in [K]$ is known. Further, assume the following assumptions on $\bm{C}$ hold for some constant $c> 1$ and $r \in [\Koracle]$:
\begin{enumerate}[label=(\roman*)]
\item $\bm{C}= \bm{X}\bm{\Omega}^{\T} + \bm{\Xi}$ for some non-random $\bm{\Omega} \in \mathbb{R}^{K}$ and random $\bm{X} \in \mathbb{R}^n$ that satisfies $\norm{n^{-1}\bm{X}^{\T}\bm{X}-\sigma_x^2}_2 = O_P(n^{-1/2})$ for some constant $\sigma_x^2>0$. If $\bm{\Omega} \neq \bm{0}$, then $\bm{X}$ is independent of $\bm{E}$. If $\bm{\Omega} = \bm{0}$, then $\bm{X}$ is independent of all but at most $c$ rows of $\bm{E}$.
\item The random matrix $\bm{\Xi} \in \mathbb{R}^{n \times K}$ is independent of $\bm{X}$, $\E(\bm{\Xi})=\bm{0}$, $\V\{\vecM{\bm{\Xi}}\} = \sum_{j=1}^b \bm{\Psi}_j \otimes \bm{B}_j$, where $\norm{\bm{\Psi}_j}_2 \leq c$ and $\V\{\vecM{\bm{\Xi}}\} \succeq c^{-1}I_{nK}$, and satisfies one of the following:
\begin{enumerate}[label=(\alph*)]
\item $\vecM(\bm{\Xi}) = \bm{A}\bm{R}$, where $\bm{A} \in \mathbb{R}^{nK \times nK}$ is a non-random matrix that satisfies $\bm{A}\bm{A}^{\T} = \sum_{j=1}^b \bm{\Psi}_j \otimes \bm{B}_j$. $\bm{R} \in \mathbb{R}^{nK}$ is a mean zero random matrix with independent entries such that $\E(\bm{R}_{i}^4) \leq c$ for all $i \in [nK]$.
\item $\E[\exp\{\vecM(\bm{\Xi})^{\T}\bm{t}\}] \leq \exp(c\norm{\bm{t}}_2^2)$ for all $\bm{t} \in \mathbb{R}^{nK}$.
\end{enumerate}
\item $\E(n^{-1}\bm{C}^{\T}\bm{C}) = I_K$ and $\Prob\left\lbrace \lamoracle_r/\lamoracle_{r-1}, \lamoracle_{r+1}/\lamoracle_r \geq 1+c^{-1} \right\rbrace \to 1$ as $n,p \to \infty$, where $\lamoracle_{0}=\infty$.
\end{enumerate}
Let $\mathcal{F} = \left\lbrace \bm{x} \in \mathbb{R}^b : \text{$(2c)^{-1}I_n \preceq \sum_{j=1}^b \bm{x}_j \bm{B}_j$ and $\norm{\bm{x}}_2 \leq 2bc$} \right\rbrace$ and for $r \in [\Koracle]$, define
\begin{align*}
f_{r,1}(\bm{\theta}) &= -n^{-1}\log\left\lbrace \abs{ P_X^{\perp}\bm{V}(\bm{\theta})P_X^{\perp} }_+ \right\rbrace - n^{-1}\{ \Coracle_{\bigcdot r} \}^{\T}\left\lbrace P_X^{\perp}\bm{V}(\bm{\theta})P_X^{\perp} \right\rbrace^{\dagger}\Coracle_{\bigcdot r}\\
\hat{\bm{\theta}}_1 &= \argmax_{\bm{\theta} \in \mathcal{F}} f_{r,1}(\bm{\theta})
\end{align*}
and
\begin{align*}
f_{r,2}(\bm{\theta}) &= -n^{-1}\log\left\lbrace \abs{ P_X^{\perp}\bm{V}(\bm{\theta})P_X^{\perp} }_+ \right\rbrace - n^{-1}\hat{\bm{C}}_{\bigcdot r}^{\T}\left\lbrace P_X^{\perp}\bm{V}(\bm{\theta})P_X^{\perp} \right\rbrace^{\dagger}\hat{\bm{C}}_{\bigcdot r}\\
\hat{\bm{\theta}}_2 &= \argmax_{\bm{\theta} \in \mathcal{F}} f_{r,2}(\bm{\theta}),
\end{align*}
and let $\hat{\bm{u}}_r \in \mathbb{R}^{K}$ be such that $\Coracle_r = \bm{C}\hat{\bm{u}}_r$. Then for $\bm{\theta}^{*} \in \mathbb{R}^b$ such that $\bm{\theta}^{*}_j = \hat{\bm{u}}_r^{\T}\bm{\Psi}_j \hat{\bm{u}}_r$,
\begin{align*}
\norm{\hat{\bm{\theta}}_1 - \bm{\theta}^*}_2 = O_P\left(n^{-1/2}\right), \quad \norm{\hat{\bm{\theta}}_2 - \bm{\theta}^*}_2 = O_P\left[ n^{-1/2} + \{n/(p\gamma_r)\}^{1/2}\right].
\end{align*}
\end{lemma} 

\begin{proof}
We first note that $\hat{\bm{u}}_r = \left( n^{-1}\bm{C}^{\T}\bm{C} \right)^{-1/2}\bm{v}$ for some unit vector $\bm{v} \in \mathbb{R}^K$, where $n^{-1}\bm{C}^{\T}\bm{C} = I_K + O_P(n^{-1/2})$. Next,
\begin{align*}
n^{-1/2}\hat{\bm{C}}_{\bigcdot r} =& \{\lamhatoracle_r\}^{-1/2} \hat{\bar{\bm{V}}}^{1/2}\left(\tilde{\bm{C}}_1\hat{\bm{v}}_1 + \tilde{\bm{C}}_2\hat{\bm{v}}_2\right)\hat{\bm{\Gamma}}_1^{1/2}\hat{\bm{U}}_{\bigcdot r} + \{\lamhatoracle_r\}^{-1/2} \hat{\bar{\bm{V}}}^{1/2} \bm{Q}_{\tilde{C}}\hat{\bm{Z}}\hat{\bm{\Gamma}}_1^{1/2}\hat{\bm{U}}_{\bigcdot r}\\
n^{-1/2}\Coracle_{\bigcdot r} =& \{\lamoracle_r\}^{-1/2}\hat{\bar{\bm{V}}}^{1/2}\tilde{\bm{C}}_1 \bm{\Gamma}_1^{1/2}\bm{U}_{\bigcdot r}.
\end{align*}
By Lemma \ref{lemma:supp:vandzhat} and for some constant $\tilde{c} > 0$,
\begin{align*}
    &\norm{ \{\lamhatoracle_r\}^{-1/2} \hat{\bar{\bm{V}}}^{1/2} \bm{Q}_{\tilde{C}}\hat{\bm{Z}}\hat{\bm{\Gamma}}_1^{1/2}\hat{\bm{U}}_{\bigcdot r} }_2 \leq \tilde{c}\norm{\{\lamhatoracle_r\}^{-1/2}\hat{\bm{Z}}\hat{\bm{\Gamma}}_1^{1/2}\hat{\bm{U}}_{\bigcdot r}}_2 = O_P\left(n^{1/2}p^{-1/2}\gamma_r^{-1/2}+n^{-1}\gamma_r^{-1}\right)\\
    &\norm{ \{\lamhatoracle_r\}^{-1/2}\tilde{\bm{C}}_2\hat{\bm{v}}_2\hat{\bm{\Gamma}}_1^{1/2}\hat{\bm{U}}_{\bigcdot r} }_2 = O_P\left( \frac{n}{p\gamma_r} + p^{-1/2}\gamma_r^{-1/2}+n^{-1}\gamma_r^{-1} \right).
\end{align*}
Further, by the proof of Lemma \ref{lemma:supp:XtC},
\begin{align*}
    \norm{ \{\lamhatoracle_r\}^{-1/2} \hat{\bar{\bm{V}}}^{1/2}\tilde{\bm{C}}_1\hat{\bm{v}}_1\hat{\bm{\Gamma}}_1^{1/2}\hat{\bm{U}}_{\bigcdot r} - \{\lamoracle_r\}^{-1/2}\hat{\bar{\bm{V}}}^{1/2}\tilde{\bm{C}}_1 \bm{\Gamma}_1^{1/2}\bm{U}_{\bigcdot r} }_2 = O_P\left( \frac{n}{p\gamma_r} + p^{-1/2}\gamma_r^{-1/2}+n^{-1}\gamma_r^{-1} \right).
\end{align*}
This shows that
\begin{align*}
    \norm{n^{-1/2}\hat{\bm{C}}_{\bigcdot r} - n^{-1/2}\Coracle_{\bigcdot r}}_2 = O_P\left(n^{1/2}p^{-1/2}\gamma_r^{-1/2}+n^{-1}\gamma_r^{-1}\right).
\end{align*}
Since $\norm{\bm{V}(\bm{\theta})}_2,\norm{\{\bm{V}(\bm{\theta})\}^{-1}}_2$ is uniformly bounded from above for all $\bm{\theta}\in\mathcal{F}$, we need only show that $\norm{\hat{\bm{\theta}}_1 - \bm{\theta}^*}_2 = O_P\left(n^{-1/2}\right)$ to complete the proof. We first see that $f_{r,1}$ can be re-written as
\begin{align*}
    f_{r,1}(\bm{\theta}) &= -n^{-1}\log\left\lbrace \abs{ \bm{Q}_X^{\T}\bm{V}(\bm{\theta})\bm{Q}_X } \right\rbrace - n^{-1}\left(\bm{Q}_X^{\T}\bm{\Xi}\hat{\bm{u}}_r\right)^{\T}\left\lbrace \bm{Q}_X^{\T}\bm{V}(\bm{\theta})\bm{Q}_X \right\rbrace^{-1}\left(\bm{Q}_X^{\T}\bm{\Xi}\hat{\bm{u}}_r\right)
\end{align*}
It is clear that by the assumptions on $\bm{\Xi}$,
\begin{align*}
    &\norm{ n^{-1}\left(\bm{Q}_X^{\T}\bm{\Xi}\right)^{\T}\left\lbrace \bm{Q}_X^{\T}\bm{V}(\bm{\theta})\bm{Q}_X \right\rbrace^{-1} \left(\bm{Q}_X^{\T}\bm{\Xi}\right) - n^{-1}\E\left[ \left(\bm{Q}_X^{\T}\bm{\Xi}\right)^{\T}\left\lbrace \bm{Q}_X^{\T}\bm{V}(\bm{\theta})\bm{Q}_X \right\rbrace^{-1} \left(\bm{Q}_X^{\T}\bm{\Xi}\right) \mid \bm{X} \right] }_2
\end{align*}
is stochastically equicontinuous, where
\begin{align*}
    n^{-1}\E\left[ \left(\bm{Q}_X^{\T}\bm{\Xi}\right)^{\T}\left\lbrace \bm{Q}_X^{\T}\bm{V}(\bm{\theta})\bm{Q}_X \right\rbrace^{-1} \left(\bm{Q}_X^{\T}\bm{\Xi}\right) \mid \bm{X} \right] = n^{-1}\sum\limits_{j=1}^b \Tr\left[\left\lbrace \bm{Q}_X^{\T}\bm{V}(\bm{\theta})\bm{Q}_X \right\rbrace^{-1}\bm{Q}_X^{\T}\bm{B}_j\bm{Q}_X\right] \bm{\Psi}_j.
\end{align*}
Since $\norm{\hat{\bm{u}}_r}_2 = 1+O_P(n^{-1/2})$, it is therefore easy to show that $\norm{\hat{\bm{\theta}}_1 - \bm{\theta}^*}_2 = o_P(1)$. The result then follows by a routine Taylor expansion argument.
\end{proof}

We can now prove Theorem \ref{theorem:XandC}.
\begin{proof}[Proof of Theorem \ref{theorem:XandC}]
Let $\hat{\bm{\theta}}_2$ be as defined in the statement of Lemma \ref{lemma:supp:REMLC} and let $\hat{\bm{D}} = \bm{V}(\hat{\bm{\theta}}_2)$. Let $\Coracle_{\bigcdot r} = \bm{C}\hat{\bm{u}}_r$, where by Lemma \ref{lemma:supp:PrelimCinference}, $\hat{\bm{u}}_r = \bm{U}\hat{\bm{y}}_r + \bm{\Delta}_r$ for $\bm{U}\in\mathbb{R}^{K\times \Koracle},\bm{\Delta}_r\in\mathbb{R}^{K}$ as defined in the statement of Lemma \ref{lemma:supp:PrelimCinference} and
\begin{align}
    \hat{\bm{y}}_{r_t} =\begin{cases}
    1+O_P\left(n^{-1/2}\right) & \text{if $t=r$}\\
    O_P\left(n^{-1/2}\right) & \text{if $t < r$}\\
    O_P\left(n^{-1/2}\gamma_t\gamma_r^{-1}\right) & \text{if $t > r$ }
    \end{cases}, \quad t \in [\Koracle].
\end{align}
Note that $\bm{U}$ is a non-random matrix. If $\Koracle=K$, $\bm{U}=I_K$ and $\bm{\Delta}_r=\bm{0}$. Lemma \ref{lemma:supp:REMLC} then implies $\hat{\bm{D}} = \bm{D} + \sum\limits_{j=1}^b \epsilon_j \bm{B}_j$, where $\bm{D}=\sum_{j=1}^b \bm{U}_{\bigcdot r}^{\T}\bm{\Psi}_j\bm{U}_{\bigcdot r} \bm{B}_j$ is a non-random matrix and $\epsilon_j = O_P\left(n^{-1/2}+n^{1/2}p^{-1/2}\gamma_r^{-1/2}\right)$ for all $j \in [b]$. Therefore,
\begin{align*}
    n^{-1}\bm{X}^{\T}\hat{\bm{D}}^{-1}\hat{\bm{C}}_{\bigcdot r} = n^{-1}\bm{X}^{\T}\bm{D}^{-1}\hat{\bm{C}}_{\bigcdot r} + n^{-1}\sum\limits_{j=1}^b \epsilon_j \bm{X}^{\T}\bm{D}^{-1}\bm{B}_j\bm{D}^{-1}\hat{\bm{C}}_{\bigcdot r} + o_P(n^{-1/2}).
\end{align*}
The proof of Lemma \ref{lemma:supp:XtC} can be easily extended to show that
\begin{align*}
    &\abs{ n^{-1}\bm{X}^{\T}\bm{D}^{-1}\hat{\bm{C}}_{\bigcdot r} - n^{-1}\bm{X}^{\T}\bm{D}^{-1}\Coracle_{\bigcdot r} } = o_P(n^{-1/2})\\
    &\abs{ n^{-1}\bm{X}^{\T}\bm{D}^{-1}\bm{B}_j\bm{D}^{-1}\hat{\bm{C}}_{\bigcdot r} - n^{-1}\bm{X}^{\T}\bm{D}^{-1}\Coracle_{\bigcdot r} } = o_P(n^{-1/2}),
\end{align*}
meaning 
\begin{align*}
    n^{-1}\bm{X}^{\T}\hat{\bm{D}}^{-1}\hat{\bm{C}}_{\bigcdot r} =& n^{-1}\bm{X}^{\T}\bm{D}^{-1}\Coracle_{\bigcdot r} + n^{-1}\sum\limits_{j=1}^b \epsilon_j \bm{X}^{\T}\bm{D}^{-1}\bm{B}_j\bm{D}^{-1}\Coracle_{\bigcdot r} + o_P(n^{-1/2})\\
    =& n^{-1}\bm{X}^{\T}\hat{\bm{D}}^{-1}\Coracle_{\bigcdot r} + o_P(n^{-1/2}).
\end{align*}
Therefore,
\begin{align*}
    \left(\bm{X}^{\T}\hat{\bm{D}}^{-1}\bm{X}\right)^{-1}\bm{X}^{\T}\hat{\bm{D}}^{-1}\hat{\bm{C}}_{\bigcdot r} =& \left(\bm{X}^{\T}\hat{\bm{D}}^{-1}\bm{X}\right)^{-1}\bm{X}^{\T}\hat{\bm{D}}^{-1}\Coracle_{\bigcdot r} + o_P(n^{-1/2})\\
    =& \bm{\Omega}^{\T}\hat{\bm{u}}_r + \left(\bm{X}^{\T}\hat{\bm{D}}^{-1}\bm{X}\right)^{-1}\bm{X}^{\T}\hat{\bm{D}}^{-1}\bm{\Xi}\hat{\bm{u}}_r + o_P(n^{-1/2}).
\end{align*}
Lastly,
\begin{align*}
    n^{-1}\bm{X}^{\T}\hat{\bm{D}}^{-1}\bm{\Xi}\bm{U}_{\bigcdot r} = n^{-1}\bm{X}^{\T}\bm{D}^{-1}\bm{\Xi}\bm{U}_{\bigcdot r} + o_P(n^{-1/2}),
\end{align*}
where $\norm{ n^{-1}\bm{X}^{\T}\bm{D}^{-1}\bm{\Xi} }_2 = O_P(n^{-1/2})$ and an application of the Lindeberg-Feller central limit theorem shows that
\begin{align}
    \left( n^{-1}\sum\limits_{j=1}^b \bm{X}^{\T}\bm{B}_j\bm{X} \bm{U}_{\bigcdot r}^{\T}\bm{\Psi}_j\bm{U}_{\bigcdot r} \right)^{-1/2}n^{-1}\bm{X}^{\T}\bm{D}^{-1}\bm{\Xi}\bm{U}_{\bigcdot r}  \tdist N\left(0,1\right)
\end{align}
as $n,p \to \infty$. Since $\norm{\hat{\bm{u}}_r-\bm{U}_{\bigcdot r}}_2 = o_P(1)$, this completes the proof.
\end{proof}

\section{Denoising the data matrix $\bm{Y}$}
\label{section:supp:denoising}
Here we prove Theorem \ref{theorem:eQTL}. The proof is given below, and utilizes Corollary \ref{corollary:supp:REML} to derive the asymptotic properties of the REML estimates for $\V\left(\bm{R}_g\right)$.

\begin{proof}[Proof of Theorem \ref{theorem:eQTL}]
For $\hat{\bm{\alpha}}_g$ as defined in the statement of Theorem \ref{theorem:eQTL}, Corollary \ref{corollary:supp:REML} shows that $\norm{ \hat{\bm{\alpha}}_g - \bm{\alpha}_g }_2 = O_P\left(n^{1/2}p^{-1/2}\gamma_K^{-1/2} + n^{-1/2}\right)$. Let $\bm{A}_g = \bm{V}(\bm{\alpha}_g)$ and $\hat{\bm{A}}_g = \bm{V}(\hat{\bm{\alpha}}_g)$. Then $\hat{\bm{A}}_g = \bm{A}_g + \sum_{j=1}^b \epsilon_j \bm{B}_j$, where $\epsilon_j = O_P\left(n^{1/2}p^{-1/2}\gamma_K^{-1/2} + n^{-1/2}\right)$ for all $j \in [b]$. Let, $\hat{\bm{M}}_g = \left(\hat{\bar{\bm{V}}}^{-1/2}\bm{X}_g  \, n^{1/2}\hat{\tilde{\bm{C}}} \right)$, where $\hat{\tilde{\bm{C}}} \in \mathbb{R}^{n \times K}$ are the first $K$ right singular vectors of $\bm{Y}\hat{\bar{\bm{V}}}^{-1/2}$. Then for $\hat{\tilde{\bm{A}}}_g = \hat{\bar{\bm{V}}}^{-1/2}\hat{\bm{A}}_g\hat{\bar{\bm{V}}}^{-1/2}$,
\begin{align*}
\hat{\bm{s}}_g = \left\lbrace \left( n^{-1}\hat{\bm{M}}_g^{\T} \hat{\tilde{\bm{A}}}_g^{-1}\hat{\bm{M}}_g \right)^{-1}\left(n^{-1}\hat{\bm{M}}_g^{\T} \hat{\tilde{\bm{A}}}_g^{-1}\hat{\bar{\bm{V}}}^{1/2}\bm{Y}_{g \bigcdot}\right) \right\rbrace_{1:d},
\end{align*}
where for $\bm{x} \in \mathbb{R}^m$, $\bm{x}_{1:d} \in \mathbb{R}^d$ is the first $d\leq m$ components of $\bm{x}$. We first see that
\begin{align*}
n^{-1}\hat{\bm{M}}_g^{\T} \hat{\tilde{\bm{A}}}_g^{-1}\hat{\bm{M}}_g = \begin{pmatrix}
n^{-1} \bm{X}_g^{\T}\hat{\bm{A}}_g^{-1}\bm{X}_g & n^{-1/2}\bm{X}_g^{\T}\hat{\bm{A}}_g^{-1} \hat{\bar{\bm{V}}}^{1/2}\hat{\tilde{\bm{C}}}\\
n^{-1/2}\hat{\tilde{\bm{C}}}^{\T} \hat{\bar{\bm{V}}}^{1/2} \hat{\bm{A}}_g^{-1}\bm{X}_g & \hat{\tilde{\bm{C}}}^{\T}\hat{\bar{\bm{V}}}^{1/2} \hat{\bm{A}}_g^{-1}\hat{\bar{\bm{V}}}^{1/2}\hat{\tilde{\bm{C}}}
\end{pmatrix}.
\end{align*}
Since $\hat{\bm{A}}_g^{-1} = \bm{A}_g^{-1}+\sum_{j=1}\epsilon_j \bm{A}_g^{-1}\bm{B}_j\bm{A}_g^{-1} + o_P(n^{1/2})$ and $\norm{\bm{A}_g^{-1}\bm{B}_j\bm{A}_g^{-1}}_2 = O(1)$, Corollaries \ref{corollary:supp:CtC} and \ref{corollary:supp:Cte} show that
\begin{align*}
    &\norm{ \hat{\tilde{\bm{C}}}^{\T}\hat{\bar{\bm{V}}}^{1/2} \hat{\bm{A}}_g^{-1}\hat{\bar{\bm{V}}}^{1/2}\hat{\tilde{\bm{C}}} - \hat{\bm{v}}^{\T}\tilde{\bm{C}}^{\T}\hat{\bar{\bm{V}}}^{1/2} \hat{\bm{A}}_g^{-1}\hat{\bar{\bm{V}}}^{1/2}\tilde{\bm{C}}\hat{\bm{v}} }_2 = o_P(n^{-1/2})\\
    &\norm{n^{-1/2}\bm{X}_g^{\T}\hat{\bm{A}}_g^{-1} \hat{\bar{\bm{V}}}^{1/2}\hat{\tilde{\bm{C}}} - n^{-1/2}\bm{X}_g^{\T}\hat{\bm{A}}_g^{-1} \hat{\bar{\bm{V}}}^{1/2}\tilde{\bm{C}} \hat{\bm{v}}}_2 = o_P(n^{-1/2}),
\end{align*}
meaning
\begin{align*}
n^{-1}\hat{\bm{M}}_g^{\T} \hat{\tilde{\bm{A}}}_g^{-1}\hat{\bm{M}}_g = \begin{pmatrix}
I_d & \bm{0}\\
\bm{0} & \hat{\bm{v}}^{\T}
\end{pmatrix}
\begin{pmatrix}
n^{-1} \bm{X}_g^{\T}\hat{\bm{A}}_g^{-1}\bm{X}_g & n^{-1/2}\bm{X}_g^{\T}\hat{\bm{A}}_g^{-1} \hat{\bar{\bm{V}}}^{1/2}\tilde{\bm{C}}\\
n^{-1/2}\tilde{\bm{C}}^{\T} \hat{\bar{\bm{V}}}^{1/2} \hat{\bm{A}}_g^{-1}\bm{X}_g & \tilde{\bm{C}}^{\T}\hat{\bar{\bm{V}}}^{1/2} \hat{\bm{A}}_g^{-1}\hat{\bar{\bm{V}}}^{1/2}\tilde{\bm{C}}
\end{pmatrix}\begin{pmatrix}
I_d & \bm{0}\\
\bm{0} & \hat{\bm{v}}
\end{pmatrix} + o_P(n^{-1/2}).
\end{align*}
Therefore,
\begin{align*}
    \hat{\bm{s}}_g = \left\lbrace \begin{pmatrix}
n^{-1} \bm{X}_g^{\T}\hat{\bm{A}}_g^{-1}\bm{X}_g & n^{-1/2}\bm{X}_g^{\T}\hat{\bm{A}}_g^{-1} \hat{\bar{\bm{V}}}^{1/2}\tilde{\bm{C}}\\
n^{-1/2}\tilde{\bm{C}}^{\T} \hat{\bar{\bm{V}}}^{1/2} \hat{\bm{A}}_g^{-1}\bm{X}_g & \tilde{\bm{C}}^{\T}\hat{\bar{\bm{V}}}^{1/2} \hat{\bm{A}}_g^{-1}\hat{\bar{\bm{V}}}^{1/2}\tilde{\bm{C}}
\end{pmatrix}^{-1}\begin{pmatrix}
n^{-1}\bm{X}_g^{\T} \hat{\bm{A}}_g^{-1}\bm{Y}_{g \bigcdot}\\
n^{-1/2}\hat{\bm{v}}^{-\T} \hat{\tilde{\bm{C}}}^{\T}\hat{\bar{\bm{V}}}^{1/2} \hat{\bm{A}}_g^{-1}\bm{Y}_{g \bigcdot}
\end{pmatrix} \right\rbrace_{1:d} + o_P(n^{-1/2}),
\end{align*}
where
\begin{align*}
    n^{-1/2}\hat{\bm{v}}^{-\T} \hat{\tilde{\bm{C}}}^{\T}\hat{\bar{\bm{V}}}^{1/2} \hat{\bm{A}}_g^{-1}\bm{Y}_{g \bigcdot} = n^{-1/2} \tilde{\bm{C}}^{\T}\hat{\bar{\bm{V}}}^{1/2} \hat{\bm{A}}_g^{-1}\bm{Y}_{g \bigcdot} + n^{-1/2}\hat{\bm{v}}^{-\T}\hat{\bm{z}}^{\T}\bm{Q}_{\tilde{C}}^{\T}\hat{\bar{\bm{V}}}^{1/2} \hat{\bm{A}}_g^{-1} \bm{Y}_{g \bigcdot}.
\end{align*}
Next,
\begin{align*}
    n^{-1/2}\hat{\bm{v}}^{-\T}\hat{\bm{z}}^{\T}\bm{Q}_{\tilde{C}}^{\T}\hat{\bar{\bm{V}}}^{1/2} \hat{\bm{A}}_g^{-1} \bm{Y}_{g \bigcdot} =& \hat{\bm{v}}^{-\T}\hat{\bm{z}}^{\T}\bm{Q}_{\tilde{C}}^{\T}\hat{\bar{\bm{V}}}^{1/2} \hat{\bm{A}}_g^{-1}\hat{\bar{\bm{V}}}^{1/2}\tilde{\bm{C}}\left(n^{-1/2}p^{1/2}\tilde{\bm{L}}_{g \bigcdot}\right) + n^{-1/2}\hat{\bm{v}}^{-\T}\hat{\bm{z}}^{\T}\bm{Q}_{\tilde{C}}^{\T}\hat{\bar{\bm{V}}}^{1/2} \hat{\bm{A}}_g^{-1}\bm{E}_{g \bigcdot},
\end{align*}
where $\norm{n^{-1/2}p^{1/2}\tilde{\bm{L}}_{g \bigcdot}}_2 = O(1)$. The proof of Corollary \ref{corollary:supp:Cte} can therefore be used to show that 
\begin{align*}
    n^{-1/2}\hat{\bm{v}}^{-\T}\hat{\bm{z}}^{\T}\bm{Q}_{\tilde{C}}^{\T}\hat{\bar{\bm{V}}}^{1/2} \hat{\bm{A}}_g^{-1} \bm{Y}_{g \bigcdot} = o_P(n^{-1/2}).
\end{align*}
Since $\im(\hat{\bar{\bm{V}}}^{1/2}\tilde{\bm{C}}) = \im(\bm{C})$,
\begin{align*}
    \hat{\bm{s}}_g = \left\lbrace \begin{pmatrix}
n^{-1} \bm{X}_g^{\T}\hat{\bm{A}}_g^{-1}\bm{X}_g & n^{-1}\bm{X}_g^{\T}\hat{\bm{A}}_g^{-1} \bm{C}\\
n^{-1}\bm{C}^{\T} \hat{\bm{A}}_g^{-1}\bm{X}_g & n^{-1}\bm{C}^{\T} \hat{\bm{A}}_g^{-1}\bm{C}
\end{pmatrix}^{-1}\begin{pmatrix}
n^{-1}\bm{X}_g^{\T} \hat{\bm{A}}_g^{-1}\bm{Y}_{g \bigcdot}\\
n^{-1} \bm{C}^{\T} \hat{\bm{A}}_g^{-1}\bm{Y}_{g \bigcdot}
\end{pmatrix} \right\rbrace_{1:d} + o_P(n^{-1/2}).
\end{align*}
Therefore,
\begin{align*}
    \hat{\bm{s}}_g = \bm{s}_g + \left\lbrace \begin{pmatrix}
n^{-1} \bm{X}_g^{\T}\hat{\bm{A}}_g^{-1}\bm{X}_g & n^{-1}\bm{X}_g^{\T}\hat{\bm{A}}_g^{-1} \bm{C}\\
n^{-1}\bm{C}^{\T} \hat{\bm{A}}_g^{-1}\bm{X}_g & n^{-1}\bm{C}^{\T} \hat{\bm{A}}_g^{-1}\bm{C}
\end{pmatrix}^{-1}\begin{pmatrix}
n^{-1}\bm{X}_g^{\T} \hat{\bm{A}}_g^{-1}\bm{R}_{g \bigcdot}\\
n^{-1} \bm{C}^{\T} \hat{\bm{A}}_g^{-1}\bm{R}_{g \bigcdot}
\end{pmatrix} \right\rbrace_{1:d} + o_P(n^{-1/2})
\end{align*}
Since $\bm{X}_g$ is mean 0, sub-Gaussian and independent of $\bm{C}$, $n^{-1}\bm{X}_g^{\T}\hat{\bm{A}}_g^{-1} \bm{C} = O_P(n^{-1/2})$, meaning
\begin{align*}
    \left\lbrace \begin{pmatrix}
n^{-1} \bm{X}_g^{\T}\hat{\bm{A}}_g^{-1}\bm{X}_g & n^{-1}\bm{X}_g^{\T}\hat{\bm{A}}_g^{-1} \bm{C}\\
n^{-1}\bm{C}^{\T} \hat{\bm{A}}_g^{-1}\bm{X}_g & n^{-1}\bm{C}^{\T} \hat{\bm{A}}_g^{-1}\bm{C}
\end{pmatrix}^{-1}\begin{pmatrix}
n^{-1}\bm{X}_g^{\T} \hat{\bm{A}}_g^{-1}\bm{R}_{g \bigcdot}\\
n^{-1} \bm{C}^{\T} \hat{\bm{A}}_g^{-1}\bm{R}_{g \bigcdot}
\end{pmatrix} \right\rbrace_{1:d} = \left( \bm{X}_g^{\T}\hat{\bm{A}}_g^{-1}\bm{X}_g\right)^{-1}\bm{X}_g^{\T} \hat{\bm{A}}_g^{-1}\bm{R}_{g \bigcdot} + o_P(n^{-1/2}).
\end{align*}
The result then follows because
\begin{align*}
    \left( \bm{X}_g^{\T}\hat{\bm{A}}_g^{-1}\bm{X}_g\right)^{-1}\bm{X}_g^{\T} \hat{\bm{A}}_g^{-1}\bm{R}_{g \bigcdot} = \left( \bm{X}_g^{\T}\bm{A}_g^{-1}\bm{X}_g\right)^{-1}\bm{X}_g^{\T} \bm{A}_g^{-1}\bm{R}_{g \bigcdot} + o_P(n^{-1/2}).
\end{align*}
\end{proof}

\section{Properties of and estimating the oracle rank $\Koracle$}
\label{subsection:supp:OracleRank}
In this section, we refer to $F$ as the number of folds in Algorithm \ref{algorithm:CBCV} use $f \in [F]$ to denote a fold. We start by stating and proving three useful lemmas.
\begin{lemma}
\label{lemma:sup:LtLRandom}
Suppose Assumptions \ref{assumption:CandL} and \ref{assumption:FALCO} hold. Let $\bm{M}^{(\pi)} \in \mathbb{R}^{p_f \times n}$ be a matrix whose rows are chosen uniformly at random, without replacement, from the rows of $\bm{M}=\E\left( \bm{Y} \mid \bm{C} \right)$. Then for any symmetric matrix $\bm{A} \in \mathbb{R}^{n \times n}$ with $\norm{\bm{A}}_2 \leq d$ and $\bm{A} \succeq d^{-1}I_n$ for some constant $d > 1$,
\begin{align*}
    \Lambda_k\left[ p_f^{-1}\left\lbrace\bm{M}^{(\pi)}\right\rbrace^{\T} \bm{A} \bm{M}^{(\pi)} \right] = \Lambda_k\left(p^{-1}\bm{M}^{\T}\bm{A}\bm{M}\right)\left[ 1+O_P\left\lbrace\left( np_f^{-1} \right)^{1/2}\right\rbrace \right] + O_P\left\lbrace\left( np_f^{-1} \right)^{1/2}\right\rbrace, \quad k \in [K]
\end{align*}
as $n,p \to \infty$, where the randomness is due to the randomness in the rows sampled from the rows of $\E\left( \bm{Y} \mid \bm{C} \right)$.
\end{lemma}
\begin{proof}
Suppose $\E\left( \bm{Y} \mid \bm{C} \right) = \bm{L}\bm{C}^{\T}$, where we assume $n^{-1}\bm{C}^{\T}\bm{A}\bm{C} = I_K$ and $np^{-1}\bm{L}^{\T}\bm{L} = \diag\left(\gamma_1,\ldots,\gamma_K\right)$ without loss of generality (I abuse notation here; $\gamma_k$ are not the same as those defined in Assumption \ref{assumption:CandL}). Therefore, we need only show that for $\pi : [p] \to [p]$ a permutation chosen uniformly at random from the set of all permutations that map $[p]$ onto itself,
\begin{align*}
    \Lambda_k\left\lbrace n p_f^{-1}\sum_{g=1}^{p_f}\bm{L}_{\pi(g) \bigcdot}\bm{L}_{\pi(g) \bigcdot}^{\T} \right\rbrace = \gamma_k\left[ 1+O_P\left\lbrace \left(np_f^{-1}\right)^{1/2}\right\rbrace \right] + O_P\left\lbrace\left( np_f^{-1} \right)^{1/2}\right\rbrace, \quad k \in [K].
\end{align*}
First,
\begin{align*}
    \E\left( n p_f^{-1}\sum_{g=1}^{p_f}\bm{L}_{\pi(g) r}\bm{L}_{\pi(g) s} \mid \bm{C} \right) = p_f^{-1}\sum\limits_{g=1}^{p_f}\E\left(n\bm{L}_{\pi(g) r}\bm{L}_{\pi(g) s}^{\T} \mid \bm{C}\right) = \gamma_r I\left(r=s\right), \quad r,s \in [K].
\end{align*}
Next, let $c = \sup_{g \in [p]}\norm{\bm{L}_{g \bigcdot}}_2$ and suppose $r \leq s \in [K]$. Then
\begin{align*}
    \V\left\lbrace np_f^{-1}\sum_{g=1}^{p_f}\bm{L}_{\pi(g) r}\bm{L}_{\pi(g) s} \mid \bm{C}\right\rbrace &\leq n^2 p_f^{-2} \sum\limits_{g=1}^{p_f}\V\left\lbrace \bm{L}_{\pi(g) r}\bm{L}_{\pi(g) s} \mid \bm{C}\right\rbrace \leq n p_f^{-2} c^2 \sum\limits_{g=1}^{p_f}\E\left\lbrace n\bm{L}_{\pi(g)s}^2\right\rbrace\\
    & = c^2 np_f^{-1}\gamma_s,
\end{align*}
where the first equality follows from the fact that the rows of $\E\left(\bm{Y}\mid \bm{C}\right)$ are being sampled with replacement, meaning $\C\left(\bm{L}_{\pi(g) r}\bm{L}_{\pi(g) s},\bm{L}_{\pi(h) r}\bm{L}_{\pi(h) s} \mid \bm{C}\right) \leq 0$ for $g \neq h$. By Assumption \ref{assumption:CandL}, $c < \tilde{c} + o_P(1)$ as $n,p \to \infty$, where $\tilde{c} > 0$ is a constant that does not depend on $n$ or $p$. 
\begin{align*}
    \hat{\bm{\Gamma}} =& n p_f^{-1}\sum_{g=1}^{p_f}\bm{L}_{\pi(g) \bigcdot}\bm{L}_{\pi(g) \bigcdot}^{\T} = \diag\left(\gamma_1,\ldots,\gamma_K\right) + \bm{R}\\
    \bm{R}_{rs} =& O_P\left\lbrace \left(\lambda_{r \vee s}n/p_f\right)^{1/2} \right\rbrace, \quad r,s \in [K].
\end{align*}
If $\limsup_{n,p \to \infty}\lambda_K = \infty$, the result follows from Lemma \ref{lemma:supp:EigBound}. Otherwise, suppose there exists a $k \in [K]$ such that $\limsup_{n,p \to \infty}\lambda_k < \infty$ but $\limsup_{n,p \to \infty}\lambda_{k-1} = \infty$, where $\lambda_0 = \infty$. Then
\begin{align*}
    \hat{\bm{\Gamma}} =& \diag\left(\gamma_1,\ldots,\gamma_K\right) + \tilde{\bm{R}} + O_P\left\lbrace \left(n/p\right)^{1/2} \right\rbrace\\
    \tilde{\bm{R}}_{rs} =& O_P\left\lbrace \left(\lambda_{r \vee s}n/p_f\right)^{1/2} \right\rbrace I\left(r\wedge s \geq k\right), \quad r,s \in [K].
\end{align*}
The result then follows by Weyl's Theorem and an application of Lemma \ref{lemma:supp:EigBound}.
\end{proof}

\begin{lemma}[Lemma S8 of \citet{CorrConf}]
\label{lemma:supp:LOOXVColumnSpace}
Let $\bm{Q},\hat{\bm{V}}^{-1/2}_{(-f)},\hat{\bm{C}}$ and $K_{\max}$ be as defined in Theorem \ref{theorem:Khat}. Then for any fold $f$ and $k \leq K_{\max}$, the loss function in \eqref{equation:LOOLoss} only depends on $\bm{Q}$, $\hat{\bm{V}}_{(-f)}$ and $\im\left(\hat{\bm{C}}\right)$.
\end{lemma}

\begin{remark}
\label{remark:supp:ColumnSpace}
Let $\bm{Y}_{(-f)},\hat{\bm{C}},\hat{\bm{V}}^{-1/2}_{(-f)}$ be as defined in Algorithm \ref{algorithm:CBCV}, and let $\bm{Y}_{(-f)}\hat{\bm{V}}^{-1/2}_{(-f)} = \bm{U}\bm{\Sigma}\bm{W}^{\T}$, where $\bm{U} \in \mathbb{R}^{p \times n}$, $\bm{\Sigma} \in \mathbb{R}^{n \times n}$ is a diagonal matrix with non-decreasing and non-negative diagonal elements, $\bm{W} \in \mathbb{R}^{n \times n}$ and $\bm{U}^{\T}\bm{U}=\bm{W}^{\T}\bm{W} = I_n$. Proposition \ref{proposition:ImageC} and Lemma \ref{lemma:supp:LOOXVColumnSpace} show that to prove Theorem \ref{theorem:Khat}, it suffices to let $\hat{\bm{V}}^{-1/2}_{(-f)}\hat{\bm{C}} = n^{1/2}\left(\bm{W}_{\bigcdot 1} \cdots \bm{W}_{\bigcdot k}\right)$ for each $k \leq K_{\max}$ and to assume $n^{-1}\bm{C}^{\T}\hat{\bm{V}}^{-1}_{(-f)}\bm{C} = I_K$ and $\bm{L}^{\T}\bm{L}$ is diagonal with non-decreasing diagonal elements.
\end{remark}

\begin{lemma}
\label{lemma:supp:VLargek}
Suppose the assumptions of Theorem \ref{theorem:Khat} hold and let $c_{max} > K$ be a constant that does not depend on $n$ or $p$. Then
\begin{align*}
    \max_{k \in \left\lbrace c_{\max}+1,\ldots,K_{\max} \right\rbrace} \norm{ \hat{\bm{V}}_{(-f)} - \bar{\bm{V}} }_2 = O_P\left(n^{1/2}p^{-1/2} + n^{-1}\right)
\end{align*}
as $n,p \to \infty$.
\end{lemma}

\begin{proof}
Let $\bar{\bm{V}}_{(-f)} = (p-p_f)^{-1}\sum\limits_{g \notin \text{fold $f$}} \bm{V}_g$. Since $p_f,p-p_f \asymp p$ and $\norm{\bar{\bm{V}} -\bar{\bm{V}}_{(-f)} }_2 = O_P\left(p^{-1/2}\right)$, it suffices drop the subscript $(-f)$ and prove the lemma using the full data matrix $\bm{Y}$. Let $k \in \left\lbrace c_{\max}+1,\ldots,K_{\max} \right\rbrace$ and let $\hat{\bm{C}} \in \mathbb{R}^{n \times k}$ be an estimate for $\bm{C}$. Step \ref{item:EstC:ImageC}\ref{item:EstC:V} of Algorithm \ref{algorithm:EstC} then estimates $\bm{V}$ as $\hat{\bm{V}} = \sum_{j=1}^b \hat{\bar{\bm{v}}}_j \bm{B}_j$, where for $\bm{S}=p^{-1}\bm{S}^{\T}\bm{S}$, $\bm{V}\left(\bm{\theta}\right) = \sum_{j=1}^b \bm{\theta}_j \bm{B}_j$ and
\begin{align*}
    \hat{f}\left(\bm{\theta}\right) = -n^{-1}\log\left\lbrace \abs{ \bm{Q}_{\hat{C}}^{\T}\bm{V}\left(\bm{\theta}\right)\bm{Q}_{\hat{C}} } \right\rbrace - n^{-1}\Tr\left[ \bm{Q}_{\hat{C}}^{\T}\bm{S}\bm{Q}_{\hat{C}} \left\lbrace \bm{Q}_{\hat{C}}^{\T}\bm{V}\left(\bm{\theta}\right)\bm{Q}_{\hat{C}} \right\rbrace^{-1} \right],
\end{align*}
$\hat{\bar{\bm{v}}} = \argmax_{\bm{\theta} \in \Theta_*}\hat{f}\left(\bm{\theta}\right)$. If $K=0$ or $\limsup_{n,p \to \infty} \lambda_1 < \infty$, then
\begin{align*}
    \hat{f}\left(\bm{\theta}\right) = -n^{-1}\log\left\lbrace \abs{ \bm{Q}_{\hat{C}}^{\T}\bm{V}\left(\bm{\theta}\right)\bm{Q}_{\hat{C}} } \right\rbrace - n^{-1}\Tr\left[ \bm{Q}_{\hat{C}}^{\T}\left(p^{-1}\bm{E}^{\T}\bm{E}\right)\bm{Q}_{\hat{C}} \left\lbrace \bm{Q}_{\hat{C}}^{\T}\bm{V}\left(\bm{\theta}\right)\bm{Q}_{\hat{C}} \right\rbrace^{-1} \right] + O_P\left(n^{-1}\right)
\end{align*}
uniformly for $\bm{\theta} \in \Theta_*$. If $\limsup_{n,p \to \infty} \lambda_1 = \infty$, define $c_{\min}$ to be such that $\limsup_{n,p \to \infty} \lambda_{c_{\min}} = \infty$ but $\limsup_{n,p \to \infty} \lambda_{c_{\min}+1} < \infty$, where $\lambda_{K+1} = 0$. Then $\bm{Q}_{\hat{C}} = \bm{Q}_{\min} \bm{U}$, where the columns of $\bm{Q}_{\min} \in \mathbb{R}^{n \times c_{\min}}$ form an orthonormal basis for $\ker\left\lbrace \left(\hat{\bm{C}}_{\bigcdot 1} \cdots \hat{\bm{C}}_{\bigcdot c_{\min}}\right)^{\T} \right\rbrace$ and $\bm{U} \in \mathbb{R}^{(n-c_{\min}) \times (n-k)}$ has orthonormal columns. By Corollary \ref{corollary:supp:RemoveBlocks} and the proof of Lemma S7 in \citet{CorrConf},
\begin{align*}
     \hat{f}\left(\bm{\theta}\right) =& -n^{-1}\log\left\lbrace \abs{ \bm{U}^{\T}\bm{Q}_{\min}^{\T}\bm{V}\left(\bm{\theta}\right)\bm{Q}_{\min}\bm{U} } \right\rbrace\\
     &- n^{-1}\Tr\left[ \bm{U}^{\T}\bm{Q}_{\min}^{\T}\left(p^{-1}\bm{E}^{\T}\bm{E}\right)\bm{Q}_{\min}\bm{U} \left\lbrace \bm{U}^{\T}\bm{Q}_{\min}^{\T}\bm{V}\left(\bm{\theta}\right)\bm{Q}_{\min}\bm{U} \right\rbrace^{-1} \right] + O_P\left(n^{-1}\right)
\end{align*}
uniformly for $\bm{\theta} \in \Theta_*$ and $k \in \left\lbrace c_{\max}+1,\ldots,K_{\max} \right\rbrace$. The latter follows from the fact that all rates in Lemma \ref{lemma:supp:vandzhat} and Corollary \ref{corollary:supp:RemoveBlocks} only depend on $k$ through $\phi_2$, which is uniformly bounded from above and satisfies $\phi_2/\lambda_r = o_P(1)$ for $r \leq c_{\min}$. An identical analysis can be used to show that we can replace $\bm{S}$ in $\nabla_{\bm{\theta}}\hat{f}\left(\bm{\theta}\right)$ and $\nabla_{\bm{\theta}}^2\hat{f}\left(\bm{\theta}\right)$ with $p^{-1}\bm{E}^{\T}\bm{E}$ at the cost of $O_P\left(n^{-1}\right)$. The result then follows by Lemma \ref{lemma:supp:EnpLarge}.
\end{proof}

\begin{proof}[Proof of Theorem \ref{theorem:Khat}]
Fix some fold $f$ and define $\bar{\bm{V}}_{f}$ to be the analogue of $\bar{\bm{V}}$ for fold $f$, and let $\delta_f^2 = \abs{\bar{\bm{V}}_{f}}^{1/n}$. Let $\bm{\pi} : [p] \to [p]$ be a permutation sampled uniformly from the set of all permutations on $[p]$. All conditional expectations and variances calculated below are with reference to the sigma algebra $\sigma\left( \bm{Y}_{(-f)}, \bm{\pi},\bm{Q}\right)$, where $\bm{Y}_{f} = \bm{L}_f\bm{C}^T + \bm{E}_f \in \mathbb{R}^{p_f \times n}$ is the test data used to estimate $\bm{L}_f$ and evaluate the loss and $\bm{Y}_{(-f)} = \bm{L}_{(-f)}\bm{C}^T + \bm{E}_{(-f)} \in \mathbb{R}^{(p-p_f) \times n}$ is the training data used to estimate $\bm{C}$ and $\bar{\bm{V}}$.\par
\indent Assumption \ref{assumption:CandL} implies
\begin{align*}
    \norm{ \bar{\bm{V}} - \bar{\bm{V}}_{f} }_2, \, \norm{ \bar{\bm{V}} - \bar{\bm{V}}_{(-f)} }_2 = O_P\left(p^{-1/2}\right).
\end{align*}
Therefore, by Lemma \ref{lemma:sup:LtLRandom},
\begin{align*}
    \Lambda_k\left\lbrace \bar{\bm{V}}_{f}^{-1/2}\bm{C}\left( p^{-1}\bm{L}_{f}^{\T}\bm{L}_{f}\right)\bm{C}^{\T}\bar{\bm{V}}_{(-f)}^{-1/2} \right\rbrace,\,\Lambda_k\left[ \bar{\bm{V}}_{(-f)}^{-1/2}\bm{C}\left\lbrace p^{-1}\bm{L}_{(-f)}^{\T}\bm{L}_{(-f)}\right\rbrace\bm{C}^{\T}\bar{\bm{V}}_{(-f)}^{-1/2} \right] = \gamma_k\left\lbrace 1+o_P(1) \right\rbrace + o_P(1)
\end{align*}
for all $k \in [K]$ as $n,p \to \infty$, where $\gamma_1,\ldots,\gamma_K$ are as defined in Assumption \ref{assumption:CandL}. Therefore, the results of Lemmas \ref{lemma:supp:UpperBlock}, \ref{lemma:supp:vandzhat} and Corollary \ref{corollary:supp:VAlgorithm} when we substitute $\bm{Y}$ with the training data $\bm{Y}_{(-f)}$.\par
\indent Let $\hat{\bm{\bar{C}}} = \bm{Q}^{\T}\hat{\bm{V}}_{(-f)}^{-1/2}\hat{\bm{C}} \in \mathbb{R}^{n \times k}$ be as defined in Algorithm \ref{algorithm:CBCV}. Lemma \ref{lemma:supp:LOOXVColumnSpace} and Remark \ref{remark:supp:ColumnSpace} show that it suffices to assume the columns of $\hat{\bm{\bar{C}}}$ are the first $k$ right singular vectors of $\bar{\bm{Y}}_{(-f)}$, where $n^{-1}\hat{\bm{\bar{C}}}^{\T}\hat{\bm{\bar{C}}}=I_k$, and that $n^{-1}\bm{C}^{\T}\hat{\bm{V}}_{(-f)}^{-1}\bm{C}=I_K$ and $\bm{L}_{(-f)}^{\T}\bm{L}_{(-f)}$ is a diagonal matrix with non-decreasing entries. We will let $\hat{\bar{h}}_i$, $i \in [n]$, be the $i$th leverage score of $\hat{\bm{\bar{C}}}$ throughout the proof. Note that $\hat{\bar{h}}_i$ is implicitly a function of $k \in \left[K_{\max}\right]$. Since $\bm{Q} \in \mathbb{R}^{n \times n}$ is sampled uniformly from the space of all unitary matrices, setting $\epsilon = n^{-1/3}$ in Lemma \ref{lemma:GaussLeverageScores} implies
\begin{align*}
    \max_{k\in\left[K_{\max}\right]}\left\lbrace\left( \max_{i \in [n]} \abs{\hat{\bar{h}}_i - k/n}\right)\right\rbrace = o_P(1)
\end{align*}
as $n \to \infty$. Let $\hat{\bar{\bm{C}}}_{(-i)} \in \mathbb{R}^{(n-1)\times k}$ be the sub-matrix of $\hat{\bar{\bm{C}}}$ with the $i$th row removed and define $\hat{\bar{\bm{H}}} = P_{\hat{\bar{C}}}$. Two useful equalities to be used throughout the proof are
\begin{align*}
&\left\lbrace \hat{\bar{\bm{C}}}_{(-i)}^T \hat{\bar{\bm{C}}}_{(-i)} \right\rbrace^{-1} \hat{\bar{\bm{C}}}_{i \bigcdot} = \frac{1}{1-\hat{\bar{h}}_i}\left( \hat{\bar{\bm{C}}}^T \hat{\bar{\bm{C}}} \right)^{-1} \hat{\bar{\bm{C}}}_{i \bigcdot}, \quad i \in [n]\\
&\hat{\bar{\bm{C}}}\left\lbrace \hat{\bar{\bm{C}}}_{(-i)}^T \hat{\bar{\bm{C}}}_{(-i)} \right\rbrace^{-1} \hat{\bar{\bm{C}}}_{i \bigcdot} = \frac{1}{1-\hat{\bar{h}}_i} \hat{\bar{\bm{H}}}_{\bigcdot i}, \quad i \in [n].
\end{align*}

Since $\hat{\bar{\bm{C}}}$ is invariant to the scale of $\hat{\bm{V}}_{(-f)}$ and we normalize the loss in \eqref{equation:LOOLoss} by $\abs{\hat{\bm{V}}_{(-f)}}^{1/n}$, it suffices to assume we have already scaled $\abs{\hat{\bm{V}}_{(-f)}}^{1/n}$ so that $\abs{\hat{\bm{V}}_{(-f)}} = 1$. Define $\bar{\bm{C}} = \bm{Q}^{\T}\hat{\bm{V}}_{(-f)}^{-1/2}\bm{C}$. A scaled version of the loss in \eqref{equation:LOOLoss} for fold $f$ can then be expressed as
\begin{align}
g(k) =& \frac{1}{\delta_{f}^2 p_f}\cv_{f}\left( k \right) = \frac{1}{\delta_{f}^2 p_f}\sum\limits_{i=1}^n \norm{\bar{\bm{Y}}_{f_{\bigcdot i}} - \hat{\bm{L}}_{f,(-i)} \hat{\bar{\bm{C}}}_{i \bigcdot}}_2^2 = \underbrace{\frac{1}{\delta_{f}^2 p_f} \sum\limits_{i=1}^n\norm{\bm{L}_f \bar{\bm{C}}_{i \bigcdot} - \hat{\bm{L}}_{f,(-i)} \hat{\bar{\bm{C}}}_{i \bigcdot}}_2^2}_{g_1(k)}\nonumber\\
\label{equation:Supp:LOOXV}
& + \underbrace{\frac{1}{\delta_{f}^2 } \Tr\left\lbrace\hat{\bm{V}}_{(-f)}^{-1}\left( p_f^{-1}\bm{E}_f^T \bm{E}_f \right)\right\rbrace}_{g_2(k)} + \underbrace{\frac{2}{\delta_{f}^2 }\sum\limits_{i=1}^n p_f^{-1} \left( \bm{L}_f \bar{\bm{C}}_{i \bigcdot} - \hat{\bm{L}}_{f,(-i)} \hat{\bar{\bm{C}}}_{i \bigcdot} \right)^T \bar{\bm{E}}_f \hat{\bar{\bm{V}}}_{(-f)}^{-1/2} \bm{a}_i}_{g_3(k)},
\end{align}
where $\bar{\bm{E}}_f = \bm{E}_f \bm{Q}$, $\hat{\bar{\bm{V}}}_{(-f)} = \bm{Q}^{\T}\hat{\bm{V}}_{(-f)}\bm{Q}$ and $\bm{a}_i \in \mathbb{R}^{n}$, $i \in [n]$, is the standard basis vector with 1 in the $i$th position and zeros everywhere else. For the remainder of the proof, define $\bm{A}_{(-i)} = I_n - \bm{a}_i\bm{a}_i^{\T}$ for $i \in [n]$. Let $\epsilon > 0$ be an arbitrarily small constant. As we did in Section \ref{subsection:supp:EigenPreliminaries}, let $k_0 = 0$ and define $k_j$ inductively as
\begin{align*}
    k_j = \min\left( \left\lbrace r \in \left\lbrace k_{j-1}+1,\ldots,K \right\rbrace: \gamma_r/\gamma_{r+1} \geq 1+\epsilon \right\rbrace \right), \quad j \in [J]
\end{align*}
where $k_J = K$. By the assumptions of Theorem \ref{theorem:Khat}, there exists an $s \in \left\lbrace 0,1,\ldots,K \right\rbrace$ such that $k_{j_s} = s$ for some $j_s \in \left\lbrace 0,1,\ldots,J \right\rbrace$ and $\gamma_{s+1} \leq \delta^2-\epsilon$. We derive the asymptotic properties of $g_1, g_2$ and $g_3$ in the three lemmas below.

\begin{lemma}
\label{lemma:supp:CBCV:g1}
If the assumptions in the statement of Theorem \ref{theorem:Khat} hold, then there exists a large constant $c_{\max} > K$ and another unrelated constant $\sigma > 0$ that do not depend on $n$ or $p$ such that
\begin{align*}
    &g_1(k) \begin{cases}
    \geq k+\sum\limits_{r=k+1}^K \gamma_r\left\lbrace 1 + O_P\left(n^{1/2}p^{-1/2} + n^{-1/2}\right) \right\rbrace & \text{if $k < s$}\\
    = k + O_P\left(np^{-1} + n^{-1/2}\right) + I\left(k < K\right)\sum\limits_{r=k+1}^K \gamma_r\left\lbrace 1 + O_P\left(n^{1/2}p^{-1/2} + n^{-1/2}\right) \right\rbrace & \text{if $k \in \left\lbrace s,\ldots,c_{\max} \right\rbrace$}\\
    \geq \sigma k\left(1+x_{k}\right) & \text{if $k \in \left\lbrace c_{\max}+1,\ldots,K_{\max} \right\rbrace$}
    \end{cases}\\
    &\max_{k \in \left\lbrace c_{\max}+1,\ldots,K_{\max} \right\rbrace} \abs{k^{1/2}x_{k}} = O_P(1)
\end{align*}
as $n,p \to \infty$.
\end{lemma}
\begin{proof}
$g_1(k)$ can be expressed as
\begin{align}
g_1(k)=&\frac{1}{\delta_{f}^2 p_f} \sum\limits_{i=1}^n\norm{\bm{L}_f \bar{\bm{C}}_{i \bigcdot} - \hat{\bm{L}}_{f,(-i)} \hat{\bar{\bm{C}}}_{i \bigcdot}}_2^2 = \frac{1}{\delta_{f}^2 p_f}\sum\limits_{i=1}^n\norm{ \frac{\bm{L}_f \bar{\bm{C}}_{i \bigcdot} - \bm{L}_f \bar{\bm{C}}^T \hat{\bar{\bm{C}}}\left( \hat{\bar{\bm{C}}}^T\hat{\bar{\bm{C}}} \right)^{-1} \hat{\bar{\bm{C}}}_{i \bigcdot}}{\left(1-\hat{\bar{h}}_i\right)} }_2^2 \nonumber\\
&+ \sum\limits_{i=1}^n\left( \frac{1}{1-\hat{\bar{h}}_i} \right)^2 \hat{\bar{\bm{H}}}_{\bigcdot i}^T \bm{A}_{(-i)} \hat{\bar{\bm{V}}}_{(-f)}^{-1/2}\left( \delta_{f}^{-2}p_f^{-1}\bar{\bm{E}_f}^T \bar{\bm{E}_f}\right)\hat{\bar{\bm{V}}}_{(-f)}^{-1/2}\bm{A}_{(-i)}\hat{\bar{\bm{H}}}_i \nonumber\\
\label{equation:Supp:LOO:Term1}
&-\frac{2}{\delta_{f}^2 n^{1/2} p_{f}^{1/2}} \sum\limits_{i=1}^n \left( \frac{1}{1-\hat{\bar{h}}_i}\right)^2\hat{\bar{\bm{H}}}_{\bigcdot i}^T \bm{A}_{(-i)} \hat{\bar{\bm{V}}}_{(-f)}^{-1/2} \bar{\bm{E}}_f^T \bm{\Delta}_i
\end{align}
where
\begin{align}
\label{equation:Supp:LOOXV:Deltai}
\bm{\Delta}_i =& \left( n^{1/2}p^{-1/2}\bm{L}_f \right)\bar{\bm{C}}_{i \bigcdot} - \left( n^{1/2}p^{-1/2}\bm{L}_f \right) \bar{\bm{C}}^T \hat{\bar{\bm{C}}}\left(\hat{\bar{\bm{C}}}^T \hat{\bar{\bm{C}}} \right)^{-1}\hat{\bar{\bm{C}}}_{i \bigcdot} = \tilde{\bm{L}}_f\bar{\bm{C}}_{i \bigcdot} - \tilde{\bm{L}}_f\bar{\bm{C}}^T \hat{\bar{\bm{C}}}\left(\hat{\bar{\bm{C}}}^T \hat{\bar{\bm{C}}} \right)^{-1}\hat{\bar{\bm{C}}}_{i \bigcdot}.
\end{align}
We derive the asymptotic properties of the three terms in  \eqref{equation:Supp:LOO:Term1} in (a), (b) and (c) below.
\begin{enumerate}[label=(\alph*)]
    \item Let $\hat{\bm{v}}_{(-f)_r}, \hat{\bm{z}}_{(-f)_r}$, estimated using $\bm{Y}_{(-f)}$, be the analogues of $\hat{\bm{v}}_r,\hat{\bm{z}}_r$, $r \in [K]$, defined in Lemma \ref{lemma:supp:vandzhat}. Similarly, let $\hat{\bm{V}}_{(-f)_j}, \hat{\bm{Z}}_{(-f)_j}$ be the analogues of $\hat{\bm{V}}_j,\hat{\bm{Z}}_j$, $j \in [J]$, defined in Lemma \ref{lemma:supp:vandzhat}. Let
    \begin{align*}
        g_{11}(k)=&\frac{1}{\delta_{f}^2 p_f}\sum\limits_{i=1}^n\norm{ \frac{\bm{L}_f \bar{\bm{C}}_{i \bigcdot} - \bm{L}_f \bar{\bm{C}}^T \hat{\bar{\bm{C}}}\left( \hat{\bar{\bm{C}}}^T\hat{\bar{\bm{C}}} \right)^{-1} \hat{\bar{\bm{C}}}_{i \bigcdot}}{\left(1-\hat{\bar{h}}_i\right)} }_2^2
    \end{align*}
    and define
    \begin{align*}
        \alpha_{-} &= \left\lbrace k/n - \max_{k \in [K_{\max}]}\left(\max_{i \in [n]} \abs{k/n - \hat{\bar{h}}_i}\right) \right\rbrace \vee 0\\
        \alpha_{+} &= \left\lbrace k/n + \max_{k \in [K_{\max}]}\left(\max_{i \in [n]} \abs{k/n - \hat{\bar{h}}_i}\right) \right\rbrace \wedge 1,
    \end{align*}
    where $\alpha_{-},\alpha_{+}=k/n + o_P\left(n^{1/2-\epsilon}\right)$ for any constant $\epsilon > 0$ as $n,p \to \infty$ by Lemma \ref{lemma:GaussLeverageScores}. Then
    \begin{align*}
        \delta_{f}^{-2}\left(1-\alpha_-\right)^{-2} \Tr\left\lbrace np_f^{-1}\bm{L}_f^T\bm{L}_f \left( n^{-1}\bar{\bm{C}}^T P_{\hat{\bar{\bm{C}}}}^{\perp}\bar{\bm{C}}\right)\right\rbrace &\leq g_{11}(k)\\
        &\leq  \delta_{f}^{-2}\left(1-\alpha_+\right)^{-2} \Tr\left\lbrace np_f^{-1}\bm{L}_f^T\bm{L}_f \left( n^{-1}\bar{\bm{C}}^T P_{\hat{\bar{\bm{C}}}}^{\perp}\bar{\bm{C}}\right)\right\rbrace
    \end{align*}
    where
    \begin{align*}
        \Tr\left\lbrace np_f^{-1}\bm{L}_f^T\bm{L}_f \left( n^{-1}\bar{\bm{C}}^T P_{\hat{\bar{\bm{C}}}}^{\perp}\bar{\bm{C}}\right)\right\rbrace = \Tr\left[ np_{f}^{-1}\bm{L}_f^T \bm{L}_f\left\lbrace I_K - \sum\limits_{r=1}^k \hat{\bm{v}}_{(-f)_r}\hat{\bm{v}}_{(-f)_r}^{\T} \right\rbrace \right].
    \end{align*}
    Since the eigenvalues of $\sum\limits_{r=1}^k \hat{\bm{v}}_{(-f)_r}\hat{\bm{v}}_{(-f)_r}^{\T}$ are bounded between 0 and 1,
    \begin{align*}
        g_{11}(k) \geq I(k < K)\delta_f^{-2}\left(\sum\limits_{r=k+1}^K \gamma_{r} \right)\left\lbrace 1+O_p(n^{1/2}p^{-1/2}) \right\rbrace
    \end{align*}
    by Lemma \ref{lemma:sup:LtLRandom}. Further for any $j \in [J]$ such that $\limsup_{n,p \to \infty}\lambda_{k_j} < \infty$,
    \begin{align*}
        \Tr\left[ np_{f}^{-1}\bm{L}_f^T \bm{L}_f\left\lbrace I_K - \sum\limits_{r=1}^{k_j} \hat{\bm{v}}_{(-f)_r}\hat{\bm{v}}_{(-f)_r}^{\T} \right\rbrace \right] =& O_P\left(np^{-1} + n^{-1}\right)\\
        &+ I(k_j < K)\left\lbrace 1+O_P\left(n^{1/2}p^{-1/2}\right) \right\rbrace\sum\limits_{r=k_j+1}^K \gamma_r
    \end{align*}
    by Lemmas \ref{lemma:supp:vandzhat} and \ref{lemma:sup:LtLRandom}. This also shows that for any $s\leq k < \tilde{k}$,
    \begin{align*}
        g_{11}(k) - g_{11}(\tilde{k}) \leq \delta_f^{-2}\gamma_{s+1} + O_P\left(n^{1/2}p^{-1/2}\right)I\left(\tilde{k} \leq K\right).
    \end{align*}
    Putting this altogether gives us
    \begin{align*}
        &\delta_f^{2} g_{11}(k) \begin{cases}
        \geq 0 & \text{if $k > K$}\\
        \geq \left(\sum\limits_{r=k+1}^K \gamma_{r} \right)\left\lbrace 1+O_p(n^{1/2}p^{-1/2}) \right\rbrace & \text{if $k \in \left\lbrace 0,1,\ldots,K \right\rbrace \setminus \{s\}$}\\
        = O_P\left(np^{-1} + n^{-1}\right) + I\left(s < K\right)\left(\sum\limits_{r=s+1}^K \gamma_r\right)\left\lbrace 1+O_P\left(n^{1/2}p^{-1/2}\right) \right\rbrace & \text{if $k=s$}
        \end{cases}\\
        &\delta_f^2\left\lbrace g_{11}(k) - g_{11}(k+1) \right\rbrace \leq \gamma_{s+1} + O_P\left(n^{1/2}p^{-1/2}\right), \quad k \in \left\lbrace s,\ldots,K \right\rbrace,
    \end{align*}
    where $\gamma_{K+1} = 0$.
\item For notational simplicity, I will assume, without loss of generality, that $\delta_f^2 = 1$ (i.e. $\abs{\bar{\bm{V}}_{f}}=1$).
\begin{align*}
    g_{12}(k) = \sum\limits_{i=1}^n\left( \frac{1}{1-\hat{\bar{h}}_i} \right)^2 \hat{\bar{\bm{H}}}_{\bigcdot i}^T \bm{A}_{(-i)} \hat{\bar{\bm{V}}}_{(-f)}^{-1/2}\left( p_f^{-1}\bar{\bm{E}_f}^T \bar{\bm{E}_f}\right)\hat{\bar{\bm{V}}}_{(-f)}^{-1/2}\bm{A}_{(-i)}\hat{\bar{\bm{H}}}_i.
\end{align*}
First,
\begin{align*}
    \E\left\lbrace g_{12}(k) \mid \bm{Y}_{(-f)},\bm{\pi},\bm{Q} \right\rbrace = \sum\limits_{i=1}^n\left( \frac{1}{1-\hat{\bar{h}}_i} \right)^2 \hat{\bar{\bm{H}}}_{\bigcdot i}^T \bm{A}_{(-i)} \hat{\bar{\bm{V}}}_{(-f)}^{-1/2}\tilde{\bm{V}}_f\hat{\bar{\bm{V}}}_{(-f)}^{-1/2}\bm{A}_{(-i)}\hat{\bar{\bm{H}}}_i,
\end{align*}
where $\tilde{\bm{V}}_{f} = \bm{Q}^{\T}\bar{\bm{V}}_f \bm{Q}$. Note that because $\bm{Q}$ is a unitary matrix, $\abs{\tilde{\bm{V}}_{f}}=1$. We then see that
\begin{align*}
    \E\left\lbrace g_{12}(k) \mid \bm{Y}_{(-f)},\bm{\pi},\bm{Q} \right\rbrace \geq &\left(1-\alpha_-\right)^{-1} \sum\limits_{i=1}^n\left(1-\hat{\bar{h}}_i\right)^{-1} \hat{\bar{\bm{H}}}_{\bigcdot i}^T \bm{A}_{(-i)} \hat{\bar{\bm{V}}}_{(-f)}^{-1/2}\tilde{\bm{V}}_f\hat{\bar{\bm{V}}}_{(-f)}^{-1/2}\bm{A}_{(-i)}\hat{\bar{\bm{H}}}_i\\
    =& \left(1-\alpha_-\right)^{-1} \Tr\left\lbrace \hat{\bar{\bm{V}}}_{(-f)}^{-1/2}\tilde{\bm{V}}_f\hat{\bar{\bm{V}}}_{(-f)}^{-1/2} \sum\limits_{i=1}^n \left(1-\hat{\bar{h}}_i\right)^{-1} \bm{A}_{(-i)}\hat{\bar{\bm{H}}}_i \hat{\bar{\bm{H}}}_i^T \bm{A}_{(-i)} \right\rbrace
\end{align*}
where $\sum\limits_{i=1}^n \left( 1-\hat{\bar{h}}_i\right)^{-1} \bm{A}_{(-i)}\hat{\bar{\bm{H}}}_i \hat{\bar{\bm{H}}}_i^T \bm{A}_{(-i)}$ is positive semi-definite with
\begin{align*}
\Tr\left\lbrace \sum\limits_{i=1}^n \left( 1-\hat{\bar{h}}_i\right)^{-1} \bm{A}_{(-i)}\hat{\bar{\bm{H}}}_i \hat{\bar{\bm{H}}}_i^T \bm{A}_{(-i)}\right\rbrace &= \Tr\left\lbrace \sum\limits_{i=1}^n \left( 1-\hat{\bar{h}}_i\right)^{-1} \left(\hat{\bar{\bm{H}}}_i - \hat{\bar{h}}_i \bm{a}_i \right) \left(\hat{\bar{\bm{H}}}_i - \hat{\bar{h}}_i \bm{a}_i \right)^T\right\rbrace\\
&= \sum\limits_{i=1}^n \hat{\bar{h}}_i = K.
\end{align*}
We first note that the minimum eigenvalue of $\hat{\bar{\bm{V}}}_{(-f)}^{-1/2}\tilde{\bm{V}}_f\hat{\bar{\bm{V}}}_{(-f)}^{-1/2}$ is uniformly bounded above $\sigma > 0$ by Assumptions \ref{assumption:CandL} and \ref{assumption:FALCO}, where $\sigma$ is a constant not dependent on $n$ or $p$. Therefore,
\begin{align*}
    \E\left\lbrace g_{12}(k) \mid \bm{Y}_{(-f)},\bm{\pi},\bm{Q} \right\rbrace \geq \left(1-\alpha_-\right)^{-1}\sigma k, \quad k \in [K_{\max}].
\end{align*}
Let $c_{\min} \in \left\lbrace 0,1,\ldots,K \right\rbrace$ be such that $\limsup_{n,p \to \infty}\lambda_{c_{\min}} = \infty$ but $\limsup_{n,p \to \infty}\lambda_{c_{\min}+1} < \infty$, where $\lambda_0 = \infty$ and $\lambda_{K+1}=0$. Let $c_{\max} > K$ be an arbitrarily large constant that does not depend on $n$ or $p$. Then Corollary \ref{corollary:supp:VAlgorithm} implies 
\begin{align*}
    \norm{\hat{\bar{\bm{V}}}_{(-f)}^{-1/2}\tilde{\bm{V}}_f\hat{\bar{\bm{V}}}_{(-f)}^{-1/2} - I_n}_2 = O_P\left( n^{-1} + p^{-1/2} \right), \quad k \in \left\lbrace c_{\min},\ldots,c_{\max} \right\rbrace.
\end{align*}
Since $\max_{i \in [n]}\abs{\hat{\bar{h}}_i - k/n} = O_P(n^{-1/2})$ for all $k \in \left\lbrace c_{\min},c_{\min}+1,\ldots,c_{\max} \right\rbrace$, putting this all together gives us
\begin{align*}
    \E\left\lbrace g_{12}(k) \mid \bm{Y}_{(-f)},\bm{\pi},\bm{Q} \right\rbrace \begin{cases}
    \geq 0 & \text{if $\limsup_{n,p \to \infty}\lambda_{k} = \infty$}\\
    = k\left\lbrace 1+O_P\left(n^{-1/2}\right) \right\rbrace & \text{if $k \in \left\lbrace c_{\min},c_{\min}+1,\ldots,c_{\max} \right\rbrace$}\\
    \geq \sigma k & \text{if $k > c_{\max}$}
    \end{cases}.
\end{align*}
Next, to calculate the conditional variance, we see that
\begin{align*}
    g_{12}(k) =& p_{f}^{-1} \Tr\left\lbrace \bar{\bm{E}}_{f}\hat{\bar{\bm{V}}}_{(-f)}^{-1/2} \hat{\bm{M}} \hat{\bar{\bm{V}}}_{(-f)}^{-1/2} \bar{\bm{E}}_{f}^T\right\rbrace = p_f^{-1}\sum\limits_{g =1}^{p_f}\bar{\bm{E}}_{f_{g \bigcdot}}^{\T}\hat{\bm{M}} \bar{\bm{E}}_{f_{g \bigcdot}}\\
    \hat{\bm{M}} =& \hat{\bar{\bm{V}}}_{(-f)}^{-1/2}\sum\limits_{i=1}^n \left( 1-\hat{\bar{h}}_i\right)^{-2} \bm{A}_{(-i)}\hat{\bar{\bm{H}}}_i \hat{\bar{\bm{H}}}_i^T \bm{A}_{(-i)}\hat{\bar{\bm{V}}}_{(-f)}^{-1/2}
\end{align*}
where for all $k \in [K_{\max}]$,
\begin{align*}
    \norm{ \sum\limits_{i=1}^n \left( 1-\hat{\bar{h}}_i\right)^{-2} \bm{A}_{(-i)}\hat{\bar{\bm{H}}}_i \hat{\bar{\bm{H}}}_i^T \bm{A}_{(-i)} }_F^2 =& \sum\limits_{i,j=1}^n \left( 1-\hat{\bar{h}}_i\right)^{-2} \left( 1-\hat{\bar{h}}_j\right)^{-2} \left\lbrace \hat{\bar{\bm{H}}}_{\bigcdot i}^{\T}\bm{A}_{(-i)}\bm{A}_{(-j)}\hat{\bar{\bm{H}}}_{\bigcdot j} \right\rbrace^2\\
    \leq & \left(1-\alpha_+\right)^{-4} \sum\limits_{i,j=1}^n \hat{\bar{\bm{H}}}_{ij}^2 = k\left(1-\alpha_+\right)^{-4}.
\end{align*}
Therefore, since $\norm{\hat{\bm{V}}_{(-f)}^{-1}}_2$ is uniformly bounded above by a constant,
\begin{align*}
    \norm{ \hat{\bm{M}} }_F^2  = O\left\lbrace k\left(1-\alpha_+\right)^{-4} \right\rbrace, \quad \norm{ \hat{\bm{M}} }_2 = O\left\lbrace \left(1-\alpha_+\right)^{-2} \right\rbrace.
\end{align*}
Since $\bar{\bm{E}}_{f_{1 \bigcdot}},\ldots,\bar{\bm{E}}_{f_{p_f \bigcdot}}$ are independent sub-Gaussian random variables with uniformly bounded sub-Gaussian norm, there exists a constant $c > 0$ that does not depend on $n$, $p$ or $k$ such that
\begin{align*}
    \Prob\left[ \abs{g_{12}(k) - \E\left\lbrace g_{12}(k) \mid \bm{Y}_{(-f)},\bm{\pi}, \bm{Q} \right\rbrace} \geq tk^{1/2} \mid \bm{Y}_{(-f)},\bm{\pi}, \bm{Q} \right] \leq \exp\left\lbrace -c\left(1-\alpha_+\right)^{-4}t^2 p \right\rbrace, \, 0 \leq t \leq c.
\end{align*}
Putting this all together gives us
\begin{align*}
    &g_{12}(k) \begin{cases}
    \geq 0 & \text{if $\limsup_{n,p \to \infty}\lambda_k = \infty$}\\
    =k \left\lbrace 1+O_P\left(n^{-1/2}\right) \right\rbrace & \text{if $k \in \left\lbrace c_{\min},\ldots,c_{\max} \right\rbrace$}\\
    \geq \sigma k\left(1+x_{n,k}\right) & \text{if $k \in \left\lbrace c_{\max}+1,\ldots,K_{\max} \right\rbrace$}
    \end{cases}\\
    &\max_{k \in \left\lbrace c_{\max}+1,\ldots,K_{\max} \right\rbrace}\abs{k^{1/2}x_{n,k}} = O_P(1)
\end{align*}
as $n,p \to \infty$.

\item I again assume $\delta_f^2 = 1$ for notational convenience, which is again without loss of generality. Define
\begin{align*}
    g_{13}(k)=&\frac{1}{n^{1/2} p_{f}^{1/2}} \sum\limits_{i=1}^n \left( \frac{1}{1-\hat{\bar{h}}_i}\right)^2\hat{\bar{\bm{H}}}_{\bigcdot i}^T \bm{A}_{(-i)} \hat{\bar{\bm{V}}}_{(-f)}^{-1/2} \bar{\bm{E}}_f^T \bm{\Delta}_i, \quad k \in [K_{\max}]\\
    \bm{\Delta}_i =& \bar{\bm{L}}_f\left\lbrace I_K - \sum\limits_{r=1}^{k} \hat{\bm{v}}_{(-f)_r}\hat{\bm{v}}_{(-f)_r}^{\T} \right\rbrace \bar{\bm{C}}_{i \bigcdot} - n^{1/2}\bar{\bm{L}}_f\sum\limits_{r=1}^{k} \hat{\bm{v}}_{(-f)_r} \hat{\bm{z}}_{(-f)_r}^{\T}\bar{\bm{q}}_i, \quad k \in [K]\\
    \bm{\Delta}_i =& \bar{\bm{L}}_f\left\lbrace I_K - \sum\limits_{r=1}^{K} \hat{\bm{v}}_{(-f)_r}\hat{\bm{v}}_{(-f)_r}^{\T} \right\rbrace \bar{\bm{C}}_{i \bigcdot} - n^{1/2}\bar{\bm{L}}_f\sum\limits_{r=1}^{K} \hat{\bm{v}}_{(-f)_r} \hat{\bm{z}}_{(-f)_r}^{\T}\bar{\bm{q}}_i\\
    & + \bar{\bm{L}}_f\left\lbrace n^{-1}\bar{\bm{C}}^{\T}\hat{\bm{R}}_{(-f)}\right\rbrace\hat{\bm{R}}_{(-f)_{i \bigcdot}}, \quad k \in \left\lbrace K+1,\ldots,K_{\max} \right\rbrace
\end{align*}
where $\bar{\bm{q}}_i$ is the $i$th row of $\bm{Q}_{\bar{C}}$ and for $k > K$,
\begin{align*}
    \hat{\bar{\bm{C}}} = \begin{pmatrix}
    \bar{\bm{C}}\hat{\bm{v}}_{(-f)} + \bm{Q}_{\bar{C}}\hat{\bm{z}}_{(-f)} & \hat{\bm{R}}_{(-f)}
    \end{pmatrix} \in \mathbb{R}^{n \times k}.
\end{align*}
It is clear that $\E\left\lbrace g_{13}(k) \mid \bm{Y}_{(-f)},\bm{\pi},\bm{Q} \right\rbrace = 0$. To understand the variation around $0$, note that
\begin{align*}
    g_{13}(k) =& p_f^{-1/2}\sum\limits_{g=1}^{p_f}\bar{\bm{L}}_{f_{g \bigcdot}}^{\T}\left\lbrace \bm{M}_1 \bm{N}_1 + \bm{M}_2 \bm{N}_2 + I\left(k > K\right)\bm{M}_3 \bm{N}_3 \right\rbrace \bar{\bm{E}}_{f_{g \bigcdot}}\\
    \bm{M}_1 =& I_K - \sum\limits_{r=1}^{k \wedge K} \hat{\bm{v}}_{(-f)_r}\hat{\bm{v}}_{(-f)_r}^{\T}, \quad \bm{N}_1 = n^{-1/2}\sum\limits_{i=1}^n \left(1-\hat{\bar{h}}_i\right)^{-2}\bar{\bm{C}}_{i \bigcdot} \hat{\bar{\bm{H}}}_{\bigcdot i}^{\T} \bm{A}_{(-i)}\hat{\bar{\bm{V}}}_{(-f)}^{-1/2}\\
    \bm{M}_2 =& \sum\limits_{r=1}^{k \wedge K} \hat{\bm{v}}_{(-f)_r} \hat{\bm{z}}_{(-f)_r}^{\T}, \quad \bm{N}_2 = \sum\limits_{i=1}^n \left(1-\hat{\bar{h}}_i\right)^{-2}\bar{\bm{q}}_i \hat{\bar{\bm{H}}}_{\bigcdot i}^{\T} \bm{A}_{(-i)}\hat{\bar{\bm{V}}}_{(-f)}^{-1/2}\\
    \bm{M}_3 =& n^{-1}\bar{\bm{C}}^{\T}\hat{\bm{R}}_{(-f)}, \quad \bm{N}_3 = n^{-1/2}\sum\limits_{i=1}^n \left(1-\hat{\bar{h}}_i\right)^{-2}\hat{\bm{R}}_{(-f)_{i \bigcdot}} \hat{\bar{\bm{H}}}_{\bigcdot i}^{\T} \bm{A}_{(-i)}\hat{\bar{\bm{V}}}_{(-f)}^{-1/2},
\end{align*}
where for some constant $c > 0$,
\begin{align*}
    \norm{\bm{N}_t}_2 \leq \left(1-\alpha_+\right)^{-2}c, \quad t \in [3].
\end{align*}
By assumption,
\begin{align*}
    e_g=\bar{\bm{L}}_{f_{g \bigcdot}}^{\T}\left\lbrace \bm{M}_1 \bm{N}_1 + \bm{M}_2 \bm{N}_2 + I\left(k > K\right)\bm{M}_3 \bm{N}_3 \right\rbrace \bar{\bm{E}}_{f_{g \bigcdot}}, \quad g \in [p_f]
\end{align*}
is sub-Gaussian with
\begin{align*}
    \log\left[\E\left\lbrace \exp\left(te_g\right) \mid \bm{Y}_{(-f)},\bm{\pi},\bm{Q} \right\rbrace\right] \leq & c t^2\left(1-\alpha_+\right)^{-4}\left\lbrace \norm{\bm{M}_1\bar{\bm{L}}_{f_{g \bigcdot}}}_2^2 + \norm{\bm{M}_2\bar{\bm{L}}_{f_{g \bigcdot}}}_2^2\right.\\
    &\left. + I\left(k > K\right)\norm{\bm{M}_3\bar{\bm{L}}_{f_{g \bigcdot}}}_2^2 \right\rbrace, \quad g \in [p_f],
\end{align*}
where $c > 0$ is a constant that does not depend on $n, p$ or $k$. And because $e_1,\ldots,e_{p_f}$ are independent and $\bar{\bm{L}}_f$ is at most rank $K$,
\begin{align*}
    \log\left(\E\left[ \exp\left\lbrace tg_{13}\left(k\right)\right\rbrace \mid \bm{Y}_{(-f)},\bm{\pi},\bm{Q} \right]\right) \leq & Kc t^2p^{-1}\left(1-\alpha_+\right)^{-4}\left\lbrace \norm{ \bar{\bm{L}}_{f}\bm{M}_1 }_2^2 + \norm{ \bar{\bm{L}}_{f}\bm{M}_2 }_2^2\right.\\
    &\left.+ I\left(k > K\right)\norm{ \bar{\bm{L}}_{f}\bm{M}_3 }_2^2 \right\rbrace, \quad k \in [K_{\max}].
\end{align*}
For notational simplicity, I will ignore the subscripts $f$ and $(-f)$ when deriving the asymptotic properties of $\bm{M}_1,\bm{M}_2$ and $\bm{M}_3$. We first see that for $j=1,2,3$, and some constant $c > 0$ that does not depend on $n,p$ or $k$,
\begin{align*}
    \norm{ \bar{\bm{L}}\bm{M}_j }_2 \leq& \norm{ \left( \bar{\bm{L}}_{\bigcdot 1} \cdots \bar{\bm{L}}_{\bigcdot c_{\min}}\right)\left( \bm{M}_{j_{\bigcdot 1}} \cdots \bm{M}_{j_{\bigcdot c_{\min}}} \right)^{\T} }_2 + \norm{ \left( \bar{\bm{L}}_{\bigcdot (c_{\min}+1)} \cdots  \bar{\bm{L}}_{\bigcdot K}\right)\left( \bm{M}_{j_{\bigcdot (c_{\min}+1)}} \cdots \bm{M}_{j_{\bigcdot K}} \right)^{\T} }_2\\
    \leq& \norm{ \left( \bar{\bm{L}}_{\bigcdot 1} \cdots \bar{\bm{L}}_{\bigcdot c_{\min}}\right)\left( \bm{M}_{j_{\bigcdot 1}} \cdots \bm{M}_{j_{\bigcdot c_{\min}}} \right)^{\T} }_2 + c \norm{\left( \bar{\bm{L}}_{\bigcdot (c_{\min}+1)} \cdots  \bar{\bm{L}}_{\bigcdot K}\right)}_2,
\end{align*}
where for some constant $\tilde{c} > 0$ that does not depend on $n,p$ or $k$, $\norm{\left( \bar{\bm{L}}_{\bigcdot (c_{\min}+1)} \cdots  \bar{\bm{L}}_{\bigcdot K}\right)}_2 \leq \tilde{c}\lambda_{c_{\min}+1}^{1/2}\left\lbrace 1+O_P\left(n^{1/2}p^{-1/2}\right) \right\rbrace$ uniformly over $k \in [K_{\max}]$ by Lemma \ref{lemma:sup:LtLRandom}. To understand the behavior of the remaining term in the above expression we consider the cases, we first note that $\phi_2 = \norm{\bar{\bm{V}} - \hat{\bar{\bm{V}}}}_2$, defined in Lemma \ref{lemma:supp:vandzhat}, satisfies $\phi_2 = O(1)$. Since the rates given in \eqref{equation:supp:vandzResults} only depend on the choice of $\hat{\bar{\bm{V}}}$ through $\phi_2$, the rates in \eqref{equation:supp:vandzResults} hold uniformly over all $k \in [K_{\max}]$. Therefore, Since $\lambda_{r} \to \infty$ for all $r \leq c_{\min}$, \eqref{equation:supp:Vhatrr}, \eqref{equation:supp:Vhatrj} and \eqref{equation:supp:zhatnorm} imply
\begin{align*}
    &\max_{k \in \left\lbrace c_{\min},\ldots,K_{\max} \right\rbrace} \left\lbrace \norm{ \left( \bar{\bm{L}}_{\bigcdot 1} \cdots \bar{\bm{L}}_{\bigcdot c_{\min}}\right)\left( \bm{M}_{j_{\bigcdot 1}} \cdots \bm{M}_{j_{\bigcdot c_{\min}}} \right)^{\T} }_2 \right\rbrace = O_P(1), \quad j \in [2],\\
    &\norm{ \left( \bar{\bm{L}}_{\bigcdot 1} \cdots \bar{\bm{L}}_{\bigcdot c_{\min}}\right)\left( \bm{M}_{j_{\bigcdot 1}} \cdots \bm{M}_{j_{\bigcdot c_{\min}}} \right)^{\T} }_2 = O_P\left(\lambda_{k+1}^{1/2}\right), \quad k\in \left[c_{\min}-1\right]; j \in [2].
\end{align*}
For $\bm{M}_3$, let $\bar{\bm{C}}_{\min} = \left( \bar{\bm{C}}_{\bigcdot 1} \cdots \bar{\bm{C}}_{\bigcdot c_{\min}} \right)$ and $\hat{\bar{\bm{C}}}_{\min} = \left( \hat{\bar{\bm{C}}}_{\bigcdot 1} \cdots \hat{\bar{\bm{C}}}_{\bigcdot c_{\min}} \right)$. Then \eqref{equation:supp:Vhatrr}, \eqref{equation:supp:Vhatrj} and \eqref{equation:supp:zhatnorm} and the fact that $\hat{\bar{\bm{C}}}^{\T}_{\min}\hat{\bm{R}} = \bm{0}$ imply
\begin{align*}
    \norm{ \left( \bar{\bm{L}}_{\bigcdot 1} \cdots \bar{\bm{L}}_{\bigcdot c_{\min}}\right)\left( \bm{M}_{3_{\bigcdot 1}} \cdots \bm{M}_{3_{\bigcdot c_{\min}}} \right)^{\T} }_2 = \norm{ \left( \bar{\bm{L}}_{\bigcdot 1} \cdots \bar{\bm{L}}_{\bigcdot c_{\min}}\right) \left(n^{-1}\bar{\bm{C}}_{\min}^{\T}\hat{\bm{R}}\right) }_2
\end{align*}
behaves exactly as $\norm{ \left( \bar{\bm{L}}_{\bigcdot 1} \cdots \bar{\bm{L}}_{\bigcdot c_{\min}}\right)\left( \bm{M}_{j_{\bigcdot 1}} \cdots \bm{M}_{j_{\bigcdot c_{\min}}} \right)^{\T} }_2$ for $j=1,2$. Putting this all together gives us that for any $\epsilon > 0$,
\begin{align*}
    &g_{13}(k) = O_P\left\lbrace \left(\lambda_{k+1} \vee 1\right)^{1/2}p^{-1/2} \right\rbrace, \quad k \in \left[ c_{\max} \right]\\
    &\max_{k \in \left\lbrace c_{\max}+1,\ldots,K_{\max} \right\rbrace} \abs{g_{13}(k)} = O_P\left(p^{-1/2+\epsilon}\right),
\end{align*}
where $\lambda_r = 0$ for $r > K$.
\end{enumerate}
This completes the proof.
\end{proof}

\begin{lemma}
\label{lemma:supp:CBCV:g2}
Let $c_{\min}$ be such that $\limsup_{n,p\to\infty}\lambda_{c_{\min}} = \infty$ but $\limsup_{n,p\to\infty}\lambda_{c_{\min}+1} < \infty$, where $\lambda_0=\infty$ and $\lambda_{K+1}=0$. Then under the assumptions of Theorem \ref{theorem:Khat} and for some arbitrarily large integer $c_{\max} > K$ that does not depend on $n$ or $p$, $g_2(k)$ satisfies
\begin{align*}
    g_2(k) \begin{cases}
    \geq n + O_P\left(n^{-1/2}p^{-1/2}\right)& \text{if $k \notin \left\lbrace c_{\min},\ldots,c_{\max} \right\rbrace$}\\
    =n + O_P\left(n^{-1} + n^{1/2}p^{-1/2}\right)& \text{if $k \in \left\lbrace c_{\min},\ldots,c_{\max} \right\rbrace$}
    \end{cases}
\end{align*}
as $n,p \to \infty$.
\end{lemma}
\begin{proof}
This is follows directly from item \ref{item:preliminaries:trace} of Lemma \ref{lemma:supp:Preliminaries} and the proof of Theorem S5 in \citet{CorrConf}.
\end{proof}

\begin{lemma}
\label{lemma:supp:CBCV:g3}
Under the assumptions of Theorem \ref{theorem:Khat} and for some constant $c>0$ that does not depend on $n$ or $p$, $g_3(k)$ satisfies
\begin{align*}
    & \abs{g_{3}(k)}=\begin{cases}
    = \gamma_{k+1}^{1/2}O_P\left(n^{1/2}p^{-1/2}\right) & \text{if $k < c_{\min}$}\\
    = O_P\left(n^{1/2}p^{-1/2} + n^{-1}\right) & \text{if $c_{\min} \leq k \leq c_{\max}$}
    \end{cases}\\
    &\max_{k \in \left\lbrace c_{\max}+1,\ldots,K_{\max} \right\rbrace} \abs{k^{-1}\abs{g_{3}(k)}} = O_P\left(n^{1/2}p^{-1/2} + n^{-1}\right),
\end{align*}
where $c_{\min}$ is such that $\limsup_{n,p \to \infty} \lambda_{c_{\min}} = \infty$ and $\limsup_{n,p \to \infty} \lambda_{c_{\min}+1} < \infty$ (for $\lambda_0 = \infty$ and $\lambda_{K+1}=0$) and $c_{\max} > K$ is an arbitrarily large constant that does not depend on $n$ or $p$.
\end{lemma}
\begin{proof}

I again assume, without loss of generality, that $\delta_f^2=\abs{\bar{\bm{V}}_{f}}=1$. Then $g_3(k)$, up to a scalar constant, can be written as
\begin{align*}
    g_3(k) =& n^{-1/2}p_f^{-1/2}\sum\limits_{i=1}^n \left( 1-\hat{\bar{h}}_i\right)^{-1}\bm{\Delta}_i^T \bar{\bm{E}}_f \hat{\bar{\bm{V}}}_{(-f)}^{-1/2}\bm{a}_i - \sum\limits_{i=1}^n \left( 1-\hat{\bar{h}}_i\right)^{-1}\hat{\bm{H}}_i^T \bm{A}_{(-i)}\hat{\bar{\bm{V}}}_{(-f)}^{-1/2}\left(p_f^{-1}\bar{\bm{E}_f}^T \bar{\bm{E}}_f\right) \hat{\bar{\bm{V}}}_{(-f)}^{-1/2} \bm{a}_i\\
    =& g_{31}(k) - g_{32}(k),
\end{align*}
where $\bm{\Delta}_i$ is defined in \eqref{equation:Supp:LOOXV:Deltai}. Clearly $\E\left\lbrace g_{31}(k) \mid \bm{Y}_{(-f)},\bm{\pi},\bm{Q} \right\rbrace = 0$. An analysis identical to that used to derive the finite sample properties of $g_{13}(k)$ in Lemma \ref{lemma:supp:CBCV:g1} can be used to show that for all $t \in \mathbb{R}$ and some constant $c > 0$ that does not depend on $n$, $p$ or $k$,
\begin{align*}
    \log\left( \E\left[ \exp\left\lbrace tg_{31}(k) \right\rbrace \right] \mid \bm{Y}_{(-f)},\bm{\pi},\bm{Q}  \right) \leq & ct^2 np^{-1}\left( 1-\alpha_+ \right)^{-2}\hat{\bm{M}}, \quad k \in \left[ K_{\max} \right]\\
    \hat{\bm{M}} =& \norm{\bar{\bm{L}}_f \bm{M}_1}_2^2 + \norm{\bar{\bm{L}}_f \bm{M}_2}_2^2 + I(k > K)\norm{\bar{\bm{L}}_f \bm{M}_3}_2^2
\end{align*}
where $\bm{M}_j$, $j=1,2,3$, are as defined in Lemma \ref{lemma:supp:CBCV:g1}. This shows that
\begin{align*}
    g_{31}(k) = \left(\gamma_{k+1} \vee 1\right)^{1/2} O_P\left(n^{1/2}p^{-1/2}\right), \quad k \in \left[c_{\max}\right],
\end{align*}
where $\gamma_r = 0$ for $r > K$. Further, a union bound shows that for all $t > 0$ and some constant $\tilde{c} > 0$ that does depend on $n$, $p$ or $k$,
\begin{align*}
    &\Prob\left\lbrace \text{$\abs{g_{31}(k)} \geq tk^{1/2}$ for at least one $k \in \left\lbrace c_{\max}+1,\ldots,K_{\max} \right\rbrace$} \mid \bm{Y}_{(-f)},\bm{\pi},\bm{Q} \right\rbrace\\
    &\leq 2\sum\limits_{k=c_{\max}+1}^{K_{\max}} \exp\left\lbrace -t^2\tilde{c}\frac{p\left(1-\alpha_+\right)^2}{n\hat{\bm{M}}} \right\rbrace^k\\
    &\leq 2 \left[ 1-\exp\left\lbrace -t^2\tilde{c}\frac{p\left(1-\alpha_+\right)^2}{n\hat{\bm{M}}} \right\rbrace^{c_{\max}+1} \right]^{-1}\exp\left\lbrace -t^2\tilde{c}\frac{p\left(1-\alpha_+\right)^2}{n\hat{\bm{M}}} \right\rbrace^{c_{\max}+1},
\end{align*}
which implies
\begin{align*}
    \max_{k \in \left\lbrace c_{\max}+1,\ldots,K_{\max} \right\rbrace}\abs{k^{-1/2}\abs{g_{31}(k)}} = O_P\left(n^{1/2}p^{-1/2}\right)
\end{align*}
as $n,p \to \infty$.\par 
\indent Define $\tilde{\bm{V}}_f = \bm{Q}^{\T}\bar{\bm{V}}_f\bm{Q}$ and $\bm{R} = \hat{\bar{\bm{V}}}_{(-f)}^{-1/2}\tilde{\bm{V}}_f\hat{\bar{\bm{V}}}_{(-f)}^{-1/2} - I_n$. Then
\begin{align*}
    \E\left\lbrace g_{32}(k) \mid \bm{Y}_{(-f)},\bm{\pi},\bm{Q}  \right\rbrace =& \sum\limits_{i=1}^n \left( 1-\hat{\bar{h}}_i\right)^{-1}\hat{\bar{\bm{H}}}_i^T \bm{A}_{(-i)}\hat{\bar{\bm{V}}}_{(-f)}^{-1/2}\tilde{\bm{V}}_f \hat{\bar{\bm{V}}}_{(-f)}^{-1/2} \bm{a}_i = \sum\limits_{i=1}^n \left( 1-\hat{\bar{h}}_i\right)^{-1}\hat{\bar{\bm{H}}}_i^T \bm{A}_{(-i)}\bm{R} \bm{a}_i\\
    =& \Tr\left( \tilde{\bm{R}}\hat{\bar{\bm{H}}} \right) - \sum\limits_{i=1}^n \hat{\bar{h}}_i\tilde{\bm{R}}_{ii}\\
    \tilde{\bm{R}} =& \bm{R} \diag\left\lbrace \left( 1-\hat{\bar{h}}_1\right)^{-1},\ldots,\left( 1-\hat{\bar{h}}_n\right)^{-1} \right\rbrace
\end{align*}
Therefore,
\begin{align*}
    \abs{\E\left\lbrace g_{32}(k) \mid \bm{Y}_{(-f)},\bm{\pi},\bm{Q}  \right\rbrace} \leq 2k\left(1-\alpha_+\right)^{-1}\norm{ \bm{R} }_2, \quad k \in [K_{\max}],
\end{align*}
which by Corollary \ref{corollary:supp:VAlgorithm} and Lemma \ref{lemma:supp:VLargek} implies for some constant $c > 0$,
\begin{align*}
    & \E\left\lbrace g_{32}(k) \mid \bm{Y}_{(-f)},\bm{\pi},\bm{Q}  \right\rbrace \begin{cases}
    \geq c\left\lbrace 1+O_P\left(n^{-1/2}\right) \right\rbrace & \text{if $k < c_{\min}$}\\
    = O_P\left( p^{-1/2} + n^{-1} \right) & \text{if $k \in \left\lbrace c_{\min},\ldots,c_{\max} \right\rbrace$}
    \end{cases}\\
    &\max_{k \in \left\lbrace c_{\max}+1,\ldots,K_{\max} \right\rbrace} \abs{k^{-1} \E\left\lbrace g_{32}(k) \mid \bm{Y}_{(-f)},\bm{\pi},\bm{Q}  \right\rbrace } = O_P\left( n^{1/2}p^{-1/2} + n^{-1} \right).
\end{align*}
We lastly need to understand the variation of $g_{32}(k)$ around its conditional mean. We first note that
\begin{align*}
    g_{32}(k) =& p_f^{-1}\sum\limits_{g=1}^{p_f} \tilde{\bm{E}}_{g \bigcdot}^{\T} \bm{M}^{(k)}\tilde{\bm{E}}_{g \bigcdot}\\
    \tilde{\bm{E}} =& \bar{\bm{E}}_f \hat{\bar{\bm{V}}}_{(-f)}^{-1/2}\\
    \bm{M}^{(k)} =& \sum\limits_{i=1}^n (1-\hat{\bar{h}}_i)^{-1}\bm{a}_i \hat{\bar{\bm{H}}}_i^{\T} \bm{A}_{(-i)} = \diag\left\lbrace (1-\hat{\bar{h}}_1)^{-1},\ldots,(1-\hat{\bar{h}}_n)^{-1} \right\rbrace \left\lbrace\hat{\bar{\bm{H}}} - \diag\left( \hat{\bar{h}}_1,\ldots,\hat{\bar{h}}_n \right)\right\rbrace.
\end{align*}
We see that $\norm{ \bm{M}^{(k)} }_2 \leq 2 \left(1-\alpha_+\right)^{-1}$ and
\begin{align*}
    \norm{ \bm{M}^{(k)} }_F^2 = \sum\limits_{i=1}^n \frac{\hat{\bar{h}}_i}{1-\hat{\bar{h}}_i} \leq k\left(1-\alpha_+\right)^{-1}, \quad k \in [K_{\max}].
\end{align*}
Further, Assumption \ref{assumption:CandL} and Proposition 2.7 and Remark 2.8 of \citet{SubGaussVariance} imply
\begin{align}
\label{equation:supp:Mk}
    \norm{ \tilde{\bm{E}}_{g \bigcdot}^{\T} \bm{M}^{(k)}\tilde{\bm{E}}_{g \bigcdot} }_{\Psi_1} \leq c \norm{ \bm{M}^{(k)} }_F, \quad g \in [p]; k \in [K_{\max}],
\end{align}
where $c > 0$ is a constant that does not depend on $n,p$ or $k$ and $\norm{ \cdot }_{\Psi_1}$ is the sub-Exponential norm applied conditionally on $\bm{Y}_{(-f)},\bm{\pi},\bm{Q}$, defined as
\begin{align*}
    \norm{ x }_{\Psi_1} = \inf_{t > 0}\left\lbrace \E\left\lbrace \exp(\abs{x/t}) \mid \bm{Y}_{(-f)},\bm{\pi},\bm{Q} \right\rbrace \leq e \right\rbrace.
\end{align*}
Since the rows of $\tilde{\bm{E}}$ are independent conditional on $\bm{Y}_{(-f)}, \bm{\pi}$ and $\bm{Q}$, Proposition 5.16 of \citet{Vershynin} implies that for all $t \geq 0$ and $\mu_k = \E\left\lbrace g_{32}(k) \mid \bm{Y}_{(-f)},\bm{\pi},\bm{Q}  \right\rbrace$,
\begin{align*}
    \Prob\left\lbrace \abs{g_{32}(k) - \mu_k} \geq tck^{1/2}\left(1-\alpha_+\right)^{-1} \mid \bm{Y}_{(-f)},\bm{\pi},\bm{Q} \right\rbrace \leq 2 \exp\left\lbrace -\tilde{c}p\min\left( t^2,t \right) \right\rbrace, \quad k \in [K_{\max}],
\end{align*}
where $c> 0$ is defined in \eqref{equation:supp:Mk} and $\tilde{c} > 0$ is a constant that does not depend on $n,p$ or $k$. This completes the proof.
\end{proof}

Aggregating the results of Lemmas \ref{lemma:supp:CBCV:g1}, \ref{lemma:supp:CBCV:g2} and \ref{lemma:supp:CBCV:g3} gives us
\begin{align*}
   & g(k)-n \begin{cases}
    \geq O_P\left(\lambda_{k+1}\right) & \text{if $k < c_{\min}$}\\
    =k + I(k < K)\sum\limits_{r=k+1}^K \delta^{-2}\gamma_r + O_P\left(n^{1/2}p^{-1/2} + n^{-1/2}\right) & \text{if $k \in \left\lbrace c_{\min},\ldots,c_{\max} \right\rbrace$}\\
    \geq \sigma k \left\lbrace 1+x_{k} \right\rbrace & \text{if $k \in \left\lbrace c_{\max}+1,\ldots,K_{\max} \right\rbrace$}
    \end{cases}\\
    & \max_{k \in \left\lbrace c_{\max}+1,\ldots,K_{\max} \right\rbrace}\abs{x_{k}} = O_P\left(n^{1/2}p^{-1/2} + n^{-1/2}\right),
\end{align*}
where $c_{\min} \geq 0$ is such that $\limsup_{n,p\to\infty}\lambda_{c_{\min}} = \infty$ but $\limsup_{n,p\to\infty}\lambda_{c_{\min}+1} < \infty$ (where $\lambda_0=\infty$ and $\lambda_{K+1}=0$), $c_{\max}>K$ is an arbitrarily large integer and $\sigma > 0$ is a constant, all of which do not depend on $n$, $p$ or $k$. This implies that for all $\epsilon > 0$, there exists a constant $M_{\epsilon}$ that does not depend on $n$ or $p$ such that if $\left(\gamma_s - \delta^2\right)\left(n^{1/2}p^{-1/2}+n^{-1/2}\right)^{-1} \geq M_{\epsilon}$, $\liminf_{n,p \to \infty}\Prob\left(\Koraclehat = s\right) \geq 1-\epsilon$. An identical argument to that presented above can be used to show that under the same conditions, $\liminf_{n,p \to \infty}\Prob\left(\Koracle = s\right) \geq 1-\epsilon$.
\end{proof}

\section{Other important results}
\label{section:supp:lemmas}

\begin{lemma}
\label{lemma:supp:EigBound}
Let $\bm{D} =\diag\left(d_1,\ldots,d_K\right)$ with $d_1 \geq \cdots \geq d_K \geq 0$ and $\bm{M} \in \mathbb{R}^{K \times K}$ be a matrix whose eigenvalues lie in the compact set $\left[c_1,c_2\right]$ where $c_1 > 0$. Then $\Lambda_k\left(\bm{M}^{1/2}\bm{D}\bm{M}^{1/2}\right) \in \left[d_kc_1,d_kc_2\right]$.
\end{lemma}
\begin{proof}
Let $\mathcal{L} \subseteq \mathbb{R}^K$ be a vector space and $\abs{\mathcal{L}} \leq K$ be its dimension. Then
\begin{align*}
    \Lambda_k\left(\bm{M}^{1/2}\bm{D}\bm{M}^{1/2}\right) &= \max_{\abs{\mathcal{L}} = k} \min_{\bm{u} \in \mathcal{L}\setminus \left\lbrace \bm{0} \right\rbrace} \frac{\bm{u}^T \bm{M}^{1/2}\bm{D}\bm{M}^{1/2} \bm{u}}{\bm{u}^T \bm{u}} = \max_{\abs{\mathcal{L}} = k} \min_{\bm{u} \in \mathcal{L}\setminus \left\lbrace \bm{0} \right\rbrace} \frac{\bm{u}^T \bm{D} \bm{u}}{\bm{u}^T \bm{M}^{-1} \bm{u}}\\
    &= \max_{\abs{\mathcal{L}} = k} \min_{\substack{\bm{u} \in \mathcal{L}\setminus \left\lbrace \bm{0} \right\rbrace \\ \bm{u}^T \bm{u} = 1}} \frac{\bm{u}^T \bm{D} \bm{u}}{\bm{u}^T \bm{M}^{-1} \bm{u}}.
\end{align*}
Consider the subspace $\mathcal{L}$ generated by the first $k \leq K$ canonical basis vectors. Then $\frac{\bm{u}^T \bm{D} \bm{u}}{\bm{u}^T \bm{M}^{-1} \bm{u}} \geq d_k c_1$, which gives the lower bound. For the upper bound,
\begin{align*}
    \Lambda_k\left(\bm{M}^{1/2}\bm{D}\bm{M}^{1/2}\right) = \min_{\abs{\mathcal{L}} = K-k+1} \max_{\substack{\bm{u} \in \mathcal{L}\setminus \left\lbrace \bm{0} \right\rbrace \\ \bm{u}^T \bm{u} = 1}} \frac{\bm{u}^T \bm{D} \bm{u}}{\bm{u}^T \bm{M}^{-1} \bm{u}}.
\end{align*}
Setting the subspace $\mathcal{L}$ to be the $K$th through $k$th canonical basis vectors gives us $\frac{\bm{u}^T \bm{D} \bm{u}}{\bm{u}^T \bm{M}^{-1} \bm{u}} \leq d_k c_2$.
\end{proof}

\begin{lemma}
\label{lemma:supp:RotatedCorrelation}
Suppose $\bm{E} \in \mathbb{R}^{p \times n}$ such that $\vecM\left(\bm{E}\right) = \bm{A}\vecM\left(\tilde{\bm{E}}\right)$ for some $\bm{A}\in\mathbb{R}^{np \times np}$, where the entries of $\tilde{\bm{E}}\in\mathbb{R}^{p \times n}$ are independent with uniformly bounded sub-Gaussian norm, and $\norm{\bm{A}}_2 = O(1)$ as $n,p \to \infty$. Let $\bm{V}=\E\left( p^{-1}\bm{E}^T\bm{E} \right)$ be a positive definite matrix with eigenvalues that are uniformly bounded above 0 and below $\infty$ as $n,p \to \infty$. Then 
\begin{align*}
    \norm{ p^{-1}\bm{E}^T\bm{E} - \bm{V} }_2 = O_P\left(n^{1/2}p^{-1/2}\right)
\end{align*}
as $n,p \to \infty$.
\end{lemma}
\begin{proof}
Since
\begin{align*}
     \norm{ p^{-1}\bm{E}^T\bm{E} - \bm{V} }_2 \leq \norm{ \bm{V}^{-1} }_2 \norm{ p^{-1}\bm{V}^{-1/2}\bm{E}^T\bm{E}\bm{V}^{-1/2} - I_n }_2
\end{align*}
and $\norm{ \bm{V}^{-1} }_2 = O(1)$ as $n,p \to \infty$, it suffices to assume $\bm{V}=I_n$. First, for any unit vector $\bm{v} \in \mathbb{R}^n$,
\begin{align*}
    p^{-1}\bm{v}^T\bm{E}^T\bm{E}\bm{v} - 1 = p^{-1}\vecM\left(\tilde{\bm{E}}\right)^T \bm{A}^T \left(\bm{v}\oplus \cdots \oplus \bm{v}\right)\left(\bm{v}\oplus \cdots \oplus \bm{v}\right)^T\bm{A} \vecM\left(\tilde{\bm{E}}\right) - 1,
\end{align*}
where $\norm{\bm{v}\oplus \cdots \oplus \bm{v}}_2 = 1$, meaning,
\begin{align*}
    \norm{ \bm{A}^T \left(\bm{v}\oplus \cdots \oplus \bm{v}\right)\left(\bm{v}\oplus \cdots \oplus \bm{v}\right)^T\bm{A} }_2 \leq \norm{\bm{A}}_2^2.
\end{align*}
Let $\bm{B}_v = \bm{A}\bm{A}^T \left(\bm{v}\oplus \cdots \oplus \bm{v}\right)\left(\bm{v}\oplus \cdots \oplus \bm{v}\right)^T \bm{A}\bm{A}^T$ and define $\bm{B}_v^{(g)} \in \mathbb{R}^{n \times n}$ be the $g$th diagonal block of $\bm{B}_v$, where $g=1,\ldots,p$. Then we also have
\begin{align*}
   \norm{ \bm{A}^T \left(\bm{v}\oplus \cdots \oplus \bm{v}\right)\left(\bm{v}\oplus \cdots \oplus \bm{v}\right)^T\bm{A} }_F^2 = \sum\limits_{g=1}^p \bm{v}^T \bm{B}_v^{(g)}\bm{v} \leq p\norm{\bm{A}}_2^2.
\end{align*}
By Remark 2.10 in \citet{SubGaussVariance}, this implies that
\begin{align*}
    \Prob\left( \abs{p^{-1}\bm{v}^T\bm{E}^T\bm{E}\bm{v} - 1} \geq t \right) \leq 2\exp\left\lbrace -p\min\left( \tilde{c}^2t^2, \tilde{c}t \right) \right\rbrace
\end{align*}
for some constant $\tilde{c}>0$ that is not a function of $\bm{v}, n$ or $p$. A standard covering argument (e.g. Theorem 5.39 in \citet{Vershynin}) then gives us the result.
\end{proof}

\begin{lemma}
\label{lemma:supp:Preliminaries}
Let $c > 0$ be a large constant, and suppose $\bm{E} \in \mathbb{R}^{p \times n}$, $\vecM\left(\bm{E}\right) \edist \bm{A}\vecM\left(\tilde{\bm{E}}\right)$, $\norm{\bm{A}}_2 \leq c$ and the entries of $\tilde{\bm{E}} \in \mathbb{R}^{p \times n}$ are independent with mean 0, variance 1 and sub-Gaussian norm bounded above by $c$. Then the following hold for any $\bm{u}_1,\bm{u}_2 \in \mathbb{S}^{n-1}$, $\bm{\ell}_1,\bm{\ell}_2 \in \mathbb{S}^{p-1}$ and positive semi-definite matrix $\bm{V} \in \mathbb{R}^n$ with $\norm{\bm{V}}_2 \leq c$:
\begin{enumerate}[label=(\roman*)]
    \item $\Prob\left[\abs{ \Tr\left( \bm{E}^{T}\bm{E}\bm{V} \right) - \E\left\lbrace\Tr\left( \bm{E}^{T}\bm{E}\bm{V} \right)\right\rbrace } \geq t\left(np\right)^{1/2}\right] \leq 2\exp\left[ -\min\left\lbrace\tilde{c}t^2,\tilde{c}t\left(np\right)^{1/2}\right\rbrace \right]$\label{item:preliminaries:trace}
    \item $\Prob\left\lbrace \abs{ \bm{u}_1^T\bm{E}^{T}\bm{E}\bm{u}_2 - \E\left( \bm{u}_1^T\bm{E}^{T}\bm{E}\bm{u}_2 \right) } \geq tp^{1/2} \right\rbrace \leq 2\exp\left\lbrace -\min\left(\tilde{c}t^2,\tilde{c}t p^{1/2}\right) \right\rbrace$\label{item:preliminaries:E1E2}
    \item For $\bm{\Sigma}_i \in \mathbb{R}^{p \times p}$ the $i$th diagonal block of $\bm{A}\bm{A}^T$ and $\bm{a}_g \in \mathbb{R}^p$ the $g$th standard basis vector, $\E\left(\bm{E} \bm{E}^T\right) = \sum\limits_{i=1}^n \bm{\Sigma}_{i}$ and
    \begin{align*}
        &\Prob\left\lbrace \abs{ \bm{\ell}_1^T\bm{E}\bm{E}^T \bm{\ell}_2 - \E\left( \bm{\ell}_1^T\bm{E}\bm{E}^T \bm{\ell}_2 \right) } \geq tn^{1/2} \right\rbrace \leq 2\exp\left\lbrace -\min\left(\tilde{c}t,\tilde{c}tn^{1/2}\right) \right\rbrace\\
        &\abs{\E\left(\bm{\ell}_1^T\bm{E}\bm{E}^T \bm{a}_g\right)} \leq n\max_{h \in [p]}\left(\bm{\ell}_{1_h}\right)\norm{n^{-1}\sum\limits_{i=1}^n \bm{\Sigma}_i}_{1}, \quad g \in [p]
    \end{align*}\label{item:preliminaries:ell}
    \item $\Prob\left\lbrace \abs{ \bm{\ell}_1^T\bm{E}\bm{u}_1 } \geq t \right\rbrace \leq \exp\left\lbrace -\tilde{c}t^2 \right\rbrace$\label{item:preliminaries:trivial}
\end{enumerate}
for all $t \geq 0$, where $\tilde{c} > 0$ only depends on $c$.
\end{lemma}
\begin{proof}
The Inequality in \ref{item:preliminaries:trivial} is trivial and follows because $\vecM\left(\tilde{\bm{E}}\right)$ has sub-Gaussian norm bounded by $c$, and \ref{item:preliminaries:E1E2}, \ref{item:preliminaries:ell} follow by the proof of Lemma \ref{lemma:supp:RotatedCorrelation}. To prove \ref{item:preliminaries:trace}, we see that
\begin{align*}
    \Tr\left( \bm{E}^{T}\bm{E}\bm{V} \right) - \E\left\lbrace\Tr\left( \bm{E}^{T}\bm{E}\bm{V} \right)\right\rbrace \edist & \vecM\left(\tilde{\bm{E}}\right)^T \bm{A}^T\left(I_p \otimes \bm{V}\right)\bm{A}\vecM\left(\tilde{\bm{E}}\right)\\
    & - \E\left\lbrace \vecM\left(\tilde{\bm{E}}\right)^T \bm{A}^T\left(I_p \otimes \bm{V}\right)\bm{A}\vecM\left(\tilde{\bm{E}}\right) \right\rbrace,
\end{align*}
where $\norm{\bm{A}^T\left(I_p \otimes \bm{V}\right)\bm{A}}_2 \leq c^3$ and $\norm{\bm{A}^T\left(I_p \otimes \bm{V}\right)\bm{A}}_F^2 \leq np c^6$. The result then follows by the proof of Lemma \ref{lemma:supp:RotatedCorrelation}.
\end{proof}

\begin{corollary}
\label{corollary:supp:Preliminaries}
Under the conditions of Lemma \ref{lemma:supp:Preliminaries}, $\norm{\bm{\ell}_1^T \bm{E}}_2, p^{-1/2}\norm{\bm{u}_1^T\bm{E}^T \bm{E} - \E\left( \bm{u}_1^T\bm{E}^T \bm{E} \right)}_2 = O_P\left(n^{1/2}\right)$.
\end{corollary}
\begin{proof}
These follow by \ref{item:preliminaries:ell} and \ref{item:preliminaries:E1E2} in the statement of Lemma \ref{lemma:supp:Preliminaries}. 
\end{proof}

\begin{lemma}[Proposition 5.1 of \cite{EigenvalueBall}]
\label{lemma:supp:EigApprox}
Let $\bm{A} \in \mathbb{R}^{K \times K}$ be a symmetric matrix and $\bm{v} \in \mathbb{R}^K$ be a unit vector such that $\bm{v}^T\bm{A}\bm{v} = \delta$. Then $\bm{A}$ has an eigenvalue in the closed ball centered at $\delta$ with radius $\norm{\bm{A}\bm{v} - \delta\bm{v}}_2$.
\end{lemma}

\begin{lemma}
\label{lemma:supp:EigVector}
Let $\bm{H} \in \mathbb{R}^{K \times K}$ be a symmetric matrix, $\bm{V} \in \mathbb{R}^{K \times r}$ have orthonormal columns, $\bm{D} = \diag\left(d_1,\ldots,d_r\right) \succ \bm{0}$, $\bm{\epsilon} \in \mathbb{R}^{r \times r}$ and $\bm{W} \in \mathbb{R}^{K \times r}$ have orthonormal columns, where $r < K$. Suppose $\bm{H}\bm{V} = \bm{V}\bm{D} + \bm{W}\bm{\epsilon}$. Then if $\bm{V}_{\epsilon} \in \mathbb{R}^{K \times r}$ is any orthonormal matrix whose columns are eigenvalues of $\bm{H}$ such that $\mathop{\min}\limits_{k \in [K-r]}\abs{d_j - \Lambda_k\left( P_{V_{\epsilon}}^{\perp}\bm{H}P_{V_{\epsilon}}^{\perp} \right)} > 0$ for all $j \in [r]$,
\begin{align*}
    \norm{P_{V_{\epsilon}} - P_{V}P_{V_{\epsilon}}}_F= \norm{P_{V} - P_{V_{\epsilon}}P_{V}}_F \leq \left\lbrace \sum\limits_{j=1}^r \frac{\bm{\epsilon}_{\bigcdot j}^T \bm{\epsilon}_{\bigcdot j}}{\mathop{\min}\limits_{k \in[K]}\left\lbrace d_j - \Lambda_k\left( P_{V_{\epsilon}}^{\perp}\bm{H}P_{V_{\epsilon}}^{\perp} \right)\right\rbrace^2} \right\rbrace^{1/2}.
\end{align*}
\end{lemma}

\begin{proof}
Let $\bm{D}_{\epsilon} \in \mathbb{R}^{r \times r}$ be the eigenvalues of $\bm{H}$ associated with the eigenvectors $\bm{V}_{\epsilon}$ and define $\bm{R} = P_{V_{\epsilon}}^{\perp}\bm{V}$. By the statement of the theorem, we need only show that
\begin{align*}
    \norm{\bm{R}}_F \leq \left\lbrace \sum\limits_{j=1}^r \frac{\bm{\epsilon}_{\bigcdot j}^T \bm{\epsilon}_{\bigcdot j}}{\mathop{\min}\limits_{k \in[K]}\left\lbrace d_j - \Lambda_k\left( P_{V_{\epsilon}}^{\perp}\bm{H}P_{V_{\epsilon}}^{\perp} \right)\right\rbrace^2} \right\rbrace^{1/2}.
\end{align*}
First,
\begin{align*}
    P_{V_{\epsilon}}\bm{V}\bm{D} + \bm{R}\bm{D} + \bm{W}\bm{\epsilon} = \bm{H}\bm{V} = \bm{H}P_{V_{\epsilon}}\bm{V} + \bm{H}\bm{R} = \bm{V}_{\epsilon}\bm{D}_{\epsilon} \bm{V}_{\epsilon}^T \bm{V} + \bm{H}\bm{R}.
\end{align*}
Next,
\begin{align*}
    \bm{D}_{\epsilon}\bm{V}_{\epsilon}^T \bm{V} = \bm{V}_{\epsilon}^T \bm{H}\bm{V} = \bm{V}_{\epsilon}^T \bm{V} \bm{D} + \bm{V}_{\epsilon}^T\bm{W}\bm{\epsilon}.
\end{align*}
Therefore,
\begin{align*}
    \bm{R}\bm{D} - \bm{H}\bm{R} = \bm{V}_{\epsilon}\left[\bm{D}_{\epsilon}\bm{V}_{\epsilon}^T \bm{V} - \bm{V}_{\epsilon}^T \bm{V}\bm{D}\right] - \bm{W}\bm{\epsilon} = -P_{V_{\epsilon}}^{\perp}\bm{W}\bm{\epsilon},
\end{align*}
which implies that
\begin{align*}
    \left(\bm{H} - d_j I_K\right)\bm{R}_{\bigcdot j} = P_{V_{\epsilon}}^{\perp}\bm{W}\bm{\epsilon}_{\bigcdot j}, \quad j \in [r].
\end{align*}
Since the columns of $\bm{R}$ lie in the orthogonal complement of $\bm{V}_{\epsilon}$, this completes the proof.
\end{proof}

\begin{corollary}
\label{corollary:supp:NormSpaceVhat}
Under the conditions of Lemma \ref{lemma:supp:EigVector},
\begin{align*}
    \norm{P_{V_{\epsilon}} - P_{V}}_F \leq 2 \left\lbrace \sum\limits_{j=1}^r \frac{\bm{\epsilon}_{\bigcdot j}^T \bm{\epsilon}_{\bigcdot j}}{\mathop{\min}\limits_{k \in[K-r]}\left\lbrace d_j - \Lambda_k\left( P_{V_{\epsilon}}^{\perp}\bm{H}P_{V_{\epsilon}}^{\perp} \right)\right\rbrace^2} \right\rbrace^{1/2}.. 
\end{align*}
\end{corollary}
\begin{proof}
\begin{align*}
    \norm{P_{V_{\epsilon}} - P_{V}}_F = \norm{\left(P_{V_{\epsilon}} - P_{V}P_{V_{\epsilon}}\right) -  P_VP_{V_{\epsilon}}^{\perp}}_F \leq \norm{P_{V_{\epsilon}} - P_{V}P_{V_{\epsilon}}}_F + \norm{P_{V_{\epsilon}}^{\perp} P_V}_F = 2\norm{P_{V_{\epsilon}} - P_{V}P_{V_{\epsilon}}}_F.
\end{align*}
\end{proof}

\begin{lemma}
\label{lemma:supp:StochEqui}
For $\bm{V}\left(\bm{\theta}\right) = \sum_{j=1}^b \bm{\theta}_j \bm{B}_j$ and $\bm{E} \in \mathbb{R}^{p \times n}$ a random matrix, let $\bm{S} = p^{-1}\bm{E}^{\T}\bm{E}$ and
\begin{align*}
    f\left(\bm{\theta}\right) = -n^{-1}\log\left\lbrace \abs{\bm{V}\left(\bm{\theta}\right)} \right\rbrace - n^{-1}\Tr\left[ \bm{S}\left\lbrace \bm{V}\left(\bm{\theta}\right) \right\rbrace^{-1} \right].
\end{align*}
If $\bm{B}_1,\ldots,\bm{B}_b$ and $\bm{E}$ satisfies Assumptions \ref{assumption:CandL} and \ref{assumption:DependenceE} and $p \gtrsim n$, $\abs{f\left\lbrace \left(\bar{v}_1,\ldots,\bar{v}_b\right)^{\T} \right\rbrace - \E\left\lbrace f\left(\bar{\bm{v}}\right) \right\rbrace} = o_P(1)$ and $f\left(\bm{\theta}\right)$ is stochastically equicontinuous on $\Theta_*$ as $n \to \infty$, where $\Theta_*$ is defined in Assumption \ref{assumption:FALCO}.
\end{lemma}

\begin{proof}
This follows from Lemma \ref{lemma:supp:Preliminaries} and the fact that for any any constant $\delta > 0$ and $\tilde{\bm{\theta}},\bm{\theta} \in \Theta_*$ such that $\norm{\tilde{\bm{\theta}} - \bm{\theta}}_2 \leq \delta$, there exists a universal constant $c>0$ not dependent on $\delta,n$ or $p$ such that
\begin{align*}
    \abs{n^{-1}\log\left\lbrace \abs{\bm{V}\left(\bm{\theta}\right)} \right\rbrace - n^{-1}\log\left\lbrace \abs{\bm{V}\left(\tilde{\bm{\theta}}\right)} \right\rbrace},\, \norm{ \left\lbrace \bm{V}\left(\bm{\theta}\right) \right\rbrace^{-1} - \left\lbrace \bm{V}\left(\tilde{\bm{\theta}}\right) \right\rbrace^{-1} }_2 \leq c\delta.
\end{align*}
\end{proof}

\begin{lemma}
\label{lemma:supp:EnpLarge}
Fix some small constant $\epsilon > 0$. In addition to the assumptions of Lemma \ref{lemma:supp:StochEqui}, suppose $n/p \to 0$ as $n,p \to \infty$. For any orthogonal projection matrix $\bm{Q} \in \mathbb{R}^{n \times n}$, define $\bm{M}^{(Q)} \in \mathbb{R}^{b \times b}$ to be $\bm{M}^{(Q)}_{ij} = n^{-1}\Tr\left(\bm{Q}\bm{B}_{i}\bm{Q}\bm{B}_j\right)$ for $i,j \in [b]$. Lastly, define
\begin{align*}
    f^{(Q)}\left(\bm{\theta}\right) &= -n^{-1}\log\left\lbrace \abs{\bm{Q}\bm{V}\left(\bm{\theta}\right)\bm{Q}}_+ \right\rbrace - n^{-1}\Tr\left[ \bm{Q}\bm{S}\bm{Q}\left\lbrace \bm{Q}\bm{V}\left(\bm{\theta}\right)\bm{Q} \right\rbrace^{\dagger} \right]\\
    \mathcal{S} &= \left\lbrace \bm{H} \in \mathbb{R}^{n \times n}, \bm{H}^{\T}=\bm{H},\bm{H}^2=\bm{H}: \bm{M}^{(H)} \succeq \epsilon I_b \right\rbrace,
\end{align*}
where $\bm{V}\left(\bm{\theta}\right)$ is as defined in Lemma \ref{lemma:supp:StochEqui}, and let $\hat{\bm{\theta}}^{(Q)} = \argmax_{\bm{\theta} \in \Theta_*}f^{(Q)}\left(\bm{\theta}\right)$. Then
\begin{align*}
    \sup_{\bm{Q} \in \mathcal{S}}\norm{\hat{\bm{\theta}}^{(Q)} - \bar{\bm{v}}}_2 = O_P\left(n^{1/2}p^{-1/2}\right).
\end{align*}
\end{lemma}

\begin{proof}
Fix a $\bm{Q} \in \mathcal{S}$ and let $\bm{U} \in \mathbb{R}^{n \times m}$ be such that $\bm{U}^{\T}\bm{U}=I_m$ and $\bm{U}\bm{U}^{\T}=\bm{Q}$. Then
\begin{align*}
    f^{(Q)}(\bm{\theta}) =& -n^{-1}\log\left\lbrace \abs{\bm{U}^{\T}\bm{V}\left(\bm{\theta}\right)\bm{U}} \right\rbrace - n^{-1}\Tr\left[ \bm{U}^{\T}\bm{S}\bm{U}\left\lbrace \bm{U}^{\T}\bm{V}\left(\bm{\theta}\right)\bm{U} \right\rbrace^{-1} \right].
\end{align*}
Define
\begin{align*}
    g^{(Q)}(\bm{\theta}) = -n^{-1}\log\left\lbrace \abs{\bm{U}^{\T}\bm{V}\left(\bm{\theta}\right)\bm{U}} \right\rbrace - n^{-1}\Tr\left[ \bm{U}^{\T}\bar{\bm{V}}\bm{U}\left\lbrace \bm{U}^{\T}\bm{V}\left(\bm{\theta}\right)\bm{U} \right\rbrace^{-1} \right].
\end{align*}
Then by Lemma \ref{lemma:supp:RotatedCorrelation},
\begin{align*}
    f^{(Q)}(\bm{\theta}) =& g^{(Q)}(\bm{\theta}) + O_P\left(n^{1/2}p^{-1/2}\right)\\
    \left\lbrace \nabla_{\bm{\theta}} f^{(Q)}(\bm{\theta}) \right\rbrace_j =& \left\lbrace \nabla_{\bm{\theta}} g^{(Q)}(\bm{\theta}) \right\rbrace_j + O_P\left(n^{1/2}p^{-1/2}\right), \quad j \in [b]\\
    \left\lbrace \nabla_{\bm{\theta}}^2 f^{(Q)}(\bm{\theta}) \right\rbrace_{ij} =& \left\lbrace \nabla_{\bm{\theta}}^2 g^{(Q)}(\bm{\theta}) \right\rbrace_{ij} + O_P\left(n^{1/2}p^{-1/2}\right), \quad i,j \in [b]
\end{align*}
uniformly for $\bm{\theta} \in \Theta_*$ and $\bm{Q} \in \mathcal{S}$. Since $\bm{Q} \in \mathcal{S}$, Lemma S4 in \citet{CorrConf} implies $\abs{g^{(Q)}(\bm{\theta}) - g^{(Q)}(\bar{\bm{v}})} \geq \delta_{\epsilon}\norm{\bm{\theta} - \bar{\bm{v}}}_2$ for some constant $\delta_{\epsilon} > 0$ that only depends on $\epsilon$. Therefore,
\begin{align*}
    \sup_{\bm{Q} \in \mathcal{S}}\norm{\hat{\bm{\theta}}^{(Q)} - \bar{\bm{v}}}_2 = o_P(1)
\end{align*}
as $n,p \to \infty$. Since $\nabla_{\bm{\theta}}^2 g^{(Q)}(\bm{\theta}) \mid_{\bm{\theta}=\bar{\bm{v}}} \succeq \tilde{\delta}_{\epsilon}I_b$ for some constant $\tilde{\delta}_{\epsilon} > 0$ that only depends on $\epsilon$, the result follows by a routine Taylor expansion argument.
\end{proof}

\begin{lemma}
\label{lemma:sqrtDeriv}
Suppose $\bm{D}^{(n)}\in \mathbb{R}^{K \times K}$, $n \geq 1$, are diagonal matrices with diagonal elements $d_{kk}^{(n)}$, $k \in [K]$, uniformly bounded from above 0 and below $\infty$. Let $\bm{M}^{(n)} \in \mathbb{R}^{K \times K}$, $n \geq 1$, be symmetric matrices. Then if $\norm{\bm{M}^{(n)}}_2 \to 0$ as $n \to \infty$,
\begin{align*}
    \left[\left\lbrace\bm{D}^{(n)} + \bm{M}^{(n)}\right\rbrace^{1/2}\right]_{ij} - d_{ii}^{(n)}I\left(i=j\right) = \frac{\bm{M}_{ij}^{(n)}}{\left(d_{ii}^{(n)}\right)^{1/2} + \left(d_{jj}^{(n)}\right)^{1/2}}\left\lbrace 1+o(1) \right\rbrace, \quad i,j \in [K]
\end{align*}
as $n \to \infty$.
\end{lemma}
\begin{proof}
We suppress the superscript $(n)$ for notational convenience. Let $\bm{X} = \left(\bm{D} + \bm{M}\right)^{1/2} - \bm{D}^{1/2}$. Then by definition,
\begin{align*}
    \bm{D}^{1/2}\bm{X} + \bm{X}\bm{D}^{1/2} + \bm{X}^2 - \bm{M} = \bm{0}.
\end{align*}

For any symmetric positive semi-definite matrix $\bm{A}$, the function that sends $\bm{A}\rightarrow \bm{A}^{1/2}$ is differentiable. Therefore, for some symmetric matrix $\bm{N} \in \mathbb{R}^{K \times K}$,
\begin{align*}
    \left(\bm{D} + \epsilon \bm{M}\right)^{1/2} - \bm{D}^{1/2} = \epsilon \bm{N} + o\left(\epsilon\right),
\end{align*}
meaning
\begin{align*}
    \bm{M} = \bm{N}\bm{D}^{1/2} + \bm{D}^{1/2}\bm{N} + o(1).
\end{align*}
Since $\bm{N}$ is symmetric, this completes the proof.
\end{proof}

\begin{lemma}
\label{lemma:supp:U}
Let $\bm{D} = \diag\left(d_1,\ldots,d_K\right)$ be such that $d_1 \geq \cdots \geq d_K \geq 0$ and $\bm{A} \succ \bm{0}$ for $\bm{A} \in \mathbb{R}^{K \times K}$, and define $\bm{U} \in \mathbb{R}^{K \times K}$ to be the eigenvectors of $\bm{D}\bm{A}\bm{D}$. Then
\begin{align*}
    \abs{\bm{U}_{rs}} \leq \kappa d_{r \vee s}/d_{r \wedge s}, \quad (r,s) \in \left\lbrace (k_1,k_2) \in [K] \times [K]: d_{k_1 \wedge k_2} > 0 \right\rbrace.
\end{align*}
where the constant $\kappa = \norm{\bm{A}}_2 \norm{\bm{A}^{-1}}_2$ is the condition number of $\bm{A}$..
\end{lemma}

\begin{proof}
Suppose $k \in [K-1]$ is such that $d_k > 0$ but $d_{k+1} = 0$. Then for $\bm{V} = \left(I_k\, \bm{0}_{k \times (K-k)}\right)^{\T}$, $\tilde{\bm{D}} = \bm{V}^{\T}\bm{D}\bm{V} = \diag\left(d_1,\ldots,d_k\right)$ and $\tilde{\bm{A}} = \bm{V}^{\T}\bm{A}\bm{V}$,
\begin{align*}
    \bm{D}\bm{A}\bm{D} = \bm{V}\tilde{\bm{D}}\tilde{\bm{A}}\tilde{\bm{D}}\bm{V}^{\T}.
\end{align*}
If the columns of $\tilde{\bm{U}} \in \mathbb{R}^{k \times k}$ contain the eigenvectors of $\tilde{\bm{D}}\tilde{\bm{A}}\title{\bm{D}}$, then $\bm{U} = \tilde{\bm{U}} \oplus \bm{W}$, where $\bm{W} \in \mathbb{R}^{(K-k)\times (K-k)}$ is an arbitrary unitary matrix. Therefore, it suffices to assume $d_K > 0$ to complete the proof.\par
\indent Suppose the eigenvalues of $\bm{A}$ lie in $\left[c_1,c_2\right]$, $c_1>0$, and let $\eta_k$ be the $k$th eigenvalue of $\bm{D}\bm{A}\bm{D}$. By Lemma \ref{lemma:supp:EigBound}, $\eta_k \in \left[d_k^2 c_1,d_k^2 c_2\right]$. Further,
\begin{align*}
   d_k^2 c_2 \geq \eta_k = \bm{U}_{\bigcdot k}^{\T}\bm{D}\bm{A}\bm{D}\bm{U}_{\bigcdot k} \geq c_1\bm{U}_{rk}^2 d_r^2, \quad r \leq k \in [K],
\end{align*}
meaning
\begin{align*}
    \abs{\bm{U}_{rk}} \leq \left(c_2/c_1\right)\frac{d_k}{d_r}, \quad r \leq k \in [K].
\end{align*}
Next, since we are assuming $\bm{D}$ is invertible,
\begin{align*}
    d_k^{-2} c_1^{-1} \geq \eta_k^{-1} = \bm{U}_{\bigcdot k}^{\T}\bm{D}^{-1}\bm{A}^{-1}\bm{D}^{-1}\bm{U}_{\bigcdot k}\geq c_2^{-1}\bm{U}_{rk}^2 d_r^{-2}, \quad r \geq k \in [K]
\end{align*}
which implies
\begin{align*}
    \abs{\bm{U}_{rk}} \leq \left(c_2/c_1\right) \frac{d_r}{d_k} \quad r \geq k \in [K].
\end{align*}
This completes the proof.
\end{proof}

\begin{lemma}
\label{lemma:supp:OracleSignal}
Let $\tilde{\bm{L}},\tilde{\bm{C}}$ be as defined in \eqref{equation:supp:Params}, $\bm{V} = \sum_{j=1}^b \bar{v}_j \bm{B}_j$ and $\hat{\bm{V}}$ be an estimate for $\bm{V}$. Define
\begin{align*}
    \bar{\bm{L}} = p^{-1/2}\bm{L}\left(\bm{C}^{\T}\bm{V}^{-1}\bm{C}\right)^{1/2}\bar{\bm{U}}, \quad \bar{\bm{C}} = \bm{C}\left(\bm{C}^{\T}\bm{V}^{-1}\bm{C}\right)^{-1/2} \bar{\bm{U}},
\end{align*}
where $\bar{\bm{U}} \in \mathbb{R}^{K \times K}$ is a unitary matrix such that $\bar{\bm{L}}^{\T}\bar{\bm{L}} = \diag\left(\bar{\gamma}_1,\ldots,\bar{\gamma}_K\right)$ for $0 < \bar{\gamma}_K \leq \cdots \leq \bar{\gamma}_1$. Suppose there exists an $s \in [K]$ such that for some large constant $c > 1$ not dependent on $n$ or $p$, $\bar{\gamma}_{s+1} \leq c$ and $\bar{\gamma}_{s}/\bar{\gamma}_{s+1} \geq c^{-1}$, where $\bar{\gamma}_{K+1}=0$. Lastly, let $\bm{A} = I_s \oplus \bm{0}_{(K-s) \times (K-s)} \in \mathbb{R}^{K \times K}$ and suppose
\begin{align*}
    \tilde{\bm{L}}\bm{A}\tilde{\bm{C}}^{\T} = \sum\limits_{k=1}^s \tilde{\mu}_k^{1/2} \tilde{\bm{w}}_k \tilde{\bm{v}}_k^{\T}, \quad \bar{\bm{L}}\bm{A}\bar{\bm{C}}^{\T} = \sum\limits_{k=1}^s \bar{\mu}_k^{1/2} \bar{\bm{w}}_k \bar{\bm{v}}_k^{\T},
\end{align*}
where $\tilde{\bm{W}}=\left(\tilde{\bm{w}}_1 \cdots \tilde{\bm{w}}_s\right), \bar{\bm{W}}=\left(\bar{\bm{w}}_1 \cdots \bar{\bm{w}}_s\right) \in \mathbb{R}^{p \times s}$ and $\tilde{\bm{V}}=\left(\tilde{\bm{v}}_1 \cdots \tilde{\bm{v}}_s\right), \bar{\bm{V}} = \left(\bar{\bm{v}}_1 \cdots \bar{\bm{v}}_s\right) \in \mathbb{R}^{n \times s}$ have orthonormal columns, $\tilde{\mu}_1 \geq \cdots \geq \tilde{\mu}_s > 0$ and $\bar{\mu}_1 \geq \cdots \geq \bar{\mu}_s > 0$. Then if $\norm{\bm{V} - \hat{\bm{V}}}_2 = O_P\left(n^{-1}\right)$ as $n,p \to \infty$ and Assumption \ref{assumption:CandL} holds, the following hold for
\begin{align*}
    &\left(\tilde{\bm{\ell}}_1 \cdots \tilde{\bm{\ell}}_p\right)^{\T} = p^{1/2}n^{-1/2}\tilde{\bm{W}} \diag\left(\tilde{\mu}_1^{1/2},\ldots,\tilde{\mu}_s^{1/2}\right),\\
    &\left(\bar{\bm{\ell}}_1 \cdots \bar{\bm{\ell}}_p\right)^{\T} = p^{1/2}n^{-1/2}\bar{\bm{W}} \diag\left(\bar{\mu}_1^{1/2},\ldots,\bar{\mu}_s^{1/2}\right)
\end{align*}
and some unitary matrix $\bm{G} \in \mathbb{R}^{s \times s}$ as $n,p \to \infty$:
\begin{enumerate}[label=(\roman*)]
    \item $\mathop{\sup}\limits_{g \in [p]} \norm{ \bm{G}^{\T}\tilde{\bm{\ell}}_g - \bar{\bm{\ell}}_g }_2 = O_P\left(n^{-1}\right)$\label{item:supp:Oracle:ell}
    \item $\tilde{\mu}_k = \bar{\mu}_k\left[ 1 + O_P\left\lbrace\left(\lambda_k n\right)^{-1} \right\rbrace \right], \quad k \in [s]$\label{item:supp:Oracle:mu}
    \item $\norm{ P_{\tilde{V}} - P_{\bar{V}} }_F^2 = O_P\left\lbrace \left(\lambda_s n\right)^{-1}\right\rbrace$.\label{item:supp:Oracle:C}
\end{enumerate}
Further, if $t \in [s]$ is such that $\bar{\mu}_{t}/\bar{\mu}_{t-1}, \bar{\mu}_{t}/\bar{\mu}_{t+1} \geq 1+c^{-1}$, then
\begin{enumerate}[label=(\roman*)]
\setcounter{enumi}{3}
    \item $\mathop{\sup}\limits_{g \in [p]} \norm{ r\tilde{\bm{\ell}}_{g_t} - \bar{\bm{\ell}}_{g_t} }_2 = O_P\left(n^{-1}\right)$\label{item:supp:Oracle:ellt}
    \item $\norm{r\tilde{\bm{v}}_t - \bar{\bm{v}}_t}_2 = O_P\left\lbrace \left(\lambda_t n\right)^{-1}\right\rbrace$\label{item:supp:Oracle:Ct}
    \item $\norm{r\tilde{\bm{w}}_t - \bar{\bm{w}}_t}_2 = O_P\left( \lambda_t^{-1/2}n^{-1}\right)$\label{item:supp:Oracle:Lt}
\end{enumerate}
for $r  \in \left\lbrace 1,-1 \right\rbrace$.
\end{lemma}

\begin{proof}
Item \ref{item:supp:Oracle:ell} and \ref{item:supp:Oracle:ellt} are straightforward. Further, all relationships are clearly true when $s=K$. Therefore, it suffices to assume $1 \leq s < K$. Let $\hat{\bm{M}} = \bar{\bm{C}}^{\T}\hat{\bm{V}}^{-1}\bar{\bm{C}}$. Then for some unitary matrix $\bm{H} \in \mathbb{R}^{K \times K}$,
\begin{align*}
    \tilde{\bm{L}} = \bar{\bm{L}}\hat{\bm{M}}^{1/2}\bm{H}, \quad \tilde{\bm{C}} = \bar{\bm{L}}\hat{\bm{M}}^{-1/2}\bm{H}^{\T}.
\end{align*}
Let $\bar{\bm{\Gamma}}_1 = \diag\left(\bar{\gamma}_1,\ldots,\bar{\gamma}_s\right)$,  $\bar{\bm{\Gamma}}_2 = \diag\left(\bar{\gamma}_1,\ldots,\bar{\gamma}_s\right)$ and $\bar{\bm{\Gamma}} = \bar{\bm{\Gamma}}_1 \oplus \bar{\bm{\Gamma}}_2$. By definition,
\begin{align*}
    &\bar{\bm{\Gamma}}^{1/2}\bm{M}^{1/2}\bm{H} = \bm{U}\bm{\Gamma}^{1/2}, \quad \bm{\Gamma} = \tilde{\bm{L}}^{\T}\tilde{\bm{L}}\\
    &\hat{\bm{H}}^{\T}\bm{M}^{-1/2} = \bm{\Gamma}^{-1/2}\bm{U} \bar{\bm{\Gamma}}^{1/2}
\end{align*}
where the columns of $\bm{U} \in \mathbb{R}^{K \times K}$ contain the eigenvectors of $\bar{\bm{\Gamma}}^{1/2} \bm{M}\bar{\bm{\Gamma}}^{1/2}$. I abuse notation and define
\begin{align*}
    \bm{L} = \bar{\bm{L}}\left(\bar{\bm{L}}^{\T}\bar{\bm{L}}\right)^{-1/2} = \left(\bm{L}_1\, \bm{L}_2\right), \quad \bar{\bm{C}} = \left(\bm{C}_1\, \bm{C}_2\right), \quad \bm{U} = \begin{pmatrix} \bm{U}_{11} & \bm{U}_{12}\\
    \bm{U}_{21} & \bm{U}_{22}
    \end{pmatrix}
\end{align*}
where $\bm{L}_1 \in \mathbb{R}^{p \times s}$, $\bm{L}_2 \in \mathbb{R}^{p \times (K-s)}$, $\bm{C}_1 \in \mathbb{R}^{n \times s}$, $\bm{C}_2 \in \mathbb{R}^{n \times (K-s)}$, $\bm{U}_{11} \in \mathbb{R}^{s \times s}$ and $\bm{U}_{21} \in \mathbb{R}^{(K-s) \times s}$. Therefore,
\begin{align*}
    \bar{\bm{L}}\bm{A}\bar{\bm{C}}^{\T} =& \bm{L}_1\bar{\bm{\Gamma}}_1^{1/2} \bm{C}_1^{\T}\\
    \tilde{\bm{L}}\bm{A}\tilde{\bm{C}}^{\T} =& \bar{\bm{L}}\left(\bar{\bm{L}}^{\T}\bar{\bm{L}}\right)^{-1/2}\bar{\bm{\Gamma}}^{1/2}\bm{M}^{1/2}\bm{H} \bm{A}\bm{H}^{\T}\bm{M}^{-1/2}\bar{\bm{C}}^{\T}\\
    =&\bar{\bm{L}}\left(\bar{\bm{L}}^{\T}\bar{\bm{L}}\right)^{-1/2} \bm{U}\bm{A}\bm{U}^{\T} \bar{\bm{\Gamma}}^{1/2} \bar{\bm{C}}^{\T}\\
    =&\left(\bm{L}_1\bm{U}_{11}\bm{U}_{11}^{\T} + \bm{L}_2\bm{U}_{21}\bm{U}_{11}^{\T}\right)\bar{\bm{\Gamma}}_1^{1/2}\bm{C}_1^{\T} + \left(\bm{L}_1\bm{U}_{11}\bm{U}_{21}^{\T} + \bm{L}_2\bm{U}_{21}\bm{U}_{21}^{\T}\right)\bar{\bm{\Gamma}}_2^{1/2}\bm{C}_2^{\T}\\
    =& \bm{L}_1\left( \bm{U}_{11}\bm{U}_{11}^{\T}\bar{\bm{\Gamma}}_1^{1/2}\bm{C}_1^{\T} + \bm{U}_{11}\bm{U}_{21}^{\T}\bar{\bm{\Gamma}}_2^{1/2}\bm{C}_2^{\T} \right) + \bm{L}_2\left( \bm{U}_{21}\bm{U}_{11}^{\T}\bar{\bm{\Gamma}}_1^{1/2}\bm{C}_1^{\T} + \bm{U}_{21}\bm{U}_{21}^{\T}\bar{\bm{\Gamma}}_2^{1/2}\bm{C}_2^{\T} \right)
\end{align*}
and
\begin{align*}
    \tilde{\bm{C}} \bm{A}\tilde{\bm{L}}^{\T}\tilde{\bm{L}}\bm{A}\tilde{\bm{C}}^{\T} =& \bm{C}_1 \bar{\bm{\Gamma}}_1^{1/2} \left\lbrace \left(\bm{U}_{11}\bm{U}_{11}^{\T}\right)^2 + \bm{U}_{11}\bm{U}_{21}^{\T}\bm{U}_{21}\bm{U}_{11}^{\T} \right\rbrace \bar{\bm{\Gamma}}_1^{1/2}\bm{C}_1^{\T}\\
    &+ \bm{C}_2 \bar{\bm{\Gamma}}_2^{1/2} \left\lbrace \bm{U}_{21}\bm{U}_{11}^{\T}\bm{U}_{11}\bm{U}_{21}^{\T} + \left(\bm{U}_{21}\bm{U}_{21}^{\T}\right)^2 \right\rbrace \bar{\bm{\Gamma}}_2^{1/2}\bm{C}_2^{\T}\\
    &+ \bm{C}_1 \bar{\bm{\Gamma}}_1^{1/2}\left( \bm{U}_{11}\bm{U}_{11}^{\T}\bm{U}_{11}\bm{U}_{21}^{\T} + \bm{U}_{11}\bm{U}_{21}^{\T}\bm{U}_{21}\bm{U}_{21}^{\T} \right)\bar{\bm{\Gamma}}_2^{1/2}\bm{C}_2^{\T}\\
    &+ \left\lbrace \bm{C}_1 \bar{\bm{\Gamma}}_1^{1/2}\left( \bm{U}_{11}\bm{U}_{11}^{\T}\bm{U}_{11}\bm{U}_{21}^{\T} + \bm{U}_{11}\bm{U}_{21}^{\T}\bm{U}_{21}\bm{U}_{21}^{\T} \right)\bar{\bm{\Gamma}}_2^{1/2}\bm{C}_2^{\T} \right\rbrace^{\T}.
\end{align*}
We therefore only have to understand how $\bm{U}_{11}$ and $\bm{U}_{21}$ behave. Using the exact same technique as used in the proof of Lemma \ref{lemma:supp:UpperBlock}, it is easy to see that
\begin{align*}
    \bm{U}_{12_{kr}},\bm{U}_{21_{rk}} = O_P\left(n^{-1}\lambda_{k}^{-1/2}\right), \quad r \in [K-s], k \in [s].
\end{align*}
Therefore,
\begin{align*}
    &\left(I_s - \bm{U}_{11}\bm{U}_{11}^{\T}\right)_{rk} = \left(\bm{U}_{12}\bm{U}_{12}^{\T}\right)_{rk} = O_P\left(n^{-2}\lambda_{r}^{-1/2}\lambda_k^{-1/2}\right), \quad r,k \in [s]\\
    &\left(I_s - \bm{U}_{11}^{\T}\bm{U}_{11}\right)_{rk} = \left(\bm{U}_{21}^{\T}\bm{U}_{21}\right)_{rk} = O_P\left(n^{-2}\lambda_{r}^{-1/2}\lambda_k^{-1/2}\right), \quad r,k \in [s],
\end{align*}
meaning
\begin{align*}
    &\bar{\bm{\Gamma}}_1^{1/2}\left(\bm{U}_{11}\bm{U}_{11}^{\T}\right)^2 \bar{\bm{\Gamma}}_1^{1/2} - \bar{\bm{\Gamma}}_1, \, \bar{\bm{\Gamma}}_2^{1/2}\bm{U}_{21}\bm{U}_{11}^{\T}\bm{U}_{11}\bm{U}_{21}^{\T}\bar{\bm{\Gamma}}_2^{1/2},\, \bar{\bm{\Gamma}}_2^{1/2}\left(\bm{U}_{21}\bm{U}_{21}^{\T}\right)^2\bar{\bm{\Gamma}}_2^{1/2},\\
    &\bar{\bm{\Gamma}}_1^{1/2}\bm{U}_{11}\bm{U}_{21}^{\T}\bm{U}_{21}\bm{U}_{21}^{\T}\bar{\bm{\Gamma}}_2^{1/2} = O_P\left(n^{-2}\right).
\end{align*}
Further, by the proof of Lemma \ref{lemma:supp:UpperBlock},
\begin{align*}
    \norm{\bm{U}_{21}\bm{U}_{11}^{\T}\bar{\bm{\Gamma}}_1^{1/2}}_2 = O_P\left(n^{-1}\right).
\end{align*}
Putting this all together,
\begin{align*}
    \tilde{\bm{C}} \bm{A}\tilde{\bm{L}}^{\T}\tilde{\bm{L}}\bm{A}\tilde{\bm{C}}^{\T} = \bar{\bm{C}}\bm{A}\bar{\bm{L}}^{\T}\bar{\bm{L}}\bm{A}\bar{\bm{C}}^{\T} + O_P\left(n^{-1}\right).
\end{align*}
Items \ref{item:supp:Oracle:mu}, \ref{item:supp:Oracle:C} and \ref{item:supp:Oracle:Ct} then follow by applications of Lemma \ref{lemma:supp:EigApprox} and Corollary \ref{corollary:supp:NormSpaceVhat}. To prove \ref{item:supp:Oracle:Lt}, let $\bm{F}_{ij} = \bm{C}_i^{\T}\bm{C}_j$ for $i,j \in [2]$ and define $\bm{Q} = \left(\bm{L}_1 \bm{S}\, \bm{L}_2\right) \in \mathbb{R}^{p \times K}$, where $\bm{S} \in \mathbb{R}^{s \times s}$ is a unitary matrix such that $\bar{\bm{\Gamma}}_1^{1/2}\bm{F}_{11}\bar{\bm{\Gamma}}_1^{1/2} = \bm{S}\bar{\bm{M}}\bm{S}^{\T}$ for $\bar{\bm{M}}=\diag\left(\bar{\mu}_1,\ldots,\bar{\mu}_s\right)$. Then for some unitary matrix $\bm{T} \in \mathbb{R}^{s \times s}$,
\begin{align*}
    \tilde{\bm{B}}=&\bm{Q}^{\T}\tilde{\bm{L}}\bm{A}\tilde{\bm{C}}^{\T}\tilde{\bm{C}}\bm{A}\tilde{\bm{L}}^{\T}\bm{Q} = \begin{pmatrix}
    \bm{A}_1 & \bm{A}_2\\
    \bm{A}_2^{\T} & \bm{0}
    \end{pmatrix} + O_P\left(n^{-1}\right)\\
    \bm{A}_1 =& \bar{\bm{M}} + \bar{\bm{M}}^{1/2}\bm{T}\bm{F}_{11}^{-1/2}\bm{F}_{12}\bm{U}_{21}\bm{U}_{11}^{\T} + \left(\bar{\bm{M}}^{1/2}\bm{T}\bm{F}_{11}^{-1/2}\bm{F}_{12}\bm{U}_{21}\bm{U}_{11}^{\T}\right)^{\T}\\
    \bm{A}_2 =& \bar{\bm{M}} \bm{S}^{\T}\bm{U}_{11}\bm{U}_{21}^{\T}
\end{align*}
where by Lemma \ref{lemma:supp:U}, which gives the structure of $\bm{S}$,
\begin{align*}
    \bm{A}_{1_{rk}} &= \bar{\mu}_r I\left(r=k\right) + O_P\left(n^{-1}\lambda_{k\wedge r}^{1/2}\lambda_{k\vee r}^{-1/2}\right), \quad r,k \in [s]\\
    \bm{A}_{2_{rk}} &= O_P\left(n^{-1}\lambda_r^{1/2}\right), \quad r \in [s]; k \in [K-s].
\end{align*}
By the proof of Lemma \ref{lemma:supp:TrueEigs}, the $k$th eigenvalues of $\bm{A}_1$ is $\bar{\mu}_k\left\lbrace 1+O_P\left(n^{-1}\lambda_k^{-1}\right) \right\rbrace$ for $k \in [s]$. Weyl's Theorem then shows that the $k$th eigenvalue of $\tilde{\bm{B}}$ is $\bar{\mu}_k\left\lbrace 1+o_P(1) \right\rbrace$ for $k \in [s]$. If $\left(\hat{\bm{v}}_{k_1}^{\T}\, \hat{\bm{v}}_{k_2}^{\T}\right) \in \mathbb{R}^{K}$, $\hat{\bm{v}}_{k_1} \in \mathbb{R}^{s}$, is the $k$th eigenvector, then
\begin{align*}
    \hat{\bm{v}}_{k_2} =& \tilde{\mu}_k^{-1}\bm{A}_2^{\T}\hat{\bm{v}}_{k_1}, \quad k \in [s]\\
    \tilde{\mu}_k\hat{\bm{v}}_{k_1} =& \left(\bm{A}_1 + \tilde{\mu}_k^{-1} \bm{A}_2\bm{A}_2^{\T}\right)\hat{\bm{v}}_{k_1}, \quad k \in [s].
\end{align*}
To prove \ref{item:supp:Oracle:ellt}, let $t \in [s]$ be such that $\bar{\mu}_{t}/\bar{\mu}_{t-1}, \bar{\mu}_{t}/\bar{\mu}_{t+1} \geq 1+c^{-1}$ for $c$ defined in the statement of the lemma. We first see that $\norm{\bm{A}_2\bm{A}_2^{\T}}_2 = O_P\left(n^{-1}\right)$. By the proof of Lemma \ref{lemma:supp:TrueEigs},
\begin{align*}
    \hat{\bm{v}}_{t_{1_r}} =& O_P\left(n^{-1}\lambda_t^{-1/2}\lambda_{r \vee t}^{-1/2}\right), \quad r \in [s]\setminus \left\lbrace t \right\rbrace\\
    \norm{\hat{\bm{v}}_{k_2}}_2 =& O_P\left(\lambda_k^{-1}n^{-1/2}\right), \quad k \in [s].
\end{align*}
Since
\begin{align*}
    \hat{\bm{v}}_{k_{1}}^{\T}\hat{\bm{v}}_{r_{1}} = -\hat{\bm{v}}_{k_{2}}^{\T}\hat{\bm{v}}_{r_{2}} = O_P\left\lbrace \left(n\lambda_k\lambda_r\right)^{-1} \right\rbrace, \quad k \neq r \in [s],
\end{align*}
\begin{align*}
    \hat{\bm{v}}_{t_{1_r}} = O_P\left( n^{-1}\lambda_{t}^{-1/2}\lambda_r^{-1/2} \right), \quad r \in [s]\setminus \left\lbrace t \right\rbrace.
\end{align*}
Therefore,
\begin{align*}
    \norm{\hat{\bm{v}}_{t_2}}_2 = \tilde{\mu}_t^{-1}\norm{\bm{A}_2^{\T}\hat{\bm{v}}_{t_{1}}}_2 = \tilde{\mu}_t^{-1}\norm{\bm{U}_{21}\bm{U}_{11}^{\T}\bm{S}\bar{\bm{M}}\hat{\bm{v}}_{t_{1}}}_2 = n^{-1}\tilde{\mu}_t^{-1} O_P\left(\norm{\bar{\bm{M}}^{1/2}\hat{\bm{v}}_{t_{1}}}_2\right) = O_P\left(n^{-1}\lambda_t^{-1/2}\right).
\end{align*}
This shows that
\begin{align*}
    1-\abs{\hat{\bm{v}}_{t_{1_t}}} = O_P\left\lbrace \left(n^{-1}\lambda_t^{-1/2}\right)^2 \right\rbrace
\end{align*}
and completes the proof.
\end{proof}

\begin{lemma}
\label{lemma:supp:PrelimCinference}
Let $\bm{M} \in \mathbb{R}^{n \times n}$ be a non-random symmetric positive definite matrix, $\bm{C} \in \mathbb{R}^{n \times K}$ be a random matrix and $\bm{L} \in \mathbb{R}^{p \times K}$ be a non-random matrix such that $np^{-1}\bm{L}^{\T}\bm{L} = \bm{D} = \diag\left( \lambda_1,\ldots,\lambda_K \right)$, where $\lambda_1 \geq \cdots \geq \lambda_K > \lambda_{K+1}= 0$. Assume the following hold for some fixed constants $s \in [K]$ and $c_1>1$:
\begin{enumerate}[label=(\roman*)]
    \item $\bm{C}$ satisfies $\norm{ n^{-1}\bm{C}^{\T}\bm{\Delta}\bm{C} - \E(n^{-1}\bm{C}^{\T}\bm{\Delta}\bm{C}) }_2 = O_P(n^{-1/2})$ for any symmetric, positive definite $\bm{\Delta} \in \mathbb{R}^{n \times n}$ such that $\norm{\bm{\Delta}}_2 \leq c_1$.\label{item:supp:PrelimCinference:C}
    \item $\E(n^{-1}\bm{C}^{\T}\bm{C}) = I_K$, $\lambda_1,\ldots,\lambda_s \in [c_1^{-1},c_1n]$, $\limsup_{n,p \to \infty}\lambda_{s+1} \leq c_2$ and $\norm{\bm{M}}_2,\norm{\bm{M}^{-1}}_2 \leq c_1$.
\end{enumerate}
Let $\bm{W} \in \mathbb{R}^{K \times K}$ be a non-random unitary matrix such that
\begin{align*}
    \bm{W}^{\T}\bm{D}^{1/2}\E\left(n^{-1}\bm{C}^{\T}\bm{M}\bm{C}\right)\bm{D}^{1/2}\bm{W} = \bm{\Gamma}=\diag\left(\gamma_1,\ldots,\gamma_K\right),
\end{align*}
where $\gamma_1 > \cdots > \gamma_K > \gamma_{K+1}=0$. For $s \in [K]$ defined above, define $\bm{W}^{(s)} = \left(\bm{W}_{\bigcdot 1} \cdots \bm{W}_{\bigcdot s}\right)$, $d_r^{(s)} = \Lambda_r\left[ \bm{D}^{1/2}\bm{W}^{(s)}\{\bm{W}^{(s)}\}^{\T}\bm{D}^{1/2} \right]$ and $\bm{C}^{(s)} \in \mathbb{R}^{n \times s}$ such that
\begin{align*}
    \left\lbrace \bm{C}^{(s)},\bm{L}^{(s)} \right\rbrace = \argmin_{(\bar{\bm{C}}, \bar{\bm{L}}) \in \mathcal{S}_{s}} \norm{ ( \bm{L}\bm{C}^{\T} - \bar{\bm{L}}\bar{\bm{C}}^{\T} )\bm{M}^{1/2} }_F^2,
\end{align*}
where $\mathcal{S}_{s}$ is as defined in Section \ref{subsection:OracleFactors}. Assume the following also hold for some constant $c_2>1$:
\begin{enumerate}[label=(\roman*)]
    \setcounter{enumi}{2}
    \item $\gamma_{k}/\gamma_{k+1} \geq 1+c_2^{-1}$ for all $k \in [K]$.
    \item $d_r^{(s)}/d_{r+1}^{(s)} \geq 1+c_2^{-1}$ for all $r \in [s]$, where $d_{s+1}^{(s)} = 0$.
\end{enumerate}
Then the following hold for some constant $\tilde{c} > 1$:
\begin{subequations}
\label{equation:supp:OracleLambaRandomC}
\begin{align}
\label{equation:supp:OracleLambaRandomC:d}
    & \text{$d_r^{(s)}/\lambda_r \in [\tilde{c}^{-1},\tilde{c}]$ and $\lambda_r - \tilde{c}\lambda_{s+1}\leq d_r^{(s)} \leq \lambda_r + \tilde{c}\lambda_{s+1}$}, \quad r \in [s]\\
\label{equation:supp:OracleLambaRandomC:hat}
    &\Lambda_r\left[ p^{-1}\bm{L}^{(s)}\{\bm{C}^{(s)}\}^{\T}\bm{C}^{(s)}\{\bm{L}^{(s)}\}^{\T} \right] = d_r^{(s)}\left\lbrace 1+O_P(n^{-1/2}) \right\rbrace, \quad r\in[s].
\end{align}
\end{subequations}
Lastly, let the non-random matrix $\bm{U} \in \mathbb{R}^{K \times s}$ be such that $\bm{U}^{\T}\bm{U}=I_s$ and whose columns are the first $s$ eigenvectors of $\bm{D}^{1/2}\bm{W}^{(s)}\{\bm{W}^{(s)}\}^{T}\bm{D}^{1/2}$. Then for all $r \in [s]$,
\begin{subequations}
\label{equation:supp:OracleCRandomC}
\begin{align}
\label{equation:supp:OracleCRandomC:U}
    &\bm{U}_{tr} = \begin{cases}
    O\left( \lambda_{s+1}/\lambda_t \right) & \text{if $t < r$}\\
    O\left\{ \lambda_{(s+1)\vee t}/\lambda_r \right\} & \text{if $t > r$}
    \end{cases}, \quad t \in [K]\setminus \{r\}\\
\label{equation:supp:OracleCRandomC:Uhat}
    &\bm{C}^{(s)}_{\bigcdot r} = \bm{C}\left( \bm{U}\hat{\bm{u}}_r + \bm{\Delta}_r \right), \quad \norm{ \bm{\Delta}_r }_2 = O_P\left( n^{-1/2}\lambda_{s+1}\lambda_r^{-1} \right), \quad \hat{\bm{u}}_{r_t} = \begin{cases}
    1+O_P\left(n^{-1/2}\right) & \text{if $t=r$}\\
    O_P\left(n^{-1/2}\right) & \text{if $t < r$}\\
    O_P\left(n^{-1/2}\lambda_t\lambda_r^{-1}\right) & \text{if $t > r$ }
    \end{cases}, \quad t \in [s]
\end{align}
\end{subequations}
\end{lemma}

\begin{remark}
\label{remark:supp:CrandomCond}
Condition \ref{item:supp:PrelimCinference:C} is quite general, and is satisfied in the following scenarios: $\bm{C} = \E(\bm{C}) + \sum\limits_{i=1}^m \bm{\Xi}_i$ for $m \leq c_1$, where $\norm{\E\left(\bm{C}\right)}_2\leq c_1n^{1/2}$ and $\bm{\Xi}_i \in \mathbb{R}^{n \times K}$ are mean $\bm{0}$, $\bm{\Xi}_i,\bm{\Xi}_j$ are independent for $i \neq j$ and satisfy one of the following:
\begin{enumerate}
    \item $\bm{\Xi}_i = \bm{A}_i\bm{R}_i$, where $\bm{A}_i\in\mathbb{R}^{nK \times nK}$ is a non-random matrix such that $\norm{\bm{A}_i}_2,\norm{\bm{A}_i^{-1}}_2 \leq c_1$ and $\bm{R}\in\mathbb{R}^{nK}$ is a mean 0 random matrix with independent entries such that $\E(\bm{R}_{i_j}^4)\leq c_1$ for all $j\in[nK]$.
    \item $\E[ \exp\{ \vecM(\bm{\Xi}_i)^{\T}\bm{t} \} ] \leq \exp( c_1\norm{\bm{t}}_2^2 )$.
\end{enumerate}
This follows from Corollary 2.6 in \citet{SubGaussVariance} and standard properties of sub-Gaussian random vectors.
\end{remark}

\begin{remark}
\label{remark:supp:Crandom1}
This lemma enumerates the properties of $\lamoracle_1,\ldots,\lamoracle_{\Koracle}$ and $\Coracle$ when $\bm{C}$ is a random matrix. Using the notation that appears in the main text, $\bm{M}$ corresponds to $\bar{\bm{V}}$, $\gamma_1,\ldots,\gamma_K$ are the same as those defined in Assumption \ref{assumption:CandL}, the index $s$ corresponds to the index $s$ defined in the statement of Theorem \ref{theorem:Khat}, $d_r^{(s)}$ corresponds to $\lamoracle_r$ and $\bm{C}^{(s)}$ corresponds to $\Coracle$. We treat $s$ as non-random because Theorem \ref{theorem:Khat} shows that $\Koraclehat=\Koracle=s$ with probability tending to 1 as $n,p\to \infty$. 
\end{remark}

\begin{remark}
\label{remark:supp:Crandom2}
This lemma is used to prove Theorem \ref{theorem:XandC}.
\end{remark}

\begin{proof}
By Lemma \ref{lemma:supp:U} and because the eigenvalues of $\E\left(n^{-1}\bm{C}^{\T}\bm{M}\bm{C}\right)$ are uniformly bounded above 0 and below $\infty$, $\bm{W}_{kt} = O\left(\lambda_{k \vee t}^{1/2}\lambda_{k \wedge t}^{-1/2}\right)$ for all $k,t\in[K]$. Next, note that $d_r^{(s)} = \Lambda_r\left[ \{\bm{W}^{(s)}\}^{\T}\bm{D}\bm{W}^{(s)} \right]$, where $\{\bm{W}^{(s)}_{\bigcdot k}\}^{\T}\bm{D}\bm{W}^{(s)}_{\bigcdot t} = O\left(\lambda_k^{1/2}\lambda_t^{1/2}\right)$ for all $k,t\in[s]$. \eqref{equation:supp:OracleLambaRandomC:d} then follows by Lemma \ref{lemma:supp:EigBound}.\par 
\indent Next, it is clear \eqref{equation:supp:OracleCRandomC:U} holds for $s=K$. If $s < K$,
\begin{align*}
    \norm{ \bm{D}-\bm{D}^{1/2}\bm{W}^{(s)}\{\bm{W}^{(s)}\}^{T}\bm{D}^{1/2} }_2 = \norm{ \bm{D}^{1/2}\sum\limits_{k=s+1}^K \bm{W}_{\bigcdot k}\bm{W}_{\bigcdot k}^{\T} \bm{D}^{1/2} }_2 = O\left(\lambda_{s+1}\right).
\end{align*}
By Weyl's Theorem and Lemma \ref{lemma:supp:EigVector}, this shows that $\norm{ \bm{U}_{\bigcdot r} - \bm{a}_r }_2 = O_P(\lambda_{s+1}/\lambda_r)$, where $\bm{a}_r \in \mathbb{R}^K$ is the $r$th standard basis vector (we have assumed $\bm{U}_{r_r} \geq 0$ without loss of generality). An identical analysis also shows that $\bm{U}_{rt} = O_P(\lambda_{s+1}/\lambda_t)$ for $t < r$. Since $\bm{U}_{\bigcdot t}^{\T}\bm{U}_{\bigcdot r}=0$ for $t < r$, this implies $\bm{U}_{tr} = O_P(\lambda_{s+1}/\lambda_t)$ for $t < r$. Next, since $\{\bm{W}^{(s)}_{\bigcdot k}\}^{\T}\bm{D}\bm{W}^{(s)}_{\bigcdot t} = O\left(\lambda_k^{1/2}\lambda_t^{1/2}\right)$ for all $k,t\in[s]$, this implies
\begin{align*}
    \bm{U} = \bm{D}^{1/2}\bm{W}^{(s)}\bm{V}\diag\left\lbrace d_1^{(s)},\ldots, d_s^{(s)} \right\rbrace^{1/2},
\end{align*}
where $\bm{V}\in\mathbb{R}^{s \times s}$ is a unitary matrix that satisfies $\bm{V}_{kt}=O_P\left(\lambda_{k \vee t}^{1/2}\lambda_{k \wedge t}^{-1/2}\right)$. Therefore, for $t > s$,
\begin{align*}
    \bm{U}_{tr} = \left\lbrace\frac{\lambda_t}{d_r^{(s)}}\right\rbrace^{1/2}\underbrace{\{\bm{W}^{(s)}_{t \bigcdot}\}^{\T} \bm{V}_{\bigcdot r}}_{=O\{\lambda_t^{1/2}\lambda_r^{-1/2}\}} = O\left(\lambda_t/\lambda_r\right),
\end{align*}
which proves \eqref{equation:supp:OracleCRandomC:U}.\par
\indent To prove the rest, we first observe that
\begin{align*}
    \norm{ n^{-1}\bm{C}^{\T}\bm{C} - I_K }_2, \norm{ n^{-1}\bm{C}^{\T}\bm{M}\bm{C} - \E\left(n^{-1}\bm{C}^{\T}\bm{M}\bm{C}\right) }_2 = O_P(n^{-1/2})
\end{align*}
by the assumptions on $\bm{C}$. Let $\hat{\bm{W}} \in \mathbb{R}^{K \times K}$ be a unitary matrix such that
\begin{align*}
    \hat{\bm{W}}^{\T}\bm{W}^{\T}\bm{D}^{1/2}\left(n^{-1}\bm{C}^{\T}\bm{M}\bm{C}\right)\bm{D}^{1/2}\bm{W}\hat{\bm{W}} = \hat{\bm{\Gamma}}=\diag\left(\hat{\gamma}_1,\ldots,\hat{\gamma_K}\right),
\end{align*}
where $\hat{\gamma}_1 \geq \cdots \geq \hat{\gamma}_K > 0$ and define $\hat{\bm{W}}^{(s)} = \left(\hat{\bm{W}}_{\bigcdot 1} \cdots \hat{\bm{W}}_{\bigcdot s}\right)$. We then get that
\begin{align*}
    p^{-1}\bm{C}^{(s)}\{\bm{L}^{(s)}\}^{\T}\bm{L}^{(s)}\{\bm{C}^{(s)}\}^{\T} = \tilde{\bm{C}} \bm{D}^{1/2}\bm{W}\hat{\bm{W}}^{(s)}\{ \hat{\bm{W}}^{(s)} \}^{\T}\bm{W}^{\T}\bm{D}^{1/2}\tilde{\bm{C}}^{\T},
\end{align*}
where $\tilde{\bm{C}} = n^{-1/2}\bm{C}$. It is therefore clear that \eqref{equation:supp:OracleLambaRandomC:hat} and \eqref{equation:supp:OracleCRandomC:Uhat} hold when $s=K$. When $s<K$, we only need to understand the behavior of $\hat{\bm{W}}^{(s)}$ to complete the proof. First, the fact that $\bm{W}_{kt} = O\left(\lambda_{k \vee t}^{1/2} \lambda_{k \wedge t}^{-1/2}\right)$ implies
\begin{align*}
\bm{G}_{kt} = \bm{W}^{\T}_{\bigcdot k}\bm{D}^{1/2}\left(n^{-1}\bm{C}^{\T}\bm{M}\bm{C}\right)\bm{D}^{1/2}\bm{W}_{\bigcdot t} = \gamma_k I(k=t) + n^{-1/2}\lambda_k^{1/2}\lambda_t^{1/2}\hat{\bm{H}}_{kt}, \quad k,t \in [K],
\end{align*}
where $\norm{\hat{\bm{H}}}_2 = O_P\left(1\right)$. By Lemma \ref{lemma:supp:EigBound},
\begin{align*}
    \hat{\gamma}_t =& \Lambda_t(\bm{G}) = \gamma_t\left\lbrace 1+O_P(n^{-1/2}) \right\rbrace, \quad t \in [K].
\end{align*}
Further, an application of the co-factor expansion argument developed in Section 7 of \citet{LowDimPCA} (which was extended in Appendix A of \citet{FanEigen} to allow $\lambda_1/\lambda_K \to \infty$) shows that $\hat{\bm{W}}_{kt} = O_P\left(n^{-1/2}\lambda_{k \vee t}^{1/2}\lambda_{k \wedge t}^{-1/2}\right)$ for $k \neq t \in [K]$. Therefore, if $\bm{W}=(\bm{W}^{(s)}\, \bm{W}_2)$ and $\bm{D}=\bm{D}_1\oplus \bm{D}_2$ for $\bm{D}_1\in\mathbb{R}^{s \times s}$,
\begin{align*}
    \hat{\bm{J}}=&\bm{D}^{1/2}\bm{W}\hat{\bm{W}}^{(s)}\{ \hat{\bm{W}}^{(s)} \}^{\T}\bm{W}^{\T}\bm{D}^{1/2} = \bm{D}^{1/2}\bm{W} \left(I - \sum\limits_{j=s+1}^K \hat{\bm{W}}_{\bigcdot j}\hat{\bm{W}}_{\bigcdot j}^{\T}\right)\bm{W}^{T} \bm{D}^{1/2} = \bm{D}^{1/2}\bm{W}^{(s)} \{\bm{W}^{(s)}\}^{\T}\bm{D}^{1/2}\\
    &+ \frac{\lambda_{s+1}}{n}\bm{D}^{1/2}\bm{W}^{(s)} \bm{D}_1^{-1/2}\hat{\bm{A}}_1\bm{D}_1^{-1/2} \{\bm{W}^{(s)}\}^{\T}\bm{D}^{1/2}\\
    &+ n^{-1/2}\bm{D}^{1/2}\bm{W}_2 \bm{D}_2^{1/2}\hat{\bm{A}}_2\bm{D}_1^{-1/2}\{\bm{W}^{(s)}\}^{\T}\bm{D}^{1/2}\\
    &+ \left[n^{-1/2}\bm{D}^{1/2}\bm{W}_2 \bm{D}_2^{1/2}\hat{\bm{A}}_2\bm{D}_1^{-1/2}\{\bm{W}^{(s)}\}^{\T}\bm{D}^{1/2}\right]^{\T} + \frac{\lambda_{s+1}}{\lambda_s n}\bm{D}^{1/2}\bm{W}_2\hat{\bm{A}}_3 \bm{W}_2^{\T} \bm{D}^{1/2},
\end{align*}
where $\norm{\hat{\bm{A}}_j}_2 = O_P(1)$ for $j=1,2,3$. We see that
\begin{align}
    \left\lbrace \bm{D}^{1/2}\bm{W}^{(s)} \bm{D}_1^{-1/2} \right\rbrace_{tk} &= O\left\lbrace \min(1,\lambda_t/\lambda_k) \right\rbrace, \quad t \in [K]; k \in [s]\\
    \norm{ \bm{D}^{1/2}\bm{W}_2 \bm{D}_2^{1/2} }_2 &= O(\lambda_{s+1}).
\end{align}
Therefore,
\begin{align*}
    p^{-1}\bm{C}^{(s)}\{\bm{L}^{(s)}\}^{\T}\bm{L}^{(s)}\{\bm{C}^{(s)}\}^{\T} =& \tilde{\bm{C}}\left\lbrace\bm{D}^{1/2}\bm{W}^{(s)} \{\bm{W}^{(s)}\}^{\T}\bm{D}^{1/2} + O_P\left(\lambda_{s+1}n^{-1/2}\right)\right\rbrace \tilde{\bm{C}}^{\T}.
\end{align*}
Note that $\hat{\bm{J}}$ is rank $s$. Define $\tilde{\bm{D}} = \diag\left\lbrace d_1^{(s)},\ldots,d_s^{(s)} \right\rbrace$ and let $\bm{U}_2 \in \mathbb{R}^{K \times (K-s)}$ have orthonormal columns such that $\bm{U}^{\T}\bm{U}_2 = \bm{0}$. Then
\begin{align*}
    \left(\bm{U}\, \bm{U}_2\right)^{\T}\hat{\bm{J}}\left(\bm{U}\, \bm{U}_2\right) = \tilde{\bm{D}}\oplus \bm{0} + O_P\left(\lambda_{s+1}n^{-1/2}\right)
\end{align*}
and $\Lambda_t(\hat{\bm{J}}) = d_t^{(s)} + O_P\left(\lambda_{s+1}n^{-1/2}\right)$ for all $t \in [s]$ by Weyl's Theorem. Since $\tilde{\bm{C}}^{\T}\tilde{\bm{C}}=I_K+O_P(n^{-1/2})$, \eqref{equation:supp:OracleLambaRandomC:d} follows from Lemma \ref{lemma:supp:EigBound}. Let $\hat{\bm{U}} \in \mathbb{R}^{K \times s}$ be the first $s$ eigenvectors of $\hat{\bm{J}}$, which are the eigenvectors corresponding to all non-zero eigenvalues. By Corollary \ref{corollary:supp:NormSpaceVhat} (and ignoring sign parity without loss of generality),
\begin{align*}
    \hat{\bm{U}} &= \bm{U}\hat{\bm{\Delta}}_1 + \bm{U}_2\hat{\bm{\Delta}}_2\\
    \hat{\bm{\Delta}}_{1_{kt}} &= \begin{cases}
    O_P\left(n^{-1/2}\lambda_{s+1}\lambda_{k \wedge t}^{-1}\right) & \text{if $t \neq k$}\\
    1-O_P\left( \lambda_{s+1}^2n^{-1}\lambda_k^{-2} \right) & \text{if $k = t$}
    \end{cases}, \quad t,k \in [s]\\
    \norm{ \hat{\bm{\Delta}}_{2_{\bigcdot t}} }_2 &= O_P\left(n^{-1/2}\lambda_{s+1}\lambda_{t}^{-1}\right), \quad t \in [s].
\end{align*}
Therefore,
\begin{align*}
    n^{-1/2}\bm{C}^{(s)} = \tilde{\bm{C}}\hat{\bm{U}}\hat{\bm{D}}^{1/2}\bm{V}\hat{\bm{\Sigma}}^{-1/2},
\end{align*}
where
\begin{align*}
    \hat{\bm{D}} &= \diag\left\lbrace \Lambda_1(\hat{\bm{J}}),\ldots,\Lambda_s(\hat{\bm{J}}) \right\rbrace\\
    \hat{\bm{\Sigma}} &= \diag\left( \Lambda_1\left[ p^{-1}\bm{L}^{(s)}\{\bm{C}^{(s)}\}^{\T}\bm{C}^{(s)}\{\bm{L}^{(s)}\}^{\T} \right],\ldots, \Lambda_s\left[ p^{-1}\bm{L}^{(s)}\{\bm{C}^{(s)}\}^{\T}\bm{C}^{(s)}\{\bm{L}^{(s)}\}^{\T} \right] \right)
\end{align*}
and $\bm{V} \in \mathbb{R}^{s \times s}$ is a unitary matrix containing the eigenvectors of $\hat{\bm{D}}^{1/2}\hat{\bm{U}}^{\T}\tilde{\bm{C}}^{\T}\tilde{\bm{C}}\hat{\bm{U}}\hat{\bm{D}}^{1/2}$. Since $\hat{\bm{U}}^{\T}\tilde{\bm{C}}^{\T}\tilde{\bm{C}}\hat{\bm{U}} = I_s + O_P(n^{-1/2})$, $\bm{V}_{kt} = O_P\left(n^{-1/2}\lambda_{k\vee t}^{1/2}\lambda_{k\wedge t}^{-1/2}\right)$ for all $k\neq t \in [s]$ by the co-factor expansion argument from \citet{LowDimPCA}. Therefore (ignoring sign parity without loss of generality),
\begin{align*}
    \left( \hat{\bm{D}}^{1/2}\bm{V}\hat{\bm{\Sigma}}^{-1/2} \right)_{kt} = \begin{cases}
    1+O_P(n^{-1/2}) & \text{if $k=t$}\\
    O_P\left(n^{-1/2}\right) & \text{if $t > k$}\\
    O_P\left(n^{-1/2}\lambda_k/\lambda_t\right) & \text{if $t < k$}
    \end{cases}, \quad t,k \in [s],
\end{align*}
meaning 
\begin{align*}
    n^{-1/2}\bm{C}^{(s)}_{\bigcdot r} = \tilde{\bm{C}}\bm{U}\left(\hat{\bm{\Delta}}_1\hat{\bm{D}}^{1/2}\bm{V}\hat{\bm{\Sigma}}^{-1/2}\right)_{\bigcdot r} + \tilde{\bm{C}} \underbrace{O_P\left(n^{-1/2}\lambda_{s+1}\lambda_r^{-1}\right)}_{K \times 1}
\end{align*}
where 
\begin{align*}
    \left(\hat{\bm{\Delta}}_1\hat{\bm{D}}^{1/2}\bm{V}\hat{\bm{\Sigma}}^{-1/2}\right)_{t r} = \begin{cases}
    1+O_P(n^{-1/2}) & \text{if $t=r$}\\
    O_P(n^{-1/2}) & \text{if $t < r$}\\
    O_P\left(n^{-1/2}\lambda_t\lambda_r^{-1}\right) & \text{if $t > r$}
    \end{cases}, \quad t \in [s].
\end{align*}
This completes the proof.
\end{proof}

\begin{lemma}
\label{lemma:GaussLeverageScores}
Let $\bm{B}\in\mathbb{R}^{n \times m}$, $m \leq n$, be any matrix with orthonormal columns, and suppose $\bm{Q} \in \mathbb{R}^{n \times n}$ is sampled uniformly from the set of all unitary matrices in $\mathbb{R}^{n \times n}$. Then if $h_{i}$, $i \in [n]$, is the $i$th leverage score of $\bm{Q}\bm{B}$,
\begin{align*}
    \Prob\left(\max_{i \in [n]}\abs{ h_i - \frac{m}{n} } \geq \epsilon\right) \leq c_1n\exp\left(-c_2 \epsilon n^{1/2}\right)
\end{align*}
for all $\epsilon > 0$, where $c_1,c_2 > 0$ are constants that do not depend on $n,m$ or $\bm{B}$.
\end{lemma}

\begin{proof}
The $h_i$ has mean $m/n$ and is identically distributed to $W/(W+Z)$, where $W \sim n^{-1}\chi^2_m$, $Z \sim n^{-1}\chi^2_{n-m}$ and $W$ and $Z$ are independent. For any fixed $\delta \in (0,1)$, we see that for some constant $c > 0$ that does not depend on $n, m$ or $\bm{B}$,
\begin{align*}
    \E\left\lbrace\left( \frac{W}{Z+W} \right)^p\right\rbrace &\leq \E\left(\tilde{W}^p\right) + \exp\left(-cn \delta^2\right)
\end{align*}
where $\tilde{W} = \left(1-\delta\right)^{-1}W$, $c_{\delta} > 0$ is a constant that only depends on $\delta$. We next see that because $n\tilde{W} \sim (1-\delta)^{-1}\chi_m^2$, there exists constants $a_{\delta},c_{\delta} > 0$ that only depend on $\delta$ such that
\begin{align*}
    \E\{(nh_i)^p\} \leq (a_{\delta}m^{1/2}p)^p + n^p \exp(-c_{\delta}n) = (a_{\delta}m^{1/2}p)^p + \left\{ p\frac{n}{p}\exp(-c_{\delta} \frac{n}{p}) \right\}^p \leq (a_{\delta}m^{1/2}p)^p + \{(ec_{\delta})^{-1}p\}^p.
\end{align*}
This shows that $nh_i$ has sub-exponential norm $\leq c m^{1/2}$, where $c > 0$ is a constant that does not depend on $n, m$ or $\bm{B}$. Using a standard sub-exponential inequality argument, we get that for some constants $\tilde{c},\bar{c} > 0$ that do not depend on $n,m$ or $\bm{B}$,
\begin{align*}
    \Prob\left( \max_{i \in [n]} \abs{h_i - \frac{m}{n}} \geq \epsilon\right) \leq n \exp\left( \bar{c}\frac{\lambda^2}{n^2}m - \lambda\epsilon \right), \quad 0 < \lambda \leq \tilde{c}nm^{-1/2}
\end{align*}
for all $\epsilon > 0$. The result then follows by setting $\lambda = \tilde{c}nm^{-1/2}$.
\end{proof}

\end{document}